\documentclass[11pt]{article}

\usepackage{amsmath,amssymb,amsthm,amsfonts,latexsym,bbm,xspace,graphicx,float,mathtools,braket,}
\usepackage[letterpaper,margin=1in]{geometry}
\usepackage[backref,colorlinks,citecolor=blue,,linkcolor=magenta,bookmarks=true]{hyperref}
\usepackage{enumitem}
\usepackage[nameinlink,capitalize]{cleveref}

\usepackage[colorinlistoftodos]{todonotes}

\newtheorem{theorem}{Theorem}[section]
\newtheorem*{namedtheorem}{\theoremname}
\newcommand{\theoremname}{testing}

\newtheorem{lemma}[theorem]{Lemma}
\newtheorem{claim}[theorem]{Claim}
\newtheorem{proposition}[theorem]{Proposition}

\newtheorem{corollary}[theorem]{Corollary}

\theoremstyle{definition}
\newtheorem{definition}[theorem]{Definition}

\newtheorem{remark}[theorem]{Remark}
\newtheorem{notation}[theorem]{Notation}
\newtheorem{example}[theorem]{Example}

\renewcommand{\Pr}{\mathop{\bf Pr\/}}
\newcommand{\E}{\mathop{\bf E\/}}

\newcommand{\tr}{\mathrm{tr}} \newcommand{\Tr}{\tr} \newcommand{\trace}{\tr}
\newcommand{\poly}{\mathrm{poly}}

\newcommand{\R}{\mathbb R}
\newcommand{\C}{\mathbb C}

\newcommand{\F}{\mathbb F}

\newcommand{\NP}{\mathsf{NP}}

\newcommand{\NEXP}{\mathsf{NEXP}} 
\newcommand{\NEEXP}{\mathsf{NEEXP}} 

\newcommand{\MIP}{\mathsf{MIP}} 

\newcommand{\QMA}{\mathsf{QMA}}

\newcommand{\eps}{\epsilon}

\newcommand{\ot}{\otimes}

\newcommand{\calA}{\mathcal{A}}
\newcommand{\calB}{\mathcal{B}}

\newcommand{\calD}{\mathcal{D}}

\newcommand{\calH}{\mathcal{H}}

\newcommand{\calO}{\mathcal{O}}
\newcommand{\calP}{\mathcal{P}}

\newcommand{\calX}{\mathcal{X}}

\newcommand{\bone}{\boldsymbol{1}}
\newcommand{\balpha}{\boldsymbol{\alpha}}

\newcommand{\btau}{\boldsymbol{\tau}}

\newcommand{\ba}{\boldsymbol{a}}
\newcommand{\bb}{\boldsymbol{b}}

\newcommand{\boldf}{\boldsymbol{f}}
\newcommand{\bg}{\boldsymbol{g}}
\newcommand{\bh}{\boldsymbol{h}}
\newcommand{\bi}{\boldsymbol{i}}

\newcommand{\bell}{\boldsymbol{\ell}}

\newcommand{\br}{\boldsymbol{r}}
\newcommand{\bs}{\boldsymbol{s}}

\newcommand{\bu}{\boldsymbol{u}}
\newcommand{\bv}{\boldsymbol{v}}

\newcommand{\bx}{{\boldsymbol{x}}}
\newcommand{\by}{\boldsymbol{y}}
\newcommand{\bz}{\boldsymbol{z}}

\newcommand{\bY}{\boldsymbol{Y}}

\newcommand{\ignore}[1]{}

\newcommand{\multisub}[2]{\mathrm{MultiSub}(#1,#2)}

\newcommand{\polyfunc}[3]{\calP(#1, #2, #3)}
\newcommand{\polymeas}[3]{\mathrm{PolyMeas}(#1, #2, #3)}
\newcommand{\polysub}[3]{\mathrm{PolySub}(#1, #2, #3)}

\DeclareMathOperator{\interp}{\mathbf{interp}}

\newcommand{\anote}[1]{}
\newcommand{\jnote}[1]{}
\newcommand{\hnote}[1]{}

\newcommand{\ainnote}[1]{}
\newcommand{\jinnote}[1]{}
\newcommand{\hinnote}[1]{}
\newcommand{\tnote}[1]{}
\newcommand{\znote}[1]{}

\newcounter{termcounter}[equation]
\renewcommand{\thetermcounter}{\the\numexpr\value{equation}+1\relax.\roman{termcounter}}
\crefname{term}{term}{terms}
\creflabelformat{term}{#2\textup{(#1)}#3}

\makeatletter
\def\term{\@ifnextchar[\term@optarg\term@noarg}
\def\term@optarg[#1]#2{%
  \textup{#1}%
  \def\@currentlabel{#1}%
  \def\cref@currentlabel{[][2147483647][]#1}%
  \cref@label[term]{#2}}
\def\term@noarg#1{%
  \refstepcounter{termcounter}%
  \textup{\thetermcounter}%
  \cref@label[term]{#1}}
\makeatother

\newcommand{\simeqbot}{\overset{\bot}{\simeq}}

\title{Quantum soundness of the classical low individual degree test}

\author{
Zhengfeng Ji\thanks{zhengfeng.ji@uts.edu.au}\\
\small{\sl University of Technology Sydney}
\and Anand Natarajan\thanks{anandn@mit.edu.  Most of this work performed while affiliated with
the California Institute of Technology.}\\
 \small{\sl Massachusetts Institute of Technology}
 \and Thomas Vidick\thanks{vidick@caltech.edu}\\
 \small{\sl California Institute of Technology}
 \and John Wright\thanks{wright@cs.utexas.edu. Most of this work performed while also affiliated with
the California Institute of Technology.}\\
 \small{\sl University of Texas at Austin}
 \and Henry Yuen\thanks{hyen@cs.toronto.edu}\\
 \small{\sl University of Toronto}\vspace*{10pt}
}

\date{}

\begin{document}

\maketitle

\begin{abstract}
Low degree tests play an important role in classical complexity theory,
serving as basic ingredients in foundational results such as $\MIP = \NEXP$~\cite{BFL91}
and the PCP theorem~\cite{AS98,ALM+98}.
Over the last ten years,
versions of these tests which are sound against quantum provers
have found increasing applications to the study of nonlocal games and the complexity class~$\MIP^*$.
The culmination of this line of work is the result $\MIP^* = \mathsf{RE}$~\cite{JNV+20}.

One of the key ingredients in the first reported proof of  $\MIP^* = \mathsf{RE}$ is a two-prover variant of the low degree test, initially  
shown to be sound against multiple quantum provers in~\cite{Vid16}.
Unfortunately a mistake was recently discovered in the latter result, invalidating the main result of~\cite{Vid16} as well as its use in subsequent works, including~\cite{JNV+20}.

We analyze a variant of the low degree test called the low individual degree test.
Our main result is that the two-player version of this test is sound against quantum provers. This soundness result is sufficient to re-derive several bounds on~$\MIP^*$ that relied on~\cite{Vid16}, including $\MIP^* = \mathsf{RE}$.
\end{abstract}

\newpage

\tableofcontents
\vfill
\thispagestyle{empty}
\newpage

\section{Introduction}\label{sec:intro}

An $m$-variate polynomial over the finite field~$\F_q$ is a function $g:\F_q^m \rightarrow \F_q$ of the form
\begin{equation*}
g(x_1, \ldots, x_m) = \sum_{i_1, \ldots, i_m} c_{i_1, \ldots, i_m} \cdot x_1^{i_1} \cdots x_{m}^{i_m},
\end{equation*}
where each coefficient~$c_{i_1, \ldots, i_m}$ is an element of~$\F_q$.
We say that~$g$ has \emph{total degree $d$}
(or \emph{degree $d$}, for short) if $i_1 + \cdots + i_m \leq d$ for each nonzero coefficient $c_{i_1, \ldots, i_m}$,
and \emph{individual degree $d$} if $i_1, \ldots, i_m \leq d$ for each nonzero coefficient $c_{i_1, \ldots, i_m}$.
Low-degree polynomials have a variety of properties which make them useful in theoretical computer science,
chief among which is their \emph{distance}:
by the Schwartz-Zippel lemma, two nonequal degree $d$ polynomials~$g$ and~$h$ agree on at most a $d/q$ fraction of the points in~$\F_q^m$.

\emph{Low (individual) degree testing} refers to the task 
of verifying that an unknown function $g:\F_q^m \rightarrow \F_q$
is representable as a polynomial of (individual) degree~$d$
by querying~$g$ on a small number of points $u \in \F_q^m$.
There is a pair of canonical tests for doing so
known as the \emph{surface-versus-point low-degree test}
and the \emph{low individual degree test}.
It is common to frame these tests
as games between a referee and two provers.
In this setting, the surface-versus-point low degree test, parameterized by an integer $k \geq 1$, is performed by the verifier as follows.
\begin{enumerate}
\item Select~$\bu \sim \F_q^m$ uniformly at random. Give it to Prover~$\mathrm{A}$. They respond with a value~$\ba \in \F_q$.
\item Select a uniformly random $k$-dimensional affine surface~$\bs$ in~$\F_q^m$ containing~$\bu$.
	Give it to Prover~$\mathrm{B}$. They respond with a degree-$d$ $k$-variate polynomial $\boldf:\bs\rightarrow \F_q$.
\item Accept if $\boldf(\bu) = \ba$.
\end{enumerate}
This test is motivated by the following ``local characterization" of low-degree polynomials:
a polynomial $g:\F_q^m \rightarrow \F_q$ is degree-$d$
if and only if $g|_s$ is degree-$d$ for all $k$-dimensional surfaces~$s$.
Hence, if~$g$ is degree-$d$,
then the provers can win with probability~$1$ by always replying with~$\ba = g(\bu)$ and $\boldf = g|_{\bs}$.
\emph{Soundness} of the low-degree test refers to the converse statement,
namely that players who succeed with high probability
must be responding based on a low-degree polynomial.
This is formalized as follows.

\begin{theorem}[Raz-Safra~\cite{RS97}]\label{thm:raz-safra}
Suppose Provers~$\mathrm{A}$ and~$\mathrm{B}$ pass the $k = 2$ surface-versus-point low-degree test with probability~$1-\eps$.
Then there exists a degree-$d$ polynomial $g:\F_q^m \rightarrow \F_q$ such that
\begin{equation*}
\Pr_{\bu \sim \F_q^m}[g(\bu) = \ba] \geq 1 - \eps - \poly(m) \cdot \poly(d/q).
\end{equation*}
\end{theorem}

A similar ``local characterization" of low individual degree polynomials states that a polynomial~$g$ has individual degree $d$
if and only if $g|_{\ell}$ is a univariate degree-$d$ polynomial for all axis-parallel lines~$\ell$.
An axis-parallel line is a line of the form $\ell = \{u + a \cdot e_i \mid i \in \F_q\}$, for $u \in \F_q^m$ and $i \in \{1, \ldots, m\}$.
Motivated by this, the low individual degree test follows the same outline as the surface-versus-point low degree test
except with the second step substituted with the following.
\begin{enumerate}
\setcounter{enumi}{1}
\item Select a uniformly random axis-parallel line $\bell$ in~$\F_q^m$ containing~$\bu$.
	Give it to Prover~$\mathrm{B}$. They respond with a degree-$d$ univariate polynomial $\boldf:\bell\rightarrow \F_q$.
\end{enumerate}
When $d = 1$, the low individual degree test
is called the \emph{multilinearity test} because a polynomial with individual degree $d = 1$ is a multilinear polynomial.
The multilinearity  and low individual degree tests were
first introduced and proven sound by Babai, Fortnow, and Lund in~\cite{BFL91}.
The analysis of its soundness was then improved by~\cite{AS98} and then further sharpened by~\cite{FHS94}.
The best bound follows from the work of Polishchuk and Spielman~\cite{PS94};
their work considers only the bivariate $m = 2$ case,
but extending it to the multivariate case yields the following result.
\begin{theorem}[Polishchuk-Spielman~\cite{PS94}]\label{thm:classical-test-soundness}
Suppose Provers~$\mathrm{A}$ and~$\mathrm{B}$ pass the low individual degree test with probability~$1-\eps$.
Then there exists a polynomial $g:\F_q^m \rightarrow \F_q$ with individual degree~$d$ such that
\begin{equation*}
\Pr_{\bu \sim \F_q^m}[g(\bu) = \ba] \geq 1 -  \poly(m) \cdot (\poly(\eps) + \poly(d/q)).
\end{equation*}
\end{theorem}
We note that the soundness error the low individual test gives is actually worse than the low degree test,
because the low individual degree function~$g$ is only $\poly(m) \cdot (\poly(\eps) + \poly(d/q))$ close to Player~$\mathrm{A}$'s strategy,
rather than $\eps + \poly(m) \cdot \poly(d/q)$.
We will discuss this weakness of the low individual degree test below.

The multilinearity test, low individual degree test, and low-degree test
form a sequence
in which each test generally enables more applications than the previous one.
The multilinearity test can be used to show that $\MIP = \NEXP$ using a polynomial number of rounds~\cite{BFL91},
the low individual degree test can reduce the number of rounds to~$1$,
and the low degree test can be used to ``scale this result down''
and prove the PCP theorem, i.e.\ $\NP = \MIP[O(\log(n)), O(1)]$~\cite{AS98,ALM+98}.

\subsection{Quantum soundness of the low degree tests}

The work of Ito and Vidick~\cite{IV12}
initiated a program of studying these tests
in the case when the players are quantum,
as a means of proving bounds on the complexity class~$\MIP^*$.
Because the provers are quantum, 
they are allowed to share an entangled state,
a resource which could potentially allow them to ``cheat" the test
and win without using a low-degree polynomial.
The goal of this program is to show that this is not possible.
In other words, the goal is to show that these tests are \emph{quantum sound},
which means that provers who succeed with high success probability
must answer their questions according to a low (individual) degree polynomial,
even if they are allowed to share quantum entanglement.\footnote{We note that \emph{quantum soundness of the low-degree tests}, which is the focus of this work,
is distinct from \emph{soundness of the quantum low-degree test}.
The ``quantum low-degree test'' is a particular test introduced by Natarajan and Vidick in~\cite{NV18a}
which gets its name from the prominent role that the low-degree test plays as a subroutine,
and its ``soundness'' is simply the result that they prove about it.
}

Correctly formalizing the notion of quantum soundness is a subtle task,
as quantum provers can in fact ace these tests
using a broader class of strategies than their classical counterparts.
For example, the two provers can use their quantum state~$\ket{\psi}$ to simulate shared randomness,
which they can use to sample a random low-degree polynomial~$\bg$  to answer their questions with;
what makes this still acceptable is that~$\bg$ depends only on their shared randomness and not their questions.
The correct formalization of quantum soundness was identified by Ito and Vidick~\cite{IV12},
which states the following:
suppose Provers~$\mathrm{A}$ and~$\mathrm{B}$ pass the low-degree test with probability close to~$1$.
For each point question~$u \in \F_q^m$,
let $A^{u} = \{A^{u}_a\}$ be the measurement that Prover~$\mathrm{A}$ applies to
their share of~$\ket{\psi}$ to produce the answer $\ba \in \F_q$.
Then the test being quantum sound
means that there should be a measurement $G = \{G_g\}$, independent of~$u \in \F_q^m$,
which outputs degree-$d$ polynomials~$g$
and ``acts like~$A$''.
In other words, rather than measuring~$A^{u}$ to produce the outcome~$\ba \in \F_q$,
Prover~$\mathrm{A}$ could have simply measured~$G$, received the polynomial~$\bg$,
and outputted its evaluation at~$\bu$, i.e.\ the value $\bg(\bu)$.
We will measure the similarity between~$A$ and~$G$
by considering the experiment where Prover~$\mathrm{A}$ measures with~$A^{\bu}$ to produce~$\ba$,
Prover~$\mathrm{B}$ measures with~$G$ to produce~$\bg$, and we check if $\bg(\bu) = \ba$.
This entails studying the quantity
\begin{equation*}
\E_{\bu \sim \F_q^m} \sum_{a \in \F_q} \sum_{g:g(\bu) = a}\bra{\psi} A^{\bu}_a \ot G_g \ket{\psi},
\end{equation*}
which we aim to show is as close to~$1$ as possible.
In this way, the provers' quantum advantage is limited to their ability to select a low-degree polynomial~$g$.

One additional quirk of the quantum setting
is that it has been historically useful to consider variants of these tests which feature more than two provers.
This allows one to use monogamy of entanglement to reduce the power that entanglement gives to the provers,
making it easier to show that a given test is quantum sound.
Low degree test results with fewer provers are more difficult to show and have more applications.

The program of showing that these tests are quantum sound was carried out for the $3$-prover multilinearity test by Ito and Vidick~\cite{IV12}
and for the $3$-prover low-degree test by Vidick~\cite{Vid16},
which was later improved to $2$-provers by Natarajan and Vidick~\cite{NV18b}.
This latter result led to a sequence of works
which culminated in the proof that $\mathsf{MIP}^* = \mathsf{RE}$
and the refutation of the Connes embedding conjecture in~\cite{JNV+20}.
We summarize this line of research in \Cref{fig:research}.

\ignore{
Quantum soundness of a three-player variant of the multilinearity test was established by Ito and Vidick in~\cite{IV12}
to prove that $\NEXP \subseteq \mathsf{MIP}^*$ with three provers.
Next, Vidick~\cite{Vid16} proved quantum soundness of a three-player variant of the low-degree test
to prove that $\NP \subseteq \mathsf{MIP}^*[O(\log(n)), O(1)]$ with three provers.
Building on this work, Natarajan and Vidick~\cite{NV18b} proved quantum soundness of the two-player low-degree test.
This result has led to a sequence of works achieving stronger and stronger results,
culminating in the proof that $\mathsf{MIP}^* = \mathsf{RE}$
and the refutation of the Connes embedding conjecture in~\cite{JNV+20}.
}

{
\floatstyle{boxed} 
\restylefloat{figure}
\begin{figure}
\begin{tabular}{l c c c}
& Test shown & Complexity-theoretic  & Number \\
&  quantum-sound &  consequence & of provers  \\
			& &&\\
1. \cite{IV12}: & multilinearity test & $\NEXP \subseteq \mathsf{MIP}^*$ & 3\\[1ex]
2. \cite{Vid16}: & low-degree test & $\NP \subseteq \mathsf{MIP}^*[O(\log(n)), O(1)]$  &3\\[1ex]
3. \cite{NV18b}: & low-degree test & $\NP \subseteq \mathsf{MIP}^*[O(\log(n)), O(1)]$ & 2\\
			& &&\\
& Consequences of \cite{NV18b}: & &\\
			& & &\\
&\multicolumn{1}{l}{\quad\qquad(a) \cite{NV18a}:}   & $\QMA \subseteq \mathsf{MIP}^*[O(\log(n)), O(1)]$ &7\\
&    & (under randomized reductions) &\\[1ex]
&\multicolumn{1}{l}{\quad\qquad (b) \cite{NW19}:}   & $\NEEXP \subseteq \mathsf{MIP}^*$ &2\\[1ex]
&\multicolumn{1}{l}{\quad\qquad (c) \cite{JNV+20}:}   & $\MIP^* = \mathsf{RE}$&2
\end{tabular}
	\caption{Prior work on quantum-sound low degree tests and their complexity-theoretic consequences.
			The first three works showed a quantum-sound test and an $\MIP^*$ bound,
			both involving the same number of provers indicated in the final column.
			The last three works use the low-degree test from~\cite{NV18b}
			to show the indicated $\MIP^*$ bound.}
\label{fig:research}
\end{figure}
}

Subsequent to the initial posting of~\cite{JNV+20} on the arXiv,
an error was discovered in the analysis of the quantum-sound low degree test
contained in~\cite{Vid16} which was propagated to~\cite{NV18b}. The error affects the proof in a manner that appears difficult to fix. As such, we currently do not know if the low-degree test is quantum-sound for any number of provers.
The invalidation of this analysis affects every result in \Cref{fig:research} except for~\cite{IV12}.

The purpose of this work is to provide a different soundness analysis, for a variant of the low-degree test, that can nevertheless be used as a replacement for it in most subsequent works. 
We do so by revisiting the three-player quantum-sound multilinearity test of~\cite{IV12}
and improving this result in two ways.
First, we generalize it to hold for the degree-$d$ low individual degree test,
of which the multilinearity test is the $d=1$ special case.
Second, using techniques introduced in~\cite{Vid16,NV18b},
we reduce the number of provers from~$3$ to~$2$.
Our main result is as follows.

\begin{theorem}[Main theorem, informal]\label{thm:main-informal}
Suppose Provers~$\mathrm{A}$ and~$\mathrm{B}$
pass the two-prover. degree-$d$ low individual degree test with probability $1-\eps$.
Let $A = \{A^u_a\}$ be the measurement the provers perform when they are given the point $u \in \F_q^m$
to produce a value $a \in \F_q$.
Then there exists a projective measurement $G = \{G_g\}$
whose outcomes~$g$ are polynomials of individual degree~$d$
such that
\begin{equation*}
\E_{\bu \sim \F_q^m}\sum_{a \in \F_q} \sum_{g:g(\bu) = a} \bra{\psi} A^{\bu}_{a} \ot G_g \ket{\psi}
\geq 1 - \poly(m) \cdot (\poly(\eps) + \poly(d/q)).
\end{equation*}
In other words, if Prover~$\mathrm{A}$ measures according to~$A^{\bu}$ to produce $\ba$ and Prover~$\mathrm{B}$ measures according to~$G$ to produce~$\bg$,
then $\bg(\bu) = \ba$ except with probability $\poly(m) \cdot (\poly(\eps) + \poly(d/q))$.
\end{theorem}
\noindent
Thus, we are able to extend \Cref{thm:classical-test-soundness}
to the case of two quantum provers (with some minor caveats; see \Cref{thm:main-formal} below for the formal statement of \Cref{thm:main-informal}).

Although \Cref{thm:main-informal} establishes quantum soundness of the low-individual degree test
and not of the low-degree test, it is still sufficient to recover the result $\NEEXP \subseteq \mathsf{MIP}^*$
from~\cite{NW19}
and the result $\mathsf{MIP}^* = \mathsf{RE}$ from~\cite{JNV+20}.
In addition,
we can use it to recover the self-test for an exponential number of EPR pairs from~\cite{NV18a}.
Edited drafts of these work to account for this change are forthcoming.
It remains open whether the complexity-theoretic consequences
of~\cite{Vid16,NV18b,NV18a} to the ``scaled down'' setting still hold.

\subsection{Total degree versus individual degree}

We now contrast the low degree test with the individual degree test
and explain why we are only able to prove the latter quantum sound.
We begin by explaining why the low individual degree test,
unlike the low degree test,
requires a $\poly(m) \cdot \poly(\eps)$ dependence in the soundness error.

\begin{example}\label{ex:bad-individual-degree-example}
Consider the degree-$(d+1)$ polynomial $h(x_1, \ldots, x_m) = x_1^{d+1}$,
and suppose that Players~$\mathrm{A}$ and~$\mathrm{B}$ play according to the following classical strategy.
\begin{itemize}
\item[$\circ$] (Player~$\mathrm{A}$): given $\bu \in \F_q^m$, return the value $\ba = h(\bu)$.
\item[$\circ$] (Player~$\mathrm{B}$): given the axis parallel line~$\bell = \{\bu + x \cdot e_{\bi} \mid x \in \F_q\}$, act as follows.
				If $\bi > 1$, then $h$ is a constant function along~$\bell$, and so return $\boldf = h|_{\bell}$.
				Otherwise, if $\bi = 1$, then $h$ is a degree-$(d+1)$ polynomial along~$\bell$.
				As the verifier expects a degree-$d$ polynomial,
				simply give up and return $\boldf \equiv 0$.
\end{itemize}
Using this strategy, $\boldf(\bu) = h|_{\bell}(\bu) = h(\bu) = \ba$ whenever $\bi \neq 1$, which occurs with probability $1- \frac{1}{m}$.
Hence, the players pass the degree-$d$ low individual degree test with probability at least $1 -\eps$ for $\eps = \frac{1}{m}$.
However, Player~$\mathrm{A}$ is responding to their questions using the polynomial~$h$
which is degree-$(d+1)$ but not degree-$d$,
and so by the aforementioned Schwartz-Zippel lemma,
any degree-$d$ polynomial~$g$ will agree with~$h$ on at most a $\frac{d+1}{q}$ fraction of all inputs.
This means that the agreement between Player~$\mathrm{A}$'s strategy and any degree-$d$ polynomial~$g$ is at most
\begin{equation*}
\Pr_{\bu \sim \F_q^m}[g(\bu) = \ba]
\leq 
1 - m \cdot \eps + \frac{d+1}{q}.
\end{equation*}
\end{example}

\Cref{ex:bad-individual-degree-example} shows that the dependence on~$m$ and~$\eps$ in \Cref{thm:main-informal} is tight up to polynomial factors.
This reveals a weakness with the low individual degree test:
one can only conclude that the players are using a low individual degree strategy
when their failure probability $\eps$ is tiny---on the order of $\frac{1}{m}$ or smaller.
The low (total) degree test is alluring because it has the potential to avoid this dependence on~$m$.

Soundness of the low total and individual degree tests
is typically proven by induction on~$m$.
For the low individual degree test,
each step of the induction incurs an error of $\poly(m) \cdot (\poly(\eps) + \poly(d/q))$.
Summing over all~$m$ steps, this gives a total error of $\poly(m) \cdot (\poly(\eps) + \poly(d/q))$,
exactly as in \Cref{thm:main-informal}.
For the low degree test,
on the other hand, each step of the induction only incurs an error of $\poly(\eps) + \poly(d/q)$.
However, if summed over all~$m$ steps,
this still gives a total error of $m \cdot (\poly(\eps) + \poly(d/q))$,
which is too large.

To account for this,
the soundness proofs in \cite{Vid16, NV18b}
introduce a technique at the end of each induction step
called \emph{Consolidation}
in which the growing error is ``reset''
down to an error $\poly(\eps) + \poly(d/q)$
which remains fixed across all levels of the induction.
This allows them to conclude with an error that was independent of the dimension~$m$.
Consolidation works as follows:
if the error of the projective measurement $G = \{G_g\}$
ever grows past this fixed error, the Consolidation step argues that on some portion of the Hilbert space,
the measurement~$G$ must be performing much worse than expected;
it then corrects this by inductively calling the low-degree test soundness to produce a better measurement on this portion of the Hilbert space.
Ultimately, however, this creates a cascade of Consolidation steps calling each other with increasing error,
which at some point grows so large that the low-degree soundness can no longer be applied.
This is the source of the bug.

\ignore{
The bug in their proofs is contained in this Consolidation section.
}

Fortunately, this work shows that the techniques of these prior works, 
aside from Consolidation,
are still sound.
So we can show soundness bounds which grow as a function of~$m$,
even if we cannot yet show bounds independent of~$m$.
In the end, this ``weakness" of the low individual degree test
is precisely what allows us to show that it, and not the low degree test, is quantum sound.

That said, we do believe, although we have not rigorously verified,
that our techniques are capable of proving a result like \Cref{thm:main-informal} for the low degree test,
i.e.\ a soundness bound of $\poly(m) \cdot (\poly(\eps) + \poly(d/q))$ rather than $\eps + \poly(m) \cdot \poly(d/q)$.
This ``weak'' quantum soundness does not rule out the possibility that quantum provers can significantly outperform their classical counterparts.
It is, however, still sufficient for the same applications as the low individual degree test,
and it hints towards the possibility of a full quantum soundness of the low degree test.

\subsection{Conclusion and open problems}

In recent years,
the classical low-degree test has played a critical role in the study of nonlocal games
and the complexity class~$\MIP^*$.
In addition to correcting previous proofs of soundness,
we hope that this new exposition
will invite new researchers to engage with this beautiful area.
We conclude with a short list of open problems for future work.

\begin{enumerate}
\item Is the classical low-degree test quantum sound? Can the Consolidation step be fixed?
\item Even if the classical low-degree test is shown quantum sound,
	there are still interesting questions to be answered about the low individual degree test.
	For example, can the diagonal lines test be removed? (See \Cref{sec:the-test} for a description of this subtest.)
	Doing so would likely simplify the proof of $\MIP^* = \mathsf{RE}$~\cite{JNV+20},
	as one could replace the complicated ``conditional linear functions'' used in the proof
	with a simpler subclass known as ``coordinate deletion functions''.
	However, as discussed at the end of \Cref{sec:technical-overview},
	we know of an example that requires the diagonal lines test
	for the low individual degree test with parameters $m = 2$, $d = 2$, and $q = 4$.
	Can we find similar examples for larger~$q$?
\item Classically,  the low degree test
	generalized in various directions,
	including tests for affine invariant properties~\cite{KS08,Sud11}
	and tests for tensor product codes (see, for example, \cite{CMS17}).
	Does quantum soundness hold for these tests as well?
\item This work continues the trend of showing that well-studied classical property testers
	are also quantum sound.
	Prior to this, the work of~\cite{IV12} (see also \cite[Chapter 2]{Vid11}) showed that the linearity tester of Blum, Luby, and Rubinfield~\cite{BLR93}
	is also quantum sound.
	The proofs of these results have so far been case-by-case adaptations of the classical proofs to the quantum setting;
	in the case of this paper, the proof is highly nontrivial
	and involves many ad hoc calculations which fortunately go in our favor.
	Is there a more conceptual reason why quantum soundness holds for these testers?
	Perhaps a reduction from the quantum case to the classical case?
\end{enumerate}

\section{Technical overview}\label{sec:technical-overview}

At a high level, our proof follows the approach of~\cite{BFL91}, and
it is useful to start by summarizing their analysis, which applies to
classical, \emph{deterministic} strategies. In this version of the test, the provers' strategy is described
by a ``points function,'' assigning a value in $\F_q$ to each point in
$\F_q^m$, and a
``lines function,'' assigning a low-degree polynomial to each
line in $\F_q^{m}$ queried in the test. From the assumption that the
points and lines functions agree at a randomly chosen point with high
probability, we would like to construct a global low-degree polynomial that has high agreement with
the points function. This is done by inductively constructing
``subspace functions'' defined on affine axis-aligned subspaces of increasing
dimension $k$. The base case is $k=1$, and is supplied by the lines
function. At each step, to construct the subspace function for a
subspace $S$ of dimension $k+1$, we pick $d+1$ parallel subspaces of
dimension $k$ that lie within $S$, and compute the unique degree-$d$
polynomial that interpolates them. The analysis shows that at each
stage, the function constructed through interpolation has high
agreement on average with the points function and lines function. In
the end, when we reach $k=m$ the ambient dimension, the subspace
function we construct is the desired global function.

For the sake of simplicity, it suffices to consider the case $d=1$,
i.e. multilinear functions. In this case, whenever we perform
interpolation, we need to combine $2$ parallel subspaces.

\paragraph{The classical zero-error case}

To start building intuition, it is useful to think about how to carry
out the above program in a highly simplified setting: the classical
zero-error case, for $m=3$. In this case, we assume that we have access to a
points function $f$ and lines function $g$ that \emph{perfectly} pass the BFL
test. Moreover, we will focus on the \emph{final} step of the
induction: thus, we assume that we have already constructed a set of
planes functions defined for every axis-parallel plane, that are
perfectly consistent with the line and point functions. In the final
step of the induction, our goal is to combine these planes functions
to create a single global function $h$
over all of $\F_q^{3}$ that is consistent with the points and lines functions. 

To do this, we will interpolate the planes as follows. Let us label
the 3 coordinates in the space $x, y, $and $z$, and consider planes
parallel to the $(x,y)$-plane. Each plane $S_z$ is specified by a
value of the $z$ coordinate:
\[ S_z = \{(x,y,z): (x, y) \in \F_q^2\}. \]
By the induction hypothesis, for every such plane $S_z$ there exists a
bilinear function $g_{z}: \F_q^{2} \to \F_q$ that agrees with the
points function. To construct a global multilinear function $h:
\F_q^3 \to \F_q$, we pick two distinct values $z_1 \neq z_2$, and ``paste''
the two plane functions $g_{z_1} $ and
$g_{z_2}$ together using \emph{polynomial
  interpolation}. Specifically, we define $h$ to be the unique
multilinear polynomial that interpolates between $g_{z_1}$ on the
plane $S_{z_1}$ and $g_{z_2}$ on the plane $S_{z_2}$.
\[ h(x,y,z) = \mathrm{interpolate}_{z_1, z_2}(g_{z_1}, g_{z_2}) = \left(\frac{z - z_2}{z_1 - z_2}\right) g_{z_1}(x,y) +
  \left(\frac{z - z_1}{z_2 - z_1}\right) g_{z_2}(x,y). \]

This procedure defines a global function $h$, and by the assumption
that $g_{z_1}$ and $g_{z_2}$ are multilinear, it follows that $h$ is
also a multilinear function. But why is $h$ consistent with the points
function? To show this, we need to consider the lines function on
lines parallel to the $z$ axis. Given a point $(x,y,z)$, let $\ell$ be
the line parallel to the $z$ axis through this point, and let $g_\ell$
be the associated lines function. By construction, $h$ agrees with
$g_\ell$ at the two points $z_1$ and $z_2$. But $g_\ell$ and the restriction
$h_{|\ell}$ of $h$ to $\ell$ are both linear functions, and hence if
they agree at two points, they must agree \emph{everywhere}. (This is
a special case of the \emph{Schwartz-Zippel lemma}, which says that if
two degree-$d$ polynomials agree at $d+1$ points, they must be equal.)
Thus, $h$
agrees with $g_\ell$ at the original point $z$ as well. By success in
the test, $g_\ell$ in turn agrees with the points function $f$ at
$(x,y,z)$, and thus, $h(x,y,z) = f(x,y,z)$. Thus, we have shown that
the global function $h$ is both multilinear and agrees with the points
function $f$ exactly.

\paragraph{Dealing with errors}
To extend the sketch above to the general case, with nonzero error,
requries some modifications. At the most basic level, we may consider
what happens when we allow for deterministic classical strategies that
succeed with probability less than $1$ in the test. Such strategies
may have ``mislabeling'' error: the points function
$f$ may be imagined to be a multilinear function that has been
corrupted at a small fraction of the points. This type of error is
handled by the analysis in~\cite{BFL91}. The main modification to the zero-error
sketch above is a careful analysis of the probability that the pasting
step produces a ``good'' interpolated polynomial for a randomly chosen
pair of planes $S_{z_1}, S_{z_2}$. This analysis makes use of the
Schwartz-Zippel lemma together with combinatorial properties of the
point-line test itself (e.g. the expansion properties of the question
graph associated with the test).

At the next level of generality, we could consider classical
\emph{randomized} strategies. Suppose we are given a randomized
strategy that succeeds in the test with probability $1 - \eps$. Any
randomized strategy can be modeled by first sampling a random seed,
and then playing a deterministic strategy conditioned on the value of
the seed. A success probability of $1 - \eps$ could have two
qualitatively different underlying causes: (1) on $O(\eps)$ fraction
of the seeds, the strategy uses a function which is totally corrupted,
and (2) on a large fraction of the seeds, the strategy uses functions
which are only $\eps$-corrupted. An analysis of randomized
strategies could naturally proceed in a ``seed-by-seed'' fashion, applying the
deterministic analysis of~\cite{BFL91} to the large fraction of ``good'' seeds (which
are each only $\eps$-corrupted), while giving up entirely on the
``bad'' seeds. 

In this document, we consider quantum strategies, which
have much richer possibilities for error. Nevertheless, we are able to
preserve some intuition from the randomized case, by working with
\emph{sub-measurements}: quantum measurements that do not always yield
an outcome. Working with a sub-measurement allows us to distinguish two
kinds of error: \emph{consistency error} (the probability that a
sub-measurement returns a wrong outcome) and \emph{completeness error}
(the probability that the sub-measurement fails to return an outcome at
all). Roughly speaking, the completeness error corresponds to the
probability of obtaining a ``bad'' seed in the randomized case, while
the consistency error corresponds to how well the strategies do on
``good'' seeds. 

The technique of using sub-measurements and managing the two types of
error separately goes back to~\cite{IV12}. That work developed a
crucial tool to convert between these two types of error called the
\emph{self-improvement lemma}, and in our analysis we make extensive
use of a refined version of this lemma (\Cref{thm:self-improvement-in-induction-section}), building on~\cite{Vid16,NV18b}. Essentially, the lemma says the
following: suppose that (a) the provers pass the test with probability
$1 - \eps$, and (b) there is a complete measurement $G_{g}$ whose
outcomes are low-degree polynomials that has
consistency error $\nu$: that is, $G$ always returns an outcome, but
has probability $\nu$ of producing an outcome $g$ that disagrees with the points measurement $A^u_a$ at a
random point $u$. Then there exists an
``improved'' sub-measurement  $H_{h}$ with consistency error $\zeta$
depending \emph{only} on $\eps$, and with completeness error
(i.e. probability of not producing an outcome at all) of $\nu
+\zeta$. Essentially, this lemma says we can always ``reset'' the
consistency error of any measurement we construct at intermediate
points in the analysis to a universal
function $\zeta$ depending only on the provers' success in the test,
at the cost of introducing some amount completeness error. Intuitively, one may
think of the action of the lemma as correcting $G$ on the portions of
Hilbert space where it is only mildly corrupted, while ``cutting out'' the
portions of Hilbert space where $G$ is too corrupted to be
correctable. In some sense, this lemma is the quantum analog of the
idea of identifying ``good'' and ``bad'' random seeds in the classical
randomized case. The proof of the lemma uses a semidefinite program
together with the combinatorial facts used in~\cite{BFL91}
(specifically, expansion of the question graph of the test, and
the Schwartz-Zippel lemma).

Armed with the self-improvement lemma, we set up the following
quantum version of the BFL
induction loop: for $k$ running from $1$ to $m$, we construct a family
of subspace measurements $G^{S}_{g}$ that returns a low-degree
polynomial $g$ for every axis-aligned affine subspace $S$ of dimension $k$.
\begin{enumerate}
  \item By the induction hypothesis, we know that there exists a
    measurement  $G^{S}_{g}$ for every $k$-dimensional subspace, which
    has consistency error $\delta(k)$ with the points
    measurement. (For the base case $k = 1$, this is the lines
    measurement from the provers' strategy). 
  \item We apply the self-improvement lemma to these measurements, yielding
    sub-measurements $\hat{G}^{S}_{g}$ that have consistency error
    $\zeta$ independent of $\delta(k)$, and completeness error $\kappa(k) = \delta(k) + \zeta$.
  \item For each subspace $S$ of dimension $k+1$, we construct a
    pasted sub-measurement, by performing a quantum version of the
    classical interpolation argument: we define the pasted
    sub-measurement by sequentially measuring several subspace
    measurements corresponding to parallel $k$-dimensional subspaces,
    and interpolate the resulting outcomes. This pasted sub- measurement has consistency slightly worse than $\zeta$, and completeness error which
    is slightly worse than $\kappa(k)$. It is at this step that it is crucial to treat the two types
    of error separately: in particular, we need the consistency error
    to be low to ensure that the interpolation produces a good result.
  \item We convert the resulting sub-measurement into a full
    measurement, by assigning a random outcome whenever the
    sub-measurement fails to yield an outcome. This measurement will
    have consistency error $\delta(k+1)$ which is larger than $\delta(k)$
    by some additive factor.
  \end{enumerate}
  At the end of the loop, when $k = m$, we obtain a single measurement that returns
  a global polynomial as desired.
\paragraph{The diagonal lines test}
An important element of the test we analyze in this document, which is
unnecessary in the classical case, is the diagonal lines test. The purpose
of this test is to certify that the points measurements used by the provers
approximately commute on average over all pairs of points $(x, y)$ in
$\F_q^m$---something which is automatically true in the classical case. This is done by asking one prover for a polynomial defined
on the line going through $x$ and $y$, and the other for the function
evaluation at either $x$ or $y$. The line through $x$ and
$y$ will not in general be axis-parallel; we refer to these general
lines as ``diagonal.''

The commutation guarantee plays an important role in our analysis of
the test, and it is an interesting question whether it is truly
necessary to test it directly with the diagonal lines test; might it
not automatically follow from success in the axis-parallel lines test? An interesting contrast can be drawn to
the Magic Square
game~\cite{mermin1990simple,peres1990incompatible,aravind2002simple},
in which questions are either cells or axis-parallel lines in a $3 \times 3$ square grid.
This has the same question distribution as
the axis-parallel line-point test over $\F_3^2$ (although the answers in the Magic Square game are strings in~$\F_2$ rather than~$\F_3$).
For the Magic Square game,
``points'' measurements along the same axis-parallel ``line'' commute,
but points that are not axis-aligned do \emph{not} commute: indeed,
for the perfect strategy, they anticommute. Despite this example,
we know that commutation between all pairs of points \emph{can} be deduced from the
axis-parallel lines test alone, rending the diagonal lines test unnecessary,
at least in the case of the bivariate ($m = 2$) multilinearity test (the low individual degree test when $d = 1$).
On the other hand, we know of a quantum strategy using noncommuting measurements
which succeeds with probability~$1$ in the $m = 2$, $d = 2$, $q = 4$ low individual degree test.
Whether this counterexample can be extended to larger~$q$ remains an open question.

\paragraph{Organization}
The rest of this document is organized as follows. In
\Cref{sec:prelims} we review some preliminaries concerning finite
fields, polynomials, and quantum measurements. In
\Cref{sec:making-measurements-projective}, we present two tools for
making quantum measurements projective, which are used in our
analysis. In \Cref{sec:induction}, we give the inductive argument
proving the main theorem. In \Cref{sec:expansion,sec:variance} we show
some properties of the hypercube graph and families of measurements
indexed by points on the hypercube, which are used in
\Cref{sec:self-improvement} to prove the self-improvement lemma. In
\Cref{sec:commutativity-points,sec:g-comm}, we show commutativity
properties of the measurements constructed in the induction, and
finally in \Cref{sec:ld-pasting} we analyze the pasting step of the induction.

\paragraph{Acknowledgments}

We thank Lewis Bowen for pointing out typos and a few minor errors in a previous version.
We also thank Madhu Sudan for help with references on the classical low individual degree test.

\section{The test}\label{sec:the-test}

\begin{definition}[Roles]
The low individual degree test will be played between two provers and a verifier.
The two provers are named Player~$\mathrm{A}$ and Player~$\mathrm{B}$.
A \emph{role} $r$ is an element of the set $\{\mathrm{A}, \mathrm{B}\}$.
Given a role~$r$, we write $\overline{r}$ for the other element of the set $\{\mathrm{A}, \mathrm{B}\}$.
\end{definition}

\begin{definition}[Low individual degree test]
Let~$m$ and~$d$ be nonnegative integers.
Let $q$ be a prime power.
Then the \emph{$(m,q,d)$-low individual degree test}
is stated in \Cref{fig:test}.
\end{definition}

{
\floatstyle{boxed} 
\restylefloat{figure}
\begin{figure}
With probability~$\tfrac{1}{3}$ each, perform one of the following three tests.
\begin{enumerate}
	\item \textbf{Axis-parallel lines test:}
		Pick a uniformly random role $\br \sim \{\mathrm{A},\mathrm{B}\}$.
		Let $\bu \sim \F_q^m$ be a uniformly random point.
		Select $\bi \sim \{1, \ldots, m\}$ uniformly at random,
		and let $\bell = \{\bu + t \cdot e_{\bi} \mid t \in \F_q\}$
		be the axis-parallel line which passes through~$\bu$ in the $\bi$-th direction.
		\begin{itemize}
		\item[$\circ$] Player~$\br$: 
			Give $\bell$;
			receive the univariate degree-$d$ polynomial $\boldf:\bell \rightarrow \F_q$.
		\item[$\circ$] Player~$\overline{\br}$:
			Give $\bu$;
			receive $\ba \in \F_q$.
		\end{itemize}
		Accept if $\boldf(\bu) = \ba$.
	\item \textbf{Self-consistency test:}
		Let $\bu \sim \F_q^m$ be a uniformly random point.
		\begin{itemize}
		\item[$\circ$] Player~$\mathrm{A}$: 
			Give $\bu$;
			receive $\ba \in \F_q$.
		\item[$\circ$] Player~$\mathrm{B}$:
			Give $\bu$;
			receive $\bb \in \F_q$.
		\end{itemize}
		Accept if $\ba = \bb$.
	\item \textbf{Diagonal lines test:}
		Pick a uniformly random role $\br \sim \{\mathrm{A},\mathrm{B}\}$.
		Let $\bu \sim \F_q^m$ be a uniformly random point.
		Select $\bi \sim \{1, \ldots, m\}$ uniformly at random,
		and let $\bv \in \F_q^m$ be a uniformly random point whose last $m - \bi$ coordinates are~$0$.
		Finally,
		let $\bell = \{\bu + t \cdot \bv \mid t \in \F_q\}$
		be the line which passes through~$\bu$ in direction~$\bv$.
		\begin{itemize}
		\item[$\circ$] Player~$\br$: 
			Give $\bell$;
			receive the univariate degree-$md$ polynomial $\boldf:\bell \rightarrow \F_q$.
		\item[$\circ$] Player~$\overline{\br}$:
			Give $\bu$;
			receive $\ba \in \F_q$.
		\end{itemize}
		Accept if $\boldf(\bu) = \ba$.
	\end{enumerate}
	\caption{The $(m,q,d)$-low individual degree test.\label{fig:test}}
\end{figure}
}

An important class of strategies are those which are \emph{symmetric}.
In this case, the bipartite state $\ket{\psi}$ Alice and Bob share is symmetric.
Furthermore, for any question Alice and Bob receive, they apply the same measurement to their share of the state.
This means that rather than, for example, keeping track of a separate points measurement for Alice and Bob,
we can use a single measurement to refer to both of their strategies,
and similarly for the axis-parallel lines and diagonal lines measurements.
This is formalized in the following definition, where we also consider strategies which are projective.

\begin{definition}[Symmetric, projective strategy]
A \emph{symmetric, projective strategy} for the $(m,q, d)$-low individual degree test
is a tuple $(\psi, A, B, L)$ defined as follows.
\begin{itemize}
\item[$\circ$] $\ket{\psi} \in \calH \ot \calH$ is a bipartite, permutation-invariant state.
\item[$\circ$] $A = \{A^u_a\}$ contains a matrix for each $u \in \F_q^m$ and $a \in \F_q$.
			For each $u$, $A^u$ is a projective measurement on~$\calH$.
\item[$\circ$] $B = \{B^{\ell}_f\}$ contains a matrix for each axis-parallel line~$\ell$ in $\F_q^m$
			and univariate degree-$d$ polynomial $f:\ell\rightarrow \F_q$.
			For each $\ell$, $B^{\ell}$ is a projective measurement on~$\calH$.
\item[$\circ$] $L = \{L^{\ell}_f\}$ contains a matrix for each line~$\ell$ in $\F_q^m$
			and univariate degree-$md$ polynomial $f:\ell\rightarrow \F_q$.
			For each $\ell$, $L^{\ell}$ is a projective measurement on~$\calH$.
\end{itemize}
\end{definition}

We can also consider more general strategies which are no longer assumed to be symmetric.
In this case, Players~$\mathrm{A}$ and~$\mathrm{B}$ each have their own versions of the measurements~$A$, $B$, and~$L$.

\begin{definition}[General projective strategy]
A \emph{projective strategy} for the $(m,q, d)$-low individual degree test
is a tuple $(\psi, A^{\mathrm{A}},B^{\mathrm{A}},L^{\mathrm{A}},A^{\mathrm{B}},B^{\mathrm{B}},L^{\mathrm{B}})$ defined as follows.
\begin{itemize}
\item[$\circ$] $\ket{\psi} \in \calH_{\mathrm{A}} \ot \calH_{\mathrm{B}}$ is a bipartite state.
\end{itemize}
Furthermore, for each~$w \in \{\mathrm{A}, \mathrm{B}\}$:
\begin{itemize}
\item[$\circ$] $A^{w} = \{A^{w,u}_a\}$ contains a matrix for each $u \in \F_q^m$ and $a \in \F_q$.
			For each $u$, $A^{w,u}$ is a projective measurement on~$\calH_{w}$.
\item[$\circ$] $B^{w} = \{B^{w,\ell}_f\}$ contains a matrix for each axis-parallel line~$\ell$ in $\F_q^m$
			and univariate degree-$d$ polynomial $f:\ell\rightarrow \F_q$.
			For each $\ell$, $B^{w,\ell}$ is a projective measurement on~$\calH_{w}$.
\item[$\circ$] $L^w = \{L^{w,\ell}_f\}$ contains a matrix for each line~$\ell$ in $\F_q^m$
			and univariate degree-$md$ polynomial $f:\ell\rightarrow \F_q$.
			For each $\ell$, $L^{w,\ell}$ is a projective measurement on~$\calH_{w}$.
\end{itemize}
\end{definition}

\begin{remark}
Throughout this work, we will only consider projective strategies.
So we will henceforth use the terms ``strategy" and ``symmetric strategy'' to refer exclusively to projective strategies and symmetric, projective strategies, respectively.

In addition, we will spend the vast majority of this work dealing solely with symmetric strategies, as they are notationally simpler to work with.
Much of this work will focus on proving~\Cref{thm:main-induction},
a variant of our main theorem for symmetric strategies.
Our main theorem for general strategies, \Cref{thm:main-formal} below, will be proven in \Cref{sec:induction} by a standard reduction to the symmetric case.
We will only see these more cumbersome-to-write general strategies
in this section, where we state \Cref{thm:main-formal},
and in \Cref{sec:induction}, where we carry out the reduction.
\end{remark}

\begin{definition}[Good strategy]
A strategy is $(\eps, \delta, \gamma)$-good if it
passes the axis-parallel lines test with probability at least $1-\eps$,
the self consistency test with probability at least~$1-\delta$,
and the diagonal lines test with probability at least~$1-\gamma$.
\end{definition}

\begin{remark}\label{rem:good-strat-characterization}
Using notation which will be introduced in \Cref{sec:comparing-measurements} below,
a symmetric strategy is $(\eps, \delta, \gamma)$-good
if and only if it satisfies the following three conditions.
For $\bell$ and $\bu$ as in the axis-parallel lines test,
\begin{equation*}
A^u_a \ot I \simeq_{\eps} I \ot B^{\ell}_{[f(u)=a]},
\quad
\text{and}
\quad
A^u_a \ot I \simeq_{\delta} I \ot A^u_a.
\end{equation*}
And for $\bell$ and $\bu$ as in the diagonal lines test,
\begin{equation*}
A^u_a \ot I \simeq_{\gamma} I \ot L^{\ell}_{[f(u)=a]}.
\end{equation*}
\end{remark}

\begin{notation}\label{not:conditioned-on-last-direction}
Our proof will be via induction,
i.e.\ proving soundness of the $(m+1,q,d)$-low individual degree test
using the soundness of the $(m,q,d)$-low individual degree test.
To do this, we will frequently use the axis-parallel line test in the specific case of $i = m+1$.
Thus, it will be convenient to introduce the following notation.
Let $(\psi, A, B, L)$ be a symmetric strategy for the $(m+1, q, d)$-low individual degree test,
Then for each $u \in \F_q^m$ we will write $B^u_f$ as shorthand for $B^{\ell}_f$,
where $\ell = \{(u, x) \mid x \in \F_q\}$.
For a function $f:\ell \rightarrow \F_q$,
we will also sometimes write $f(x)$ as shorthand for $f(u,x)$.
\end{notation}

\begin{definition}
Consider the $(m,q,d)$-low individual degree test.
For $j \in \{1, \ldots, m\}$,
we refer to the \emph{$j$-restricted diagonal lines test}
as the diagonal lines test conditioned on $\bi = j$.
For example, in the $m$-restricted diagonal lines test,
the line $\bell$ is simply a uniformly random line in $\F_q^m$.
\end{definition}

Now we state our main theorem using notation which will be introduced in \Cref{sec:prelims} below.
This the formal version of \Cref{thm:main-informal}.

\begin{theorem}[Main theorem; quantum soundness of the low individual degree test]\label{thm:main-formal}
  Consider a projective strategy $(\psi, A^{\mathrm{A}},B^{\mathrm{A}},L^{\mathrm{A}},A^{\mathrm{B}},B^{\mathrm{B}},L^{\mathrm{B}})$ 
   which passes the  $(m,q,d)$-low individual degree test with probability at least $1-\eps$.
  Let $k \geq md$ be an integer.
  Let
  \begin{equation*}
  \nu = 100000 k^2m^4 \cdot \Big(\eps^{1/40000} + (d/q)^{1/40000} + e^{-k/(2560000m^2)}\Big).
  \end{equation*}
Then there exists projective measurements $G^{\mathrm{A}}, G^{\mathrm{B}} \in \polymeas{m}{q}{d}$ with the following properties:
\begin{enumerate}
  \item (Consistency with~$A$):  On average over $\bu \sim \F_q^{m}$,
	  \begin{align*}
	  A^{\mathrm{A},u}_a \otimes I
	  	&\simeq_{\nu} I \otimes G^{\mathrm{B}}_{[g(u)=a]},\\
	I \otimes A^{\mathrm{B},u}_a
	  	&\simeq_{\nu}  G^{\mathrm{A}}_{[g(u)=a]} \ot I.
	  \end{align*}
\item (Self-consistency):
	\begin{equation*}
	G^{\mathrm{A}}_g \ot I \simeq_{\nu} I \ot G^{\mathrm{B}}_g.
	\end{equation*}
\end{enumerate}
\end{theorem}

We note that there is a tradeoff in \Cref{thm:main-formal} specified by the parameter~$k$.
As~$k$ increases, $\nu$'s prefactor $k^2$ increases.
On the other hand, as $k$ increases, the term $e^{- k/(2560000m^2)}$ which also occurs in $\nu$ \emph{decreases}.
Thus, when applying this theorem, one must select~$k$ to balance these competing demands.
Typically, choosing $k = \poly(m)$ should be more than sufficient for applications.
We believe that this parameter~$k$ is an artifact of our proof,
and we hope that it will be removed in the future.


\section{Preliminaries}\label{sec:prelims}

We use \textbf{boldface font} to denote random variables.
For two complex numbers $\alpha, \beta \in \C$, we write
$\alpha \approx_{\eps} \beta$ if
\begin{equation*}
|\alpha - \beta| \leq \eps.
\end{equation*}
We note the following triangle inequality for numbers, which we will use repeatedly:
\begin{equation}\label{eq:triangle-inequality-for-numbers}
\text{if $\alpha \approx_{\eps} \beta$ and $\beta \approx_{\delta} \gamma$, then $\alpha \approx_{\eps+\delta}\gamma$.}
\end{equation}

\subsection{Finite fields}

A \emph{finite field} is a field with a finite number of elements.
There is a unique finite field~$\F_q$ of~$q$ elements for each prime power $q = p^t$,
and there are no other finite fields.
We write $\omega$ for the $p$-th root of unity $\omega = e^{2 \pi i/p}$.

\begin{definition}[Finite field trace]
The \emph{finite field trace} is the function $\tr:\F_q \rightarrow \F_p$ defined as
\begin{equation*}
\tr[x] = \sum_{\ell=0}^{t-1} x^{p^{\ell}}.
\end{equation*}
\end{definition}

\begin{proposition}
\label{prop:fourier-fact-scalar}
Let $a \in \F_q$. Then
\begin{equation*}
\E_{\bx \sim \F_q} \omega^{\tr[\bx \cdot a]}
= \left\{\begin{array}{rl}
	1 & \text{if $a = 0$},\\
	0 & \text{otherwise}.
	\end{array}\right.
\end{equation*}
\end{proposition}
\begin{proof}
If $a = 0$, then $\tr[x \cdot a] = 0$ for all $x \in \F_q$.
As a result,
\begin{equation*}
\E_{\bx \sim \F_q} \omega^{\tr[\bx \cdot a]}
= \E_{\bx \sim \F_q} \omega^{0}
= \E_{\bx \sim \F_q} 1
= 1.
\end{equation*}
On the other hand, if $a \neq 0$, then there exists a $y \in \F_q$
such that $\tr[a \cdot y] \neq 0$. As a result,
\begin{equation*}
C
:= \E_{\bx \sim \F_q} \omega^{\tr[\bx \cdot a]}
= \E_{\bx \sim \F_q} \omega^{\tr[(\bx+y) \cdot a]}
= \E_{\bx \sim \F_q} (\omega^{\tr[\bx \cdot a]} \cdot \omega^{\tr[y \cdot a]})
= C \cdot \omega^{\tr[y\cdot a]}.
\end{equation*}
But because $\tr[y \cdot a] \neq 0$,
$\omega^{\tr[y\cdot a]} \neq 1$.
This implies that $C = 0$.
\end{proof}

\begin{proposition}
\label{prop:fourier-fact-vector}
Let $v \in \F_q^m$. Then
\begin{equation*}
\E_{\bu \sim \F_q^m} \omega^{\tr[\bu \cdot v]}
= \left\{\begin{array}{rl}
	1 & \text{if $v = 0$},\\
	0 & \text{otherwise}.
	\end{array}\right.
\end{equation*}
\end{proposition}
\begin{proof}
By the linearity of the trace,
\begin{equation*}
\E_{\bu \sim \F_q^m} \omega^{\tr[\bu \cdot v]}
= \E_{\bu \sim \F_q^m} (\omega^{\tr[\bu_1 \cdot v_1]} \cdots \omega^{\tr[\bu_m \cdot v_m]})
= \prod_{i=1}^m (\omega^{\tr[\bu_i \cdot v_i]}),
\end{equation*}
which by \Cref{prop:fourier-fact-scalar} is $1$ if $v_i = 0$ for all $i$ and $0$ otherwise.
\end{proof}

\subsection{Polynomials over finite fields}

\begin{definition}[Polynomials of low individual degree]
Let $q$ be a prime power,
and let $m$ and $d$ be nonnegative integers.
We define $\polyfunc{m}{q}{d}$ to be the set of polynomials in $\F_q^m$
with individual degree~$d$.
\end{definition}

\begin{remark}
Note that as defined in \Cref{sec:intro},
in this work a polynomial with individual degree~$d$ is one in which the degree~$d_i$ of each coordinate~$i$
is \emph{at most~$d$}.
This allows us to say ``individual degree~$d$'' rather than the wordier ``individual degree at most~$d$''.
For example, under this definition, $\polyfunc{m}{q}{d}$ is contained in $\polyfunc{m}{q}{d+1}$ for each~$d$.
\end{remark}

The most important fact about low-degree polynomials is that they have large distance from each other.
This is shown by the following lemma.

\begin{lemma}[Schwartz-Zippel lemma~\cite{Sch80,Zip79}]\label{lem:schwartz-zippel-total-degree}
Let $g, h:\F_q^m \rightarrow \F_q$ be two distinct polynomials of total degree~$d$.
Then
\begin{equation*}
\Pr_{\bx \sim \F_q^m}[g(\bx) = h(\bx)] \leq \frac{d}{q}.
\end{equation*}
\end{lemma}

Since any polynomial with individual degree~$d$ has total degree~$md$,
\Cref{lem:schwartz-zippel-total-degree} implies the following corollary.

\begin{corollary}[Schwartz-Zippel for individual degree]
Let $g, h \in \polyfunc{m}{q}{d}$ be distinct. Then
\begin{equation*}
\Pr_{\bx \sim \F_q^m}[g(\bx) = h(\bx)] \leq \frac{md}{q}.
\end{equation*}
\end{corollary}

\subsection{Measurements}

\begin{definition}[Measurements and sub-measurements]
Let $\calH$ be a Hilbert space and $\calA$ be a set of outcomes.
A \emph{sub-measurement}
is a set of Hermitian, positive-semidefinite matrices
$A = \{A_a\}_{a \in \calA}$
acting on~$\calH$
such that $\sum_a A_a \leq I$.
The sub-measurement is \emph{projective} if $(A_a)^2 = A_a$ for each~$a$.
It is a \emph{measurement}
if it satisfies the stronger condition $\sum_a A_a = I$.
\end{definition}

An important class of sub-measurements
are those that output polynomials~$g$ of individual degree~$d$.
These are defined as follows.

\begin{definition}[Low-degree polynomial measurements]
We write $\polysub{m}{q}{d}$ for the set of sub-measurements $G=\{G_g\}$
with outcomes $g \in \polyfunc{m}{q}{d}$.
We write $\polymeas{m}{q}{d}$ for the subset of $\polysub{m}{q}{d}$
containing only those $G = \{G_g\}$ which are measurements.
\end{definition}

\begin{definition}[Post-processing measurements]\label{def:post-processing}
Let $A = \{A_a\}_{a \in \calA}$ be a set of matrices,
and let $f:\calA \rightarrow \calB$ be a function.
Then for each $b \in \calB$,
we define the matrix
\begin{equation*}
A_{[f(a)=b]} = \sum_{a:f(a)=b} A_a.
\end{equation*}
\end{definition}

\begin{remark}
We note that \Cref{def:post-processing}
agrees with the notation for post-processing measurements
used in \cite{NW19},
but it disagrees with the notation used in \cite{JNV+20}.
That work uses the notation ``$A_{[f(\cdot)=b]}$'' rather
than the notation ``$A_{[f(a)=b]}$'' that we use in this work. 
\end{remark}

An easy-to-prove fact is that measurements remain measurements after post-processing their outcomes.

\begin{proposition}
Let $A = \{A_a\}_{a \in \calA}$ be a set of matrices,
and let $f:\calA \rightarrow \calB$ be a function.
Then
\begin{equation*}
\sum_a A_a = \sum_b A_{[f(a)=b]}.
\end{equation*}
Thus, if~$\{A_a\}$ is a sub-measurement (respectively, measurement),
then~$\{A_{[f(a)=b]}\}$ is also a sub-measurement (respectively, measurement).
\end{proposition}

\begin{remark}
We note that there is some ambiguity in the notation $A_{[f(a)=b]}$
because it requires one to know which of~$f$ or~$a$ is the measurement outcome
and which is the function applied to it.
For example, given a sub-measurement $G = \{G_g\} \in \polysub{m}{q}{d}$,
we will often consider evaluating its outputs at a point $u \in \F_q^m$.
This entails looking at the sub-measurement
\begin{equation*}
\{G_{[g(u)=a]}\}_{a \in \F_q},
\quad
\text{where}
\quad
G_{[g(u)=a]} = \sum_{g:g(u) = a} G_g.
\end{equation*}
Here, $g$ is the measurement outcome of~$G$,
and the function being applied maps it to its evaluation on the point~$u$, i.e.\ $g(u)$.
In general, it should always be clear from context
what the measurement outcome and the function being applied to it are.
\end{remark}

\begin{notation}[Measurements indexed by questions]
We will frequently encounter sets of sub-measurements $A^x= \{A^x_{a}\}$ indexed by elements~$x$ from a set of ``questions" $\calX$. We will write $A = \{A^x_a\}$ for this set. 
We will typically refer to this set~$A$ as a sub-measurement,
and we will refer to it as a measurement if each~$A^x$ is a measurement.
In addition, we will refer to it as projective if each~$A^x$ is projective.
\end{notation}

\begin{notation}[Complete part of sub-measurement]
Given a sub-measurement~$A = \{A_a\}$, we will write $A = \sum_a A_a$.
Similarly, given a sub-measurement~$A = \{A^x_a\}$, we will write $A^x = \sum_a A^x_a$ and $A = \E_{\bx} A^{\bx}$.
\end{notation}

The complete part of a sub-measurement
contrasts with its \emph{incomplete part},
which is the matrix $I - A^x$.
We will sometimes
view~$A$ as a measurement
by throwing in its incomplete part as an additional POVM element.
This is formalized in the following definition.

\begin{definition}[Completing a sub-measurement]\label{def:measurement-completion}
Let $A = \{A^x_a\}_{a \in \calA}$ be a sub-measurement.
We define the \emph{completion of~$A$},
denoted $\mathrm{completion}(A)$,
to be the measurement $\widehat{A} = \{\widehat{A}^x_a\}$
with outcome set $\widehat{\calA} = \calA \cup \{\bot\}$ 
such that for each $a \in \widehat{\calA}$,
\begin{equation*}
A^x_a
= \left\{\begin{array}{cl}
		A^x_a & \text{if $a \in \calA$,}\\
		I - A^x & \text{if $a = \bot$}.
		\end{array}\right.
\end{equation*}
\end{definition}

\subsection{Comparing measurements}
\label{sec:comparing-measurements}

An central problem in this paper is recognizing when two measurements are close to each other. 
We will survey two methods of doing so,
using the \emph{consistency} and the \emph{state dependent distance}.
This section largely mirrors Sections 4.4 and 4.5 of \cite{NW19},
which prove numerous properties of these two distances.
However, we are unable to cite their results directly
because they are mostly stated and proven only for measurements,
whereas we will need to apply them to sub-measurements as well.

\subsubsection{Consistency between measurements}

The most basic notion of similarity between two measurements is given by their \emph{consistency}.

\begin{definition}[Consistency]\label{def:simeq}
Let $\ket{\psi}$ be a state in $\calH_{\mathrm{A}} \ot \calH_{\mathrm{B}}$.
Let $A = \{A^x_a\}$ be a sub-measurement acting on $\calH_{\mathrm{A}}$
and $B = \{B^x_a\}$ be a sub-measurement acting on $\calH_{\mathrm{B}}$.
Finally, let $\calD$ be a distribution on the question set~$\calX$.
Then we say that
 \begin{equation*}
 A^x_a \ot I \simeq_{\delta} I \ot B^x_a
 \end{equation*}
  on state $\ket{\psi}$ and distribution $\calD$ if
  \begin{equation}\label{eq:no-big-Oh}
  \E_{\bx \sim \calD} \sum_{a \neq b} \bra{\psi} A^{\bx}_a \ot B^{\bx}_b
    \ket{\psi} \leq \delta.
    \end{equation}
\end{definition}

This is simply the probability that the provers receive different outcomes
when they measure with~$A$ and~$B$,
assuming the sub-measurements do return an outcome.

\begin{notation}[Simplifying notation]
Because the state $\ket{\psi}$ and distribution $\calD$ are typically clear from context,
we will often write ``$ A^x_a \ot I \simeq_{\delta} I \ot B^x_a$"
as shorthand for ``$ A^x_a \ot I \simeq_{\delta} I \ot B^x_a$ on state $\ket{\psi}$ and distribution $\calD$".
In the case when $\calD$ is not clear by context,
we might specify it implicitly in terms of a random variable~$\bx$ distributed according to~$\calD$.
For example, suppose the distribution~$\calD$ is supposed to be uniform on the question set $\F_q$.
Then we might say ``on average over $\bx \sim \F_q$,
\begin{equation*}
A^x_a \ot I \simeq_{\delta} I \ot B^x_a"
\end{equation*}
as shorthand for ``$ A^x_a \ot I \simeq_{\delta} I \ot B^x_a$ on distribution $\calD$".
\end{notation}

We note that there are two differences between the definition of consistency in \Cref{def:simeq}
and the original definition of consistency given in \cite[Definition 4.11]{NW19}.
The first of these is that the~\cite{NW19} definition
allows the right-hand side of \Cref{eq:no-big-Oh} to be $O(\delta)$
rather than strictly~$\delta$.
The benefit of this is that they do not need to keep track of constant prefactors in their proofs;
we have elected to use this more concrete definition
to make our proofs more easily verifiable, at the expense of tracking these constant prefactors.
The second difference is that is that they define their consistency as the quantity
\begin{equation*}
1 - \E_{\bx} \sum_a \bra{\psi} A^{\bx}_a \ot B^{\bx}_a \ket{\psi}.
\end{equation*}
As we show below in \Cref{prop:simeq-for-measurements},
this agrees with \Cref{def:simeq}
when~$A$ and~$B$ are both measurements.
However, these two definitions disagree
when~$A$ and~$B$ are \emph{sub}-measurements,
which is a case we will frequently encounter throughout this paper.

\begin{proposition}[Consistency for measurements]\label{prop:simeq-for-measurements}
Let $A = \{A^x_a\}$ and $B = \{B^x_a\}$ be two measurements.
Then
\begin{equation*}
A^x_a \ot I \simeq_{\delta} I \ot B^x_a
\end{equation*}
if and only if
\begin{equation*}
\E_{\bx} \sum_a \bra{\psi} A^{\bx}_a \ot B^{\bx}_a \ket{\psi} \geq 1 - \delta.
\end{equation*}
\end{proposition}
\begin{proof}
We compute
\begin{align*}
&\E_{\bx} \sum_{a \neq b}\bra{\psi} A^{\bx}_a \ot B^{\bx}_b \ket{\psi}\\
=~& \E_{\bx} \sum_b \bra{\psi} \Big(\sum_{a \neq b} A^{\bx}_a\Big) \ot B^{\bx}_b\ket{\psi}\\
=~&\E_{\bx} \sum_b \bra{\psi} (I - A^{\bx}_b) \ot B^{\bx}_b\ket{\psi}
	\tag{because~$A$ is a measurement}\\
=~&\E_{\bx} \sum_b \bra{\psi} I \ot B^{\bx}_b\ket{\psi}
	- \E_{\bx} \sum_b \bra{\psi} A^{\bx}_b \ot B^{\bx}_b\ket{\psi}\\
=~&1
	- \E_{\bx} \sum_b \bra{\psi} A^{\bx}_b \ot B^{\bx}_b\ket{\psi}. \tag{because~$B$ is a measurement}
\end{align*}
Hence, the first expression is at most~$\delta$ if and only if the last one is as well,
and we are done.
\end{proof}

\subsubsection{The state-dependent distance}

Suppose we have three measurements $\{A^x_a\}$, $\{B^x_a\}$, and $\{C^x_a\}$,
and we know that
\begin{equation*}
A^x_a \ot I \simeq_{\eps} I \ot C^x_a.
\end{equation*}
What property of~$A$ and~$B$ allows us to conclude that
\begin{equation}\label{eq:can-we-use-approx-delta-to-derive-this}
B^x_a \ot I \simeq_{\eps} I \ot C^x_a?
\end{equation}
The answer is provided by the \emph{state-dependent distance}.

\begin{definition}[The state-dependent distance]\label{def:approx_delta}
Let $\ket{\psi}$ be a state in $\calH$.
Let $A = \{A^x_a\}$ and $B = \{B^x_a\}$ be
sets of matrices acting on $\calH$.
Finally, let $\calD$ be a distribution on the question set~$\calX$.
Then we say that
  \begin{equation*}
  A^x_a \approx_{\delta}  B^x_a
  \end{equation*}
 on state $\ket{\psi}$ and distribution $\calD$ if
  \begin{equation*}
  \E_{\bx \sim \calD} \sum_{a} \Vert (A^{\bx}_a - B^{\bx}_a) \ket{\psi} \Vert^2 \leq \delta.
    \end{equation*}
\end{definition}

We note one odd feature of \Cref{def:approx_delta} in comparison to the consistency,
which is that $\ket{\psi}$ is not assumed to have a bipartition $\calH_{\mathrm{A}} \ot \calH_{\mathrm{B}}$
in which~$A$ and~$B$ are applied on opposite sides.
We will address this in \Cref{sec:consistency-from-state-dependent-distance} below.
Before doing so, we will answer our question above,
even in the case when~$C$ is allowed to be a sub-measurement. 

\begin{proposition}[Transfering ``$\simeq$" using ``$\approx$"]
Let $\{A^x_a\}$ and $\{B^x_a\}$ be measurements,
and let $\{C^x_a\}$ be a sub-measurement.
Suppose that $A^x_a \ot I \simeq_{\delta} I \ot C^x_a$
and $A^x_a \ot I \approx_{\eps} B^x_a \ot I$.
Then $B^x_a \ot I \simeq_{\delta + \sqrt{\eps}} I \ot C^x_a$.
  \label{prop:triangle-sub}
\end{proposition}

\begin{proof}
  First, we can
  rewrite the consistency between $A$ and $C$ as
  \begin{align*}
    \delta \geq \E_{\bx} \sum_{a \neq b} \bra{\psi} A^{\bx}_a \ot C^{\bx}_b
    \ket{\psi} 
    & = \E_{\bx} \sum_{a} \bra{\psi} A^{\bx}_a \ot (C^{\bx} - C^{\bx}_a) \ket{\psi} \\
    & =  \E_{\bx} \sum_{a} \bra{\psi} A^{\bx}_a \ot C^{\bx} \ket{\psi}
    			- \E_{\bx} \sum_{a} \bra{\psi} A^{\bx}_a \ot C^{\bx}_a \ket{\psi} \\
    & =  \E_{\bx} \bra{\psi} I \ot C^{\bx} \ket{\psi}
    			- \E_{\bx} \sum_{a} \bra{\psi} A^{\bx}_a \ot C^{\bx}_a \ket{\psi} \tag{because~$A$ is a measurement}\\
    & = \bra{\psi} I \ot C \ket{\psi} - \E_{\bx} \sum_a
                 \bra{\psi} A^{\bx}_a \ot C^{\bx}_a \ket{\psi}.
  \end{align*}
  Likewise, we can rewrite the
  inconsistency between $B$ and $C$ as
  \begin{align*}
    \E_{\bx} \sum_{a \neq b} \bra{\psi} B^{\bx}_a \ot C^{\bx}_b
    \ket{\psi} &= \bra{\psi} I \ot C \ket{\psi} - \E_{\bx} \sum_a
                 \bra{\psi} B^{\bx}_a \ot C^{\bx}_a \ket{\psi}.
  \end{align*}
  We want to show that the inconsistency between~$B$ and~$C$ is close to the inconsistency between $A$
  and $B$. In particular, we claim that
  \begin{equation*}
  \E_{\bx} \sum_{a \neq b} \bra{\psi} A^{\bx}_a \ot C^{\bx}_b
    \ket{\psi}
    \approx_{\sqrt{\eps}} \E_{\bx} \sum_{a \neq b} \bra{\psi} B^{\bx}_a \ot C^{\bx}_b
    \ket{\psi}.
  \end{equation*}
  To show this, we bound the magnitude of the difference using Cauchy-Schwarz.
  \begin{align*}
   & \Big| \E_{\bx} \sum_{a } \bra{\psi} (A^{\bx}_a - B^{\bx}_a) \ot
                     C^{\bx}_a\ket{\psi} \Big| \nonumber
    \\
    \leq~& \sqrt{ \E_{\bx} \sum_{a} \bra{\psi}
                            (A^{\bx}_a - B^{\bx}_a)^2 \ot I \ket{\psi}}
                            \cdot \sqrt{\E_{\bx} \sum_{a} \bra{\psi} I \ot (C^{\bx}_a)^2 \ket{\psi}} \\
    \leq~& \sqrt{\eps} \cdot 1. \tag{because~$C$ is a sub-measurement}
  \end{align*}
  This completes the proof.
\end{proof}

Next, we show that in the case of measurements,
the state-dependent distance is a weakening of the consistency.

\begin{proposition}[``$\simeq$" implies ``$\approx$" for measurements]\label{prop:simeq-to-approx}
Let $A = \{A^x_a\}$ and $B = \{B^x_a\}$ be two measurements such that
\begin{equation*}
A^x_a \ot I \simeq_{\delta} I \ot B^x_a.
\end{equation*}
Then
\begin{equation*}
A^x_a \ot I \approx_{2\delta} I \ot B^x_a.
\end{equation*}
This is an ``if and only if'' if~$A$ and~$B$ are both projective measurements.
\end{proposition}
\begin{proof}
Our goal is to bound
\begin{align*}
&\E_{\bx} \sum_a\Vert (A^{\bx}_a \ot I - I \ot B^{\bx}_a) \ket{\psi} \Vert^2\\
=~&\E_{\bx} \sum_a \bra{\psi} (A^{\bx}_a \ot I - I \ot B^{\bx}_a)^2 \ket{\psi}\\
=~&\E_{\bx} \sum_a \bra{\psi} (A^{\bx}_a)^2 \ot I \ket{\psi}
	+\E_{\bx} \sum_a \bra{\psi} I \ot (B^{\bx}_a)^2 \ket{\psi}
	- 2 \E_{\bx} \sum_a \bra{\psi} A^{\bx}_a \ot B^{\bx}_a \ket{\psi}\\
\leq~&\E_{\bx} \sum_a \bra{\psi} A^{\bx}_a \ot I \ket{\psi}
	+\E_{\bx} \sum_a \bra{\psi} I \ot B^{\bx}_a \ket{\psi}
	- 2 \E_{\bx} \sum_a \bra{\psi} A^{\bx}_a \ot B^{\bx}_a \ket{\psi}\\
=~&2
	- 2 \E_{\bx} \sum_a \bra{\psi} A^{\bx}_a \ot B^{\bx}_a \ket{\psi} \tag{because~$A$ and~$B$ are measurements}\\
\leq~& 2 - 2 \cdot(1-\delta) \tag{by~\Cref{prop:simeq-for-measurements} and the fact that~$A$ and~$B$ are measurements}\\
=~&2 \delta.
\end{align*}
This completes the proof.
When  $A$ and~$B$ are projective, ``if and only if'' 
follows from the first inequality becoming an equality.
\end{proof}

\begin{remark}
We note that \Cref{prop:simeq-to-approx} certainly does \emph{not} hold for sub-measurements.
For example, if $A^x_a = 0$ for all~$a$,
then $A_a^x \ot I \simeq_{0} I \ot B_a^x$,
but
\begin{equation*}
\E_{\bx} \sum_a \Vert (A^{\bx}_a \ot I - I \ot B^{\bx}_a)\ket{\psi} \Vert^2
= \E_{\bx} \sum_a \Vert (I \ot B^{\bx}_a)\ket{\psi} \Vert^2,
\end{equation*}
which is nonzero unless $(I \ot B^x_a) \ket{\psi} = 0$ for all~$x$ and~$a$.
\end{remark}

\subsubsection{Deriving consistency relations from the state-dependent distance}
\label{sec:consistency-from-state-dependent-distance}

\Cref{prop:triangle-sub} shows the purpose of the state-dependent distance,
which is to derive new consistency relations from old ones.
In addition, its proof uses a strategy
which will recur frequently throughout this paper.
In this strategy, we would like to demonstrate a sequence of 
expressions in which each expression is close to the previous one:
\begin{equation*}
\E_{\bx} \sum_a \bra{\psi} (A_0)^{\bx}_a (B_0)^{\bx}_a \ket{\psi}
\approx_{\eps_1} \E_{\bx} \sum_a \bra{\psi} (A_1)^{\bx}_a (B_1)^{\bx}_a \ket{\psi}
\approx_{\eps_2} \cdots
\approx_{\eps_t} \E_{\bx} \sum_a \bra{\psi} (A_t)^{\bx}_a (B_t)^{\bx}_a \ket{\psi}.
\end{equation*}
By the triangle inequality, we can therefore conclude that the $0$-th expression
is close to the $t$-th expression.
To show that the $i$-th quantity is close to the $(i+1)$-st,
we will typically arrange for $A_{i} = A_{i+1}$,
and we will swap out $B_i$ for $B_{i+1}$
using an approximation relation such as $(B_i)^x_a \approx (B_{i+1})^x_a$,
with the help of the Cauchy-Schwarz inequality
(or the same might occur with the roles of~$A$ and~$B$ reversed). 
Even if $A_0$ and $B_0$ can be written as local measurements applied to either side of a bipartition,
e.g.\ $(A_0)^x_a = A^x_a \ot I$ and $(B_0)^x_a = I \ot B^x_a$,
and likewise for $A_t$ and $B_t$,
the intermediate steps may feature matrices which do not decompose nicely across a bipartition.
This is why \Cref{def:approx_delta} is phrased so broadly,
with no mention of a bipartition.

In general, the $A_i$'s and $B_i$'s encountered in this sequence of steps
may be quite unstructured: for example, not sub-measurements,
and possibly not even Hermitian.
Thus, we are interested in determining which conditions
to place on these measurements are sufficient to carry out this proof strategy.
The following proposition gives a broad condition under which this can be accomplished.

\begin{proposition}
\label{prop:closeness-of-ip}
	Let $\{A^x_a\}$, $\{B^x_a\}$, and $\{C^x_{a,b} \}$ be matrices.
	Suppose that  $A^x_a \approx_\gamma B^x_a$ and that for all~$x$, $\sum_a (\sum_b C^x_{a,b}) (\sum_b C^x_{a,b})^\dagger  \leq I$. Then
	\begin{equation}
	\E_{\bx} \sum_{a,b} \bra{\psi} C^{\bx}_{a,b} A^{\bx}_a  \ket{\psi} \approx_{\sqrt{\gamma}} \E_{\bx} \sum_{a,b} \bra{\psi} C^{\bx}_{a,b} B^{\bx}_a  \ket{\psi} \,. \label{eq:closeness3}
	\end{equation}
	Similarly, suppose that  $(A^x_a)^\dagger \approx_\gamma (B^x_a)^\dagger$ and that for all~$x$, $\sum_a (\sum_b C^x_{a,b})^\dagger (\sum_b C^x_{a,b}) \leq I$. Then 
	\begin{equation}
	\E_{\bx} \sum_{a,b} \bra{\psi} A^{\bx}_a  C^{\bx}_{a,b} \ket{\psi} \approx_{\sqrt{\gamma}} \E_{\bx} \sum_{a,b} \bra{\psi} B^{\bx}_a  C^{\bx}_{a,b}  \ket{\psi} \,. \label{eq:closeness4}
	\end{equation}	
\end{proposition}
\begin{proof}
	We begin by showing~\Cref{eq:closeness3}.
	\begin{align*}
		&\Big | \E_{\bx} \sum_{a,b} \bra{\psi} C^{\bx}_{a,b} (A^{\bx}_a - B^{\bx}_a) \ket{\psi} \Big | \\
		&= \Big | \E_{\bx} \sum_{a} \bra{\psi} \Big ( \sum_b C^{\bx}_{a,b} \Big ) \cdot (A^{\bx}_a - B^{\bx}_a) \ket{\psi} \Big |\\
		&\leq \Big ( \E_{\bx} \sum_a \bra{\psi} \Big(\sum_b C^{\bx}_{a,b}\Big) \Big(\sum_b C^{\bx}_{a,b}\Big)^\dagger \ket{\psi} \Big )^{1/2} \cdot \Big ( \E_{\bx} \sum_a \bra{\psi} (A^{\bx}_a - B^{\bx}_a)^\dagger  (A^{\bx}_a - B^{\bx}_a) \ket{\psi} \Big)^{1/2} \\
		&\leq \sqrt{\gamma}.
	\end{align*}
	The third line uses Cauchy-Schwarz, and the fourth line uses the assumption $\sum_a (\sum_b C^x_{a,b}) (\sum_b C^x_{a,b})^\dagger  \leq I$ to bound the first factor by $1$ and  the assumption $A^x_a \approx_\gamma B^x_a$ to bound the second factor by $\sqrt{\gamma}$.
	As for \Cref{eq:closeness4}, we want to bound
	\begin{equation*}
	\Big | \E_{\bx} \sum_{a,b} \bra{\psi} (A^{\bx}_a - B^{\bx}_a) C^{\bx}_{a,b}  \ket{\psi} \Big |
	= \Big | \E_{\bx} \sum_{a,b} \bra{\psi} (C^{\bx}_{a,b})^\dagger ((A^{\bx}_a)^\dagger - (B^{\bx}_a)^\dagger) \ket{\psi} \Big |.
      \end{equation*}
      It then follows from \Cref{eq:closeness3} that this is at most $\sqrt{\gamma}$.
    \end{proof}
    
\Cref{prop:closeness-of-ip} is broad enough to capture
almost all of our applications of the state-dependent distance.
Unfortunately, defining the $C^x_{a, b}$ matrices and showing that 
they satisfy the inequality $\sum_a (\sum_b C^x_{a,b}) (\sum_b C^x_{a,b})^\dagger  \leq I$
can be somewhat cumbersome.
As a result, we will usually carry out these Cauchy-Schwarz calculations by hand.
However, we will occasionally use the following proposition
which simplifies \Cref{prop:closeness-of-ip}.

\begin{proposition}\label{prop:easy-approx-from-approx-delta}
Let $A = \{A^x_a\}$, $B = \{B^x_a\}$, and $C = \{C^x_a\}$ be sub-measurements
such that $A^x_a \approx_{\delta} B^x_a$. Then
\begin{equation*}
\E_{\bx} \sum_a \bra{\psi} A^{\bx}_a  C^{\bx}_a \ket{\psi}
\approx_{\sqrt{\delta}} \E_{\bx} \sum_a \bra{\psi} B^{\bx}_a  C^{\bx}_a \ket{\psi}.
\end{equation*}
\end{proposition}
\begin{proof}
To show this, we bound the magnitude of the difference.
\begin{align*}
&\Big| \E_{\bx} \sum_a \bra{\psi} (A^{\bx}_a - B^{\bx}_a) \cdot (C^{\bx}_a) \ket{\psi}\Big|\\
\leq~& \sqrt{\E_{\bx} \sum_a \bra{\psi} (A^{\bx}_a - B^{\bx}_a)^2 \ket{\psi}}
	\cdot \sqrt{\E_{\bx} \sum_a \bra{\psi} (C^{\bx}_a)^2 \ket{\psi}}\\
\leq~& \sqrt{\delta} \cdot \sqrt{1}.
\end{align*}
This completes the proof.
\end{proof}

A proposition similar to \Cref{prop:closeness-of-ip},
but for  ``$\approx$'',
holds as well.
    This is \cite[Fact~$4.20$]{NW19}.
    \begin{proposition}
      \label{prop:cab-approx-delta}
      Let $\{A^x_a\}, \{B^x_a\},$ and $\{C^x_{a,b}\}$ be
      matrices. Suppose that $A^{x}_a \approx_\delta B^{x}_a$ and that
      for all $x$ and $a$, $\sum_b (C^{x}_{a,b})^\dagger
      (C^{x}_{a,b}) \leq I$. Then
      \[ C^{x}_{a,b} A^x_{a} \approx_{\delta} C^{x}_{a,b} B^{x}_a. \]
    \end{proposition}
    \begin{proof}
      The error we wish to bound is
      \begin{align*}
        \E_{\bx} \sum_{a,b} \bra{\psi} 
        (A^{\bx}_a - B^{\bx}_a)^\dagger
        (C^{\bx}_{a,b})^\dagger C^{\bx}_{a,b} (A^{\bx}_a - 
        B^{\bx}_a) \ket{\psi} &\leq \E_{\bx} \sum_{a} \bra{\psi}
                                (A^{\bx}_a - B^{\bx}_a)^\dagger
                                (A^{\bx}_a - B^{\bx}_a) \ket{\psi} \\
                              &\leq \delta. \qedhere
      \end{align*}
    \end{proof}

\subsubsection{Miscellaneous distance properties}

We now state a few miscellaneous properties of our two distances.

\begin{proposition}[Triangle inequality for vectors squared]\label{prop:triangle-inequality-for-vectors-squared}
Let $\ket{\psi_1}, \ldots, \ket{\psi_k}$ be vectors.
Then
\begin{equation*}
\Vert \ket{\psi_1} + \cdots + \ket{\psi_k} \Vert^2 \leq k \cdot(\Vert \ket{\psi_1}\Vert^2 + \cdots + \Vert \ket{\psi_k}\Vert^2).
\end{equation*}
\end{proposition}
\begin{proof}
First, if $x_1, \ldots, x_k \in \R$, then
\begin{equation}\label{eq:prop-for-real-numbers}
(x_1 + \cdots + x_k)^2 
= \sum_{i, j= 1}^k x_i x_j
\leq \sum_{i, j= 1}^k \frac{x_i^2 + x_j^2}{2}
= \sum_{i, j= 1}^k x_i^2
= \sum_{i=1}^k k \cdot x_i^2.
\end{equation}
Next, by the triangle inequality
\begin{align*}
\Vert \ket{\psi_1} + \cdots + \ket{\psi_k} \Vert^2
& = (\Vert \ket{\psi_1} + \cdots + \ket{\psi_k} \Vert)^2\\
& \leq (\Vert \ket{\psi_1}\Vert + \cdots + \Vert \ket{\psi_k} \Vert)^2\\
& \leq k \cdot(\Vert \ket{\psi_1}\Vert^2 + \cdots + \Vert \ket{\psi_k}\Vert^2).
\end{align*}
where the last step uses \Cref{eq:prop-for-real-numbers} applied to the case of $x_i = \Vert \ket{\psi_i}\Vert$.
\end{proof}

\begin{proposition}[Triangle inequality for ``$\approx_{\delta}$"]
\label{prop:triangle-inequality-for-approx_delta}
Suppose $A_1 = \{(A_1)^x_a\}, \ldots, A_{k+1} = \{(A_{k+1})^x_a\}$ is a set of matrices such that
\begin{equation*}
(A_i)^x_a \approx_{\delta_i} (A_{i+1})^x_a
\end{equation*}
for all $i \in [k]$. Then
\begin{equation*}
(A_1)^x_a \approx_{k \cdot (\delta_1 + \cdots + \delta_{k})} (A_{k+1})^x_a.
\end{equation*}
\end{proposition}
\begin{proof}
We want to bound
\begin{align*}
&~\E_\bx \sum_a \Vert ((A_1)^{\bx}_a - (A_{k+1})^{\bx}_a) \ket{\psi}\Vert^2\\
= &~\E_\bx \sum_a \Vert (((A_1)^{\bx}_a - (A_2)^{\bx}_a) + \cdots.+ ((A_{k})^{\bx}_a - (A_{k+1})^{\bx}_a))\ket{\psi}\Vert^2\\
\leq &~\E_\bx \sum_a k\cdot(\Vert ((A_1)^{\bx}_a - (A_2)^{\bx}_a) \ket{\psi}\Vert^2 + \cdots.+ \Vert((A_k)^{\bx}_a - (A_{k+1})^{\bx}_a)\ket{\psi}\Vert^2)\\
\leq &~k\cdot(\delta_1 + \cdots + \delta_{k}),
\end{align*}
where the inequality uses \Cref{prop:triangle-inequality-for-vectors-squared} applied to the vectors $((A_i)^{\bx}_a - (A_{i+1})^{\bx}_a) \ket{\psi}$ for $i \in [k]$.
\end{proof}

We note that \Cref{prop:triangle-inequality-for-approx_delta}
contrasts with the triangle inequality for ``$\approx_{\delta}$'' when applied to numbers,
i.e.\ \Cref{eq:triangle-inequality-for-numbers},
for which no multiplicative factor of~$k$ appears in the error.

The following is Fact~$4.29$ from~\cite{NW19};
however, they incorrectly claimed a final bound of $A^x_a \ot I \simeq_{\eps +\delta + \gamma} I \ot D^x_a$.
We give a new proof of this statement, albeit with a slightly weaker quantitative bound.

\begin{proposition}[Triangle inequality for ``$\simeq$"]\label{prop:simeq-triangle-inequality}
		Suppose that $A$, $B$, $C$, and~$D$ are measurements such that
		\begin{equation*}
		A^x_a \ot I \simeq_{\eps} I \ot B^x_a,
		\quad
		C^x_a \ot I \simeq_{\delta} I \ot B^x_a,
		\quad
		C^x_a \ot I \simeq_{\gamma} I \ot D^x_a.
		\end{equation*}
		Then
		\begin{equation*}
		A^x_a \ot I \simeq_{\eps + 2\sqrt{\delta + \gamma}} I \ot D^x_a.
		\end{equation*}
		\end{proposition}
\begin{proof}
Because $B$, $C$, and~$D$ are measurements, \Cref{prop:simeq-to-approx} implies that
\begin{equation*}
C^x_a \ot I \approx_{2\delta} I \ot B^x_a,
\quad
C^x_a \ot I \approx_{2\gamma} I \ot D^x_a.
\end{equation*}
The triangle inequality, \Cref{prop:triangle-inequality-for-approx_delta}, then implies that
\begin{equation*}
I \ot B^x_a \approx_{4 \delta + 4 \gamma} I \ot D^x_a.
\end{equation*}
Finally, \Cref{prop:triangle-sub} implies that
\begin{equation*}
A^x_a \ot I
\simeq_{\eps + \sqrt{4 \delta + 4 \gamma}} I \ot D^x_a.
\end{equation*}
This completes the proof.
\end{proof}

\begin{proposition}[Data processing for ``$\simeq$'']
\label{prop:simeq-data-processing}
Let $A = \{A^x_a\}$ and $B = \{B^x_a\}$ be two measurements such that
\begin{equation*}
A^x_a \ot I \simeq_{\delta} I \ot B^x_a.
\end{equation*}
Then for any function~$f$,
\begin{equation*}
A^x_{[f(a)=b]} \ot I \simeq_{\delta} I \ot B^x_{[f(a)=b]}.
\end{equation*}
\end{proposition}
\begin{proof}
We want to bound
\begin{align*}
\E_{\bx} \sum_{b \neq b'} \bra{\psi} A^{\bx}_{[f(a)=b]} \ot B^{\bx}_{[f(a)=b']}\ket{\psi}
& = \E_{\bx} \sum_{b \neq b'} \sum_{a: f(a) = b} \sum_{a': f(a') = b'} \bra{\psi} A^{\bx}_{a} \ot B^{\bx}_{a'}\ket{\psi}\\
& \leq \E_{\bx} \sum_{a \neq a'} \bra{\psi} A^{\bx}_{a} \ot B^{\bx}_{a'}\ket{\psi}\\
& \leq \delta.
\end{align*}
This completes the proof.
\end{proof}
    
The following fact is useful for translating between statements about consistency and closeness between sub-measurements.
\begin{proposition}
\label{prop:cons-sub-meas} 
Let $\{A^x_a\}$ be a sub-measurement and let $\{B^x_a\}$ be a measurement such that on average over $x$,
\[
	A^x_a \ot I \simeq_\gamma I \ot B^x_a \,.
\]
Then the following hold
	\begin{gather}
		A^x_a \ot I \approx_{\gamma} A^x_a \ot B^x_a \approx_{\gamma} A^x \ot B^x_a,
		\label{eq:closeness5}
	\end{gather}
	where $A^x = \sum_a A^x_a$.
	As a result, by \Cref{prop:triangle-inequality-for-approx_delta},
	\begin{equation*}
	A^x_a \ot I \approx_{4\gamma} A^x \ot B^x_a
	\end{equation*}
\end{proposition}
\begin{proof}
	We establish the first approximation in \Cref{eq:closeness5}:
	\begin{align*}
		&\E_{\bx} \sum_a \bra{\psi} \Big (A^{\bx}_a \ot (I - B^{\bx}_a) \Big)^2 \ket{\psi}  \\
		&\leq \E_{\bx} \sum_a \bra{\psi} A^{\bx}_a \ot (I - B^{\bx}_a) \ket{\psi} \\
		&= \E_{\bx} \sum_{\substack{a,b : \\ b \neq a}} \bra{\psi} A^{\bx}_a \ot B^{\bx}_b \ket{\psi} \\
		&\leq \gamma \,.
	\end{align*}
	The first inequality folows from the fact that $A^x_a \ot (I - B^x_a)$ has operator norm at most $1$, and the third line follows from the fact that $\{B^x_b\}$ is a complete measurement, and the last line follows from the assumption of consistency between the $A$ and $B$ (sub-)measurements.
	
	To establish the second approximation in \Cref{eq:closeness5}, we compute the difference:
	\begin{align*}
	&\E_{\bx} \sum_a \bra{\psi} \Big ( (A^{\bx} - A^{\bx}_a) \ot B^{\bx}_a \Big)^2 \ket{\psi} \\
	&\leq  \E_{\bx} \sum_a \bra{\psi} (A^{\bx} - A^{\bx}_a) \ot B^{\bx}_a  \ket{\psi} \\
	&= \E_{\bx} \sum_{\substack{a,a' : \\ a \neq a'}} \bra{\psi} A^{\bx}_{a'} \ot B^{\bx}_a  \ket{\psi} \\
	&\leq \gamma \,.
	\end{align*}
	The second line follows from the fact that  $(A^x - A^x_a) \ot B^x_a$ has operator norm at most $1$, and the last inequality follows from the consistency between the $A$ and $B$ (sub-)measurements.
\end{proof}

\begin{proposition}  \label{prop:switch-sandwich}
  Suppose $\{A^x_a\}$ is a projective sub-measurement
  satisfying
  \begin{equation}
    A^x_a \ot I \approx_{\delta} I \ot A^x_a. \label{eq:Aapproxd}
  \end{equation}
  Then for any $0 \leq B \leq I$, the following holds:
  \begin{equation}
    \E_{\bx}\sum_a \bra{\psi} A^{\bx}_a B A^{\bx}_a \ot I \ket{\psi}
    \approx_{2\sqrt{\delta}} \E_{\bx} \sum_a \bra{\psi} B
    \ot A^{\bx}_a
    \ket{\psi} \approx_{\sqrt{\delta}} \E_{\bx} \sum_a \bra{\psi}
    BA^{\bx}_a \ot I \ket{\psi} \label{eq:switch-sandwich} \end{equation}

\end{proposition}
\begin{proof}
  We will show the first approximation in \Cref{eq:switch-sandwich} in two steps. First, we show that
  \begin{equation} \label{eq:shift-right-A}
    \E_{\bx} \sum_a \bra{\psi} A^{\bx}_a B A^{\bx}_a \ot I \ket{\psi} \approx_{\sqrt{\delta}}
    \E_{\bx} \sum_a \bra{\psi} A^{\bx}_a B \ot A^{\bx}_a \ket{\psi}.
  \end{equation}
  To do so, we bound the magnitude of the difference.
  \begin{align*}
   &\Big| \E_{\bx}\sum_a \bra{\psi} (A^{\bx}_a B \otimes I)\cdot (A^{\bx}_a \ot I  - I \ot A^{\bx}_a)  \ket{\psi} \Big|\\
    \leq~&
    \Big( \E_{\bx} \sum_a \bra{\psi} A^{\bx}_a
                 B^2
                 A^{\bx}_a \otimes I \ket{\psi} \Big)^{1/2}  
                \cdot \Big(\E_{\bx} \sum_a
                 \bra{\psi} (A^{\bx}_a \ot I - I \ot A^{\bx}_a)^2
                 \ket{\psi} \Big)^{1/2} \\
               \leq~& \sqrt{\delta}.\tag{because $B \leq I$ and \eqref{eq:Aapproxd}}
  \end{align*}
  Next, we show that
  \begin{equation*}
  \eqref{eq:shift-right-A} =  \E_{\bx} \sum_a \bra{\psi} A^{\bx}_a B \ot A^{\bx}_a \ket{\psi} \approx_{\sqrt{\delta}} \E_{\bx} \sum_ a\bra{\psi} B \ot A^{\bx}_a \ket{\psi}.
  \end{equation*}
  To do so, we bound the magnitude of the difference.
  \begin{align*}
    &\Big|\E_{\bx} \sum_a \bra{\psi} (A^{\bx}_a \ot I - I \ot A^{\bx}_a) \cdot(B
    \ot A^{\bx}_a) \ket{\psi}\Big| \tag{because~$A$ is projective}\\ \leq~& \Big( \E_{\bx} \sum_a \bra{\psi}
                               (A^{\bx}_a \ot I - I \ot A^{\bx}_a)^2
                               \ket{\psi} \Big)^{1/2}  
                              \cdot \Big(\E_{\bx} \sum_a
                               \bra{\psi} B^2 \ot A^{\bx}_a \ket{\psi}
                               \Big)^{1/2} \\
                             \leq~& \sqrt{\delta}.\tag{because $B \leq I$ and \eqref{eq:Aapproxd}}
  \end{align*}
  Thus, the first approximation in \Cref{eq:switch-sandwich}
  follows. To show the second approximation, we bound the magnitude of the difference.
  \begin{align*}
  &\Big|\E_{\bx}\sum_a \bra{\psi} (B\otimes I)  \cdot ( A^{\bx}_a \ot I - I \ot A^{\bx}_a )
    \ket{\psi}\Big|\\
    =~&\Big|\E_{\bx} \bra{\psi} (B\otimes I)\cdot \sum_a ( A^{\bx}_a \ot I - I \ot A^{\bx}_a )
    \ket{\psi}\Big|\\
    \leq~& \Big( \E_{\bx} \bra{\psi} B^2 \otimes I \ket{\psi} \Big)^{1/2}
    \cdot \Big( \E_{\bx} \sum_{a,b} \bra{\psi}
    (A^{\bx}_a \ot I - I \ot A^{\bx}_a) \cdot(A^{\bx}_b \ot I
    - I \ot A^{\bx}_b) \ket{\psi} \Big)^{1/2}.
  \end{align*}
  The first term in the product is at most $1$ because $B \leq I$.
  We bound the expression inside the second square root as follows.
  \begin{align*}
    &\E_{\bx} \sum_{a,b} \bra{\psi}    (A^{\bx}_a \ot I - I \ot A^{\bx}_a) \cdot(A^{\bx}_b \ot I
    - I \ot A^{\bx}_b) \ket{\psi}  \\
    =~&\E_{\bx} \sum_{a, b} \bra{\psi}(A^{\bx}_a A^{\bx}_b \ot I + I \otimes A^{\bx}_a A^{\bx}_b - A^{\bx}_a \otimes A^{\bx}_b - A^{\bx}_b \otimes A^{\bx}_a)\ket{\psi}\\
    =~& \E_{\bx}\sum_a
                                                \bra{\psi} (A^{\bx}_a)^2
                                                \ot I \ket{\psi} +
                                                 \E_{\bx}\sum_a
                                                \bra{\psi} I \ot (A^{\bx}_a)^2 \ket{\psi} -
                                                 2\cdot\E_{\bx} \sum_{a,b} \bra{\psi}A^{\bx}_a
                                                \ot A^{\bx}_b
                                                \ket{\psi} \tag{because~$A$ is projective}
                                                \\
    \leq~& \E_{\bx}\sum_a
                                                \bra{\psi} (A^{\bx}_a)^2
                                                \ot I \ket{\psi} +
                                                 \E_{\bx}\sum_a
                                                \bra{\psi} I \ot (A^{\bx}_a)^2 \ket{\psi} -
                                                 2\cdot\E_{\bx} \sum_{a} \bra{\psi}A^{\bx}_a
                                                \ot A^{\bx}_a
                                                \ket{\psi}
                                                \\
       =~& \E_{\bx}\sum_a
                                                \bra{\psi} (A^{\bx}_a
                                                \ot I  - I \ot A^{\bx}_a)^2
                                                \ket{\psi}
       \\
   \leq~& \delta,
  \end{align*}
  where the last step uses \Cref{eq:Aapproxd}.
\end{proof}

\begin{proposition}\label{prop:completeness-transfer-projective-P}
Let $A = \{A^x_a\}$ be a sub-measurement 
and let $P$ be a projective sub-measurement such that $A^x_a \ot I \approx_{\eps} P_a^x \ot I$.
Then
\begin{equation*}
\bra{\psi} A \ot I \ket{\psi}
\geq  \bra{\psi} P \ot I \ket{\psi} - 2\sqrt{\eps}.
\end{equation*}
\end{proposition}
\begin{proof}
We calculate:
\begin{align*}
\bra{\psi} P \ot I \ket{\psi}
&= \E_{\bx} \sum_a \bra{\psi} P^{\bx}_a \ot I \ket{\psi} \\
&= \E_{\bx} \sum_a \bra{\psi} (P^{\bx}_a)^2 \ot I \ket{\psi} \tag{because~$P$ is projective}\\
&\approx_{\sqrt{\eps}} \E_{\bx} \sum_a \bra{\psi} (A^{\bx}_a \cdot P^{\bx}_a) \ot I \ket{\psi}
									\tag{by \Cref{prop:easy-approx-from-approx-delta}} \\
&\approx_{\sqrt{\eps}} \E_{\bx} \sum_a \bra{\psi} (A^{\bx}_a)^2 \ot I \ket{\psi}
									\tag{by \Cref{prop:easy-approx-from-approx-delta}} \\
&\leq \E_{\bx} \sum_a \bra{\psi} A^{\bx}_a \ot I \ket{\psi}\\
& = \bra{\psi} A \ot I \ket{\psi}.
\end{align*}
This completes the proof.
\end{proof}

\subsubsection{Strong self-consistency}

An important property of a sub-measurement~$A = \{A^x_a\}$
is that if both provers measure using~$A$,
then they receive the same outcome.
It seems natural to study this using \emph{self-consistency} of~$A$, i.e. the number~$\delta$ such that
\begin{equation*}
A^x_a \ot I \simeq_{\delta} I \ot A^x_a.
\end{equation*}
However, when~$A$ is a sub-measurement,
being $\delta$-self consistent
only implies the following weaker condition:
if both provers measure using~$A$ and one of them receives~$\ba$,
then the other will most likely either receive~$\ba$ \emph{or not receive any outcome whatsoever}.
This motivates defining the following stronger notion of self-consistency.

\begin{definition}[Strong self consistency]\label{def:strong-self-consistency}
Let $\ket{\psi}$ be a permutation-invariant state in $\calH \ot \calH$
and let $A = \{A^x_a\}$ be a sub-measurement acting on~$\calH$.
Then $A$ is \emph{$\delta$-strongly self consistent} if
\begin{equation*}
\E_{\bx} \sum_a \bra{\psi} A^{\bx}_a \otimes A^{\bx}_a \ket{\psi}
\geq
\bra{\psi} A \otimes I \ket{\psi} -  \delta.
\end{equation*}
\end{definition}

We now relate strong self-consistency to our two notions of similarity.
First, we show that strong self-consistency is indeed
a stronger condition than $A^x_a \ot I \simeq_{\delta} I \ot A^x_a$,
at least for sub-measurements.

\begin{proposition}\label{prop:other-two-notions-of-self-consistency}
Let $\ket{\psi}$ be a permutation-invariant state,
and let $A = \{A^x_a\}$ be a sub-measurement. If
\begin{equation*}
\E_{\bx} \sum_a \bra{\psi} A^{\bx}_a \otimes A^{\bx}_a \ket{\psi} \geq \bra{\psi} A \otimes I \ket{\psi} - \delta
\end{equation*}
then $A^x_a \otimes I \simeq_{\delta} I \otimes A^x_a$.
This is an ``if and only if" if $A$ is a measurement.
\end{proposition}
\begin{proof}
For a sub-measurement~$A$,
\begin{align*}
&\E_{\bx} \sum_{a \neq b} \bra{\psi} A^{\bx}_a \ot A^{\bx}_b \ket{\psi}\\
=~& \E_{\bx} \sum_{a} \bra{\psi} A^{\bx}_a \ot (A^{\bx}- A^{\bx}_a) \ket{\psi}\\
\leq~&\E_{\bx} \sum_{a} \bra{\psi} A^{\bx}_a \ot (I- A^{\bx}_a) \ket{\psi}\\
=~&\E_{\bx} \sum_{a} \bra{\psi} A^{\bx}_a \ot I \ket{\psi}
		- \E_{\bx} \sum_{a} \bra{\psi} A^{\bx}_a \ot A^{\bx}_a \ket{\psi}\\
=~& \bra{\psi} A \ot I \ket{\psi}
		- \E_{\bx} \sum_{a} \bra{\psi} A^{\bx}_a \ot A^{\bx}_a \ket{\psi}.
\end{align*}
This is at most~$\delta$ if $A$ is $\delta$-strongly self-consistent.
On the other hand, if~$A$ is a measurement, then the inequality becomes an equality.
Hence, if $A^x_a \ot I \simeq_{\delta} I \otimes A^x_a$,
then $A$ is $\delta$-strongly self-consistent.
\end{proof}

Next, we show that strong self-consistency is also a stronger 
condition than $A^x_a \ot I \approx_{2\delta} I \ot A^x_a$, at least for non-projective measurements.

\begin{proposition}\label{prop:two-notions-of-self-consistency}
Let $\ket{\psi}$ be a permutation-invariant state,
and let $A = \{A^x_a\}$ be a sub-measurement. If
\begin{equation*}
\E_{\bx} \sum_a \bra{\psi} A^{\bx}_a \otimes A^{\bx}_a \ket{\psi} \geq \bra{\psi} A \otimes I \ket{\psi} - \delta
\end{equation*}
then $A^x_a \otimes I \approx_{2\delta} I \otimes A^x_a$.
This is an ``if and only if" if $A$ is projective.
\end{proposition}
\begin{proof}
For general (i.e.\ not necessarily projective)~$A$
\begin{align}
&\E_{\bx} \sum_a \Vert(A^{\bx}_a \otimes I - I \otimes A^{\bx}_a) \ket{\psi}\Vert^2\nonumber\\
=~&\E_{\bx} \sum_a \bra{\psi} (A^{\bx}_a \otimes I - I \otimes A^{\bx}_a)^2 \ket{\psi}\nonumber\\
=~&2\cdot \Big( \E_{\bx} \sum_a \bra{\psi} (A^{\bx}_a)^2 \otimes I \ket{\psi} -  \E_{\bx} \sum_a \bra{\psi} A^{\bx}_a \otimes A^{\bx}_a \ket{\psi}\Big)\nonumber\\
\leq~&2\cdot \Big( \E_{\bx} \sum_a \bra{\psi} A^{\bx}_a \otimes I \ket{\psi} -  \E_{\bx} \sum_a \bra{\psi} A^{\bx}_a \otimes A^{\bx}_a \ket{\psi}\Big).\label{eq:here's-where-projectivity-would-help}
\end{align}
This is at most~$2\cdot \delta$ if
\begin{equation*}
\E_{\bx} \sum_a \bra{\psi} A^{\bx}_a \otimes A^{\bx}_a \ket{\psi} \geq \E_{\bx} \sum_a \bra{\psi} A^{\bx}_a \otimes I \ket{\psi} - \delta.
\end{equation*}
If~$A$ is projective, then \Cref{eq:here's-where-projectivity-would-help} becomes an equality,
and so $A^x_a \otimes I \approx_{2\delta} I \otimes A^x_a$ implies that
\begin{equation*}
 \E_{\bx} \sum_a \bra{\psi} A^{\bx}_a \otimes I \ket{\psi} -  \E_{\bx} \sum_a \bra{\psi} A^{\bx}_a \otimes A^{\bx}_a \ket{\psi}
 = \frac{1}{2} \cdot \left(\E_{\bx} \sum_a \Vert(A^{\bx}_a \otimes I - I \otimes A^{\bx}_a) \ket{\psi}\Vert^2\right) \leq \delta.\qedhere
\end{equation*}
\end{proof}

Hence, we may also refer to the condition $A^x_a \ot I \approx_{2\delta} I \ot A^x_a$
as ``strong self-consistency'' if~$A$ is projective.

For the remainder of the section, we will prove various properties of strongly self-consistent sub-measurements.

\begin{proposition}\label{prop:two-notions-of-self-consistency-after-evaluation}
Let $\ket{\psi}$ be a permutation-invariant state,
and let $A = \{A^x_a\}$ be a sub-measurement such that
\begin{equation*}
\E_{\bx} \sum_a \bra{\psi} A^{\bx}_a \otimes A^{\bx}_a \ket{\psi} \geq  \bra{\psi} A \otimes I \ket{\psi} - \delta.
\end{equation*}
Then for any function~$f$,
$A^x_{[f(a)=b]} \otimes I \approx_{2\delta} I \otimes A^x_{[f(a)=b]}$.
\end{proposition}
\begin{proof}
Our goal is to bound
\begin{align}
&\E_{\bx} \sum_b \Vert(A^{\bx}_{[f(a)=b]} \otimes I - I \otimes A^{\bx}_{[f(a)=b]}) \ket{\psi}\Vert^2\nonumber\\
=~&\E_{\bx} \sum_b \bra{\psi} (A^{\bx}_{[f(a)=b]} \otimes I - I \otimes A^{\bx}_{[f(a)=b]})^2 \ket{\psi}\nonumber\\
=~&2\cdot \Big( \E_{\bx} \sum_b \bra{\psi} (A^{\bx}_{[f(a)=b]})^2 \otimes I \ket{\psi} -  \E_{\bx} \sum_b \bra{\psi} A^{\bx}_{[f(a)=b]} \otimes A^{\bx}_{[f(a)=b]} \ket{\psi}\Big)\nonumber\\
\leq~&2\cdot \Big( \E_{\bx} \sum_b \bra{\psi} A^{\bx}_{[f(a)=b]} \otimes I \ket{\psi} -  \E_{\bx} \sum_b \bra{\psi} A^{\bx}_{[f(a)=b]} \otimes A^{\bx}_{[f(a)=b]} \ket{\psi}\Big).\nonumber\\
=~&2\cdot \Big(  \bra{\psi} A \otimes I \ket{\psi} -  \E_{\bx} \sum_b \bra{\psi} A^{\bx}_{[f(a)=b]} \otimes A^{\bx}_{[f(a)=b]} \ket{\psi}\Big).\label{eq:finishing-this-up}
\end{align}
The second term in \Cref{eq:finishing-this-up} can be bounded by
\begin{align*}
 \E_{\bx} \sum_b \bra{\psi} A^{\bx}_{[f(a)=b]} \otimes A^{\bx}_{[f(a)=b]} \ket{\psi}
 &=  \E_{\bx} \sum_{a, a': f(a) = f(a')} \bra{\psi} A^{\bx}_{a} \otimes A^{\bx}_{a'} \ket{\psi}\\
 &\geq\E_{\bx} \sum_{a} \bra{\psi} A^{\bx}_{a} \otimes A^{\bx}_{a} \ket{\psi}\\
 &\geq \bra{\psi} A \ot I \ket{\psi} -\delta.
\end{align*}
Hence,
\begin{equation*}
\eqref{eq:finishing-this-up} \leq 
	2\cdot \Big( \bra{\psi} A \otimes I \ket{\psi} - (\bra{\psi} A \ot I \ket{\psi} -\delta)\Big)
= 2\delta.
\end{equation*}
This concludes the proof.
\end{proof}

\begin{proposition}\label{prop:completeness-transfer-self-consistent-A}
Let $\ket{\psi}$ be a permutation-invariant state,
and let $A = \{A^x_a\}$ be a sub-measurement such that
\begin{equation*}
\E_{\bx} \sum_a \bra{\psi} A^{\bx}_a \otimes A^{\bx}_a \ket{\psi} \geq  \bra{\psi} A \otimes I \ket{\psi} - \delta.
\end{equation*}
In addition, let $B$ be a sub-measurement such that $A^x_a \ot I \approx_{\eps} B_a^x \ot I$.
Then
\begin{equation*}
\bra{\psi} B \ot I \ket{\psi}
\geq \bra{\psi} A \ot I \ket{\psi} - \delta - 2\sqrt{\eps}.
\end{equation*}
\end{proposition}
\begin{proof}
We calculate:
\begin{align*}
\bra{\psi} B \ot I \ket{\psi}
&= \E_{\bx} \sum_a \bra{\psi} B^{\bx}_a \ot I \ket{\psi} \\
&\geq \E_{\bx} \sum_a \bra{\psi} B^{\bx}_a \ot B^{\bx}_a \ket{\psi} \\
& \geq \E_{\bx} \sum_a \bra{\psi} A^{\bx}_a \ot B^{\bx}_a \ket{\psi} - \sqrt{\eps}
				\tag{by \Cref{prop:easy-approx-from-approx-delta}} \\
& \geq \E_{\bx} \sum_a \bra{\psi} A^{\bx}_a \ot A^{\bx}_a \ket{\psi} - 2\sqrt{\eps}
				\tag{by \Cref{prop:easy-approx-from-approx-delta}} \\
& \geq  \bra{\psi} A \ot I \ket{\psi} - \delta - 2\sqrt{\eps}.
\end{align*}
This completes the proof.
\end{proof}

\begin{proposition}\label{prop:self-consistency-implies-data-processing}
Let $\ket{\psi}$ be a permutation-invariant state,
and let $A = \{A^x_a\}$ be a sub-measurement such that
\begin{equation*}
\E_{\bx} \sum_a \bra{\psi} A^{\bx}_a \otimes A^{\bx}_a \ket{\psi} \geq  \bra{\psi} A \otimes I \ket{\psi} - \delta.
\end{equation*}
In addition, let $P$ be a projective sub-measurement such that $P^x_a \ot I \approx_{\eps} A_a^x \ot I$.
Then for any function~$f$,
\begin{equation*}
P^x_{[f(a)=b]} \ot I \approx_{8\delta + 8\sqrt{\eps}} A_{[f(a)=b]}^x \ot I.
\end{equation*}
\end{proposition}
\begin{proof}
We will begin by showing that
\begin{equation}\label{eq;what-we-want-to-prove-but-on-wrong-side}
P^x_{[f(a)=b]} \ot I \approx_{2\delta + 4\sqrt{\eps}} I \ot A_{[f(a)=b]}^x.
\end{equation}
To do so, our goal is to bound
\begin{align}
&\E_{\bx} \sum_b \Vert(P^{\bx}_{[f(a)=b]} \ot I - I \ot A_{[f(a)=b]}^{\bx}) \ket{\psi} \Vert^2\nonumber\\
=~&\E_{\bx} \sum_b \bra{\psi} (P^{\bx}_{[f(a)=b]}  \ot I -  I \ot A_{[f(a)=b]}^{\bx})^2 \ket{\psi} \nonumber\\
=~&\E_{\bx} \sum_b \bra{\psi} \Big((P^{\bx}_{[f(a)=b]})^2 \ot I 
		+ I \ot (A^{\bx}_{[f(a)=b]})^2 
		- 2\cdot P^{\bx}_{[f(a)=b]} \ot A^{\bx}_{[f(a)=b]}\Big) \ket{\psi}\nonumber\\
\leq~&\E_{\bx} \sum_b \bra{\psi} \Big(P^{\bx}_{[f(a)=b]} \ot I 
		+ I \ot A^{\bx}_{[f(a)=b]}
		- 2\cdot P^{\bx}_{[f(a)=b]} \ot A^{\bx}_{[f(a)=b]}\Big) \ket{\psi}\nonumber\\
=~&\bra{\psi} P \ot I \ket{\psi} + \bra{\psi} I \ot A \ket{\psi}
		- 2\cdot \E_{\bx} \sum_b \bra{\psi}  P^{\bx}_{[f(a)=b]} \ot A^{\bx}_{[f(a)=b]} \ket{\psi}.
				\label{eq:gonna-handle-third-term}
\end{align}
By \Cref{prop:completeness-transfer-projective-P}, the first term in \Cref{eq:gonna-handle-third-term} is at most
\begin{equation*}
\bra{\psi} P \ot I \ket{\psi}
\leq \bra{\psi} A \ot I \ket{\psi} + 2\sqrt{\eps}.
\end{equation*}
As for the third term in \Cref{eq:gonna-handle-third-term}, we can bound it by
\begin{align*}
\E_{\bx} \sum_b \bra{\psi}  P^{\bx}_{[f(a)=b]} \ot A^{\bx}_{[f(a)=b]} \ket{\psi}
&= \E_{\bx} \sum_{a, a': f(a) = f(a')} \bra{\psi}  P^{\bx}_{a} \ot A^{\bx}_{a'} \ket{\psi}\\
&\geq \E_{\bx} \sum_{a} \bra{\psi}  P^{\bx}_{a} \ot A^{\bx}_{a} \ket{\psi}\\
&\geq \E_{\bx} \sum_{a} \bra{\psi}  A^{\bx}_{a} \ot A^{\bx}_{a} \ket{\psi} - \sqrt{\eps}
						\tag{by \Cref{prop:easy-approx-from-approx-delta}} \\
&\geq \bra{\psi} A \otimes I \ket{\psi} - \delta - \sqrt{\eps}.
\end{align*}
Putting everything together,
\begin{equation*}
\eqref{eq:gonna-handle-third-term}
\leq (\bra{\psi} A \ot I \ket{\psi} + 2\sqrt{\eps}) + \bra{\psi} A \ot I \ket{\psi}
	- 2 \cdot(\bra{\psi} A \otimes I \ket{\psi} - \delta - \sqrt{\eps})
= 2\delta + 4\sqrt{\eps}.
\end{equation*}
This proves \Cref{eq;what-we-want-to-prove-but-on-wrong-side}.

\Cref{prop:two-notions-of-self-consistency-after-evaluation} implies that
\begin{equation*}
A^x_{[f(a)=b]} \ot I \approx_{2\delta} I \ot A^x_{[f(a)=b]}.
\end{equation*}
Hence, by \Cref{eq;what-we-want-to-prove-but-on-wrong-side},
\begin{equation*}
P^x_{[f(a)=b]} \ot I
\approx_{2\delta + 4\sqrt{\eps}} I \ot A_{[f(a)=b]}^x
\approx_{2\delta} A^x_{[f(a)=b]} \ot I.
\end{equation*}
Thus, by \Cref{prop:triangle-inequality-for-approx_delta},
\begin{equation*}
P^x_{[f(a)=b]} \ot I
\approx_{8\delta + 8\sqrt{\eps}}A^x_{[f(a)=b]} \ot I.
\end{equation*}
This concludes the proof.
\end{proof}

\begin{proposition}\label{prop:completing-to-measurement}
		Let $\ket{\psi}$ be a permutation-invariant state,
		and let $A = \{A_a\}$ be a measurement such that
		\begin{equation*}
		\E_{\bx} \sum_a \bra{\psi} A_a \otimes A_a \ket{\psi} \geq  \bra{\psi} A \otimes I \ket{\psi} - \zeta.
		\end{equation*}
		Suppose $B = \{B_a\}$ is a sub-measurement such that
		$A_a\ot I \approx_{\delta} B_a\ot I$.
		Let $C = \{C_a\}$ be a measurement in which there is an~$a^*$
		such that $C_{a^*} = B_{a^*} + (I-B)$ and $C_a = B_a$ for all $a \neq a^*$.
		Then $A_a \ot I \approx_{2\delta + 4 \sqrt{\delta} + 2\zeta} C_a \ot I$.
		\end{proposition}
		
		Before proving this, we need the following proposition.
		\begin{proposition}\label{prop:cool-prop}
		Let $\ket{\psi}$ be a permutation-invariant state,
		and let $A = \{A_a\}$ be a sub-measurement such that
		\begin{equation*}
		\E_{\bx} \sum_a \bra{\psi} A_a \otimes A_a \ket{\psi} \geq  \bra{\psi} A \otimes I \ket{\psi} - \zeta.
		\end{equation*}
		Then
		\begin{equation*}
		\sum_a \bra{\psi} (A_a)^2 \otimes I \ket{\psi} \geq \sum_a \bra{\psi} A_a \otimes I \ket{\psi} -\zeta.
		\end{equation*}
		\end{proposition}
		\begin{proof}
		Applying Cauchy-Schwarz, we have
\begin{align}
\sum_a \bra{\psi} A_a \otimes A_a \ket{\psi}
&= \sum_a \bra{\psi} (A_a \ot I) \cdot (I \ot A_a) \ket{\psi}\nonumber\\
&\leq \sqrt{\sum_a \bra{\psi} (A_a)^2 \ot I \ket{\psi}} \cdot \sqrt{\sum_a \bra{\psi} I \ot (A_a)^2 \ket{\psi}}\nonumber\\
& = \sum_a \bra{\psi} (A_a)^2 \ot I \ket{\psi}.\label{eq:As-on-same-side-rewrite}
\end{align}
As a result,
\begin{align*}
\sum_a \bra{\psi} (A_a)^2 \otimes I \ket{\psi}
&\geq \sum_a \bra{\psi} A_a \otimes A_a \ket{\psi}\tag{by \Cref{eq:As-on-same-side-rewrite}}\\
&\geq \sum_a \bra{\psi} A_a \otimes I \ket{\psi} -\zeta. \tag{by self-consistency of~$A$}
\end{align*}
This completes the proof.
\end{proof}
		Now we prove \Cref{prop:completing-to-measurement}.
		\begin{proof}[Proof of \Cref{prop:completing-to-measurement}]
		By \Cref{prop:triangle-inequality-for-vectors-squared} (the triangle inequality for vectors squared),
		\begin{align*}
		&\phantom{=}\sum_a \Vert(A_a - C_a)\ot I \ket{\psi} \Vert^2\\
		& = \sum_{a \neq a^*} \Vert(A_a - B_a)\ot I \ket{\psi} \Vert^2 + \Vert(A_{a^*}  - (I-B + B_{a^*}))\ot I\ket{\psi}\Vert^2\\
		& \leq \sum_{a \neq a^*} \Vert(A_a - B_a)\ot I \ket{\psi} \Vert^2 + 2 \cdot \Vert(A_{a^*}  -  B_{a^*})\ot I\ket{\psi}\Vert^2 + 2 \cdot \Vert(I - B) \ot I\ket{\psi}\Vert^2\\
		& \leq 2\sum_{a \neq a^*} \Vert(A_a - B_a)\ot I \ket{\psi} \Vert^2 + 2 \cdot \Vert(A_{a^*}  -  B_{a^*})\ot I\ket{\psi}\Vert^2 + 2 \cdot \Vert(I - B) \ot I\ket{\psi}\Vert^2\\
		& = 2\sum_{a} \Vert(A_a - B_a)\ot I \ket{\psi} \Vert^2 + 2 \cdot \Vert(I - B)\ot I \ket{\psi}\Vert^2\\
		& \leq 2 \delta + 2 \cdot \Vert(I - B)\ot I \ket{\psi}\Vert^2. \tag{because $A_a \ot I \approx_{\delta} B_a \ot I$}
		\end{align*}
		The second term we can bound as follows.
		\begin{align}
		\Vert(I - B) \ot I \ket{\psi}\Vert^2
		& = \bra{\psi} (I-B)^2 \ot I \ket{\psi}\nonumber\\
		& \leq \bra{\psi} (I-B)\ot I \ket{\psi} \tag{because~$B$ is a sub-measurement}\\
		& = 1 - \bra{\psi} B\ot I \ket{\psi}\nonumber\\
		& = 1 - \sum_a \bra{\psi} B_a \ot I \ket{\psi}\nonumber\\
		& \leq 1 -\sum_a \bra{\psi} B_a^2 \ot I \ket{\psi}.\label{eq:to-return-later-whatevs}
		\end{align}
		Now, by \Cref{prop:easy-approx-from-approx-delta} and the fact that $A_a\ot I \approx_{\delta} B_a\ot I$,
		\begin{align*}
		\sum_a \bra{\psi} B_a^2 \ot I \ket{\psi}
		& \approx_{\sqrt{\delta}} \sum_a \bra{\psi} (A_a \cdot B_a)\ot I \ket{\psi}\\
		& \approx_{\sqrt{\delta}} \sum_a \bra{\psi} A_a^2\ot I \ket{\psi}\\
		& \geq  \bra{\psi} A \ot I \ket{\psi}- \zeta \tag{by \Cref{prop:cool-prop}}\\
		& = 1 - \zeta. \tag{because~$A$ is a measurement}
		\end{align*}
		Hence, by \Cref{eq:to-return-later-whatevs},
		\begin{equation*}
		\Vert(I - B) \ot I \ket{\psi}\Vert^2 \leq 1 -\sum_a \bra{\psi} B_a^2 \ot I \ket{\psi} \leq 2\sqrt{\delta} + \zeta.
		\end{equation*}
		This concludes the proof.
		\end{proof}


\section{Making measurements projective}
\label{sec:making-measurements-projective}

A recurring theme in $\MIP^*$ research is that 
projective measurements are significantly easier to manipulate
than general POVM measurements.
As just one example among many,
when $A = \{A^x_a\}$ and $B = \{B^x_a\}$ are projective measurements,
the ``$\approx_\delta$" and ``$\simeq_\delta$" distances are equivalent.
In other words:
\begin{equation*}
A^x_a \ot I \simeq_{\delta} I \ot B^x_a
\qquad
\iff
\qquad
A^x_a \ot I \approx_{2\delta} I \ot B^x_a.
\end{equation*}
This is Fact~$4.13$ from~\cite{NW19}.
On the other hand when~$A$ and~$B$
are not projective,
one can find examples of measurements
where $A^x_a \ot I \approx_{\delta} I \ot B^x_a$
for $\delta \rightarrow 0$
but one can show $A^x_a \ot I \simeq_{\eps} I \ot B^x_a$
only for $\eps \rightarrow 1$;
this is Remark~$4.15$ from~\cite{NW19}.
As a result, it is important to be able to convert
POVM measurements to projective measurements
whenever possible.

In this section, we will survey two tools for doing so.
The first of these is the textbook Naimark dilation theorem.
In our setting, it states the following.

\begin{theorem}[Naimark dilation]\label{thm:naimark}
Let $\ket{\psi}$ be a state in $\calH_{\mathrm{A}} \ot \calH_{\mathrm{B}}$.
Let $A = \{A^x_a\}$ be a sub-measurement acting on $\calH_{\mathrm{A}}$
and $B=\{B^y_b\}$ be a sub-measurement acting on $\calH_{\mathrm{B}}$.
Then there exists
\begin{enumerate}
\item Hilbert spaces $\calH_{\mathrm{A}_{\mathsf{aux}}}$ and $\calH_{\mathrm{B}_{\mathsf{aux}}}$,
\item a state $\ket{\mathsf{aux}} \in \calH_{\mathrm{A}_{\mathsf{aux}}} \ot \calH_{\mathrm{B}_{\mathsf{aux}}}$,
\item and two measurements  $\widehat{A} = \{\widehat{A}^x_a\}$ and $\widehat{B}=\{\widehat{B}^y_b\}$
acting on $\calH_{\mathrm{A}} \ot \calH_{\mathrm{A}_{\mathsf{aux}}}$
and $\calH_{\mathrm{B}} \ot \calH_{\mathrm{B}_{\mathsf{aux}}}$,
respectively,
\end{enumerate}
such that the following is true.
If we write $\ket{\widehat{\psi}} = \ket{\psi} \otimes \ket{\mathsf{aux}}$,
then for all $x,y, a, b$,
 \begin{equation*}
 \bra{\psi} A^x_a \ot B^y_b \ket{\psi}
 = \bra{\widehat{\psi}} \widehat{A}^x_a \ot \widehat{B}^y_b \ket{\widehat{\psi}}.
 \end{equation*}
 In addition, $\ket{\mathsf{aux}}$ is a \emph{product state},
 meaning that we can write it as
 \begin{equation*}
 \ket{\mathsf{aux}}= \ket{\mathsf{aux}_{\mathrm{A}}} \ot \ket{\mathsf{aux}_{\mathrm{B}}},
 \end{equation*}
for $\ket{\mathsf{aux}_{\mathrm{A}}}$ in $\calH_{\mathrm{A}_{\mathsf{aux}}}$
and $\ket{\mathsf{aux}_{\mathrm{B}}}$ in $\calH_{\mathrm{B}_{\mathsf{aux}}}$.
\end{theorem}
The second of these is the ``orthogonalization lemma" from~\cite{KV11}.
In the setting of symmetric strategies, it states the following.

\begin{theorem}[Orthogonalization lemma]\label{thm:orthonormalization}
Let $\ket{\psi}$ be a permutation-invariant state,
and let $A = \{A_a\}$ be a sub-measurement with strong self-consistency
\begin{equation*}
\sum_a \bra{\psi} A_a \otimes A_a \ket{\psi} \geq \sum_a \bra{\psi} A_a \otimes I \ket{\psi} -\zeta.
\end{equation*}
Then there exists a projective sub-measurement $P = \{P_a\}$ such that
\begin{equation*}
A_a \otimes I \approx_{100 \zeta^{1/4}} P_a \otimes I.
\end{equation*}
\end{theorem}

These two results serve largely the same purpose,
and for our applications it will typically suffice to use either one.
That said, there are tradeoffs when choosing to use one or the other.
Naimark dilation is convenient because it can be applied to any family of sub-measurements
and it preserves measurement outcome probabilities exactly.
On the other hand, it requires swapping out the state $\ket{\psi}$ and measurements~$A$ and~$B$
for $\ket{\widehat{\psi}}$ and $\widehat{A}$ and $\widehat{B}$, respectively,
which can be notationally cumbersome.
Most works opt to skip adding the hats,
and instead say something to the effect of ``using Naimark, we can assume that~$A$ and~$B$ are projective."
However, there is subtlety in doing so,
which is that Naimark dilation does not necessarily preserve ``$\approx_{\delta}$" statements;
in other words, $A^x_a \ot I \approx_{\delta} I \ot B^x_a$
does not necessarily imply $\widehat{A}^x_a \ot I \approx_{\delta} I \ot \widehat{B}^x_a$.
(As an example, \Cref{ex:easy-but-long} below
provides a case in which $A^x_a \ot I \approx_{0} I \ot B^x_a$,
but $\widehat{A}^x_a \ot I \approx_{\delta} I \ot \widehat{B}^x_a$ only holds for $\delta \geq 1$.)
This can lead to trouble if the same notation is used for~$A, B$ and $\widehat{A}, \widehat{B}$,
as ``$\approx_{\delta}$" statements derived before applying Naimark might not necessarily hold after applying Naimark.
(That said, as pointed out at the end of Section~$4.4$ of~\cite{NW19},
``$\approx_{\delta}$" statements are often derived as a consequence of ``$\simeq_{\delta}$" statements.
And as Naimark preserves ``$\simeq_{\delta}$" statements,
one could use them to simply rederive any ``$\approx_{\delta}$" statements post-Naimark.)

The downsides of using the orthogonalization lemma
are that it can only be applied to measurements which are strongly self-consistent
and it introduces additional error.
As we will see below, its proof is significantly more complicated than Naimark dilation.
On the plus side, it has none of the notational baggage that Naimark brings,
and so it is more concrete to use.
In addition, it can lead to stronger results,
as it does not require introducing an auxiliary state $\ket{\mathsf{aux}}$.

In this work, we will opt to use the orthogonalization lemma rather than Naimark dilation.
However, we will include proofs of both for completeness.
We will prove \Cref{thm:naimark} in \Cref{sec:naimark}
and then \Cref{thm:orthonormalization} in \Cref{sec:orthogonalization}.

\subsection{Naimark dilation}\label{sec:naimark}

To prove \Cref{thm:naimark}, we will first need the following lemma.

\begin{lemma}
  \label{lem:naimark-helper}
  Let $A = \{A_{a}\}$ be a sub-measurement with $k$ distinct outcomes $a \in \calA$, and let $\ket{\mathsf{aux}} \in \C^{k+1}$ be any state. Then there exists a projective sub-measurement $\widehat{A} = \{\widehat{A}_{a}\}$ such that
  for each outcome~$a$,
  \[ (I \ot \bra{\mathsf{aux}} ) \cdot \widehat{A}_{a} \cdot (I \ot \ket{\mathsf{aux}}) =
    A_{a}. \]
\end{lemma}
\begin{proof}
We will consider an orthonormal basis for $\C^{k+1}$
consisting of a vector $\ket{a}$ for each $a \in \calA$ and the vector $\ket{\bot}$.
Let~$U$ be any unitary such that for each vector $\ket{\psi} \in \C^d$,
\begin{align*}
U \cdot (\ket{\psi} \ot \ket{\mathsf{aux}})
&= \sum_{a \in \calA} ((A_{a})^{1/2} \ket{\psi}) \ot \ket{a} + ((1 - A)^{1/2}\ket{\psi}) \ot \ket{\bot}\\
&= \Big(\sum_{a \in \calA} (A_{a})^{1/2} \ot \ket{a} + (1 - A)^{1/2} \ot \ket{\bot}\Big)
		\cdot \ket{\psi}.
\end{align*}
This implies that
\begin{equation*}
U \cdot (I \ot \ket{\mathsf{aux}})
= \sum_{a \in \calA} (A_{a})^{1/2} \ot \ket{a} + (1 - A)^{1/2} \ot \ket{\bot}.
\end{equation*}
Hence, for any $a \in \calA$,
\begin{equation*}
(I \ot \bra{a}) \cdot U \cdot (I \ot \ket{\mathsf{aux}})
= (I \ot \bra{a}) \cdot \Big(\sum_{a \in \calA} (A_{a})^{1/2} \ot \ket{a} + (1 - A)^{1/2} \ot \ket{\bot}\Big)
= (A_a)^{1/2}.
\end{equation*}
  Then the desired projective sub-measurement is
  \[ \widehat{A}_a = U^{\dagger} \cdot (I \ot \ket{a} \bra{a}) \cdot U. \]
  This is because
\begin{align*}
 (I \ot \bra{\mathsf{aux}} ) \cdot \widehat{A}_{a} \cdot (I \ot \ket{\mathsf{aux}})
&=(I \ot \bra{\mathsf{aux}} ) \cdot (U^{\dagger} \cdot (I \ot \ket{a} \bra{a}) \cdot U) \cdot (I \ot \ket{\mathsf{aux}})\\
&=(I \ot \bra{\mathsf{aux}} ) \cdot U^{\dagger} \cdot (I \ot \ket{a}) \cdot (I\ot \bra{a}) \cdot U \cdot (I \ot \ket{\mathsf{aux}})\\
&= (A_a)^{1/2} \cdot (A_a)^{1/2} = A_a.
\end{align*}
This completes the proof.
\end{proof}

Now we prove \Cref{thm:naimark}.
\begin{proof}[Proof of \Cref{thm:naimark}]
For each $x$, let $k$ be the number of outcomes in $A^x$,
and let $\ket{\mathsf{aux}_{\mathrm{A},x}}$ be a state of dimensionality $k+1$.
Let $\widetilde{A}^x$ be the sub-measurement guaranteed by \Cref{lem:naimark-helper}.
Define $\ket{\mathsf{aux}_{\mathrm{B},x}}$ and $\widetilde{B}^x$ similarly.
We define
\begin{equation*}
\ket{\widehat{\psi}} = \ket{\psi} \ot \big(\ot_{x} \ket{\mathsf{aux}_{\mathrm{A},x}}\big)
					\ot \big(\ot_{x} \ket{\mathsf{aux}_{\mathrm{B},x}}\big)
\end{equation*}
and
\begin{align*}
\widehat{A}^x_a
&= \widetilde{A}^x_a \ot \big(\ot_{z \neq x} I_{\mathsf{aux}_{\mathrm{A},z}}\big)\\
\widehat{B}^y_b
&= \widetilde{B}^y_b \ot \big(\ot_{z\neq y} I_{\mathsf{aux}_{\mathrm{B},z}}\big).
\end{align*}
Then for all $x,y,z,b$,
\begin{align*}
\bra{\widehat{\psi}} \widehat{A}^x_a \ot \widehat{B}^y_b \ket{\widehat{\psi}}
& = \bra{\psi} \ot \bra{\mathsf{aux}_{\mathrm{A},x}} \ot \bra{\mathsf{aux}_{\mathrm{B},y}}
	\cdot \widetilde{A}^x_a \ot \widetilde{B}^y_b \cdot \ket{\psi} \ot \ket{\mathsf{aux}_{\mathrm{A},x}} \ot \ket{\mathsf{aux}_{\mathrm{B},y}}\\
&= 	\bra{\psi} A^x_a \ot B^y_b \ket{\psi}.
\end{align*}
This completes the proof.
\end{proof}

\begin{example}[Naimark does not preserve ``$\approx_{\delta}$'']\label{ex:easy-but-long}
We now carry out a simple example
in which Naimark preserves ``$\simeq_\delta$'' statements without preserving ``$\approx_\delta$'' statements.
Let $\ket{\psi}$ be an arbitrary state in $\calH_{\mathrm{A}} \ot \calH_{\mathrm{B}}$,
where $\calH_{\mathrm{A}} = \calH_{\mathrm{B}} = \C^d$.
In addition, let $A = \{A_0, A_1\}$ and $B = \{B_0, B_1\}$ be the two-outcome measurements
in which $A_0 = A_1 = B_0 = B_1 = \frac{1}{2} \cdot I_{d \times d}$.
Then
\begin{align*}
\sum_{a \neq b} \bra{\psi} A_a \ot B_b\ket{\psi}
&= \bra{\psi} A_0 \ot B_1\ket{\psi} + \bra{\psi} A_1 \ot B_0\ket{\psi}\\
&= \frac{1}{4} \cdot \bra{\psi} I \ot I \ket{\psi} + \frac{1}{4} \cdot \bra{\psi} I \ot I \ket{\psi}
= \frac{1}{2},
\end{align*}
and so $A_a \ot I \simeq_{1/2} I \ot A_a$.
In addition,
\begin{equation*}
\sum_a \Vert (A_a \ot I - I \ot B_a) \ket{\psi}\Vert^2
= \sum_a \Big\Vert \Big(\frac{1}{2} \cdot I \ot I - \frac{1}{2} \cdot I \ot I\Big) \ket{\psi}\Big\Vert^2
= 0,
\end{equation*}
and so $A_a \ot I \approx_0 I \ot B_a$.

Now we carry out the steps of Naimark dilation.
Because~$A$ and~$B$ are two-outcome,
we will set $\calH_{\mathrm{A}_{\mathsf{aux}}} = \calH_{\mathrm{B}_{\mathsf{aux}}} = \C^2$,
spanned by the basis vectors $\ket{0}$ and $\ket{1}$.
(\Cref{lem:naimark-helper} actually asks that the auxiliary space have dimension $2+1 = 3$,
where the third dimension is spanned by the basis vector $\ket{\bot}$.
However, this is unnecessary for this example because~$A$ and~$B$
are measurements rather than sub-measurements.)
We will choose our two local auxiliary states
to be
\begin{equation*}
\ket{\mathsf{aux}_{\mathrm{A}}} = \ket{\mathsf{aux}_{\mathrm{B}}}
	= \ket{+} = \frac{1}{\sqrt{2}}\ket{0} + \frac{1}{\sqrt{2}}\ket{1}.
\end{equation*}
We first construct $\widehat{A}$ by following the proof of \Cref{lem:naimark-helper}.
To begin, it asks for~$U$ to be a unitary such that for any $\ket{\phi}$ in $\calH_{\mathrm{A}}$,
\begin{align*}
U \cdot (\ket{\phi} \ot \ket{\mathsf{aux}_{\mathrm{A}}})
&=U \cdot (\ket{\phi} \ot \ket{+})\\
&=  ((A_{0})^{1/2} \ket{\phi}) \ot \ket{0} + ((A_{1})^{1/2} \ket{\phi}) \ot \ket{1}\\
&= \Big(\Big(\frac{1}{2}\cdot I\Big)^{1/2} \ket{\phi}\Big) \ot \ket{0}
	+ \Big(\Big(\frac{1}{2}\cdot I\Big)^{1/2} \ket{\phi}\Big) \ot \ket{1}\\
&=\frac{1}{\sqrt{2}} \ket{\phi} \ot \ket{0} + \frac{1}{\sqrt{2}} \ket{\phi} \ot \ket{1}\\
&= \ket{\phi} \ot \ket{+}.
\end{align*}
As a result, we can simply take~$U$ to be the identity matrix.
Thus,
\begin{equation*}
\widehat{A}_a = U^{\dagger} \cdot (I \ot \ket{a} \bra{a}) \cdot U = I \ot \ket{a}\bra{a}.
\end{equation*}
By a similar argument, we can take $\widehat{B}_a = I \ot \ket{a}\bra{a}$ for $a \in \{0, 1\}$ as well.

Now we set $\ket{\widehat{\psi}} = \ket{\psi}\ot \ket{\mathsf{aux}_{\mathrm{A}}} \ot \ket{\mathsf{aux}_{\mathrm{B}}}$.
By \Cref{thm:naimark},
we already know that $\widehat{A}_a \ot I \approx_0 I \ot \widehat{B}_a$
on state $\ket{\widehat{\phi}}$ because this holds for the un-hatted state and measurements.
On the other hand,
we will now show that $\widehat{A}_a \ot I \approx_{\delta} I \ot \widehat{B}_a$
only for $\delta \geq 1$, in contrast to the un-hatted case where it holds for $\delta = 0$.
To see this, we calculate as follows.
\begin{align}
&\sum_a \Vert (\widehat{A}_a \ot I_{\mathrm{B}, \mathsf{aux}_{\mathrm{B}}}
	- I_{\mathrm{A}, \mathsf{aux}_{\mathrm{A}}} \ot \widehat{B}_a) \ket{\widehat{\psi}}\Vert^2\nonumber\\
 ={}& \sum_a \Vert ((I_{\mathrm{A}} \ot \ket{a}\bra{a}) \ot I_{\mathrm{B}, \mathsf{aux}_{\mathrm{B}}}
	- I_{\mathrm{A}, \mathsf{aux}_{\mathrm{A}}} \ot (I_{\mathrm{B}} \ot \ket{a}\bra{a}))
		\ket{\psi} \ot \ket{+} \ot \ket{+}\Vert^2\nonumber\\
={}& \sum_a \Big\Vert \Big(\frac{1}{\sqrt{2}} \ket{\psi} \ot \ket{a}\ot\ket{+}
			- \frac{1}{\sqrt{2}} \ket{\psi} \ot \ket{+} \ot \ket{a}\Big) \Big\Vert^2\nonumber\\
={}& \sum_a \frac{1}{2} \cdot \braket{\psi \mid \psi} 
	\cdot (\bra{a} \ot \bra{+}- \bra{+} \ot \bra{a}) \cdot (\ket{a} \ot \ket{+}- \ket{+} \ot \ket{a}).\label{eq:awefea}
\end{align}
For each~$a$, we have that
\begin{align*}
&(\bra{a} \ot \bra{+}- \bra{+} \ot \bra{a}) \cdot (\ket{a} \ot \ket{+}- \ket{+} \ot \ket{a})\\
={}& 2 - \braket{a \mid +} \cdot \braket{+ \mid a} - \braket{+ \mid a} \cdot \braket{a \mid +}\\
={}& 2 - \frac{1}{\sqrt{2}}\cdot \frac{1}{\sqrt{2}} - \frac{1}{\sqrt{2}} \cdot \frac{1}{\sqrt{2}}\\
={}& 1.
\end{align*}
Plugging this into \Cref{eq:awefea}, we arrive at our bound.
\end{example}

\subsection{Orthogonalization lemma}\label{sec:orthogonalization}

To show \Cref{thm:orthonormalization},
we first show it for the case when~$A$ is a measurement, rather than a sub-measurement.
Our proof of this case will even work in the more general setting
when the strategy is not assumed to be symmetric,
meaning that~$\ket{\psi}$ is not necessarily permutation-invariant and Player~$\mathrm{B}$'s measurement~$B$ may not be equal to Player~$\mathrm{A}$'s measurement.
This is the content of the following lemma, whose proof we defer till later. 

\begin{lemma}[Orthogonalization lemma for measurements]\label{lem:orthonormalization-main-lemma}
Let $\ket{\psi}$ be a state (which is not necessarily permutation-invariant),
and let $A = \{A_a\}$ and $B = \{B_a\}$ be measurements such that
\begin{equation*}
A_a \ot I \simeq_{\zeta} I \ot B_a
\end{equation*}
on state $\ket{\psi}$.
Then there exists a projective sub-measurement $P = \{P_a\}$ such that
\begin{equation*}
A_a \otimes I \approx_{84 \zeta^{1/4}} P_a \otimes I.
\end{equation*}
\end{lemma}

(Note that when $\ket{\psi}$ is permutation-invariant and $B = A$, the condition $A_a \ot I \simeq_{\zeta} I \ot B_a$
is equivalent to~$A_a$ being $\zeta$-strongly self-consistent, by \Cref{prop:other-two-notions-of-self-consistency}.
This is because~$A$ is a measurement.)

\ignore{
\begin{lemma}[Orthogonalization lemma for measurements]\label{lem:orthonormalization-main-lemma}
Let $\ket{\psi}$ be a permutation-invariant state,
and let $A = \{A_a\}$ be a measurement with strong self-consistency
\begin{equation*}
\sum_a \bra{\psi} A_a \otimes A_a \ket{\psi} \geq 1 -\zeta.
\end{equation*}
Then there exists a projective sub-measurement $P = \{P_a\}$ such that
\begin{equation*}
A_a \otimes I \approx_{76 \zeta^{1/4}} P_a \otimes I.
\end{equation*}
\end{lemma}
}

We now prove \Cref{thm:orthonormalization} by reducing the general (sub-measurement) case
to the case when~$A$ is a measurement.

\begin{proof}[Proof of \Cref{thm:orthonormalization} assuming \Cref{lem:orthonormalization-main-lemma}]
Let $A = \{A_a\}$ be a sub-measurement with outcomes $a \in \calA$ whose strong self-consistency is
\begin{equation*}
\sum_a \bra{\psi} A_a \otimes A_a \ket{\psi} \geq \sum_a \bra{\psi} A_a \otimes I \ket{\psi} -\zeta.
\end{equation*}
Note that
\begin{equation*}
\bra{\psi} A \otimes A \ket{\psi}
= \sum_{a, b} \bra{\psi} A_a \otimes A_b \ket{\psi}
\geq \sum_{a} \bra{\psi} A_a \otimes A_a \ket{\psi}
\geq  \bra{\psi} A \otimes I \ket{\psi} -\zeta.
\end{equation*}
Rearranging,
\begin{align*}
\bra{\psi} A \otimes (I - A)\ket{\psi}
&= \bra{\psi} A \otimes I \ket{\psi} - \bra{\psi} A \otimes A \ket{\psi}\\
&\leq (\bra{\psi} A \otimes A \ket{\psi} + \zeta) - \bra{\psi} A \otimes A \ket{\psi}
 = \zeta.
\end{align*}
As a result,
\begin{align*}
\bra{\psi} (I-A) \otimes (I-A) \ket{\psi}
&= \bra{\psi} (I-A) \otimes I \ket{\psi} - \bra{\psi} (I-A) \otimes A \ket{\psi}\\
&\geq \bra{\psi} (I-A) \otimes I \ket{\psi} - \zeta.
\end{align*}

Let $\widehat{A}$ be the POVM measurement with outcomes in $\widehat{\calA} = \calA \cup \{\bot\}$ defined as
\begin{equation*}
\widehat{A}_a
= \left\{\begin{array}{cl}
		A_a & \text{if $a \in \calA$},\\
		(I-A) & \text{if $a = \bot$.}
		\end{array}\right.
\end{equation*}
Then the strong self-consistency of $\widehat{A}$ is
\begin{align*}
\sum_{a \in \widehat{\calA}} \bra{\psi} \widehat{A}_a \otimes \widehat{A}_a \ket{\psi}
& = \sum_{a \in \calA} \bra{\psi} A_a \otimes A_a \ket{\psi}
	+ \bra{\psi} \widehat{A}_{\bot} \otimes \widehat{A}_{\bot} \ket{\psi}\\
& \geq \Big(\sum_{a \in \calA} \bra{\psi} A_a \otimes I \ket{\psi} - \zeta\Big)
	+ \bra{\psi} (I-A) \otimes (I-A) \ket{\psi}\\
& \geq \Big( \bra{\psi} A \otimes I \ket{\psi} - \zeta\Big)
	+ \Big(\bra{\psi} (I-A) \otimes I \ket{\psi} - \zeta\Big)\\
& = 1 - 2\zeta.
\end{align*}
Because $\widehat{A}$ is a measurement,
\Cref{prop:other-two-notions-of-self-consistency} states that this is equivalent to
\begin{equation*}
\widehat{A}_a \ot I \simeq_{2\zeta} I \ot \widehat{A}_a.
\end{equation*}
As a result, \Cref{lem:orthonormalization-main-lemma} implies the existence of a projective sub-measurement
$\widehat{P} = \{\widehat{P}_a\}_{a \in \widehat{\calA}}$ such that 
\begin{equation*}
\widehat{A}_a \ot I \approx_{84\cdot(2\zeta)^{1/4}} \widehat{P}_a \ot I.
\end{equation*}
If we define the projective sub-measurement $P = \{P_a\}_{a \in \calA}$
by $P_a = \widehat{P}_a$ for all $a \in \calA$, then
\begin{equation*}
\sum_{a \in \calA} \Vert (A_a - P_a) \ot I \ket{\psi} \Vert^2
\leq \sum_{a \in \widehat{\calA}} \Vert (\widehat{A}_a - \widehat{P}_a) \ot I \ket{\psi} \Vert^2
\leq 84 \cdot (2\zeta)^{1/4}.
\end{equation*}
Thus, $A_a \ot I \approx_{84 \cdot (2\zeta)^{1/4}} P_a \ot I$.
We conclude the proof by noting that $84 \cdot (2\zeta)^{1/4} \leq 100 \cdot \zeta^{1/4}$
because $84 \cdot (2)^{1/4} \approx 99.89\leq100$.
\end{proof}

Now we prove \Cref{lem:orthonormalization-main-lemma}.

\begin{proof}[Proof of~\Cref{lem:orthonormalization-main-lemma}]
We note that the lemma as stated is trivial when $\zeta > 1/4$.
As a result, we will assume that
\begin{equation}\label{eq:assumption-on-zeta}
\zeta \leq 1/4.
\end{equation}
Let $A = \{A_a\}$ and~$B = \{B_a\}$ be POVM measurements such that
\begin{equation*}
A_a \ot I \simeq_{\zeta} I \ot B_a.
\end{equation*}
\ignore{
By \Cref{prop:two-notions-of-self-consistency},
\begin{equation}\label{eq:A-approx-delta}
A_a \otimes I\approx_{2\zeta} I\otimes A_a.
\end{equation}
}
By \Cref{prop:simeq-for-measurements}, this implies that
\begin{equation*}
\sum_a \bra{\psi}A_a \ot B_a \ket{\psi} \geq 1 - \zeta.
\end{equation*}
Applying Cauchy-Schwarz, we have
\begin{align*}
\sum_a \bra{\psi} A_a \otimes B_a \ket{\psi}
&= \sum_a \bra{\psi} (A_a \ot I) \cdot (I \ot B_a) \ket{\psi}\nonumber\\
&\leq \sqrt{\sum_a \bra{\psi} (A_a)^2 \ot I \ket{\psi}} \cdot \sqrt{\sum_a \bra{\psi} I \ot (B_a)^2 \ket{\psi}}\nonumber\\
& \leq \sqrt{\sum_a \bra{\psi} (A_a)^2 \ot I \ket{\psi}} \cdot 1.
\end{align*}
Taking the square of both sides,
\begin{equation*}
\sum_a \bra{\psi} (A_a)^2 \ot I \ket{\psi} \geq \Big(\sum_a \bra{\psi} A_a \otimes B_a \ket{\psi}\Big)^2 \geq (1-\zeta)^2
\geq 1 - 2\zeta
= \sum_a \bra{\psi} A_a \ot I \ket{\psi} - 2\zeta. \tag{because~$A$ is a measurement}
\end{equation*}
\ignore{
\begin{align}
\sum_a \bra{\psi} A_a \otimes A_a \ket{\psi}
&= \sum_a \bra{\psi} (A_a \ot I) \cdot (I \ot A_a) \ket{\psi}\nonumber\\
&\leq \sqrt{\sum_a \bra{\psi} (A_a)^2 \ot I \ket{\psi}} \cdot \sqrt{\sum_a \bra{\psi} I \ot (A_a)^2 \ket{\psi}}\nonumber\\
& = \sum_a \bra{\psi} (A_a)^2 \ot I \ket{\psi}.\label{eq:As-on-same-side}
\end{align}
\ignore{We claim that
\begin{equation}\label{eq:As-on-same-side}
\sum_a \bra{\psi} A_a \otimes A_a \ket{\psi}
\approx_{\sqrt{2\zeta}} \sum_a \bra{\psi} (A_a)^2 \otimes I \ket{\psi}.
\end{equation}
To show this, we bound the magnitude of the difference.
\begin{align*}
&\Big|\sum_a \bra{\psi} (A_a \otimes I) \cdot (A_a \otimes I - I \otimes A_a) \ket{\psi}\Big|\\
\leq~&
	\sqrt{\sum_a \bra{\psi} (A_a)^2 \otimes I \ket{\psi}}
	\cdot \sqrt{\sum_a \bra{\psi}(A_a \otimes I - I \otimes A_a)^2 \ket{\psi}}\\
\leq~& 1 \cdot \sqrt{2\zeta}. \tag{by \Cref{eq:A-approx-delta}}
\end{align*}}
As a result,
\begin{align*}
\sum_a \bra{\psi} (A_a)^2 \otimes I \ket{\psi}
&\geq \sum_a \bra{\psi} A_a \otimes A_a \ket{\psi}\tag{by \Cref{eq:As-on-same-side}}\\
&\geq \sum_a \bra{\psi} A_a \otimes I \ket{\psi} -\zeta. \tag{by self-consistency of~$A$}
\end{align*}
}
Rearranging,
\begin{equation}\label{eq:A-looks-projective}
\sum_a \bra{\psi} (A_a - (A_a)^2) \otimes I \ket{\psi}
\leq 2\zeta.
\end{equation}
This is all we need the~$B$ measurement for; henceforth, we will derive consequences of~\Cref{eq:A-looks-projective}.

In our next lemma, we convert each~$A_a$ to a projective matrix~$R_a$
by rounding each of $A_a$'s large eigenvalues to~$1$
and small eigenvalues to~$0$.

\begin{lemma}[Rounding to projectors]\label{lem:projective-non-measurement}
There exists a set of projective matrices $\{R_a\}$ such that
\begin{equation*}
A_a \ot I \approx_{2\sqrt{\zeta}} R_a \ot I.
\end{equation*}
and
\begin{equation*}
R:= \sum_a R_a \leq (1+2\sqrt{\zeta}) \cdot I.
\end{equation*}
\end{lemma}
\begin{proof}
For each~$a$, we write the eigendecomposition of $A_a$ as follows: 
\begin{equation*}
A_a = \sum_i \lambda_{a, i} \cdot\ket{u_{a, i}} \bra{u_{a, i}}.
\end{equation*}
Let
\begin{equation}\label{eq:bound-on-delta}
0 < \delta \leq 1/2
\end{equation}
be a number to be decided later.
Let $\mathsf{trunc}_\delta:[0, 1] \rightarrow \{0, 1\}$ be the truncation function defined as
\begin{equation*}
\mathsf{trunc}_\delta(x) =
\left\{\begin{array}{rl}
	1 & \text{if } x \geq 1- \delta,\\
	0 & \text{otherwise.}
	\end{array}\right.
\end{equation*}
Then for each~$a$ we define the matrix $R_a$
\begin{equation*}
R_a = \mathsf{trunc}_\delta(A_a) = \sum_i \mathsf{trunc}_\delta(\lambda_{a, i}) \cdot \ket{u_{a, i}} \bra{u_{a, i}}.
\end{equation*}
To analyze this, we will require the following technical lemma.
\begin{lemma}\label{lem:trunc-inequality}
For any $x \in [0, 1]$,
\begin{equation*}
(x - \mathsf{trunc}_\delta(x))^2 \leq \frac{1}{\delta} \cdot (x - x^2).
\end{equation*}
\end{lemma}
\begin{proof}
This is proved by case analysis.
First, suppose that $x \geq 1 - \delta$.
This implies that $\mathsf{trunc}_\delta(x) = 1$.
In addition, because $\delta \leq 1/2$ by \Cref{eq:bound-on-delta},
$
x \geq 1 - \delta \geq 1/2 \geq \delta.
$
Thus,
\begin{align*}
(x - \mathsf{trunc}_\delta(x))^2
& = (1-x)^2\\
 &\leq (1-x)\\
&\leq (1-x) \cdot \frac{x}{\delta} \tag{because $x \geq \delta$}\\
&= \frac{1}{\delta} \cdot (x - x^2).
\end{align*}
Next, suppose that $x < 1 -\delta$.
This implies that $\mathsf{trunc}_\delta(x) = 0$.
Thus,
\begin{align*}
(x - \mathsf{trunc}_\delta(x))^2
& = x^2\\
 &\leq x\\
&\leq x \cdot \frac{(1-x)}{\delta} \tag{because $x \leq 1- \delta$}\\
&= \frac{1}{\delta} \cdot (x - x^2).
\end{align*}
This concludes the proof.
\end{proof}

As a result, for each~$a$, \Cref{lem:trunc-inequality} implies that
\begin{equation*}
(A_a - R_a)^2 = (A_a - \mathsf{trunc}_\delta(A_a))^2 \leq \frac{1}{\delta} \cdot (A_a - (A_a)^2).
\end{equation*}
Thus,
\begin{align*}
\sum_a \Vert (A_a - R_a) \ot I \ket{\psi}\Vert^2
= \sum_a \bra{\psi} (A_a - R_a)^2 \ot I\ket{\psi}
&\leq \frac{1}{\delta} \cdot \sum_a \bra{\psi} (A_a - (A_a)^2) \ot I \ket{\psi}\\
&\leq \frac{1}{\delta} \cdot 2\zeta. \tag{by \Cref{eq:A-looks-projective}}
\end{align*}
This implies that $A_a \otimes I \approx_{2\zeta/\delta} R_a \otimes I$.

Next, for each $x \in [0, 1]$, it follows from the definition of $\mathsf{trunc}_\delta$ that
\begin{equation*}
\mathsf{trunc}_\delta(x) \leq \left(\frac{1}{1-\delta}\right) \cdot x.
\end{equation*}
Thus, for each~$a$
\begin{equation*}
R_a = \mathsf{trunc}_\delta(A_a) \leq \left(\frac{1}{1-\delta}\right) \cdot A_a.
\end{equation*}
Summing over all~$a$,
\begin{equation*}
R = \sum_a R_a \leq \left(\frac{1}{1-\delta}\right) \cdot \sum_a A_a = \left(\frac{1}{1-\delta}\right) \cdot A
= \left(\frac{1}{1-\delta}\right) \cdot I,
\end{equation*}
because~$A$ is a measurement.
We note that
\begin{equation*}
\frac{1}{1-\delta}
= \frac{1}{1-\delta} \cdot \frac{1 + 2\delta}{1 + 2\delta}
= \frac{1 + 2\delta}{1 + \delta - 2\delta^2}
\leq 1 + 2\delta,
\end{equation*}
because
\begin{equation*}
1 + \delta - 2\delta^2 = 1 + \delta\cdot(1-2\delta)\geq 1
\end{equation*}
when $\delta \leq 1/2$,
which we assumed in \Cref{eq:bound-on-delta}.
Hence,
\begin{equation*}
R \leq (1+ 2\delta) \cdot I.
\end{equation*}

The lemma now follows by setting $\delta = \sqrt{\zeta}$. Note that we required that $\delta$ be at most $1/2$,
which follows from our assumption that $\zeta \leq 1/4$ from \Cref{eq:assumption-on-zeta}.
\ignore{
Finally,
\begin{align*}
&\sum_a R_a (R - R_a) R_a\\
=~&\sum_a R_a R R_a - \sum_a (R_a)^3\\
\leq ~&\left(\frac{1}{1-\delta}\right) \cdot \sum_a R_a A R_a - \sum_a (R_a)^3 \tag{by XXX}\\
\leq ~&\left(\frac{1}{1-\delta}\right) \cdot \sum_a  (R_a)^2 - \sum_a (R_a)^3 \tag{because~$A$ is a sub-measurement}\\
= ~&\left(\frac{1}{1-\delta}\right) \cdot \sum_a  R_a - \sum_a R_a \tag{because the $R_a$'s are projectors}\\
= ~&\left(\frac{1}{1-\delta}\right) \cdot R - R\\
= ~&\left(\frac{1}{1-\delta} -  1\right) \cdot R\\
\leq~& 2\delta \cdot R. \tag{because $\delta \leq 1/2$}
\end{align*}
The lemma now follows by setting $\delta = XXX$. Note that we required that $\delta$ be at most $1/2$,
which follows from our assumption on $\zeta$ from XXX.
}
\end{proof}

Write $d$ for the dimension of the~$A$ matrices.
If the~$\{R_a\}$ matrices from \Cref{lem:projective-non-measurement} formed a projective sub-measurement,
then their total rank would be at most~$d$.
Even if this is not true,
the next lemma shows that we can still post-process them
to reduce their total rank to at most~$d$.

\begin{lemma}[Rank reduction]\label{lem:projective-low-rank-sum}
There exists a set of projection matrices $\{Q_a\}$ such that
\begin{equation*}
A_a \ot I \approx_{12\sqrt{\zeta}} Q_a \ot I.
\end{equation*}
and
\begin{equation*}
Q:= \sum_a Q_a \leq (1+2\sqrt{\zeta}) \cdot I.
\end{equation*}
Furthermore,  $Q$ has bounded total rank:
\begin{equation*}
\sum_a \mathrm{rank}(Q_a) \leq d.
\end{equation*}
\end{lemma}
\begin{proof}
Let $\{R_a\}$ be the set of projective matrices given by \Cref{lem:projective-non-measurement}.
For each $a$, let $r_a$ be the rank of $R_a$.
Let $r = \sum_a r_a$.
If $r \leq d$, then the lemma is trivially satisfied by taking $Q$ to be~$R$ and applying \Cref{lem:projective-non-measurement}.
Henceforth, we will assume that $r > d$.
We want to reduce~$r$ so that it is at most~$d$.
Fortunately, it is already not too much larger than~$d$:
\begin{equation}\label{eq:bound-on-r}
r
= \sum_{a} r_a
= \sum_{a} \trace(R_a)
= \trace(R)
\leq (1 + 2\sqrt{\zeta}) \cdot \trace(I)
= (1+2\sqrt{\zeta}) \cdot d.
\end{equation}
Let $\ket{v_{a, 1}}, \ldots, \ket{v_{a, r_a}}$ be an orthonormal basis for the range of $R_a$, so that
\begin{equation*}
R_a = \sum_{i=1}^{r_a} \ket{v_{a, i}}\bra{v_{a,i}}.
\end{equation*}
To reduce~$r$,
we will throw out those $\ket{v_{a,i}}$'s whose overlap with $\ket{\psi}$ is small.
For each $a$ and $1 \leq i \leq r_a$, we denote the overlap of $\ket{v_{a,i}}$ and $\ket{\psi}$ by
\begin{equation*}
o_{a, i} = \bra{\psi} \cdot (\ket{v_{a,i}}\bra {v_{a,i}} \ot I) \cdot\ket{\psi}.
\end{equation*}
The total overlap is given by
\begin{align}
\sum_a \sum_{i=1}^{r_a} o_{a, i}
&= \sum_a \sum_{i=1}^{r_a} \bra{\psi} \cdot (\ket{v_{a,i}}\bra {v_{a,i}} \ot I) \cdot\ket{\psi}\nonumber\\
&= \sum_a \bra{\psi} R_a \ot I \ket{\psi}\nonumber\\
&= \bra{\psi} R \ot I \ket{\psi}\nonumber\\
&\leq (1+2\sqrt{\zeta}) \cdot \bra{\psi} I \ot I \ket{\psi}\nonumber\\
&= 1+2\sqrt{\zeta}.\label{eq:total-overlap}
\end{align}
Now we define $\mathsf{Large}$ to be the set of large overlaps:
\begin{equation*}
\mathsf{Large} = \{(a, i) \mid \text{$o_{a, i}$ is among the~$d$ largest of the $o_{b, j}$'s}\},
\end{equation*}
where we break ties arbitrarily,
and we define $\mathsf{Small}$ to be the set containing the remaining $(a, i)$'s.
We note that $\mathsf{Large}$ is well-defined and has size~$d$ because $r > d$.
Hence, \Cref{eq:bound-on-r} implies that
\begin{equation}\label{eq:size-of-small-is-small}
|\mathsf{Small}|
= r - |\mathsf{Large}| 
= r - d
\leq (1 + 2\sqrt{\zeta}) \cdot d - d
= 2\sqrt{\zeta} \cdot d
\leq 2\sqrt{\zeta} \cdot r.
\end{equation}
Thus, the small $(a,i)$'s have small total overlap:
\begin{align}
\sum_{(a, i) \in \mathsf{Small}} o_{a, i}
\leq \frac{|\mathsf{Small}|}{r} \sum_{a, i} o_{a,i}
&\leq 2\sqrt{\zeta} \cdot \sum_{a, i} o_{a,i} \tag{by \Cref{eq:size-of-small-is-small}}\\
&\leq 2\sqrt{\zeta} \cdot (1+2\sqrt{\zeta}) \tag{by \Cref{eq:total-overlap}}\\
& \leq 4\sqrt{\zeta}, \label{eq:small-overlaps}
\end{align}
where the final step uses the assumption that $\zeta \leq 1/4$ from \Cref{eq:assumption-on-zeta}.

For each~$a$ we let $\mathsf{Large}_a$ to be set of $i$'s such that $(a,i)$ is contained in $\mathsf{Large}$,
and we define $\mathsf{Small}_a$ similarly.
We define the matrix
\begin{equation*}
Q_a = \sum_{i\in \mathsf{Large}_a} \ket{v_{a, i}}\bra{v_{a,i}}.
\end{equation*}
Then clearly
\begin{equation*}
\sum_a \mathrm{rank}(Q_a)
= \sum_a |\mathsf{Large}_a|
= |\mathsf{Large}|
\leq d.
\end{equation*}
We can compute the difference
\begin{equation*}
R_a - Q_a
= \sum_{i=1}^{r_a} \ket{v_{a, i}}\bra{v_{a,i}} - \sum_{i \in \mathsf{Large}_a}  \ket{v_{a, i}}\bra{v_{a,i}}
= \sum_{i\in \mathsf{Small}_a} \ket{v_{a, i}}\bra{v_{a,i}}.
\end{equation*}
This is a projective Hermitian matrix, which implies that $Q_a \leq R_a$.
As a result,
\begin{equation*}
Q = \sum_a Q_a \leq \sum_a R_a = R \leq (1+2\sqrt{\zeta}) \cdot I.
\end{equation*}
In addition,
\begin{align*}
\sum_a \Vert (R_a - Q_a) \ot I \ket{\psi} \Vert^2
&= \sum_a \bra{\psi}(R_a - Q_a)^2 \ot I \ket{\psi}\\
&= \sum_a \bra{\psi}(R_a - Q_a)\ot I \ket{\psi} \tag{because $R_a - Q_a$ is a projector}\\
&= \sum_a \sum_{i \in \mathsf{Small}_a} \bra{\psi} \cdot (\ket{v_{a,i}}\bra {v_{a,i}} \ot I) \cdot\ket{\psi}\\
&= \sum_a \sum_{i \in \mathsf{Small}_a} o_{a, i}\\
&\leq 4\sqrt{\zeta}. \tag{by \Cref{eq:small-overlaps}}
\end{align*}
This means that
\begin{equation*}
R_a \ot I \approx_{4\sqrt{\zeta}} Q_a \ot I.
\end{equation*}
Since we know that $A_a \ot I \approx_{2\sqrt{\zeta}} R_a \ot I$ by \Cref{lem:projective-non-measurement},
\Cref{prop:triangle-inequality-for-approx_delta} implies that
\begin{equation*}
A_a \ot I \approx_{12\sqrt{\zeta}} Q_a \ot I.
\end{equation*}
\ignore{
Finally,
\begin{align*}
&\sum_a Q_a (Q - Q_a) Q_a\\
=~&\sum_a Q_a Q Q_a - \sum_a (Q_a)^3\\
\leq ~&XXX \cdot \sum_a Q_a A Q_a - \sum_a (Q_a)^3 \tag{by XXX}\\
\leq ~&XXX \cdot \sum_a  (Q_a)^2 - \sum_a (Q_a)^3 \tag{because~$A$ is a sub-measurement}\\
= ~&XXX \cdot \sum_a  Q_a - \sum_a Q_a \tag{because the $Q_a$'s are projectors}\\
= ~&XXX \cdot Q - Q\\
= ~&\left(XXX -  1\right) \cdot Q\\
\leq~& 2\delta \cdot Q.
\end{align*}
}
This completes the proof.
\end{proof}

Henceforth, we let $Q = \{Q_a\}$ be the set of projective matrices given by \Cref{lem:projective-low-rank-sum}.
We now derive a few properties of~$Q$ that follow as a consequence of \Cref{lem:projective-low-rank-sum}.
To begin, we show that~$Q$ is almost as complete as~$A$.

\begin{lemma}[Completeness of~$Q$]\label{lem:Q-completeness}
\begin{equation*}
\bra{\psi} Q \otimes I \ket{\psi}
\geq 1 - 11 \zeta^{1/4}.
\end{equation*}
\end{lemma}
\begin{proof}
To begin, we claim that
\begin{align}
\bra{\psi} Q \otimes I \ket{\psi}
& = \sum_a \bra{\psi} Q_a \otimes I \ket{\psi}\nonumber\\
& = \sum_a \bra{\psi} (Q_a)^2 \otimes I \ket{\psi}\tag{because the $Q_a$'s are projective}\\
& \approx_{5 \zeta^{1/4}} \sum_a \bra{\psi} (Q_a \cdot A_a) \otimes I \ket{\psi}. \label{eq:Q-for-an-A}
\end{align}
To show this, we bound the magnitude of the difference.
\begin{align*}
&\Big|\sum_a \bra{\psi} (Q_a \otimes I) \cdot ((Q_a - A_a) \otimes I) \ket{\psi}\Big|\\
\leq~&\sqrt{\sum_a\bra{\psi} (Q_a)^2 \otimes I \ket{\psi}}
	\cdot \sqrt{\sum_a \bra{\psi} (Q_a - A_a)^2 \otimes I \ket{\psi}}\\
\leq~&  \sqrt{1+2\sqrt{\zeta}} \cdot \sqrt{12\sqrt{\zeta}} \tag{by \Cref{lem:projective-low-rank-sum}}\\
\leq~& \sqrt{2} \cdot \sqrt{12\sqrt{\zeta}},
\end{align*}
where the last line uses the assumption that $\zeta \leq 1/4$ from \Cref{eq:assumption-on-zeta}.
Next, we claim that
\begin{equation}\label{eq:another-Q-for-A}
\eqref{eq:Q-for-an-A}
= \sum_a \bra{\psi} (Q_a \cdot A_a) \otimes I \ket{\psi}
\approx_{4 \zeta^{1/4}} \sum_a \bra{\psi} (A_a)^2 \otimes I \ket{\psi}.
\end{equation}
To show this, we bound the magnitude of the difference.
\begin{align*}
&\Big|\sum_a \bra{\psi} ((Q_a - A_a) \otimes I) \cdot (A_a \otimes I) \ket{\psi}\Big|\\
\leq~&\sqrt{\sum_a \bra{\psi} (Q_a - A_a)^2 \otimes I \ket{\psi}}
	\cdot \sqrt{\sum_a\bra{\psi} (A_a)^2 \otimes I \ket{\psi}}\\
\leq~& \sqrt{12 \sqrt{\zeta}} \cdot 1. \tag{by \Cref{lem:projective-low-rank-sum}}
\end{align*}
In conclusion,
\begin{align*}
\bra{\psi} Q \otimes I \ket{\psi}
& \geq \sum_a \bra{\psi} (A_a)^2 \otimes I \ket{\psi} - 9 \zeta^{1/4} \tag{by \Cref{eq:Q-for-an-A,eq:another-Q-for-A}}\\
& \geq \sum_a \bra{\psi} A_a \otimes I \ket{\psi} - 2\zeta - 9\zeta^{1/4} \tag{by \Cref{eq:A-looks-projective}}\\
& \geq \bra{\psi} A \otimes I \ket{\psi} - 11 \zeta^{1/4} \nonumber\\
& = 1 - 11 \zeta^{1/4}.
\end{align*}
This completes the proof.
\end{proof}

We will also need the following bound on the completeness of the \emph{square root} of~$Q$.
Note that such a bound follows trivially from \Cref{lem:Q-completeness}
when~$Q$ is a sub-measurement because $\sqrt{Q} \geq Q$ when $Q \leq I$.

\begin{lemma}[Completeness of~$\sqrt{Q}$]\label{lem:sqrt-Q-completeness}
\begin{equation*}
\bra{\psi} \sqrt{Q} \otimes I \ket{\psi}
\geq 1 - 12 \zeta^{1/4}.
\end{equation*}
\end{lemma}
\begin{proof}
Let $Q = \sum_i \nu_i \ket{u_i} \bra{u_i}$ be the eigendecomposition of~$Q$.
Then because $Q \leq (1 + 2\sqrt{\zeta}) \cdot I$, each eigenvalue $\nu_i$ is at most $1+2\sqrt{\zeta}$.
Thus,
\begin{equation*}
\sqrt{Q}
= \sum_i \sqrt{\nu_i} \ket{u_i}\bra{u_i}
\geq \frac{1}{\sqrt{1 + 2\sqrt{\zeta}}} \cdot \sum_i \nu_i \ket{u_i}\bra{u_i}
= \frac{1}{\sqrt{1 + 2\sqrt{\zeta}}} \cdot Q.
\end{equation*}
We note that
\begin{equation*}
\frac{1}{\sqrt{1 + 2\sqrt{\zeta}}}
\geq \frac{1}{\sqrt{1 + 2\sqrt{\zeta} + \zeta}}
= \frac{1}{1 + \sqrt{\zeta}}
= \frac{1}{1 + \sqrt{\zeta}} \cdot \Big(\frac{1-\sqrt{\zeta}}{1-\sqrt{\zeta}}\Big)
= \frac{1-\sqrt{\zeta}}{1 - \zeta}
\geq 1 - \sqrt{\zeta}.
\end{equation*}
Hence, $\sqrt{Q} \geq (1-\sqrt{\zeta}) \cdot Q$.
As a result, \Cref{lem:Q-completeness} implies that
\begin{equation*}
\bra{\psi} \sqrt{Q} \otimes I \ket{\psi}
\geq (1-\sqrt{\zeta}) \cdot\bra{\psi} Q \otimes I \ket{\psi}
\geq (1-\sqrt{\zeta}) \cdot \left(1 - 11 \zeta^{1/4}\right)
\geq 1 -\sqrt{\zeta} - 11 \zeta^{1/4}.
\end{equation*}
This concludes the proof.
\end{proof}

Finally, we show the following lemma,
which quantifies a sense in which~$Q$ is ``almost projective".

\begin{lemma}[$Q$ is almost projective]\label{lem:q-almost-projective}
\begin{equation*}
\sum_a (Q_a \cdot Q  \cdot Q_a - Q_a) \leq  4\sqrt{\zeta} \cdot I.
\end{equation*}
\end{lemma}
\begin{proof}
\Cref{lem:projective-low-rank-sum} implies that $Q \leq (1 + 2\sqrt{\zeta}) \cdot I$. As a result,
\begin{align*}
&\sum_a Q_a \cdot Q  \cdot Q_a - \sum_a Q_a\\
\leq ~&(1 + 2\sqrt{\zeta}) \cdot \sum_a Q_a \cdot I  \cdot Q_a - \sum_a Q_a \\
= ~&(1+2\sqrt{\zeta}) \cdot \sum_a  Q_a - \sum_a Q_a \tag{because the $Q_a$'s are projectors}\\
= ~&(1+2\sqrt{\zeta}) \cdot Q - Q\\
= ~& 2\sqrt{\zeta} \cdot Q\\
\leq~& 2\sqrt{\zeta} \cdot (1 + 2\sqrt{\zeta}) \cdot I\\
\leq~& 2 \sqrt{\zeta} \cdot 2 \cdot I. \tag{by \Cref{eq:assumption-on-zeta}}
\end{align*}
This completes the proof.
\end{proof}

We now arrive at the most important definition in this proof,
which is a natural matrix decomposition for the $Q_a$ matrices.

\begin{definition}[Matrix decomposition of~$Q_a$]
For each $a$, let $m_a$ be the rank of $Q_a$.
Let $\ket{v_{a, 1}}, \ldots, \ket{v_{a, m_a}}$ be an orthonormal basis for the range of $Q_a$, so that
\begin{equation*}
Q_a = \sum_{i=1}^{m_a} \ket{v_{a, i}}\bra{v_{a,i}}.
\end{equation*}
Let $m = \sum_a m_a$, and consider an orthonormal basis of $\C^m$ consisting of vectors $\ket{a, i}$ for each $a$ and $1 \leq i \leq m_a$.
For each $a$, define the matrix
\begin{equation*}
X_a = \sum_{i=1}^{m_a} \ket{a, i} \bra{v_{a, i}}.
\end{equation*}
In addition, define the matrix
\begin{equation}\label{eq:looks-like-singular-value-decomposition}
X = \sum_a X_a =  \sum_a \sum_{i=1}^{m_a} \ket{a, i} \bra{v_{a, i}}.
\end{equation}
Finally, we let $T = \{T_a\}$ be the projective measurement on $\C^m$ defined as
\begin{equation*}
T_a = \sum_{i=1}^{m_a} \ket{a, i}\bra{a, i}.
\end{equation*}
\end{definition}

The next pair of lemmas will prove some basic properties of the~$X_a$ matrices.

\begin{lemma}\label{lem:xa-t}
For each~$a$, $\displaystyle
X_a = T_a \cdot X.
$
\end{lemma}
\begin{proof}
This is a simple calculation:
\begin{equation*}
T_a \cdot X = \Big(\sum_i^{m_a} \ket{a, i} \bra{a,i}\Big)\cdot \Big( \sum_a \sum_{i=1}^{m_a}  \ket{a, i}\bra{v_{a, i}}\Big) 
= \sum_{i=1}^{m_a} \ket{a, i}\bra{v_{a, i}} = X_a.\qedhere
\end{equation*}
\end{proof}

\begin{lemma}[$Q_a$ restated]\label{lem:qa-restated}
For each~$a$,
\begin{equation*}
Q_a = X^\dagger_a \cdot X_a = X^\dagger \cdot T_a \cdot X = X_a^\dagger \cdot X.
\end{equation*}
\end{lemma}
\begin{proof}
The first equality follows from
\begin{align*}
X_a^\dagger \cdot X_a
&= \Big(\sum_{i=1}^{m_a} \ket{v_{a, i}}\bra{a, i} \Big) \cdot \Big(\sum_{j=1}^{m_a} \ket{v_{a, j}}\bra{a, j} \Big)^\dagger\nonumber\\
&= \sum_{i, j=1}^{m_a} \ket{v_{a, i}}\bra{a, i} \cdot \ket{a, j} \bra{v_{a, j}}
= \sum_{i=1}^{m_a} \ket{v_{a,i}}\bra{v_{a,i}} = Q_a.
\end{align*}
The remaining equalities follow from \Cref{lem:xa-t} and the fact that~$T$ is a projective measurement.
\begin{equation*}
X_a^\dagger \cdot X_a
= (X^\dagger \cdot T_a) \cdot (T_a  \cdot X)
= X^\dagger \cdot T_a \cdot X = X_a^\dagger \cdot X.\qedhere
\end{equation*}
\end{proof}

Now we introduce our main tool for studying~$X$, which is via its singular value decomposition.

\begin{definition}[SVD of~$X$]
Let $X = U \cdot \Sigma_{m \times d} \cdot V^\dagger$ be the singular value decomposition (SVD) of~$X$.
Because~$X$ is an $m \times d$ matrix,
the definition of the SVD states that
$U$ is an $m \times m$ unitary matrix,
$V$ is a $d \times d$ unitary matrix,
and $\Sigma_{m\times d}$ is an $m \times d$ diagonal matrix
with nonnegative real numbers on its diagonal.
\end{definition}

\begin{notation}
For positive integers $h$ and~$w$,
we will write $I_{h \times w}$ for the $h \times w$ matrix with $1$'s on its diagonal and $0$'s everywhere else.

For integers $h, w \geq m$, we also write $\Sigma_{h \times w}$
for the $h \times w$ diagonal matrix whose diagonal agrees with $\Sigma_{m \times d}$'s;
namely, $(\Sigma_{h \times w})_{i, i} = (\Sigma_{m \times d})_{i,i}$ for all $1 \leq i \leq m$
and $(\Sigma_{h \times w})_{i, i} = 0$ for all $i > m$.
We note that because $\Sigma_{m \times d}$ is a real-valued diagonal matrix,
$(\Sigma_{h \times w})^\dagger = \Sigma_{w \times h}$.
\end{notation}

With this definition,
we can give a helpful expression for the square of~$X$.

\begin{lemma}[$X$ squared]\label{lem:X-squared}
\begin{equation*}
X \cdot X^\dagger = U  \cdot (\Sigma_{m \times m})^2 \cdot U^\dagger,
\quad
\text{and}
\quad
X^\dagger \cdot X = Q = V \cdot (\Sigma_{d \times d})^2 \cdot V^\dagger.
\end{equation*}
\end{lemma}
\begin{proof}
First,
\begin{align*}
X \cdot X^\dagger
&= (U  \cdot \Sigma_{m \times d} \cdot V^\dagger) \cdot (U  \cdot \Sigma_{m \times d} \cdot V^\dagger)^\dagger\\
&= U  \cdot \Sigma_{m \times d} \cdot V^\dagger \cdot V \cdot \Sigma_{d \times m} \cdot U^\dagger\\
&= U  \cdot \Sigma_{m \times d} \cdot I_{d \times d} \cdot \Sigma_{d \times m} \cdot U^\dagger \tag{because~$V$ is a $d \times d$ unitary}\\
&= U  \cdot (\Sigma_{m \times m})^2 \cdot U^\dagger \tag{because~$m \leq d$}.
\end{align*}
Second,
\begin{align*}
X^\dagger \cdot X
& = \sum_a X^\dagger \cdot T_a X  \tag{because~$T$ is a measurement}\\
&= \sum_a Q_a \tag{by \Cref{lem:qa-restated}}\\
&= Q.
\end{align*}
In addition, we can rewrite $X^\dagger \cdot X$ as
\begin{align*}
X^\dagger \cdot X
&= (U  \cdot \Sigma_{m \times d} \cdot V^\dagger)^\dagger \cdot (U  \cdot \Sigma_{m \times d} \cdot V^\dagger)\\
&= V \cdot \Sigma_{d \times m} \cdot U^\dagger \cdot U \cdot \Sigma_{m \times d} \cdot V^\dagger\\
&= V \cdot \Sigma_{d \times m} \cdot I_{m \times m} \cdot \Sigma_{m \times d} \cdot V^\dagger \tag{because~$U$ is an $m \times m$ unitary}\\
&= V \cdot (\Sigma_{d \times d})^2 \cdot V^\dagger.
\end{align*}
This completes the proof.
\end{proof}

The following lemma relates an expression in the $X_a$'s
with an expression in the~$Q_a$'s that appeared previously in \Cref{lem:q-almost-projective}.

\begin{lemma}\label{lem:X-expression-to-Q-expression}
For each~$a$,
\begin{equation*}
X_a^\dagger \cdot (X \cdot X^\dagger - I_{m \times m})^2 \cdot X_a
= Q_a \cdot Q \cdot Q_a - Q_a.
\end{equation*}
\end{lemma}
\begin{proof}
Using~\Cref{lem:X-squared} and~\Cref{lem:qa-restated},
\begin{equation*}
X_a^\dagger \cdot (X \cdot X^\dagger \cdot X \cdot X^\dagger) \cdot X_a
= X_a^\dagger \cdot X \cdot Q \cdot X^\dagger \cdot X_a
= Q_a \cdot Q \cdot Q_a.
\end{equation*}
Similarly, because $Q_a$ is projective,
\begin{equation*}
X_a^\dagger \cdot (X \cdot X^\dagger) \cdot X_a
= (X_a^\dagger \cdot X) \cdot (X^\dagger \cdot X_a)
= Q_a \cdot Q_a
= Q_a.
\end{equation*}
Finally,
\begin{equation*}
X_a^\dagger \cdot (I_{m \times m}) \cdot X_a
= X_a^\dagger \cdot X_a = Q_a.
\end{equation*}
Putting these together,
\begin{align*}
&X_a^\dagger \cdot (X \cdot X^\dagger - I_{m \times m})^2 \cdot X_a\\
 =~& X_a^\dagger \cdot (X \cdot X^\dagger \cdot X \cdot X^\dagger  - 2 \cdot X\cdot X^\dagger + I_{m \times m}) \cdot X_a\\
    =~& Q_a \cdot Q \cdot Q_a  - 2 \cdot Q_a  + Q_a\\
    =~& Q_a \cdot Q \cdot Q_a  -  Q_a.
\end{align*}
This completes the proof.
\end{proof}

Now we are ready to state the projective sub-measurement~$P$
which should approximate~$A$.
Before doing so, we give some intuition for the construction.

\begin{remark}
Suppose that the $Q_a$'s actually formed a projective measurement.
This would imply that the vectors $\ket{v_{a, i}}$, over all $a$ and $1 \leq i \leq m_a$ form an orthonormal set.
Then the SVD would actually have already been provided in \Cref{eq:looks-like-singular-value-decomposition};
for each $a$ and $1 \leq i \leq m_a$,
the corresponding singular value would be~$1$ and the corresponding left- and right-singular vectors would be $\ket{a, i}$ and $\ket{v_{a, i}}$, respectively.
In particular, we would have $\Sigma = I_{m \times d}$ and $X = U \cdot I_{m \times d}\cdot V^\dagger$.
\end{remark}

In reality, we don't know that $\Sigma = I_{m \times d}$.
However, we will construct~$P = \{P_a\}$ as if it were.
This motivates the following definition.

\begin{definition}[Definition of~$P$]
Define the matrix 
\begin{equation*}
\widehat{X} = U  \cdot I_{m \times d} \cdot V^\dagger.
\end{equation*}
In addition, for each $a$, define the matrices
\begin{equation*}
\widehat{X}_a  = T_a \cdot \widehat{X},
\quad P_a  = \widehat{X}_a^\dagger \cdot \widehat{X}_a.
\end{equation*}
\end{definition}

We now give analogues of \Cref{lem:qa-restated,lem:X-squared} for the~$P$ matrices.

\begin{lemma}[$P_a$ restated]\label{lem:pa-restated}
For each~$a$,
\begin{equation*}
P_a = \widehat{X}^\dagger \cdot T_a \cdot \widehat{X} = \widehat{X}_a^\dagger \cdot \widehat{X}.
\end{equation*}
\end{lemma}
\begin{proof}
This follows from the definition of $\widehat{X}^\dagger_a$ and the fact that $T$ is a projective measurement:
\begin{equation*}
P_a
= \widehat{X}_a^\dagger \cdot \widehat{X}_a
= (\widehat{X}^\dagger \cdot T_a) \cdot (T_a  \cdot \widehat{X})
= \widehat{X}^\dagger \cdot T_a \cdot \widehat{X} = \widehat{X}_a^\dagger \cdot \widehat{X}.\qedhere
\end{equation*}
\end{proof}

\begin{lemma}[$\widehat{X}$ squared]\label{lem:X-hat-squared}
\begin{equation*}
\widehat{X} \cdot \widehat{X}^\dagger = I_{m \times m}.
\ignore{
\quad
\text{and}
\quad
\widehat{X}^\dagger \cdot \widehat{X} = V \cdot I_{d \times m} \cdot I_{m \times d} \cdot V^\dagger.
}
\end{equation*}
\end{lemma}
\begin{proof}
First,
\begin{align*}
\widehat{X} \cdot \widehat{X}^\dagger
&= (U  \cdot I_{m \times d} \cdot V^\dagger) \cdot (U  \cdot I_{m \times d} \cdot V^\dagger)^\dagger\\
&= U  \cdot I_{m \times d} \cdot V^\dagger \cdot V \cdot I_{d \times m} \cdot U^\dagger\\
&= U  \cdot I_{m \times d} \cdot I_{d \times d} \cdot I_{d \times m} \cdot U^\dagger \tag{because~$V$ is a $d \times d$ unitary}\\
&= U  \cdot I_{m \times m} \cdot U^\dagger \tag{because~$m \leq d$}\\
&= I_{m \times m},
\end{align*}
where the last step uses the fact that~$U$ is an $m \times m$ unitary.
\ignore{
Second,
\begin{align*}
\widehat{X}^\dagger \cdot \widehat{X}
& = (U  \cdot I_{m \times d} \cdot V^\dagger)^\dagger \cdot (U  \cdot I_{m \times d} \cdot V^\dagger)\\
& = V \cdot I_{d \times m} \cdot U^\dagger \cdot U  \cdot I_{m \times d} \cdot V^\dagger\\
& = V \cdot I_{d \times m} \cdot I_{m \times m}  \cdot I_{m \times d} \cdot V^\dagger \tag{because~$U$ is an $m \times m$ unitary}\\
& = V \cdot I_{d \times m} \cdot I_{m \times d} \cdot V^\dagger.
\end{align*}
We note that $m \leq d$, and so $I_{d \times m} \cdot I_{m \times d}$ is not necessarily equal to $I_{d \times d}$.
}
\end{proof}

Finally, we show two lemmas on quantities involving both $X$ and~$\widehat{X}$.

\begin{lemma}[$X$ times $\widehat{X}$]\label{lem:X-times-X-hat}
\begin{equation*}
X \cdot \widehat{X}^\dagger = U  \cdot \Sigma_{m \times m} \cdot U^\dagger,
\quad
\text{and}
\quad
X^\dagger \cdot \widehat{X} = \sqrt{Q}.
\end{equation*}
\end{lemma}
\begin{proof}
First,
\begin{align*}
X \cdot \widehat{X}^\dagger
&= (U  \cdot \Sigma_{m \times d} \cdot V^\dagger) \cdot (U  \cdot I_{m \times d} \cdot V^\dagger)^\dagger\\
&= U  \cdot \Sigma_{m \times d} \cdot V^\dagger \cdot V \cdot I_{d \times m} \cdot U^\dagger\\
&= U  \cdot \Sigma_{m \times d} \cdot I_{d \times d} \cdot I_{d \times m} \cdot U^\dagger \tag{because~$V$ is a $d \times d$ unitary}\\
&= U  \cdot \Sigma_{m \times m} \cdot U^\dagger \tag{because~$m \leq d$}.
\end{align*}
Second,
\begin{align*}
X^\dagger \cdot \widehat{X}
&= (U  \cdot \Sigma_{m \times d} \cdot V^\dagger)^\dagger \cdot (U  \cdot I_{m \times d} \cdot V^\dagger)\\
&= V \cdot \Sigma_{d \times m} \cdot U^\dagger \cdot U \cdot I_{m \times d} \cdot V^\dagger\\
&= V \cdot \Sigma_{d \times m} \cdot I_{m \times m} \cdot I_{m \times d} \cdot V^\dagger \tag{because~$U$ is an $m \times m$ unitary}\\
&= V \cdot \Sigma_{d \times d} \cdot V^\dagger\\
&= \sqrt{V \cdot (\Sigma_{d \times d})^2 \cdot V^\dagger}\\
&= \sqrt{Q}. \tag{by \Cref{lem:X-squared}}
\end{align*}
This completes the proof.
\end{proof}

\begin{lemma}[Squared difference]\label{lem:squared-difference}
\begin{equation*}
(X - \widehat{X}) \cdot (X - \widehat{X})^\dagger \leq (X \cdot X^\dagger - I_{m \times m})^2.
\end{equation*}
\end{lemma}
\begin{proof}
By~\Cref{lem:X-squared,lem:X-hat-squared,lem:X-times-X-hat},
\begin{align*}
(X - \widehat{X}) \cdot (X - \widehat{X})^\dagger
& = X \cdot X^\dagger  - X\cdot \widehat{X}^\dagger - \widehat{X} \cdot X^\dagger + \widehat{X} \cdot \widehat{X}^\dagger\\
& = U \cdot \Sigma_{m \times m}^2 \cdot U^\dagger  - 2 \cdot U \cdot \Sigma_{m \times m} \cdot U^\dagger+ I_{m \times m}\\
& = U \cdot \left(\Sigma_{m \times m}^2   - 2 \cdot  \Sigma_{m \times m} +  I_{m \times m} \right) \cdot U^\dagger\\
& = U \cdot \left(\Sigma_{m \times m}  -  I_{m \times m}\right)^2 \cdot U^\dagger.
\end{align*}
Because $\Sigma_{m \times m}$ and $I_{m \times m}$ are commuting,
\begin{equation*}
(\Sigma_{m \times m} - I_{m \times m})^2
\leq (\Sigma_{m \times m} + I_{m \times m})^2 \cdot (\Sigma_{m \times m} - I_{m \times m})^2
= (\Sigma_{m\times m}^2 - I_{m \times m})^2.
\end{equation*}
As a result,
\begin{align*}
U \cdot \left(\Sigma_{m \times m}  -  I_{m \times m}\right)^2 \cdot U^\dagger
& \leq U \cdot \left(\Sigma_{m \times m}^2 -  I_{m \times m}\right)^2 \cdot U^\dagger\\
& = \left(U \cdot \Sigma_{m \times m}^2 \cdot U^\dagger - I_{m \times m} \right)^2\\
&= (X \cdot X^\dagger - I_{m \times m})^2. \tag{by \Cref{lem:X-squared}}
\end{align*}
This completes the proof.
\end{proof}

The first property we need of~$P$ is that it is a projective sub-measurement.
This is shown in the following lemma.

\begin{lemma}[Projectivity of~$P$]\label{lem:P-projectivity}
$P = \{P_a\}$ forms a projective sub-measurement.
\end{lemma}
\begin{proof}
Let $a, b$ be (possibly distinct) outcomes.
Then
\begin{align*}
P_a \cdot P_b
&= (\widehat{X}^\dagger \cdot T_a \cdot \widehat{X})
	\cdot (\widehat{X}^\dagger \cdot T_b \cdot \widehat{X}) \tag{by \Cref{lem:pa-restated}}\\
&= (\widehat{X}^\dagger \cdot T_a \cdot I_{m \times m} \cdot T_b \cdot \widehat{X}) \tag{by \Cref{lem:X-hat-squared}}\\
&= (\widehat{X}^\dagger \cdot T_a  \cdot \widehat{X}) \cdot \bone[a = b] \tag{because~$T$ is a projective measurement}\\
&= P_a  \cdot \bone[a=b]. \tag{by \Cref{lem:pa-restated}}
\end{align*}
This completes the proof.
\end{proof}

The second property we need of~$P$ is that it is close to~$A$.
We will first show that it is close to~$Q$.

\begin{lemma}[$P$ is close to~$Q$]\label{lem:P-Q-approx}
\begin{equation*}
Q_a \otimes I \approx_{30\zeta^{1/4}} P_a \otimes I.
\end{equation*}
\end{lemma}
\begin{proof}
Our goal is to upper-bound the quantity
\begin{align}
&\sum_a \bra{\psi} (Q_a - P_a)^2 \otimes I \ket{\psi}\nonumber\\
=~& \sum_a \bra{\psi} (Q_a)^2 \otimes I \ket{\psi} + \sum_a \bra{\psi} (P_a)^2 \otimes I \ket{\psi} - \sum_a \bra{\psi} Q_a P_a \otimes I \ket{\psi}
	- \sum_a \bra{\psi} P_a Q_a \otimes I \ket{\psi}\nonumber\\
=~& \bra{\psi} Q \otimes I \ket{\psi} +  \bra{\psi} P\otimes I \ket{\psi} - \sum_a \bra{\psi} Q_a P_a \otimes I \ket{\psi}
	- \sum_a \bra{\psi} P_a Q_a \otimes I \ket{\psi},\label{eq:P-Q-thing-to-bound}
\end{align}
where the last step uses the projectivity of~$Q$ and~$P$.
We bound the four terms in \Cref{eq:P-Q-thing-to-bound} separately.
First, by \Cref{lem:projective-low-rank-sum},
\begin{equation*}
 \bra{\psi} Q \otimes I \ket{\psi}
 \leq (1+2\sqrt{\zeta}) \cdot \bra{\psi} I \otimes I \ket{\psi} \leq 1+2\sqrt{\zeta}.
\end{equation*}
Second, because~$P$ is a sub-measurement,
\begin{equation*}
\bra{\psi} P \ot I\ket{\psi} \leq 1.
\end{equation*}

The third term and fourth terms are significantly more complicated to bound.
As they are complex conjugates of each other, we can write their sum as
\begin{equation}\label{eq:complex-conjugates}
\sum_a \bra{\psi} Q_a P_a \otimes I \ket{\psi}
	+ \sum_a \bra{\psi} P_a Q_a \otimes I \ket{\psi}
= 2 \cdot \mathfrak{R}\Big(\sum_a \bra{\psi} Q_a P_a \otimes I \ket{\psi}\Big).
\end{equation}
We now focus on the expression on the right-hand side of \Cref{eq:complex-conjugates}.
To begin, we use \Cref{lem:qa-restated} to rewrite it as
\begin{equation*}
\sum_a \bra{\psi} Q_a P_a \otimes I \ket{\psi}
= \sum_a \bra{\psi} ((X_a^\dagger \cdot X) \cdot P_a)  \otimes I \ket{\psi}.
\end{equation*}
The main step will be to show that we can exchange the second $X$ for an $\widehat{X}$, i.e.
\begin{equation}\label{eq:add-a-hat}
\sum_a \bra{\psi} (X_a^\dagger \cdot X \cdot P_a) \otimes I \ket{\psi}
\approx_{2\zeta^{1/4}} \sum_a \bra{\psi} (X_a^\dagger \cdot \widehat{X} \cdot P_a)  \otimes I \ket{\psi}.
\end{equation}
To show this, we bound the magnitude of the difference.
\begin{multline*}
\Big|\sum_a \bra{\psi} ((X_a^\dagger \cdot (X - \widehat{X})) \ot I) \cdot(P_a \ot I) \ket{\psi}\Big|\\
\leq \sqrt{\sum_a \bra{\psi}(X_a^\dagger \cdot (X - \widehat{X})  \cdot (X - \widehat{X})^\dagger \cdot X_a) \ot I \ket{\psi}}\cdot \sqrt{\sum_a \bra{\psi} (P_a)^2 \ot I \ket{\psi}}.
\end{multline*}
The expression inside the first square root is
\begin{align*}
& \sum_a \bra{\psi}(X_a^\dagger \cdot (X - \widehat{X})  \cdot (X - \widehat{X})^\dagger \cdot X_a) \ot I \ket{\psi}\\
\leq~&  \sum_a \bra{\psi} (X_a^\dagger \cdot (X \cdot X^\dagger - I_{m \times m})^2 \cdot X_a) \ot I \ket{\psi}\tag{by \Cref{lem:squared-difference}}\\
=~&  \sum_a \bra{\psi} (Q_a \cdot Q \cdot Q_a - Q_a)  \ot I \ket{\psi}\tag{by \Cref{lem:X-expression-to-Q-expression}}\\
\leq~& 4\sqrt{\zeta} \cdot \bra{\psi} I \ot I \ket{\psi} \tag{by~\Cref{lem:q-almost-projective}}\\
=~& 4\sqrt{\zeta}.
\end{align*}
The expression inside the second square root is at most~$1$ because~$P$ is a sub-measurement.
Next, we claim that the \Cref{eq:add-a-hat} is in fact a real-valued expression.
To see this,
\begin{align*}
\eqref{eq:add-a-hat}
& = \sum_a \bra{\psi} (X_a^\dagger \cdot \widehat{X} \cdot P_a)  \otimes I \ket{\psi}\\
&= \sum_a \bra{\psi} (X^\dagger \cdot T_a \cdot \widehat{X} \cdot \widehat{X}^\dagger \cdot T_a \cdot \widehat{X}) \ot I \ket{\psi}
			\tag{by \Cref{lem:xa-t,lem:pa-restated}}\\
&= \sum_a \bra{\psi} (X^\dagger \cdot T_a \cdot I_{m \times m} \cdot T_a \cdot \widehat{X}) \ot I \ket{\psi} \tag{by \Cref{lem:X-hat-squared}}\\
&= \sum_a \bra{\psi} (X^\dagger \cdot T_a  \cdot \widehat{X}) \ot I \ket{\psi}\\
&=  \bra{\psi} (X^\dagger  \cdot \widehat{X}) \ot I \ket{\psi} \tag{because~$T$ is a measurement}\\
&=  \bra{\psi} \sqrt{Q} \ot I \ket{\psi} \tag{by \Cref{lem:X-times-X-hat}},
\end{align*}
which is real-valued because~$Q$ is positive semidefinite.
As a result, we have
\begin{align*}
\mathfrak{R}\Big(\sum_a \bra{\psi} Q_a P_a \otimes I \ket{\psi}\Big)
& \geq \mathfrak{R}\Big(\sum_a \bra{\psi} (X_a^\dagger \cdot \widehat{X} \cdot P_a)  \otimes I \ket{\psi}\Big) - 2\zeta^{1/4} \tag{by \Cref{eq:add-a-hat}}\\
& = \bra{\psi} \sqrt{Q} \ot I \ket{\psi}- 2\zeta^{1/4}\\
& \geq (1 - 12 \zeta^{1/4}) - 2\zeta^{1/4}. \tag{by \Cref{lem:sqrt-Q-completeness}}
\end{align*}
In total, \Cref{eq:complex-conjugates} shows that
\begin{equation*}
\sum_a \bra{\psi} Q_a P_a \otimes I \ket{\psi}
	+ \sum_a \bra{\psi} P_a Q_a \otimes I \ket{\psi}
\geq 2 \cdot(1 - 14\zeta^{1/4}).
\end{equation*}
Putting everything together, we conclude that
\begin{equation*}
\eqref{eq:P-Q-thing-to-bound}
\leq (1 + 2\sqrt{\zeta}) + 1 - 2\cdot (1-14\zeta^{1/4}) 
= 2\sqrt{\zeta} + 28 \zeta^{1/4}
\leq 30 \zeta^{1/4}.
\end{equation*}
This completes the proof.
\end{proof}

Finally, we have that
\begin{align*}
A_a \ot I
&\approx_{12 \sqrt{\zeta}} Q_a \ot I \tag{by \Cref{lem:projective-low-rank-sum}}\\
&\approx_{30 \zeta^{1/4}} P_a \ot I. \tag{by \Cref{lem:P-Q-approx}}
\end{align*}
Hence, \Cref{prop:triangle-inequality-for-approx_delta} implies that
\begin{equation*}
A_a \ot I \approx_{84 \zeta^{1/4}} P_a \ot I.
\end{equation*}
This completes the proof.
\end{proof}

\section{The main induction step}
\label{sec:induction}

We will now carry out the main inductive argument. 
The inductive hypothesis is stated as follows.

\begin{theorem}[Main induction]\label{thm:main-induction}
  Let $(\psi, A, B, L)$ be an $(\eps, \delta, \gamma)$-good symmetric strategy
  for the $(m,q,d)$ low individual degree test. Let $k \geq md$ be an integer.
  Then there exists a measurement $G \in \polymeas{m}{q}{d}$ such that on
  average over $\bu \sim \F_q^{m}$,
	\begin{equation*}
	  A^{u}_a \otimes I \simeq_{\sigma} I \otimes G_{[g(u)=a]},
	\end{equation*}
  where $\sigma = m^2 \cdot \bigl(\nu + e^{-k/(80000m^2)}\bigr)$ and
  $\nu = 1000k^2 m^2 \cdot\big(\eps^{1/1024} + \delta^{1/1024} + \gamma^{1/1024} + (d/q)^{1/1024}\big)$.
  \ignore{
\item (Completeness): \label{item:self-improvement-G-completeness}  If $G = \sum_g G_g$, then
  	\begin{equation*}
	\bra{\psi} G \otimes I \ket{\psi} \geq 1 - \kappa,
	\end{equation*}
	where
	\begin{equation*}
	\kappa = m^2 \cdot\Big(\nu + e^{- k/(80000m^2)}\Big).
	\end{equation*}
}
\end{theorem}

Comparing with our main theorem, \Cref{thm:main-formal},
\Cref{thm:main-induction} produces a measurement which is consistent with~$A$,
but it is not projective or self-consistent.
In addition, the strategy is assumed to be symmetric.
We correct these deficiencies in the following proof.

\begin{proof}[Proof of~\Cref{thm:main-formal} assuming \Cref{thm:main-induction}]
Suppose $(\psi, A^{\mathrm{A}},B^{\mathrm{A}},L^{\mathrm{A}},A^{\mathrm{B}},B^{\mathrm{B}},L^{\mathrm{B}})$
is a (not necessarily symmetric) strategy
 which passes the $(m, q, d)$-low individual degree test
with probability $(1-\eps)$.
Throughout this proof we will refer to this as the \emph{original strategy}.
Then because each of the three subtests occurs with probability $1/3$,
$(\psi, A^{\mathrm{A}},B^{\mathrm{A}},L^{\mathrm{A}},A^{\mathrm{B}},B^{\mathrm{B}},L^{\mathrm{B}})$ is a $(3\eps, 3\eps, 3\eps)$-good strategy.
We now would like to apply \Cref{thm:main-induction},
but it only applies to strategies which are symmetric.
So we will apply a standard construction to ``symmetrize'' our strategy,
apply \Cref{thm:main-induction},
and then ``unsymmetrize'' the resulting $\{G_g\}$ measurement to obtain $\{G^{\mathrm{A}}_g\}$ and $\{G^{\mathrm{B}}_g\}$ measurements.

For simplicity, we will assume that Player~$\mathrm{A}$ and~$\mathrm{B}$'s
Hilbert spaces $\calH_{\mathrm{A}}$ and $\calH_{\mathrm{B}}$
are both $\C^d$, for some~$d$.
(This argument can be extended to the case of different dimensions in a straightforward manner.)
We will introduce two additional registers, $\calH_{\mathrm{A}'} =\calH_{\mathrm{B'}} = \C^2$,
referred to as the \emph{role registers}.
Then the symmetrized state is given by
\begin{equation*}
\ket{\psi_{\mathrm{sym}}} = \ket{0}_{\mathrm{A}'}\ket{1}_{\mathrm{B}'} \ket{\psi}_{\mathrm{A},\mathrm{B}}
					+\ket{1}_{\mathrm{A}'}\ket{0}_{\mathrm{B}'} \ket{\psi_{\mathrm{swap}}}_{\mathrm{A},\mathrm{B}}
					\in (\C^2_{\mathrm{A}'} \ot \C^d_{\mathrm{A}}) \ot 
						(\C^2_{\mathrm{B}'} \ot \C^d_{\mathrm{B}}),
\end{equation*}
where $\ket{\psi_{\mathrm{swap}}}$ denotes $\ket{\psi}$ with its two registers swapped.
Note that the resulting state $\ket{\psi_{\mathrm{sym}}}$ is symmetric under the exchange of its two registers.
Next, we define the symmetrized measurement $A_{\mathrm{sym}} = \{(A_{\mathrm{sym}})^u_a\}$ as follows
\begin{equation*}
(A_{\mathrm{sym}})^u_a = \ket{0}\bra{0} \ot A^{\mathrm{A},u}_a + \ket{1} \bra{1} \ot A^{\mathrm{B},u}_a,
\end{equation*}
and we define $B_{\mathrm{sym}}$ and $L_{\mathrm{sym}}$ similarly.
The strategy $(\psi_{\mathrm{sym}}, A_{\mathrm{sym}}, B_{\mathrm{sym}}, L_{\mathrm{sym}})$
is symmetric; we refer to it as the \emph{symmetrized strategy}.
It has the following interpretation:
the two players measure their respective role registers in the standard basis;
the one that receives a ``$0$'' will act as Player~A in the original strategy,
and the one that receives a ``$1$'' will act as Player~B in the original strategy.
As a result, the symmetrized strategy is also a $(3\eps, 3\eps, 3\eps)$-good strategy.

Now we apply \cref{thm:main-induction} to the symmetrized strategy.
To do so, set
\begin{align*}
&1000k^2 m^2 \cdot\Big((3\eps)^{1/1024} + (3\eps)^{1/1024} + (3\eps)^{1/1024} + (d/q)^{1/1024}\Big)\\
\leq{}& 10000 k^2 m^2 \cdot \Big(\eps^{1/1024} +  (d/q)^{1/1024}\Big)
=: \nu,
\end{align*}
and set $\sigma = m^2\cdot\left(\nu + e^{-k/(80000m^2)}\right)$.
Then \Cref{thm:main-induction} produces a measurement $G = \{G_g\} \in \polymeas{m}{q}{d}$ such that
\begin{equation}\label{eq:just-applied-induction}
(A_{\mathrm{sym}})^{u}_a \otimes I \simeq_{\sigma} I \otimes G_{[g(u)=a]}, \quad
\text{and}
\quad
G_{[g(u)=a]} \otimes I \simeq_{\sigma} I \otimes (A_{\mathrm{sym}})^u_a.
\end{equation}

Now we unsymmetrize the symmetrized strategy
to derive measurements $\{G^w_g\}$ for $w \in \{\mathrm{A},\mathrm{B}\}$.
Letting $I_d$ denote the $d \times d$ identity operator, define the operators
\begin{gather*}
    G^{\mathrm{A}}_g = (\bra{0} \otimes I_d) \cdot G_g\cdot (\ket{0} \otimes I_d) \\
    G^{\mathrm{B}}_g = (\bra{1} \otimes I_d) \cdot G_g\cdot (\ket{1} \otimes I_d),
\end{gather*}
where $\ket{0}$ and $\ket{1}$ act on the $\C^2$ part of $G_g$. Thus, $G^{\mathrm{A}}_g, G^{\mathrm{B}}_g$ are positive operators acting on $\C^d$, and furthermore they form POVMs:
\[
    \sum_g G^{\mathrm{A}}_g = \sum_g (\bra{0} \otimes I_d) \cdot G_g \cdot (\ket{0} \otimes I_d) = (\bra{0} \otimes I_d) \cdot \Big(\sum_g G_g\Big) \cdot (\ket{0} \otimes I_d) = I_d\;.
\]
The same derivation holds for $G^{\mathrm{B}}_g$. 

Next we verify that the $\{G^{\mathrm{B}}_g\}$ measurements are consistent with the $\{A^{\mathrm{A},u}_a\}$ measurements:
\begin{align*}
    &\E_{\bu} \sum_{g,a \neq g(\bu)} \bra{\psi} A^{\mathrm{A},\bu}_a \otimes G^{\mathrm{B}}_g \ket{\psi} \notag \\
    ={}& \E_{\bu} \sum_{g,a \neq g(\bu)} \bra{\psi} A^{\mathrm{A},\bu}_a \otimes (\bra{1} \otimes I_d) \cdot G_g \cdot (\ket{1} \otimes I_d) \ket{\psi} \notag \\
    ={}& \E_{\bu} \sum_{g,a \neq g(\bu)} (\bra{0,1}_{\mathrm{A'B'}} \otimes \bra{\psi}_{\mathrm{AB}}) \cdot (A_{\mathrm{sym}})^{\bu}_a \otimes G_g \cdot (\ket{0,1}_{\mathrm{A'B'}} \otimes \ket{\psi}_{\mathrm{AB}}) \notag \\
    \leq{}& 2 \cdot \E_{\bu} \sum_{g,a \neq g(\bu)} \bra{\psi_{\mathrm{sym}}} (A_{\mathrm{sym}})^{\bu}_a \otimes G_g \ket{\psi_{\mathrm{sym}}} \notag \\
    \leq{}& 2\sigma. \tag{by \Cref{eq:just-applied-induction}}
\end{align*}
The first inequality follows from the fact that the cross-terms $\bra{0} (A_{\mathrm{sym}})_a^{u} \ket{1}$ and $\bra{1} (A_{\mathrm{sym}})_a^{u} \ket{0}$ vanish by construction of $A_{\mathrm{sym}}$.  Combined with a similar derivation for the $G^{\mathrm{A}}_g$ and $A^{\mathrm{B},u}_a$ operators, we deduce
\begin{gather}
    G^{\mathrm{A}}_{[g(u)=a]} \otimes I \simeq_{2\sigma} I \otimes A^{\mathrm{B},u}_a \label{eq:cons-a}\\
    I \otimes G^{\mathrm{B}}_{[g(u)=a]} \simeq_{2\sigma}  A^{\mathrm{A},u}_a \otimes I \label{eq:cons-b}
\end{gather}
In addition, because the original strategy is $(3\eps, 3\eps, 3\eps)$-good,
\begin{equation*}
A^{\mathrm{A},u}_a \otimes I \simeq_{3\eps} I \otimes A^{\mathrm{B},u}_a.
\end{equation*}
Hence, \Cref{prop:simeq-triangle-inequality} implies that
\begin{equation*}
G^{\mathrm{A}}_{[g(u)=a]} \otimes I \simeq_{2\sigma + 2\sqrt{3\eps + 2\sigma}} I \otimes G^{\mathrm{B}}_{[g(u)=a]}.
\end{equation*}
This implies that
\begin{align*}
2\sigma + 2\sqrt{3\eps + 2\sigma}
& \geq \E_{\bu} \sum_{a \neq b} \bra{\psi} G^{\mathrm{A}}_{[g(\bu) =a]} \ot G^{\mathrm{B}}_{[g(\bu) = b]} \ket{\psi}\\
& = \E_{\bu} \sum_{g \neq h}  \bone[g(\bu) \neq h(\bu)] \cdot \bra{\psi} G^{\mathrm{A}}_{g} \ot G^{\mathrm{B}}_{h} \ket{\psi}\\
& = \E_{\bu} \sum_{g \neq h}  \bra{\psi} G^{\mathrm{A}}_{g} \ot G^{\mathrm{B}}_{h} \ket{\psi}
	- \E_{\bu} \sum_{g \neq h}  \bone[g(\bu) = h(\bu)] \cdot \bra{\psi} G^{\mathrm{A}}_{g} \ot G^{\mathrm{B}}_{h} \ket{\psi}\\
& \geq \E_{\bu} \sum_{g \neq h}  \bra{\psi} G^{\mathrm{A}}_{g} \ot G^{\mathrm{B}}_{h} \ket{\psi}
	-  \sum_{g \neq h}  \frac{md}{q}\cdot \bra{\psi} G^{\mathrm{A}}_{g} \ot G^{\mathrm{B}}_{h} \ket{\psi} \tag{by Schwartz-Zippel}\\
& \geq \E_{\bu} \sum_{g \neq h}  \bra{\psi} G^{\mathrm{A}}_{g} \ot G^{\mathrm{B}}_{h} \ket{\psi} - \frac{md}{q}.
\end{align*}
Rearranging, we get
\begin{equation}\label{eq:G-self-consistency}
G^{\mathrm{A}}_g \otimes I \simeq_{\zeta_1} I \ot G^{\mathrm{B}}_g.
\end{equation}
where $\zeta_1 = 2\sigma + 2\sqrt{3\eps + 2\sigma} + md/q$.

We have now derived everything we wanted, except that~$G$ is not necessarily projective.
To remedy this, we apply the orthogonalization lemma for measurements (\Cref{lem:orthonormalization-main-lemma})
to \Cref{eq:G-self-consistency}.
It implies the existence of two projective sub-measurements $P^{\mathrm{A}} = \{P^{\mathrm{A}}\}, P^{\mathrm{B}} = \{P^{\mathrm{B}}\} \in \polysub{m}{q}{d}$ such that
\begin{align*}
G^{\mathrm{A}}_g \ot I &\approx_{100\zeta_1^{1/4}} P^{\mathrm{A}}_g \ot I,\\
I \ot G^{\mathrm{B}}_g &\approx_{100\zeta_1^{1/4}} I \ot P^{\mathrm{B}}_g.
\end{align*}
Hence, \Cref{prop:completing-to-measurement} implies that we can complete $P^{\mathrm{A}}$ and~$P^{\mathrm{B}}$ to projective measurements
$Q^{\mathrm{A}} = \{Q^{\mathrm{A}}\}, Q^{\mathrm{B}} = \{Q^{\mathrm{B}}\} \in \polymeas{m}{q}{d}$ such that
\begin{align}
G^{\mathrm{A}}_g \ot I &\approx_{\zeta_2} Q^{\mathrm{A}}_g \ot I,\label{eq:G-with-Q-A}\\
I \ot G^{\mathrm{B}}_g &\approx_{\zeta_2} I \ot Q^{\mathrm{B}}_g.\nonumber
\end{align}
where $\zeta_2= 200\zeta_1^{1/4} + 40\zeta_1^{1/8}$.
Now,  \Cref{eq:G-self-consistency} and \Cref{prop:simeq-to-approx} imply that
\begin{equation*}
G^{\mathrm{A}}_g \otimes I \approx_{2\zeta_1} I \ot G^{\mathrm{B}}_g.
\end{equation*}
By the triangle inequality (\Cref{prop:triangle-inequality-for-approx_delta}),
\begin{equation*}
Q_g^{\mathrm{A}} \ot I \approx_{\zeta_3} I \ot Q_g^{\mathrm{B}},
\end{equation*}
where $\zeta_3 = 6\zeta_1 + 6\zeta_2$.
Because both of these measurements are projective, \Cref{prop:simeq-to-approx} then implies that
\begin{equation}\label{eq:third-goal}
Q_g^{\mathrm{A}} \ot I \simeq_{\zeta_3/2} I \ot Q_g^{\mathrm{B}},
\end{equation}
By the data processing inequality (\Cref{prop:simeq-data-processing}),
\begin{equation}\label{eq:just-data-processed-the-heck-outta-this}
Q_{[g(u)=a]}^{\mathrm{A}} \ot I \simeq_{\zeta_3/2} I \ot Q_{[g(u)=a]}^{\mathrm{B}},
\end{equation}
Next, \Cref{prop:triangle-sub}, applied to \Cref{eq:G-self-consistency} and \Cref{eq:G-with-Q-A} implies that
\begin{equation*}
Q^{\mathrm{A}}_g \otimes I \simeq_{\zeta_1} I \ot G^{\mathrm{B}}_g.
\end{equation*}
By data processing (\Cref{prop:simeq-data-processing}),
\begin{equation}\label{eq:ok-almost-there-ok}
Q^{\mathrm{A}}_{[g(u)=a]} \otimes I \simeq_{\zeta_1} I \ot G^{\mathrm{B}}_{[g(u)=a]}.
\end{equation}
Now we apply the triangle inequality (\Cref{prop:simeq-triangle-inequality})
to \Cref{eq:cons-b,eq:ok-almost-there-ok,eq:just-data-processed-the-heck-outta-this},
which implies that
\begin{equation}\label{eq:one-goal}
A^{\mathrm{A},u}_a \ot I \simeq_{\zeta_4} I \ot Q^{\mathrm{B}}_{[g(u)=a]},
\end{equation}
where $\zeta_4 = 2\sigma + 2 \sqrt{\zeta_1 + \zeta_3/2}$.
A similar argument shows that
\begin{equation}\label{eq:another-goal}
I \ot A^{\mathrm{B},u}_a \simeq_{\zeta_4} Q^{\mathrm{A}}_{[g(u)=a]} \ot I.
\end{equation}

Now we calculate the error.
First,
\begin{align*}
\sigma
&= m^2\cdot\left(10000 k^2 m^2 \cdot \Big(\eps^{1/1024} +  (d/q)^{1/1024}\Big) + e^{-k/(80000m^2)}\right)\\
&\leq 10000 k^2 m^4 \cdot \Big(\eps^{1/1024} +  (d/q)^{1/1024} + e^{-k/(80000m^2)}\Big).
\end{align*}
Next, using $10000^{1/2} = 100$,
\begin{align*}
\zeta_1 &= 2\sigma + 2\sqrt{3\eps + 2\sigma} + md/q\\
& \leq 2\cdot\Big(10000 k^2 m^4 \cdot \Big(\eps^{1/1024} +  (d/q)^{1/1024} + e^{-k/(80000m^2)}\Big)\Big)\\
&		\qquad\qquad+ 2\cdot \Big( 3 \eps + 10000 k^2 m^4 \cdot \Big(\eps^{1/1024} +  (d/q)^{1/1024} + e^{-k/(80000m^2)}\Big)\Big)^{1/2} + md/q\\
& \leq 2\cdot\Big(10000 k^2 m^4 \cdot \Big(\eps^{1/1024} +  (d/q)^{1/1024} + e^{-k/(80000m^2)}\Big)\Big)\\
&		\qquad\qquad+ 2\cdot \Big( 2 \eps^{1/2} + 100 k m^2 \cdot \Big(\eps^{1/2048} +  (d/q)^{1/2048} + e^{-k/(160000m^2)}\Big)\Big) + md/q\\
& \leq 20204 k^2 m^4 \cdot \Big(\eps^{1/2048} +  (d/q)^{1/2048} + e^{-k/(160000m^2)}\Big).
\end{align*}
Next, using the fact that $20204^{1/4} \leq 12$ and $20204^{1/8} \leq 4$,
\begin{align*}
\zeta_2
&= 200\zeta_1^{1/4} + 40\zeta_1^{1/8}\\
& \leq 200 \Big(20204 k^2 m^4 \cdot \Big(\eps^{1/2048} +  (d/q)^{1/2048} + e^{-k/(160000m^2)}\Big)\Big)^{1/4}\\
&		\qquad\qquad+ 40 \Big(20204 k^2 m^4 \cdot \Big(\eps^{1/2048} +  (d/q)^{1/2048} + e^{-k/(160000m^2)}\Big)\Big)^{1/8}\\
& \leq 200 \Big(12 k m \cdot \Big(\eps^{1/8192} +  (d/q)^{1/8192} + e^{-k/(640000m^2)}\Big)\Big)\\
&		\qquad\qquad+ 40 \Big(4 k m
		\cdot \Big(\eps^{1/16384} +  (d/q)^{1/16384} + e^{-k/(1280000m^2)}\Big)\Big)\\
& \leq 2560 km \cdot \Big(\eps^{1/16384} + (d/q)^{1/16384} + e^{-k/(1280000m^2)}\Big).
\end{align*}
Next, using $6\cdot 20204 + 6 \cdot 2560 \leq 150000$,
\begin{equation*}
\zeta_3 = 6 \zeta_1 + 6 \zeta_2
\leq 150000 k^2 m^4 \cdot \Big(\eps^{1/16384} + (d/q)^{1/16384} + e^{-k/(1280000m^2)}\Big).
\end{equation*}
Next, using $\sqrt{20204} \leq 143$, $\sqrt{150000} \leq 388$, and $2 \cdot (10000+143+388) \leq 40000$,
\begin{align*}
\zeta_4  &= 2\sigma + 2 \sqrt{\zeta_1 + \zeta_3/2}\\
& \leq 2 \cdot \Big(10000 k^2 m^4 \cdot \Big(\eps^{1/1024} +  (d/q)^{1/1024} + e^{-k/(80000m^2)}\Big)\Big)\\
& \qquad \qquad + 2 \cdot \Big(20204 k^2 m^4 \cdot \Big(\eps^{1/2048} +  (d/q)^{1/2048} + e^{-k/(160000m^2)}\Big) \\
& \qquad \qquad \qquad  \qquad+ 150000 k^2 m^4 \cdot \Big(\eps^{1/16384} + (d/q)^{1/16384} + e^{-k/(1280000m^2)}\Big)\Big)^{1/2}\\
& \leq 2 \cdot \Big(10000 k^2 m^4 \cdot \Big(\eps^{1/1024} +  (d/q)^{1/1024} + e^{-k/(80000m^2)}\Big)\Big)\\
& \qquad \qquad + 2 \cdot \Big(143 k m^2 \cdot \Big(\eps^{1/4096} +  (d/q)^{1/4096} + e^{-k/(320000m^2)}\Big) \\
& \qquad \qquad \qquad  \qquad+ 388 k m^2 \cdot \Big(\eps^{1/32768} + (d/q)^{1/32768} + e^{-k/(2560000m^2)}\Big)\Big)\\
& \leq 40000k^2m^4 \cdot\Big(\eps^{1/32768} + (d/q)^{1/32768} + e^{-k/(2560000m^2)}\Big)\Big).
\end{align*}
Both this and $\zeta_3/2$ are less than
\begin{equation*}
100000k^2m^4 \cdot\Big(\eps^{1/40000} + (d/q)^{1/40000} + e^{-k/(2560000m^2)}\Big)\Big).
\end{equation*}
Hence, \Cref{eq:one-goal,eq:another-goal,eq:third-goal} provide the three bounds we want.
This concludes the proof.
\end{proof}

The remainder of this section is organized as follows:
first, in \Cref{sec:self-improvement-and-pasting}, we define the two main steps in the proof of \Cref{thm:main-induction}
known as self-improvement and pasting.
Following that, we prove \Cref{thm:main-induction} in \Cref{sec:proof-of-main-induction}.

\subsection{Self-improvement and pasting}\label{sec:self-improvement-and-pasting}

There are two main steps in the proof of \Cref{thm:main-induction}.
The first is self-improvement, which is stated as follows.

\begin{theorem}[Self-improvement]\label{thm:self-improvement-in-induction-section}
  Let $(\psi, A, B, L)$ be an $(\eps, \delta, \gamma)$-good symmetric strategy for the $(m,q,d)$ low individual degree test.
Let $G \in \polysub{m}{q}{d}$ be a sub-measurement with the following properties:
\begin{enumerate}
\ignore{
\item (Completeness): \label{item:si-inductive-G-completeness}  If $G =  \sum_g G_g$, then
  	\begin{equation*}
	\bra{\psi} G \otimes I \ket{\psi} \geq 1 - \kappa.
	\end{equation*}
}
  \item (Consistency with~$A$): \label{item:si-inductive-G-consistency} On average over $\bu \sim \F_q^{m}$,
	  \begin{equation*}
	  A^{u}_a \otimes I \simeq_{\nu} I \otimes G_{[g(u)=a]}.
	  \end{equation*}
	  \end{enumerate}
Let
\begin{equation*}
\zeta = 3000m\cdot \Big(\eps^{1/32} + \delta^{1/32} + (d/q)^{1/32}\Big).
\end{equation*}
Then there exists a projective sub-measurement $H \in \polysub{m}{q}{d}$ with the following properties:
\begin{enumerate}
    \item(Completeness):  \label{item:si-inductive-completeness} If $H =  \sum_h H_h$, then
    \begin{equation*}
    \bra{\psi} H \otimes I \ket{\psi} \geq (1-\nu)-\zeta.
    \end{equation*}
    \item(Consistency with~$A$):\label{item:si-inductive-A-consistency} On average over $\bu \sim \F_q^m$,
        \begin{equation*}
            A^u_a \otimes I \simeq_{\zeta} I \otimes H_{[h(u) = a]}.
        \end{equation*}
    \item(Strong self-consistency): \label{item:si-inductive-self}
    	\begin{equation*}
		H_h \otimes I \approx_{\zeta} I \otimes H_h.
	\end{equation*}
    \item(Boundedness): \label{item:si-inductive-boundedness} There exists a positive-semidefinite matrix~$Z$ such that
    \begin{equation*}
        \bra{\psi} Z \otimes (I - H) \ket{\psi} \leq \zeta
    \end{equation*}
    and for each $h \in \polyfunc{m}{q}{d}$,
	\begin{equation*}
	Z \geq \left(\E_{\bu} A^{\bu}_{h(\bu)}\right).
	\end{equation*}
\end{enumerate}
\end{theorem}

We note that by \Cref{prop:two-notions-of-self-consistency} the condition in \Cref{item:si-inductive-self}
is equivalent to~$H$'s strong self-consistency because~$H$ is projective.
Self-improvement states that we can take a measurement~$G$
whose consistency error with~$A$ is~$\nu$
and produce another measurement~$H$
which has negligible consistency error with~$A$
and incompleteness~$\nu$.
Hence, we have ``moved'' $G$'s consistency error onto $H$'s incompleteness.

The second main step is pasting, which is stated as follows.

\begin{theorem}[Pasting]\label{thm:ld-pasting-in-induction-section}
  Let $(\psi, A, B, L)$ be an $(\eps, \delta, \gamma)$-good symmetric strategy for the $(m+1,q,d)$ low individual degree test.
  Let $\{G^x\}_{x \in \F_q}$ denote a set of projective sub-measurements in $\polysub{m}{q}{d}$ with the following properties:
  \begin{enumerate}
  \item (Completeness): \label{item:ld-pasting-inductive-completeness}  If $G = \E_{\bx} \sum_g G^{\bx}_g$, then
  	\begin{equation*}
	\bra{\psi} G \otimes I \ket{\psi} \geq 1 - \kappa.
	\end{equation*}
  \item (Consistency with~$A$): \label{item:ld-pasting-inductive-consistency} On average over $(\bu, \bx) \sim \F_q^{m+1}$,
	  \begin{equation*}
	  A^{u, x}_a \otimes I \simeq_{\zeta} I \otimes G^x_{[g(u)=a]}.
	  \end{equation*}
  \item (Strong self-consistency): \label{item:ld-pasting-inductive-self-consistency} On average over~$\bx \sim \F_q$,
  	\begin{equation*}
		G^x_g \otimes I \approx_{\zeta} I \otimes G^x_g.
	\end{equation*}
  \item (Boundedness): \label{item:ld-pasting-inductive-boundedness} There exists a positive-semidefinite matrix $Z^x$ for each $x \in \F_q$ such that
  	\begin{equation*}
		\E_{\bx} \bra{\psi} (I-G^{\bx})\otimes Z^{\bx} \ket{\psi} \leq \zeta
	\end{equation*}
	and for each $x \in \F_q$ and $g \in \polyfunc{m}{q}{d}$,
	\begin{equation*}
	Z^x \geq \left(\E_{\bu} A^{\bu, x}_{g(\bu)}\right).
	\end{equation*}
  \end{enumerate}
  Let $k \geq 400md$ be an integer. Let
 \begin{align*}
 \nu &= 100 k^2m \cdot \left(\eps^{1/32} + \delta^{1/32} + \gamma^{1/32} + \zeta^{1/32} + (d/q)^{1/32}\right),\\
 \sigma& = \kappa \cdot \left(1 + \frac{1}{100m}\right)
    + 2\nu + e^{- k/(80000m^2)}.
 \end{align*}
  Then there exists a ``pasted" measurement $H \in \polymeas{m+1}{q}{d}$ which satisfies the following property.
  \begin{enumerate}
  \item (Consistency with~$A$): \label{item:ld-pasting-inductive-N-consistency} On average over $\bu \sim \F_q^{m+1}$,
	  \begin{equation*}
	  A^{u}_a \otimes I \simeq_{\sigma} I \otimes H_{[h(u)=a]}.
	  \end{equation*}
	  \ignore{
  \item (Completeness): \label{item:ld-pasting-inductive-N-completeness} If $H =  \sum_h H_h$, then
  	\begin{equation*}
	\bra{\psi} H \otimes I \ket{\psi} \geq 1 - \kappa \cdot \left(1 + \frac{1}{100m}\right)
    - \nu - e^{- k/(80000m^2)}.
	\end{equation*}
	}
  \end{enumerate}
\end{theorem}

Intuitively, $\sigma$ should be thought of as being roughly~$\kappa$,
plus a small amount of error which does not depend on~$\kappa$.
Hence, pasting states that we can take a family of sub-measurements $\{G^x\}$ with incompleteness $\kappa$
and produce a pasted measurement~$H$ whose consistency error with~$A$ is roughly~$\kappa$,
plus a small amount of new error.
Thus, the overall inductive step looks as follows:
given a family of measurements with some inconsistency error,
we ``move'' the error into the incompleteness using self-improvement,
and then we paste the measurements together to form a single measurement whose error is roughly the same as the original error.

We note that the main error term $\nu$ depends only on the ``small" parameters $\eps$, $\delta$, $\zeta$, $\gamma$, and $d/q$ and not on the ``large" parameter $\kappa$.

\subsection{Proof of \Cref{thm:main-induction}}\label{sec:proof-of-main-induction}

\begin{definition}
Let $x \in \F_q$.
For each line $\ell \in \F_q^m$,
we define $\mathrm{append}_x(\ell)$ to be the line in $\F_q^{m+1}$ containing every point $(u, x)$ such that $u \in \ell$.
In addition, for each function $f : \ell \rightarrow \F_q$,
we define $\mathrm{append}_x(f) : \mathrm{append}_x(\ell) \rightarrow \F_q$
such that for each $(u, x) \in \mathrm{append}_x(\ell)$, $\mathrm{append}_x(f)(u, x) = f(u)$.
\end{definition}

\begin{definition}[$x$-restricted low-degree strategy]
Let $(\psi, A, B, L)$  be a symmetric strategy the $(m+1,q,d)$-low individual degree test.
Given $x \in \F_q$, we define the \emph{$x$-restricted strategy}
$(\psi, A^x, B^x, L^x)$ for the $(m,q,d)$-low individual degree test as follows.
\begin{enumerate}
\item For each $u \in \F_q^m$, $(A^x)^u_a = A^{u, x}_a$.
\item For each axis parallel line $\ell \in \F_q^m$, $(B^x)^\ell_f = B^{\mathrm{append}_x(\ell)}_{\mathrm{append}_x(f)}$.
\item For each  line $\ell \in \F_q^m$, $(L^x)^\ell_f = L^{\mathrm{append}_x(\ell)}_{\mathrm{append}_x(f)}$.
\end{enumerate}
\end{definition}

\begin{lemma}\label{lem:restricted-probabilities}
Let $(\psi, A, B, L)$ be  an $(\eps, \delta, \gamma)$-good symmetric strategy
		for the $(m+1, d, q)$-low individual degree test.
For each $x \in \F_q$, let $(\psi, A^x, B^x, L^x)$ be the corresponding $x$-restricted strategy.
In addition, write $\eps_x$ for the probability that it fails the axis-parallel lines test,
$\delta_x$ for the probability it fails the self-consistency test,
and $\gamma_x$ for the probability it fails the diagonal lines test.
Then
\begin{equation*}
\E_{\bx} \eps_{\bx} \leq \left(\frac{m+1}{m}\right) \cdot\eps,
\quad
\E_{\bx} \delta_{\bx} \leq  \delta,
\quad
\E_{\bx} \gamma_{\bx} \leq \left(\frac{m+1}{m}\right) \cdot \gamma.
\end{equation*}
\end{lemma}
\begin{proof}
We begin with the self-consistency test.
Here, both provers are given a uniformly random point $(\bu, \bx)$,
they measure using $A^{\bu, \bx} = (A^{\bx})^{\bu}$, and they succeed if their outcomes are the same.
This is equivalent to performing the self-consistency test on the $\bx$-restricted strategy, averaged over $\bx$,
and so $\E_{\bx} \delta_{\bx} \leq  \delta$. (It is a ``$\leq$" rather than an ``$=$" because $\delta$ is just an upper-bound on the failure probability.)

Next, we consider the axis-parallel lines test.
Suppose the provers are sent the line $\bell$ and the point $(\bu, \bx) \in \bell$.
With probability $\frac{1}{m+1}$, $\bell$ is parallel to the $(m+1)$-st direction.
When it is not, then $\bell = \mathrm{append}_x(\bell')$, where $\bell'$ is an axis-parallel line in $\F_q^m$.
In this case, the points prover measures with the measurement $A^{\bu, \bx} = (A^\bx)^{\bu}$ and receives an outcome $\ba$,
the lines prover measures with the measurement $B^{\bell} = (B^{\bx})^{\bell'}$ and receives an outcome $\boldf = \mathrm{append}_{\bx}(\boldf')$, and they succeed if $\boldf(\bu, \bx) = \ba$, or, equivalently, if $\boldf'(\bu) = \ba$.
Hence, the probability that they succeed is equal to the probability that the $\bx$-restricted strategy passes the $(m, q, d)$-low individual degree test, which is $\eps_{\bx}$. As a result,
\begin{align*}
\eps
&\geq \Pr_{\bell, \bu, \bx}[\text{$A$ and $B$ succeed given $\bell$, $(\bu, \bx)$}]\\
&\geq \left(\frac{m}{m+1}\right) \cdot \Pr_{\bell, \bu, \bx}[\text{$A$ and $B$ succeed given $\bell$, $(\bu, \bx)$} \mid \text{$\bell$ is not parallel to direction $m+1$}]\\
& =  \left(\frac{m}{m+1}\right) \cdot\E_{\bx} \eps_{\bx}.
\end{align*}

Finally, the analysis of the diagonal lines test follows the same proof as the axis-parallel lines test, and we omit it here.
\end{proof}

\begin{proof}[Proof of \Cref{thm:main-induction}]
We note that the bound we are proving is trivial when at least one of $\eps$, $\delta$, $\gamma$,  or $d/q$ is $\geq 1$, as $\nu$ is at least~$1$ in that case. Hence, we may assume that $\eps, \delta, \gamma, d/q \leq 1$. This will aid us when carrying out the error calculations, as it allows us to bound terms like $(d/q)^{1/2}$ by terms like $(d/q)^{1/4}$.

The proof is by induction on~$m$.
The base case is when $m = 1$.
In this case, there is only one axis-parallel line~$\ell$ in $\F_q^m$,
and so $B^\ell \in \polymeas{m}{q}{d}$.
Because this strategy fails the axis-parallel line test with probability at most~$\eps$,
\begin{equation*}
A^u_a \ot I \simeq_{\eps} I \ot B^\ell_{[f(u) = a]},
\end{equation*}
on average over $\bu \sim \F_q$.
We note that this bound holds independent of the value of~$k$.
This is even better than the theorem demands, and so the theorem is proved.

Now we perform the induction step.
Assuming that \Cref{thm:main-induction} holds for~$m \geq 1$,
we will show that it holds for~$m+1$ as well.
Let $(\psi, A, B, L)$ be an $(\eps, \delta, \gamma)$-good symmetric strategy
			for the $(m+1,q,d)$ low individual degree test.
Let $k \geq (m+1)d$ be an integer.

For each $x \in \F_q$, let $(\psi, A^x, B^x, L^x)$ be the corresponding $x$-restricted strategy.
In addition, write $\eps_x$ for the probability that it fails the axis-parallel lines test,
$\delta_x$ for the probability it fails the self-consistency test,
and $\gamma_x$ for the probability it fails the diagonal lines test.

For each $x \in \F_q$, we apply the inductive hypothesis to the $x$-restricted strategy with the same integer~$k$.
This is possible because $k \geq (m+1)d \geq md$.
Let
\begin{equation*}
\nu_x = 1000k^2 m^2 \cdot\Big(\eps_x^{1/1024} + \delta_x^{1/1024} + \gamma_x^{1/1024} + (d/q)^{1/1024}\Big),
\end{equation*}
and
\begin{equation*}
\sigma_x =  m^2 \cdot\Big(\nu_x + e^{- k/(80000m^2)}\Big)
\end{equation*}
Then the inductive hypothesis states that there exists a measurement $G^x \in \polymeas{m}{q}{d}$ such that
\begin{equation*}
(A^x)^u_a\ot I \simeq_{\sigma_x} I \ot G^x_{[g(u)=a]}.
\end{equation*}

Next, we apply self-improvement to each~$G^x$.
Let 
\begin{equation*}
\zeta_x = 3000m \cdot \Big(\eps_x^{1/32} + \delta_x^{1/32} + (d/q)^{1/32}\Big).
\end{equation*}
Then \Cref{thm:self-improvement-in-induction-section} produces
a projective sub-measurement $\widehat{G}^x \in \polysub{m}{q}{d}$ such that
for each $x \in \F_q$, the following statements hold.
\begin{enumerate}
    \item(Completeness): If $\widehat{G}^x =  \sum_g \widehat{G}^x_g$, then
    \begin{equation*}
    \bra{\psi} \widehat{G}^x \ot I \ket{\psi} \geq (1-\sigma_x) - \zeta_x.
    \end{equation*}
    \item(Consistency with~$A^x$): On average over $\bu \sim \F_q^m$,
        \begin{equation*}
            (A^x)^u_a \ot I \simeq_{\zeta_x} I \ot \widehat{G}^x_{[g(u) = a]}.
        \end{equation*}
    \item(Strong self-consistency):
    	\begin{equation*}
		\widehat{G}^x_g \ot I \approx_{\zeta_x} I \ot \widehat{G}^x_g.
	\end{equation*}
    \item(Boundedness): There exists a positive-semidefinite matrix~$Z^x$ such that
    \begin{equation*}
        \bra{\psi} Z^x \ot (I - \widehat{G}^x) \ket{\psi} \leq \zeta_x
    \end{equation*}
    and for each $g \in \polyfunc{m}{q}{d}$,
	\begin{equation*}
	Z^x \geq \Big( \E_{\bu} (A^x)^{\bu}_{g(\bu)}\Big).
	\end{equation*}
\end{enumerate}

Having produced the $\widehat{G}^x$'s, we would like to paste them together.
To do so, we need bounds for the above four properties
which are stated on average over~$\bx \sim \F_q$ rather than for each $x \in \F_q$ individually.
This involves computing ``averaged" versions of our error parameters  $\nu_x$, $\sigma_x$, and $\zeta_x$.
In these derivations, we will crucially use the fact that $\alpha \mapsto \alpha^c$ is concave when $c \leq 1$,
and hence $\E (\balpha)^c \leq (\E \balpha)^c$.
\begin{align*}
\E_{\bx} \nu_{\bx}
& = \E_{\bx}\Big(1000k^2 m^2 \cdot\Big(\eps_{\bx}^{1/1024} + \delta_{\bx}^{1/1024} + \gamma_{\bx}^{1/1024} + (d/q)^{1/1024}\Big)\Big)\\
& \leq1000k^2 m^2 \cdot\Big((\E_{\bx}\eps_{\bx})^{1/1024} + (\E_{\bx} \delta_{\bx})^{1/1024} + (\E_{\bx}\gamma_{\bx})^{1/1024} + (d/q)^{1/1024}\Big)\tag{by concavity}\\
&\leq1000k^2 m^2 \cdot\Big(\Big(\frac{(m+1)}{m} \cdot \eps\Big)^{1/1024} + \delta^{1/1024} + \Big(\frac{(m+1)}{m} \cdot \gamma\Big)^{1/1024} + (d/q)^{1/1024}\Big) \tag{by \Cref{lem:restricted-probabilities}}\\ 
&\leq 1000k^2 (m+1)^2 \cdot\Big(\eps^{1/1024} + \delta^{1/1024} + \gamma^{1/1024} + (d/q)^{1/1024}\Big).
\end{align*}
We call this value~$\nu$.
Next, if we define
\begin{equation*}
\sigma = m^2 \cdot\Big(\nu + e^{- k/(80000m^2)}\Big),
\end{equation*}
then
\begin{equation*}
\sigma \geq m^2 \cdot\Big(\E_{\bx} \nu_{\bx} + e^{- k/(80000m^2)}\Big)
	= \E_{\bx} \sigma_{\bx}.
\end{equation*}
Finally,
\begin{align*}
\E_{\bx} \zeta_{\bx}
& = \E_{\bx} \Big(3000m \cdot \Big(\eps_{\bx}^{1/32} + \delta_{\bx}^{1/32} + (d/q)^{1/32}\Big)\Big)\\
& \leq 3000m \cdot \Big((\E_{\bx}\eps_{\bx})^{1/32} + (\E_{\bx} \delta_{\bx})^{1/32} + (d/q)^{1/32}\Big)
	\tag{by concavity}\\
&\leq 3000 m \cdot \Big(\Big( \frac{(m+1)}{m} \cdot \eps\Big)^{1/32} + \delta^{1/32} + (d/q)^{1/32}\Big)
	\tag{by \Cref{lem:restricted-probabilities}}\\
&\leq3000 (m+1) \cdot \Big(\eps^{1/32} + \delta^{1/32} + (d/q)^{1/32}\Big).
\end{align*}
We call this value~$\zeta$.
We note for later that
\begin{align}
\zeta & = 3000 (m+1) \cdot \Big(\eps^{1/32} + \delta^{1/32} + (d/q)^{1/32}\Big)\nonumber\\
& \leq 1000 k^2 (m+1)^2\cdot \Big(\eps^{1/32} + \delta^{1/32} + (d/q)^{1/32}\Big) \tag{because $m \geq 2$}\nonumber\\
&\leq 1000k^2 (m+1)^2 \cdot\Big(\eps^{1/1024} + \delta^{1/1024} + \gamma^{1/1024} + (d/q)^{1/1024}\Big)\nonumber\\
& = \nu. \label{eq:zeta-smaller-than-nu}
\end{align}
Having defined these, the following statements hold.
\begin{enumerate}
    \item(Completeness): If $\widehat{G} =  \E_{\bx} \sum_g \widehat{G}^{\bx}_g$, then
    \begin{equation*}
    \bra{\psi} \widehat{G} \ot I \ket{\psi} \geq (1-\sigma) - \zeta.
    \end{equation*}
    \item(Consistency with~$A$): On average over $(\bu,\bx) \sim \F_q^{m+1}$,
        \begin{equation*}
            A^{u,x}_a \ot I \simeq_{\zeta} I \ot \widehat{G}^x_{[g(u) = a]}.
        \end{equation*}
    \item(Strong self-consistency): On average over $\bx \sim \F_q$,
    	\begin{equation*}
		\widehat{G}^x_g \ot I \approx_{\zeta} I \ot \widehat{G}^x_g.
	\end{equation*}
    \item(Boundedness): There exists a positive-semidefinite matrix~$Z^x$ for each $x \in \F_q$ such that
    \begin{equation*}
        \E_{\bx} \bra{\psi} Z^{\bx} \ot (I - \widehat{G}^{\bx}) \ket{\psi} \leq \zeta
    \end{equation*}
    and for each $x \in\F_q$ and $g \in \polyfunc{m}{q}{d}$,
	\begin{equation*}
	Z^x \geq \Big( \E_{\bu} A^{\bu,x}_{g(\bu)}\Big).
	\end{equation*}
\end{enumerate}

We are now ready to apply \Cref{thm:ld-pasting-in-induction-section}.
To do so, we note that because $(3000)^{1/32} \leq 2$
and $32 \cdot 32 = 1024$,
\begin{align*}
\zeta^{1/32}
&= \Big(3000 (m+1) \cdot \Big(\eps^{1/32} + \delta^{1/32} + (d/q)^{1/32}\Big)\Big)^{1/32}\\
&\leq 2 (m+1) \cdot \Big(\eps^{1/1024} + \delta^{1/1024} + (d/q)^{1/1024}\Big).
\end{align*}
Hence,
\begin{align*}
&100 k^2m \cdot \left(\eps^{1/32} + \delta^{1/32} + \gamma^{1/32} + \zeta^{1/32} + (d/q)^{1/32}\right)\\
\leq~&100 k^2m \cdot \left(\eps^{1/32} + \delta^{1/32} + \gamma^{1/32} + 2 (m+1) \cdot \Big(\eps^{1/1024} + \delta^{1/1024} + (d/q)^{1/1024}\Big)+ (d/q)^{1/32}\right)\\
\leq~& 200 k^2 m(m+1) \cdot\left(\eps^{1/1024} + \delta^{1/1024} + \gamma^{1/1024} + \eps^{1/1024} + \delta^{1/1024} + (d/q)^{1/1024}+ (d/q)^{1/1024}\right)\\
\leq~& 1000 k^2 (m+1)^2 \cdot\left(\eps^{1/1024} + \delta^{1/1024} + \gamma^{1/1024} + (d/q)^{1/1024}\right)\\
=~& \nu.
\end{align*}
Then \Cref{thm:ld-pasting-in-induction-section} implies the existence of a
pasted measurement $H \in \polysub{m+1}{q}{d}$ which satisfies the following
property.
On average over $\bu \sim \F_q^{m+1}$,
\begin{equation*}
  A^{u}_a \otimes I \simeq_{\sigma^*} I \otimes H_{[h(u)=a]},
\end{equation*}
where
\begin{equation*}
  \sigma^* = (\sigma + \zeta) \cdot \left(1 + \frac{1}{100m}\right)
  + 2\nu + e^{- k/(80000m^2)}.
\end{equation*}
	  \ignore{
  \item (Completeness): \label{item:ld-pasting-inductive-N-completeness} If $H =  \sum_h H_h$, then
  	\begin{equation*}
	\bra{\psi} H \otimes I \ket{\psi} \geq 1 - \kappa \cdot \left(1 + \frac{1}{100m}\right)
    - \nu - e^{- k/(80000m^2)}.
	\end{equation*}
	}

The consistency with~$A$ is as guaranteed in the theorem statement.
Hence, we need only verify that the completeness bound implies the one in the theorem statement as well.
\begin{align}
\sigma^*
&= (\sigma + \zeta) \cdot \left(1 + \frac{1}{100m}\right)
						    + 2\nu + e^{- k/(80000m^2)}\nonumber\\
&\leq  (\sigma + \nu) \cdot \left(1 + \frac{1}{100m}\right)
						    + 2\nu + e^{- k/(80000m^2)} \tag{by \Cref{eq:zeta-smaller-than-nu}}\nonumber\\
& = \left(1 + \frac{1}{100m}\right) \cdot \left(m^2 \cdot\Big(\nu + e^{- k/(80000m^2)}\Big) + \nu\right)
		+ 2\nu + e^{- k/(80000m^2)}\nonumber\\
&\leq \left(1 + \frac{1}{100m}\right) \cdot (m^2+3) \cdot\Big(\nu + e^{- k/(80000m^2)}\Big).\label{eq:gonna-bound-m-function}
\end{align}
Now, because $m \geq 2$,
\begin{equation*}
\frac{1}{100m} \cdot
(m^2+3)
\leq
\frac{1}{100m}\cdot
(m^2+4m-5)
= 
\frac{1}{100m}\cdot
(m-1) (m+5)
\leq m-1
\leq 2(m-1).
\end{equation*}
Hence,
\begin{equation*}
\left(1 + \frac{1}{100m}\right) \cdot (m^2+3)
= m^2 + 3 + \frac{1}{100m} \cdot
(m^2+3)
\leq m^2 + 3 + 2(m-1)
= m^2 + 2m  + 1
= (m+1)^2.
\end{equation*}
As a result,
\begin{equation*}
\eqref{eq:gonna-bound-m-function}
\leq (m+1)^2\cdot\Big(\nu + e^{- k/(80000m^2)}\Big)
\leq (m+1)^2\cdot\Big(\nu + e^{- k/(80000(m+1)^2)}\Big).
\end{equation*}
\ignore{
\begin{align*}
&\kappa \cdot \left(1 + \frac{1}{100m}\right) + \nu + e^{- k/(80000m^2)}\\
=~&\left(1 + \frac{1}{100m}\right) \cdot m^2 \cdot\Big(\nu + e^{- k/(80000m^2)}\Big) + \nu+ e^{- k/(80000m^2)}\\
\leq~&\left(1 + \frac{1}{m}\right) \cdot m^2 \cdot\Big(\nu + e^{- k/(80000m^2)}\Big) + \nu+ e^{- k/(80000m^2)}\\
=~& m(m+1) \cdot\Big(\nu + e^{- k/(80000m^2)}\Big) + \nu+ e^{- k/(80000m^2)}\\
\leq~& (m+1)^2 \cdot \Big(\nu + e^{- k/(80000m^2)}\Big)\\
\leq~& (m+1)^2 \cdot \Big(\nu + e^{- k/(80000(m+1)^2)}\Big).
\end{align*}
}
This is the bound guaranteed by the theorem and so it completes the proof.
\end{proof}


\section{Expansion in the hypercube graph}
\label{sec:expansion}

\begin{definition}[Hypercube graph]
The \emph{hypercube graph} $C = (V,E)$ is the graph with vertex set $V = \F_q^m$ and an edge between $u, v \in V$ whenever $u$ and $v$ disagree in at most one coordinate (so that every vertex is connected to itself).
A random edge in~$C$, denoted $(\bu, \bv) \sim C$,
is distributed as follows:
draw $\bu \sim \F_q^m$, $\bi \sim \{1, \ldots, m\}$,  and $\bx \sim \F_q$, all uniformly at random, and set $\bv = \bu + \bx \cdot e_{\bi}$.
\end{definition}

We will use $M$ to denote the number of vertices in~$C$, i.e.\ $M = q^m$.
\ignore{
We note that $|E| =\frac{1}{2} \cdot n N$.
The following proposition gives a second interpretation for a uniformly random edge in~$C$.

\begin{proposition}\label{prop:rerandomize-coord}
The following two distributions are identical.
\begin{enumerate}
    \item Output $(\bu, \bv) \sim C$.
    \item Draw $\bu \sim \F_q^n$, $\bi \sim [n]$,  and $\bx \sim \F_q$, all uniformly at random, and output
    $(\bu, \bu + \bx \cdot e_{\bi})$.
\end{enumerate}
\end{proposition}
}
\subsection{Eigenvalues of the hypercube graph}

\begin{definition}[Adjacency matrix]
The \emph{normalized adjacency matrix of~$C$} is the matrix~$K$ defined as
\begin{equation*}
K =  \E_{(\bu, \bv) \sim C} \ket{\bu}\bra{\bv}.
\end{equation*}
The \emph{Laplacian of~$C$} is the matrix
\begin{equation*}
L = \frac{1}{M} \cdot I - K.
\end{equation*}
\end{definition}

The following proposition gives another convenient way of writing the Laplacian of~$C$.

\begin{proposition}\label{prop:laplacian-rewrite}
$\displaystyle L = \frac{1}{2} \cdot \E_{(\bu, \bv) \sim C} (\ket{\bu} - \ket{\bv}) \cdot (\bra{\bu} - \bra{\bv}).$
\end{proposition}
\begin{proof}
If we draw $(\bu, \bv) \sim C$, then both $\bu$ and~$\bv$ are distributed as uniformly random elements of~$\F_q^m$. As a result,
\begin{equation*}
\frac{1}{M} \cdot I = \E_{\bu \in \F_q^m} \ket{\bu}\bra{\bu} = \frac{1}{2}\cdot\E_{(\bu, \bv) \sim C}  \ket{\bu}\bra{\bu} + \ket{\bv}\bra{\bv}.
\end{equation*}
In addition, $(\bu, \bv)$ is distributed identically to $(\bv, \bu)$. As a result,
\begin{equation*}
K =  \E_{(\bu, \bv) \sim C} \ket{\bu}\bra{\bv} = \frac{1}{2} \cdot \E_{(\bu, \bv) \sim C} \ket{\bu}\bra{\bv} + \ket{\bv}\bra{\bu}.
\end{equation*}
Combining these two,
\begin{equation*}
L = \frac{1}{2}\cdot\E_{(\bu, \bv) \sim C} \left[\ket{\bu}\bra{\bu} + \ket{\bv}\bra{\bv} - \ket{\bu}\bra{\bv} - \ket{\bv}\bra{\bu}\right]
= \frac{1}{2} \cdot \E_{(\bu, \bv) \sim C} (\ket{\bu} - \ket{\bv}) \cdot (\bra{\bu} - \bra{\bv}).\qedhere
\end{equation*}
\end{proof}

The most important properties of the adjacency matrix are its eigenvalues and eigenvectors. These are provided in the next proposition, which is standard in the literature.

\begin{proposition}\label{prop:eigenvectors}
For each $\alpha \in \F_q^m$, define the vector
\begin{equation*}
\ket{\varphi_\alpha} := \frac{1}{M^{1/2}} \cdot \sum_{u \in \F_q^m} \omega^{\mathrm{tr}[u \cdot \alpha]}\cdot \ket{u}.
\end{equation*}
Then the following two statements are true.
\begin{enumerate}
\item The $\ket{\varphi_\alpha}$'s form an orthonormal basis of $\C^V$.\label{item:orthonormal}
\item For each $\alpha \in \F_q^m$, $\ket{\varphi_\alpha}$ is an eigenvector for~$K$ with eigenvalue~$\frac{1}{M} \cdot  \frac{m - |\alpha|}{m}$,
	where $|\alpha|$ is the number of nonzero coordinates in~$\alpha$.\label{item:eigenvector}
\end{enumerate}
\end{proposition}
\begin{proof}
First, we prove \Cref{item:orthonormal}.
Given $\alpha, \beta \in \F_q^m$,
\begin{equation*}
\braket{\varphi_\alpha \mid \varphi_\beta}
= \frac{1}{M} \sum_{u \in \F_q^m} \omega^{\mathrm{tr}[u \cdot (\beta-\alpha)]}
= \left\{\begin{array}{rl}
	1 & \text{if } \alpha = \beta,\\
	0 & \text{otherwise}.
	\end{array} \right.\tag{by~\Cref{prop:fourier-fact-vector}}
\end{equation*}
As a result, the $\ket{\varphi_\alpha}$ vectors form an orthonormal basis of~$\C^V$.

Next, we prove \Cref{item:eigenvector}.
Given $\alpha \in \F_q^m$,
\begin{equation}\label{eq:eigenvector-calculation}
K \cdot \ket{\varphi_\alpha}
= \left(\E_{(\bu, \bv) \sim C} \ket{\bu}\bra{\bv}\right) \cdot \bigg(\frac{1}{M^{1/2}} \cdot \sum_{u \in \F_q^m} \omega^{\mathrm{tr}[u \cdot \alpha]}\cdot \ket{u}\bigg)\\
= \frac{1}{M^{1/2}} \cdot \E_{(\bu, \bv) \sim C} \omega^{\mathrm{tr}[\bv \cdot \alpha]} \ket{\bu}.
\end{equation}
By definition of a random edge,
we can replace~$\bv$ with $\bu + \bx \cdot e_{\bi}$, where $\bi$ is a uniformly random index in $\{1, \ldots, m\}$ and $\bx$ is a uniformly random element of $\F_q$.
As a result,
\begin{equation*}
\eqref{eq:eigenvector-calculation} 
= \frac{1}{M^{1/2}} \cdot \E_{\bu,\bi, \bx} \omega^{\mathrm{tr}[(\bu + \bx \cdot e_{\bi}) \cdot \alpha]} \ket{\bu}
= \left(\E_{\bi, \bx}\omega^{\mathrm{tr}[(\bx \cdot e_{\bi}) \cdot \alpha]}\right)  \cdot \frac{1}{M^{1/2}} \cdot \E_{\bu} \omega^{\mathrm{tr}[\bu \cdot \alpha]} \ket{\bu}
= \frac{1}{M}\left(\E_{\bi, \bx}\omega^{\mathrm{tr}[\bx \cdot \alpha_{\bi}]}\right)  \cdot \ket{\varphi_\alpha}.
\end{equation*}
Hence, $\ket{\varphi_\alpha}$ is an eigenvector of~$K$ with eigenvalue
\begin{equation*}
\frac{1}{M} \cdot \E_{\bi} \left[\E_{\bx}\omega^{\mathrm{tr}[\bx \cdot \alpha_{\bi}]}\right]
= \frac{1}{M} \cdot \E_{\bi} \left[\bone[\alpha_{\bi} = 0]\right]
= \frac{1}{M} \cdot  \frac{m - |\alpha|}{m}. \tag{by~\Cref{prop:fourier-fact-scalar}}
\end{equation*}
This concludes the proof.
\end{proof}

We will use the following corollary of \Cref{prop:eigenvectors}.

\begin{corollary}\label{cor:laplacian-spectral-gap}
Let $\lambda_1 \leq \lambda_2 \leq \cdots \leq \lambda_M$ be the eigenvalues of~$L$. Then $\lambda_1 = 0$ and $\lambda_2 = \frac{1}{m M}$.
\end{corollary}
\begin{proof}
Let $\mu_1 \geq \mu_2 \geq \cdots \geq \mu_M$ be the eigenvalues of~$K$. Then $\mu_i = \frac{1}{M} - \lambda_i$.
Thus, it suffices to show that $\mu_1 = \frac{1}{M}$ and $\mu_2 = \frac{1}{M} \cdot \frac{m-1}{m}$.
By \Cref{prop:eigenvectors}, $\ket{\varphi_\alpha}$ has eigenvalue $\frac{1}{M}$ when $|\alpha| = 0$ and eigenvalue $\frac{1}{M} \cdot \frac{m-1}{m}$ when $|\alpha| =1$.
\end{proof}

\subsection{Local and global variance}

In this section, $\ket{\psi}$ will denote a vector (not necessarily normalized) in $\calH_A \otimes \calH_B$,
and for each $u \in \F_q^m$, $0 \leq A^u \leq I$ will be a matrix acting on $\calH_A$.

\begin{definition}\label{def:local-and-variance}
The \emph{local variance of~$A$ on $\ket{\psi}$} is defined as
\begin{equation*}
\mathbf{Var}_{\mathrm{local}}(A, \psi):=\frac{1}{2} \cdot\E_{(\bu, \bv) \sim C} \bra{\psi} (A^{\bu} - A^{\bv})^2 \otimes I \ket{\psi}.
\end{equation*}
The \emph{global variance of~$A$ on $\ket{\psi}$} is defined as
\begin{equation*}
\mathbf{Var}_{\mathrm{global}}(A, \psi):=\frac{1}{2} \cdot\E_{\bu, \bv \sim \F_q^m} \bra{\psi} (A^{\bu} - A^{\bv})^2 \otimes I \ket{\psi}.
\end{equation*}
\end{definition}

The global variance differs from the local variance
because $\bu, \bv$ are chosen independently from~$\F_q^m$
rather than from the edges of $C$.
A standard fact from spectral graph theory allows us to 
use the expansion of $C$ to relate these two quantities.

\begin{lemma}\label{lem:local-to-global}
$\displaystyle
\mathbf{Var}_{\mathrm{global}}(A, \psi) \leq  m \cdot \mathbf{Var}_{\mathrm{local}}(A,\psi).
$
\end{lemma}

Before proving \Cref{lem:local-to-global},
we will give nice expressions for the local and global variances.
To begin, we show how to rewrite the local variance in terms of the Laplacian of~$C$.

\begin{lemma}\label{lem:local-rewrite}
Define the matrix
\begin{equation*}
A_{\mathrm{combine}} = \sum_{u \in \F_q^m} \ket{u} \otimes A^u \otimes I.
\end{equation*}
Then
\begin{equation*}
\mathrm{Tr}(A_{\mathrm{combine}}^\dagger \cdot (L \otimes \ket{\psi}\bra{\psi}) \cdot A_{\mathrm{combine}})= \mathbf{Var}_{\mathrm{local}}(A, \psi).
\end{equation*}
\end{lemma}
\begin{proof}
For any $u, v \in \F_q^m$,
\begin{align}
((\bra{u} - \bra{v}) \otimes \bra{\psi})\cdot A_{\mathrm{combine}}
&= ((\bra{u} - \bra{v}) \otimes \bra{\psi})\cdot \sum_{w \in \F_q^m} \ket{w} \otimes A^w \otimes I\nonumber\\
&=  \bra{\psi}\cdot  ((A^u- A^v) \otimes I).\label{eq:reader-probably-has-no-idea-whats-going-on-yet}
\end{align}
As a result,
\begin{align*}
&A_{\mathrm{combine}}^\dagger \cdot L \otimes \ket{\psi}\bra{\psi} \cdot A_{\mathrm{combine}}\\
=~& \frac{1}{2} \cdot \E_{(\bu, \bv) \sim C} A_{\mathrm{combine}}^\dagger ((\ket{\bu} - \ket{\bv}) \cdot (\bra{\bu} - \bra{\bv})\otimes \ket{\psi} \bra{\psi})\cdot A_{\mathrm{combine}}\tag{by \Cref{prop:laplacian-rewrite}}\\
=~&  \frac{1}{2} \cdot \E_{(\bu, \bv) \sim C} ((A^{\bu} - A^{\bv}) \otimes I) \cdot \ket{\psi}\bra{\psi} \cdot ((A^{\bu} - A^{\bv}) \otimes I). \tag{by \eqref{eq:reader-probably-has-no-idea-whats-going-on-yet}}
\end{align*}
Thus, if we take the trace,
\begin{equation*}
\mathrm{Tr}(A_{\mathrm{combine}}^\dagger \cdot (L \otimes \ket{\psi}\bra{\psi}) \cdot A_{\mathrm{combine}})
= \frac{1}{2} \cdot \E_{(\bu, \bv) \sim C} \bra{\psi} (A^{\bu} - A^{\bv})^2 \otimes I \cdot \ket{\psi}
= \mathbf{Var}_{\mathrm{local}}(A, \psi).
\end{equation*}
This completes the proof.
\end{proof}

Next, we give a simple expression for the global variance.

\begin{lemma}\label{lem:global-rewrite}
Expand $A_{\mathrm{combine}}$ as
\begin{equation*}
A_{\mathrm{combine}} = \ket{\varphi_0} \otimes A_0 + \ket{\varphi_{\perp}} \otimes A_{\perp},
\end{equation*}
where $\ket{\varphi_{\perp}}$ is orthogonal to $\ket{\varphi_0}$.
(Here we are writing $\ket{\varphi_0}$ for the vector $\ket{\varphi_\alpha}$ from \Cref{prop:eigenvectors} in the case of $\alpha = (0, \ldots, 0)$.)
Then
\begin{equation*}
\frac{1}{M} \cdot \mathrm{Tr}(\bra{\varphi_{\perp}} \otimes A_{\perp} \cdot (I \otimes \ket{\psi}\bra{\psi}) \cdot \ket{\varphi_{\perp}} \otimes A_{\perp})= \mathbf{Var}_{\mathrm{global}}(A, \psi).
\end{equation*}
\end{lemma}
\begin{proof}
We begin by computing~$A_0$:
\begin{equation*}
A_0
= \bra{\varphi_0} \otimes I \cdot A_{\mathrm{combine}}
= \bigg(\frac{1}{M^{1/2}} \cdot \sum_{u \in \F_q^m} \bra{u}\bigg) \otimes I \cdot \bigg(\sum_{u \in \F_q^m} \ket{u} \otimes A^u \otimes I\bigg)
= \frac{1}{M^{1/2}} \cdot \sum_{u \in \F_q^m} A^u \otimes I.
\end{equation*}
Then
\begin{align*}
\ket{\varphi_0} \otimes A_0
= \bigg(\frac{1}{M^{1/2}} \cdot \sum_{u \in \F_q^m} \ket{u}\bigg)\otimes \bigg(\frac{1}{M^{1/2}} \cdot \sum_{u \in \F_q^m} A^u \otimes I\bigg)
&= \frac{1}{M} \sum_{u \in \F_q^m} \ket{u} \otimes \sum_{v \in \F_q^m} A^v \otimes I\\
&= \sum_{u \in \F_q^m} \ket{u} \otimes A_{\mathrm{avg}} \otimes I,
\end{align*}
where we have written $A_{\mathrm{avg}} = \E_{\bu} A^{\bu}$.
As a result,
\begin{equation*}
\ket{\varphi_{\perp}} \otimes A_{\perp}
= A_{\mathrm{combine}} - \ket{\varphi_0} \otimes A_0
= \sum_{u \in \F_q^n} \ket{u} \otimes (A^u - A_{\mathrm{avg}}) \otimes I.
\end{equation*}
Thus,
\begin{align*}
&\bra{\varphi_{\perp}} \otimes A_{\perp} \cdot (I \otimes \ket{\psi}\bra{\psi}) \cdot \ket{\varphi_{\perp}} \otimes A_{\perp}\\
=&\sum_{u, v \in \F_q^m} \braket{u \mid v} \otimes ((A^u - A_{\mathrm{avg}}) \otimes I) \cdot \ket{\psi}\bra{\psi} \cdot ((A^v - A_{\mathrm{avg}}) \otimes I)\\
=&\sum_{u \in \F_q^m}  ((A^u - A_{\mathrm{avg}}) \otimes I) \cdot \ket{\psi}\bra{\psi} \cdot ((A^u - A_{\mathrm{avg}}) \otimes I).
\end{align*}
As a result, if we take the trace,
\begin{equation}\label{eq:just-took-trace}
\frac{1}{M} \cdot \mathrm{Tr}(\bra{\varphi_{\perp}} \otimes A_{\perp} \cdot (I \otimes \ket{\psi}\bra{\psi}) \cdot \ket{\varphi_{\perp}} \otimes A_{\perp})
= \E_{\bu \sim \F_q^m} \bra{\psi} (A^{\bu} - A_{\mathrm{avg}})^2 \otimes I \ket{\psi}.
\end{equation}
We can rewrite the squared expression as
\begin{align*}
\E_{\bu \sim \F_q^m}(A^{\bu} - A_{\mathrm{avg}})^2
&= \E_{\bu \sim \F_q^m}((A^{\bu})^2 + (A_{\mathrm{avg}})^2 - A^{\bu} \cdot A_{\mathrm{avg}} - A_{\mathrm{avg}} \cdot A^{\bu})\\
&= \E_{\bu \sim \F_q^m}((A^{\bu})^2 - (A_{\mathrm{avg}})^2)\\
&= \frac{1}{2}\cdot\E_{\bu, \bv \sim \F_q^m}((A^{\bu})^2 + (A^{\bv})^2 - A^{\bu} \cdot A^{\bv} - A^{\bv} \cdot A^{\bu})\\
&= \frac{1}{2}\cdot\E_{\bu, \bv \sim \F_q^m}(A^{\bu} - A^{\bv})^2.
\end{align*}
Thus,
\begin{equation*}
\eqref{eq:just-took-trace}
= \frac{1}{2} \cdot\E_{\bu, \bv \sim \F_q^m} \bra{\psi} (A^{\bu} - A^{\bv})^2 \otimes I \ket{\psi}
= \mathbf{Var}_{\mathrm{global}}(A, \psi).
\end{equation*}
This completes the proof.
\end{proof}

Now we prove \Cref{lem:local-to-global}.

\begin{proof}[Proof of \Cref{lem:local-to-global}]
We begin by computing
\begin{align}
A_{\mathrm{combine}}^\dagger \cdot L \otimes \ket{\psi}\bra{\psi} \cdot A_{\mathrm{combine}}
 &=(\bra{\varphi_0} \otimes A_0 + \bra{\varphi_{\perp}} \otimes A_{\perp}) \cdot L \otimes \ket{\psi}\bra{\psi} \cdot (\ket{\varphi_0} \otimes A_0 + \ket{\varphi_{\perp}} \otimes A_{\perp})\nonumber\\
  &=\bra{\varphi_{\perp}} \otimes A_{\perp} \cdot L \otimes \ket{\psi}\bra{\psi} \cdot \ket{\varphi_{\perp}} \otimes A_{\perp} \nonumber\\
  & = \bra{\varphi_{\perp}} L \ket{\varphi_{\perp}} \cdot A_{\perp}  \ket{\psi}\bra{\psi}  A_{\perp} ,\label{eq:used-0-eigenvector}
\end{align}
where the second step follows from the fact that $\ket{\varphi_{0}}$ is a $0$-eigenvector for $L$.
Note that because $\ket{\varphi_{\perp}}$ is orthogonal to $\ket{\varphi_0}$,
\begin{equation*}
\bra{\varphi_{\perp}} L \ket{\varphi_{\perp}} \geq \frac{1}{mM} \cdot \braket{\varphi_{\perp} \mid \varphi_{\perp}}
\end{equation*}
by \Cref{cor:laplacian-spectral-gap}.
As a result,
\begin{align*}
 \mathbf{Var}_{\mathrm{local}}(A,\psi) 
 & = \mathrm{Tr}(A_{\mathrm{combine}}^\dagger \cdot L \otimes \ket{\psi}\bra{\psi} \cdot A_{\mathrm{combine}}) \tag{by \Cref{lem:local-rewrite}}\\
  & = \bra{\varphi_{\perp}} L \ket{\varphi_{\perp}} \cdot \mathrm{Tr}(A_{\perp} \cdot  \ket{\psi}\bra{\psi} \cdot  A_{\perp}) \tag{by \Cref{eq:used-0-eigenvector}}\\
  & \geq \frac{1}{mM} \cdot \braket{\varphi_{\perp} \mid \varphi_{\perp}} \cdot \mathrm{Tr}(A_{\perp} \cdot  \ket{\psi}\bra{\psi} \cdot  A_{\perp})\\
    & = \frac{1}{mM} \cdot \mathrm{Tr}(\bra{\varphi_{\perp}} \otimes A_{\perp} \cdot (I \otimes \ket{\psi}\bra{\psi}) \cdot \ket{\varphi_{\perp}} \otimes A_{\perp})\\
    &=  \frac{1}{m} \cdot \mathbf{Var}_{\mathrm{global}}(A,\psi) . \tag{by \Cref{lem:global-rewrite}}
\end{align*}
This concludes the proof.
\end{proof}

\section{Global variance of the points measurements}
\label{sec:variance}
Throughout this section, $(\psi, A, B, L)$ will denote a fixed $(\eps, \delta, \gamma)$-good symmetric strategy for the $(m,q,d)$-low individual degree test.

\begin{lemma}\label{lem:generalize-b}
Let $G \in \polysub{m}{q}{d}$. Then
\begin{equation*}
B^\ell_{[f(u) = g(u)]} \otimes (G_g)^{1/2} \approx_{md/q} B^\ell_{g|_\ell} \otimes (G_g)^{1/2}
\end{equation*}
on the axis-parallel lines test distribution.
\end{lemma}
\begin{proof}
We want to bound the quantity
\begin{align*}
&~\E_{\bu, \bell} \sum_{g \in \polyfunc{m}{d}{q}} \Vert (B^{\bell}_{[f(\bu) = g(\bu)]}  - B^{\bell}_{g|_{\bell}}) \otimes (G_g)^{1/2} \ket{\psi}\Vert^2\\
=&~\E_{\bu, \bell} \sum_{g \in \polyfunc{m}{d}{q}} \bra{\psi} \bigg(\sum_{f : f \neq g|_{\bell}} \bone[f(\bu) = g(\bu)] \cdot B^{\bell}_f\bigg)^2 \otimes G_g \ket{\psi}\\
\leq&~\E_{\bu, \bell} \sum_{g \in \polyfunc{m}{d}{q}} \bra{\psi} \bigg(\sum_{f : f \neq g|_{\bell}} \bone[f(\bu) = g(\bu)] \cdot B^{\bell}_f\bigg) \otimes G_g \ket{\psi} \\
= &~\E_{\bell} \sum_{g \in \polyfunc{m}{d}{q}} \sum_{f : f \neq g|_{\bell}} \bra{\psi}  B^{\bell}_f \otimes G_g \ket{\psi} \cdot \left( \E_{\bu} \bone[f(\bu) = g(\bu)]\right)\\
\leq &~\E_{\bell} \sum_{g \in \polyfunc{m}{d}{q}} \sum_{f : f \neq g|_{\bell}} \bra{\psi}  B^{\bell}_f \otimes G_g \ket{\psi} \cdot \frac{md}{q}\tag{by Schwartz-Zippel}\\
\leq &~\frac{md}{q}.\qedhere
\end{align*}
\end{proof}

\begin{lemma}\label{lem:local-variance-of-points}
Let $G \in \polysub{m}{q}{d}$.  Then
\begin{equation}\label{eq:local-variance-of-points-equation}
A^u_{g(u)} \otimes (G_g)^{1/2} \approx_{24\cdot(\eps + \delta + \frac{md}{q})} A^v_{g(v)} \otimes (G_g)^{1/2}
\end{equation}
on the distribution $(\bu, \bv) \sim C$.
\end{lemma}
\begin{proof}
Let $\bu$ and $\bell$ be distributed as in the axis-parallel lines test,
and sample $\bv \sim \bell$.
Then $\bv$ and $\bell$ are also distributed as in the axis-parallel lines test.
As a result,
\begin{align*}
A^u_{g(u)} \otimes (G_g)^{1/2}
&\approx_{2\delta} I \otimes (G_g)^{1/2} A^u_{g(u)}\tag{by \Cref{prop:simeq-to-approx}}\\
&\approx_{2\eps} B^\ell_{[f(u) = g(u)]} \otimes (G_g)^{1/2}\tag{by \Cref{prop:simeq-to-approx}}\\
&\approx_{\frac{md}{q}} B^\ell_{g|_{\ell}} \otimes (G_g)^{1/2}\tag{by \Cref{lem:generalize-b}}\\
&\approx_{\frac{md}{q}} B^\ell_{[f(v) = g(v)]} \otimes (G_g)^{1/2}\tag{by \Cref{lem:generalize-b}}\\
&\approx_{2\eps} I \otimes (G_g)^{1/2} A^v_{g(v)} \tag{by \Cref{prop:simeq-to-approx}}\\
&\approx_{2\delta} A^v_{g(v)} \otimes (G_g)^{1/2}.\tag{by \Cref{prop:simeq-to-approx}}
\end{align*}
Steps 1, 2, 5, and 6 are also using \Cref{rem:good-strat-characterization,prop:cab-approx-delta}.
The lemma now follows from \Cref{prop:triangle-inequality-for-approx_delta}.
\end{proof}

We note that \Cref{eq:local-variance-of-points-equation} is equivalent to the statement that
\begin{equation}\label{eq:equivalent-local-variance}
\sum_{g \in \polyfunc{m}{q}{d}}\E_{(\bu, \bv) \sim C} \bra{\psi} (A^{\bu}_{g(\bu)}  - A^{\bv}_{g(\bv)})^2 \otimes G_g \ket{\psi}
\leq 24\left(\eps + \delta + \frac{md}{q}\right).
\end{equation}
This can be viewed as a form of local variance for the points measurements.
We now derive the corresponding expression for the global variance of the points measurements.

\begin{lemma}\label{lem:global-variance-of-points}
Let $G \in \polysub{m}{q}{d}$.  Then
\begin{equation}\label{eq:global-variance-of-points-equation}
A^u_{g(u)} \otimes (G_g)^{1/2} \approx_{24m\cdot(\eps + \delta + \frac{md}{q})} A^v_{g(v)} \otimes (G_g)^{1/2}
\end{equation}
on the distribution $\bu, \bv \sim \F_q^m$.
\end{lemma}
\begin{proof}
We want to bound 
\begin{equation}\label{eq:TODO:bound-this!}
\E_{\bu, \bv \sim \F_q^m} \sum_{g \in \polyfunc{m}{q}{d}} \Vert (A^{\bu}_{g(\bu)}  - A^{\bv}_{g(\bv)} ) \otimes (G_g)^{1/2} \ket{\psi}\Vert^2
= \E_{\bu, \bv \sim \F_q^m} \sum_{g \in \polyfunc{m}{q}{d}} \bra{\psi} (A^{\bu}_{g(\bu)}  - A^{\bv}_{g(\bv)} )^2 \otimes G_g \ket{\psi}.
\end{equation}
For each $g \in \polyfunc{m}{q}{d}$, define
\begin{equation*}
\forall u \in \F_q^m, ~ A(g)^u := A^u_{g(u)}, \qquad \text{and} \qquad \ket{\psi_g} := I \otimes (G_g)^{1/2} \ket{\psi}.
\end{equation*}
Then
\begin{align*}
\eqref{eq:TODO:bound-this!}
& = \sum_{g \in \polyfunc{m}{q}{d}}\E_{\bu, \bv \sim \F_q^n} \bra{\psi_g} (A(g)^{\bu}  - A(g)^{\bv})^2 \otimes I \ket{\psi_g}\\
& = \sum_{g \in \polyfunc{m}{q}{d}}2\cdot  \mathbf{Var}_{\mathrm{global}}(A(g), \psi_g)\\
& \leq \sum_{g \in \polyfunc{m}{q}{d}}2m\cdot \mathbf{Var}_{\mathrm{local}}(A(g), \psi_g)\tag{by \Cref{lem:local-to-global}}\\
& = m \cdot \sum_{g \in \polyfunc{m}{q}{d}}\E_{(\bu, \bv) \sim C} \bra{\psi_g} (A(g)^{\bu}  - A(g)^{\bv})^2 \otimes I \ket{\psi_g}\\
& = m \cdot \sum_{g \in \polyfunc{m}{q}{d}}\E_{(\bu, \bv) \sim C} \bra{\psi} (A^{\bu}_{g(\bu)}  - A^{\bv}_{g(\bv)})^2 \otimes G_g \ket{\psi}\\
& \leq m \cdot 24\left(\eps + \delta + \frac{md}{q}\right).\tag{by \Cref{eq:equivalent-local-variance}}
\end{align*}
This concludes the proof.
\end{proof}

\section{Self-improvement}
\label{sec:self-improvement}

Throughout this section, $(\psi, A, B, L)$ will denote a fixed $(\eps, \delta, \gamma)$-good symmetric strategy for the $(m,q,d)$ low individual degree test.
The majority of this section will be devoted to proving \Cref{lem:self-improvement-helper} below,
which is a slightly weaker form of \Cref{thm:self-improvement-in-induction-section}.
The key difference is that the measurement~$H$ it outputs is allowed to be non-projective,
rather than the projective measurement given by \Cref{thm:self-improvement-in-induction-section}.
Having proven this, we can apply \Cref{thm:orthonormalization}
to produce a projective measurement;
this is done in \Cref{sec:self-improvement-projective} below,
completing the proof of \Cref{thm:self-improvement-in-induction-section}.

We now highlight other differences between \Cref{lem:self-improvement-helper}
and \Cref{thm:self-improvement-in-induction-section}.
Since~$H$ is non-projective, 
we have stated its strong self-consistency in \Cref{item:self-improvement-self}
in terms of \Cref{def:strong-self-consistency}; see \Cref{prop:two-notions-of-self-consistency}
for a proof that these conditions are equivalent for projective sub-measurements.
The other key difference is that the boundedness condition is modified slightly in \Cref{item:self-improvement-boundedness}.
Finally, the error $\zeta$ is substantially smaller.

\begin{lemma}[Self-improvement with non-projective output]\label{lem:self-improvement-helper}
Let $G \in \polymeas{m}{q}{d}$ be a measurement with the following property:
\begin{enumerate}
\ignore{
\item (Completeness): \label{item:self-improvement-G-completeness}  If $G = \sum_g G_g$, then
  	\begin{equation*}
	\bra{\psi} G \otimes I \ket{\psi} \geq 1 - \kappa.
	\end{equation*}
}
  \item (Consistency with~$A$): \label{item:self-improvement-G-consistency} On average over $\bu \sim \F_q^{m}$,
	  \begin{equation*}
	  A^{u}_a \otimes I \simeq_{\nu} I \otimes G_{[g(u)=a]}.
	  \end{equation*}
	  \end{enumerate}
Let
\begin{equation*}
\zeta = 100m\cdot \Big(\eps^{1/2} + \delta^{1/2} + (d/q)^{1/2}\Big).
\end{equation*}
Then there exists $H \in \polysub{m}{q}{d}$ with the following properties:
\begin{enumerate}
    \item(Completeness):  \label{item:self-improvement-completeness} If $H =  \sum_h H_h$, then
    \begin{equation*}
    \bra{\psi} H \otimes I \ket{\psi} \geq (1-\nu)-\zeta.
    \end{equation*}
    \item(Consistency with~$A$):\label{item:self-improvement-A-consistency} On average over $\bu \sim \F_q^m$,
        \begin{equation*}
            A^u_a \otimes I \simeq_{\zeta} I \otimes H_{[h(u) = a]}.
        \end{equation*}
    \item(Strong self-consistency): \label{item:self-improvement-self}
    	\begin{equation*}
		\sum_{h} \bra{\psi} H_h \ot H_h \ket{\psi} \geq \bra{\psi} H \ot I \ket{\psi} - \zeta.
		\ignore{H_h \otimes I \simeq_{\zeta} I \otimes H_h.}
	\end{equation*}
    \item(Boundedness): \label{item:self-improvement-boundedness} There exists a positive-semidefinite matrix~$Z$ such that
    \begin{equation*}
        \bra{\psi} Z \otimes I \ket{\psi} -\E_{\bu} \sum_a \bra{\psi}  A^{\bu}_{a} \otimes H_{[h(\bu)=a]} \ket{\psi} \leq \zeta
    \end{equation*}
    and for each $h \in \polyfunc{m}{q}{d}$,
	\begin{equation*}
	Z \geq \left(\E_{\bu} A^{\bu}_{h(\bu)}\right).
	\end{equation*}
\end{enumerate}
\end{lemma}

\subsection{A semidefinite program}

A key element in the proof of \Cref{lem:self-improvement-helper} will be a pair of primal and dual semidefinite programs. 
To define them, it will be convenient to introduce the notational shorthand
\begin{equation*}
A_g = \E_{\bu \sim \F_q^m} A^{\bu}_{g(\bu)}.
\end{equation*}
Then the primal is
\begin{align}
 \sup &\quad \sum_g \,\Tr(T_g \cdot A_g) \label{eq:primal-objective}\\
  \text{s.t.} &\quad   T_g \geq 0\qquad\forall g\in\polyfunc{m}{q}{d}\;,\notag\\
	&\quad \sum_g T_g \leq I,\notag
\end{align}
and the dual is
\begin{align}
  \inf &\quad \Tr(Z) \label{eq:dual-objective}\\
 \text{ s.t.} &\quad Z \geq A_g. \label{eq:dual-constraint}
\end{align}
We will prove that these two program are indeed dual to each other in \Cref{lem:sdp} below.

\begin{lemma}\label{lem:sdp}
The semidefinite programs~\eqref{eq:primal-objective} and~\eqref{eq:dual-objective} are dual to each other. 
Moreover there is an optimal pair of solutions $\{T_g\}$ to~\eqref{eq:primal-objective} and $Z$ to~\eqref{eq:dual-objective} such that $\sum_g T_g = I$ and  
\begin{equation}\label{eq:slater}
T_g Z  \,=\, T_g {A_g},\qquad\forall g\in \polyfunc{m}{q}{d}.
\end{equation}
\end{lemma}

\begin{proof}
To show that~\eqref{eq:primal-objective} and~\eqref{eq:dual-objective} are dual to each other,
we begin by rewriting the primal~\eqref{eq:primal-objective} in canonical form.
Let $r$ be the dimension of the space on which~$A$ acts.
Let $M = |\polyfunc{m}{q}{d}|$ be the number of polynomials with individual degree $d$. 
We will assume some ordering of these polynomials $g_1, \ldots, g_M \in \polyfunc{m}{q}{d}$ which is allowed to be arbitrary.
Consider the following semidefinite program:
\begin{align}
\sup &\quad \Tr(C^\dagger X) \label{eq:primal-canonical}\\
  \text{s.t.} &\quad   \Tr(D_{ij}^\dagger X) = b_{ij} \qquad\forall i,j\in\{1,\ldots,r\},\nonumber\\
	&\quad X \geq 0,\nonumber
\end{align}
where the variables~$C$, $D_{ij}$, and $b_{ij}$ are defined as follows:
\begin{equation*}
C = \sum_{i = 1}^M \ket{i}\bra{i} \otimes A_{g_i}
\end{equation*}
\begin{equation*}
\forall i,j \in [r], \quad
D_{ij} = \sum_{k=1}^{M+1} \ket{k}\bra{k} \otimes \ket{i}\bra{j},
\qquad
b_{ij}  = \left\{\begin{array}{rl}
		1 & \text{if } i = j,\\
		0 & \text{otherwise}.
		\end{array}\right.
\end{equation*}
We claim that \eqref{eq:primal-canonical} is equivalent to \eqref{eq:primal-objective}.
To see this, let
\begin{equation*}
X = \sum_{i, j =1}^{M+1} \ket{i}\bra{j} \otimes X_{ij}
\end{equation*}
be a feasible solution to \eqref{eq:primal-canonical}.
Then because $X \geq 0$, $X_{i i} \geq 0$ for each $i \in [M+1]$. In addition
\begin{align*}
\Tr(D_{i_1 i_2}^\dagger X)
&= \sum_k \sum_{j_1 j_2} \Tr((\ket{k}\bra{k} \otimes \ket{i_1}\bra{i_2})\cdot(\ket{j_1} \bra{j_2} \otimes X_{j_1 j_2}))\\
&= \sum_k  \bra{i_2} X_{k k} \ket{i_1}\\
&= \left\{\begin{array}{rl}
		1 & \text{if } i_1 = i_2,\\
		0 & \text{otherwise}.
		\end{array}\right.
\end{align*}
This is equivalent to the statement $\sum_{i=1}^{M+1} X_{i i} = I$, which implies that $\sum_{i=1}^M X_{i i} \leq I$.
Finally, the objective value \eqref{eq:primal-canonical} is
\begin{align*}
\Tr(C^\dagger X)
& = \sum_{i = 1}^M  \sum_{j,k = 1}^{M+1} \Tr((\ket{i}\bra{i} \otimes A_{g_i}^\dagger) \cdot (\ket{j}\bra{k} \otimes X_{j, k}))\\
&=  \sum_{i = 1}^M \Tr(A_{g_i}^\dagger \cdot  X_{i, i})\\
&=  \sum_{i = 1}^M \Tr(A_{g_i} \cdot  X_{i, i}). \tag{because~$A$ is Hermitian}
\end{align*}
As a result, setting $T_{g_i} = X_{i i}$ for each $i \in [M]$ gives a feasible solution to the original semidefinite program~\eqref{eq:primal-objective}
with a matching objective value. A similar transformation allows us to convert solutions of~\eqref{eq:primal-objective} to~\eqref{eq:primal-canonical}.

The dual of \eqref{eq:primal-canonical} is 
\begin{align}
\inf &\quad \sum_i\, z_{ij} b_{ij} \label{eq:dual-canonical}\\
  \text{s.t.} &\quad  \sum_{i,j} \,z_{ij} D_{ij} \geq C.\label{eq:dual-canonical-constraint}
\end{align}
We claim that \eqref{eq:dual-canonical} is equivalent to \eqref{eq:dual-objective}.
To see this, we first calculate
\begin{align*}
\sum_{i,j=1}^r \,z_{ij} D_{ij}
= \sum_{i,j=1}^r z_{ij} \cdot \sum_{k=1}^{M+1} \ket{k}\bra{k} \otimes \ket{i}\bra{j}
&= \sum_{k=1}^{M+1} \ket{k}\bra{k} \otimes \left(\sum_{i, j=1}^r z_{ij}\ket{i}\bra{j}\right)\\
&=: \sum_{k=1}^{M+1}\ket{k}\bra{k} \otimes Z.
\end{align*}
Then the constraint \eqref{eq:dual-canonical-constraint} states that
\begin{equation*}
\sum_{i=1}^{M+1}\ket{i}\bra{i} \otimes Z \geq C = \sum_{i = 1}^M \ket{i}\bra{i} \otimes A_{g_i},
\end{equation*}
which is equivalent to the statement that $Z \geq A_{g_i}$ for all $i \in [M]$ and $Z \geq 0$.
In other words, $Z$ is a feasible solution to the original dual SDP \eqref{eq:dual-objective}, with value
\begin{equation*}
\Tr(Z) = \sum_{i=1}^r Z_{i i} = \sum_{i,j=1}^r b_{i j} Z_{i j} =  \sum_{i,j=1}^r b_{i j} z_{i j},
\end{equation*}
the same value as in \eqref{eq:dual-canonical}.
Hence, the two dual programs are the same as well,
which implies that \eqref{eq:primal-objective} and \eqref{eq:dual-objective} form a primal/dual pair.

To show that a primal/dual pair satisfies \emph{strong duality}, i.e.\ that their optimum values are the same,
we use Slater's condition~\cite[Section~5.2.3]{cvxbook} and show that they satisfy \emph{strict feasibility}, which means both have a feasible solution which is positive \emph{definite}
that satisfies all constraints with a strict inequality.
It can be checked that the following two solutions to \eqref{eq:primal-objective} and \eqref{eq:dual-objective} satisfy this property:
\begin{equation*}
\forall g \in \polyfunc{m}{q}{d}, \quad T_g = \frac{1}{2M} \cdot I,
\qquad\qquad\qquad
Z = 2I.
\end{equation*}
Because they satisfy strong duality, their optimal solutions satisfy the \emph{complementary slackness} condition (see~\cite{AHO97}).
If $X$ and $(z_{ij})$ are an optimal pair of solutions to~\eqref{eq:primal-canonical} and~\eqref{eq:dual-canonical} respectively,  this implies that
\begin{equation}\label{eq:complementary-slackness}
X\Big( \sum_{i,j} \,z_{ij} D_{ij} - C \Big) \,=\, 0 \;.
\end{equation}
Clearly for any optimal pair $(X,z_{ij})$ we can assume without loss of generality that $X$ is block-diagonal.
Then if we translate \eqref{eq:complementary-slackness} back to the variables $\{T_g\}$ and $Z$, we get
\begin{align*}
0 = X\Big( \sum_{i,j} \,z_{ij} D_{ij} - C \Big)
&= \sum_{i=1}^{M+1} \ket{i}\bra{i} \otimes X_{ii} \cdot \bigg(\sum_{i=1}^{M+1} \ket{i}\bra{i} \otimes Z - \sum_{i=1}^M \ket{i}\bra{i}\otimes A_{g_i}\bigg)\\
&= \sum_{i=1}^{M} \ket{i}\bra{i} \otimes (X_{ii} \cdot (Z - A_{g_i})) + \ket{M+1}\bra{M+1} \otimes (X_{M+1,M+1} \cdot Z).
\end{align*}
This implies that $X_{ii} \cdot (Z - A_{g_i}) = 0$ for $i \in [M]$ and $X_{M+1, M+1} = 0$.
Translating back to the variables $\{T_g\}$ and $Z$, this gives $T_g(Z-A_g)=0$ for all $g\in\polyfunc{m}{q}{d}$ and $\sum_g T_g = I$. 
\end{proof}

\subsection{Proof of \Cref{lem:self-improvement-helper}}
We let $T = \{T_g\}$ and $Z$ be the optimal solutions to the SDPs~\eqref{eq:primal-objective} and \eqref{eq:dual-objective} respectively
given by \Cref{lem:sdp}.
Then~$T$ is a measurement, and
\begin{align}
 \forall g \in \polyfunc{m}{q}{d}: \qquad
 Z &\geq (\E_{\bu} A^{\bu}_{g(\bu)}),  \label{eq:Z-greater-than-A}\\
T_g \cdot Z &= T_g \cdot (\E_{\bu} A^{\bu}_{g(\bu)}).\label{eq:swap-Z-for-A}
\end{align}
For each $u \in \F_q^m$, define $H^u = \{H^u_h\}_{h \in \polyfunc{m}{q}{d}}$ as
\begin{equation*}
    H^u_h := A^u_{h(u)}\cdot T_h \cdot A^u_{h(u)}.
\end{equation*}
Let $u \in \F_q^m$. Then
\begin{align*}
    \sum_{h \in \polyfunc{m}{q}{d}} H^u_h
    & = \sum_{h \in \polyfunc{m}{q}{d}} A^u_{h(u)} \cdot T_h \cdot A^u_{h(u)}\\
    & = \sum_{a \in \F_q}A^u_{a} \cdot \bigg(\sum_{h:h(u) = a}  T_h\bigg) \cdot A^u_{a}\\
    & \leq \sum_{a \in \F_q} (A^u_{a})^2\tag{because~$T$ is a measurement}\\
    & = \sum_{a \in \F_q} A^u_{a}\tag{$A$ is projective}\\
    & = I.
\end{align*}
Hence, $H^u$ is a sub-measurement,
and therefore $H^u \in \polysub{m}{q}{d}$.
Next, define $H = \{H_h\}_{h \in \polyfunc{m}{q}{d}}$ as
\begin{equation*}
H_h := \E_{\bu \sim \F_q^m}H^{\bu}_h.
\end{equation*}
Then
\begin{equation*}
\sum_{h \in \polyfunc{m}{q}{d}} H_h
= \E_{\bu \sim \F_q^m} \sum_{h \in \polyfunc{m}{q}{d}} H^{\bu}_h
\leq I. \tag{$H^u$ is a sub-measurement}
\end{equation*}
Hence, $H$ is a sub-measurement,
and therefore $H \in \polysub{m}{q}{d}$.

Set
\begin{equation*}
\zeta_{\mathrm{variance}} = 24m\cdot\Big(\eps + \delta + \frac{md}{q}\Big)
\end{equation*}
to be the error in \Cref{eq:local-variance-of-points-equation}.
Prior to showing that~$H$ satisfies \Cref{item:self-improvement-A-consistency,item:self-improvement-completeness,item:self-improvement-self,item:self-improvement-boundedness}, we will prove the following technical lemma.

\begin{lemma}\label{lem:add-in-u}
Suppose $M = \{M^u_o\}$ is a sub-measurement with outcomes in some set $\calO$.
For each $u \in \F_q^m$, let $S_u$ be a subset of $\calO \otimes \polyfunc{m}{q}{d}$.
Then
\begin{equation*}
\E_{\bu \sim \F_q^m} \sum_{(o, h) \in S_{\bu}} \bra{\psi} M^{\bu}_o \otimes H_h \ket{\psi}
\approx_{4\sqrt{\zeta_{\mathrm{variance}}}} \E_{\bu \sim \F_q^m} \sum_{(o, h) \in S_{\bu}} \bra{\psi} (A^{\bu}_{h(\bu)} \cdot M^{\bu}_o \cdot A^{\bu}_{h(\bu)}) \otimes T_h \ket{\psi}.
\end{equation*}
\end{lemma}
\begin{proof}
We begin by expanding
\begin{align}
\E_{\bu \sim \F_q^m} \sum_{(o, h) \in S_{\bu}} \bra{\psi} M^{\bu}_o \otimes H_h \ket{\psi}
&= \E_{\bu, \bv} \sum_{(o, h) \in S_{\bu}} \bra{\psi} M^{\bu}_o \otimes H^{\bv}_h \ket{\psi}\nonumber\\
&= \E_{\bu, \bv} \sum_{(o, h) \in S_{\bu}} \bra{\psi} M^{\bu}_o \otimes (A^{\bv}_{h(\bv)} \cdot T_h \cdot A^{\bv}_{h(\bv)}) \ket{\psi}.\label{eq:expand-that-H}
\end{align}
We claim that
\begin{equation}\label{eq:move-one}
\eqref{eq:expand-that-H}
\approx_{\sqrt{2\delta}}
\E_{\bu, \bv} \sum_{(o, h) \in S_{\bu}} \bra{\psi} (A^{\bv}_{h(\bv)} \cdot M^{\bu}_o) \otimes (T_h \cdot A^{\bv}_{h(\bv)}) \ket{\psi}.
\end{equation}
To show this, we bound the magnitude of the difference.
\begin{multline}
\Big|\E_{\bu, \bv} \sum_{(o, h) \in S_{\bu}} \bra{\psi} (A^{\bv}_{h(\bv)} \otimes I - I \otimes A^{\bv}_{h(\bv)}) \cdot (M^{\bu}_o \otimes (T_h \cdot A^{\bv}_{h(\bv)}))\ket{\psi}\Big|\\
\leq
\sqrt{\E_{\bu, \bv} \sum_{(o, h) \in S_{\bu}} \bra{\psi} (A^{\bv}_{h(\bv)} \otimes I - I \otimes A^{\bv}_{h(\bv)}) \cdot (M^{\bu}_o \otimes T_h) \cdot (A^{\bv}_{h(\bv)} \otimes I - I \otimes A^{\bv}_{h(\bv)}) \ket{\psi}}\\
\cdot \sqrt{\E_{\bu, \bv} \sum_{(o, h) \in S_{\bu}} \bra{\psi}  M^{\bu}_o \otimes (A^{\bv}_{h(\bv)} \cdot T_h\cdot  A^{\bv}_{h(\bv)}) \ket{\psi}}.\label{eq:move-one-cauchy-schwarz}
\end{multline}
The term inside the first square root is
\begin{align*}
&\E_{\bu, \bv} \sum_{a \in \F_q} \bra{\psi} (A^{\bv}_a \otimes I - I \otimes A^{\bv}_a) \cdot \bigg( \sum_{(o, h) \in S_{\bu},h(\bv) = a}M^{\bu}_o \otimes T_h\bigg) \cdot (A^{\bv}_a \otimes I - I \otimes A^{\bv}_a) \ket{\psi}\\
\leq~&\E_{\bu, \bv} \sum_{a \in \F_q} \bra{\psi} (A^{\bv}_a \otimes I - I \otimes A^{\bv}_a)^2 \ket{\psi}, \tag{because $M^{\bu}$ and $T_h$ are sub-measurements}
\end{align*}
which is at most~$2\delta$ by \Cref{prop:simeq-to-approx} and because~$A$ is $\delta$-self-consistent.
The term inside the second square root is
\begin{equation*}
\E_{\bu, \bv} \sum_{(o, h) \in S_{\bu}} \bra{\psi}  M^{\bu}_o \otimes H^{\bv}_h \ket{\psi},
\end{equation*}
which is at most~$1$ because $M^{\bu}$ and $H^{\bv}$ are sub-measurements.
Next, we claim that
\begin{equation}\label{eq:move-another}
\eqref{eq:move-one}
\approx_{\sqrt{2\delta}}
\E_{\bu, \bv} \sum_{(o, h) \in S_{\bu}} \bra{\psi} (A^{\bv}_{h(\bv)} \cdot M^{\bu}_o \cdot A^{\bv}_{h(\bv)}) \otimes T_h  \ket{\psi}.
\end{equation}
To show this, we bound the magnitude of the difference.
\begin{align}\label{eq:move-another-cauchy-schwarz}
&\Big|\E_{\bu, \bv} \sum_{(o, h) \in S_{\bu}} \bra{\psi} ((A^{\bv}_{h(\bv)} \cdot M^{\bu}_o) \otimes T_h) \cdot (A^{\bv}_{h(\bv)} \otimes I - I \otimes A^{\bv}_{h(\bv)})  \ket{\psi}\Big|\nonumber\\
&\leq\sqrt{\E_{\bu, \bv} \sum_{(o, h) \in S_{\bu}} \bra{\psi} (A^{\bv}_{h(\bv)} \cdot M^{\bu}_o \cdot A^{\bv}_{h(\bv)}) \otimes T_h  \ket{\psi}}\nonumber\\
&\quad\cdot \sqrt{\E_{\bu, \bv} \sum_{(o, h) \in S_{\bu}} \bra{\psi} (A^{\bv}_{h(\bv)} \otimes I - I \otimes A^{\bv}_{h(\bv)}) \cdot (M^{\bu}_o \otimes T_h) \cdot (A^{\bv}_{h(\bv)} \otimes I - I \otimes A^{\bv}_{h(\bv)}) \ket{\psi}}.
\end{align}
The term inside the first square root is
\begin{align*}
&\E_{\bu, \bv} \sum_h \bra{\psi} (A^{\bv}_{h(\bv)} \cdot \bigg(\sum_{o:(o, h) \in S_{\bu}} M^{\bu}_o \bigg)\cdot A^{\bv}_{h(\bv)}) \otimes T_h  \ket{\psi}\\
\leq~&\E_{\bu, \bv} \sum_h \bra{\psi} (A^{\bv}_{h(\bv)})^2 \otimes T_h  \ket{\psi}\tag{because~$M$ is a sub-measurement}\\
\leq~&\E_{\bu, \bv} \sum_h \bra{\psi} I \otimes T_h  \ket{\psi}\tag{because~$A^{\bv}_{h(\bv)} \leq I$}\\
=~&1.\tag{because~$T$ is a measurement}
\end{align*}
As for the term inside the second square root, it is equal to the term inside the first square root in \Cref{eq:move-one-cauchy-schwarz}, which we showed was at most $2\delta$.

Having moved both $A$'s to the left-hand side, we want to show that
\begin{equation}\label{eq:change-one}
\eqref{eq:move-another}
\approx_{\sqrt{\zeta_{\mathrm{variance}}}}
\E_{\bu, \bv} \sum_{(o, h) \in S_{\bu}} \bra{\psi} (A^{\bu}_{h(\bu)} \cdot M^{\bu}_o \cdot A^{\bv}_{h(\bv)}) \otimes T_h  \ket{\psi}.
\end{equation} 
To show this, we bound the magnitude of the difference.
\begin{multline}\label{eq:change-one-cauchy-schwarz}
\Big|\E_{\bu, \bv} \sum_{(o, h) \in S_{\bu}} \bra{\psi} ((A^{\bv}_{h(\bv)} - A^{\bu}_{h(\bu)}) \cdot M^{\bu}_o \cdot A^{\bv}_{h(\bv)}) \otimes T_h  \ket{\psi}\Big|\\
\leq
\sqrt{\E_{\bu, \bv} \sum_{(o, h) \in S_{\bu}} \bra{\psi} ((A^{\bv}_{h(\bv)} - A^{\bu}_{h(\bu)}) \cdot M^{\bu}_o \cdot (A^{\bv}_{h(\bv)} - A^{\bu}_{h(\bu)})) \otimes T_h  \ket{\psi}}\\
\cdot \sqrt{\E_{\bu, \bv} \sum_{(o, h) \in S_{\bu}} \bra{\psi} (A^{\bv}_{h(\bv)} \cdot M^{\bu}_o \cdot A^{\bv}_{h(\bv)}) \otimes T_h  \ket{\psi}}.
\end{multline}
The term inside the first square root is
\begin{align*}
&\E_{\bu, \bv} \sum_h \bra{\psi} ((A^{\bv}_{h(\bv)} - A^{\bu}_{h(\bu)}) \cdot \bigg( \sum_{o:(o, h) \in S_{\bu}}M^{\bu}_o\bigg) \cdot (A^{\bv}_{h(\bv)} - A^{\bu}_{h(\bu)})) \otimes T_h  \ket{\psi}\\
\leq~&\E_{\bu, \bv} \sum_h \bra{\psi} (A^{\bv}_{h(\bv)} - A^{\bu}_{h(\bu)})^2 \otimes T_h  \ket{\psi}. \tag{because~$M^{\bu}$ is a sub-measurement}
\end{align*}
But $T \in \polysub{m}{q}{d}$, and so by \Cref{lem:global-variance-of-points} this expression is at most $\zeta_{\mathrm{variance}}$.
As for the term inside the second square root, it is equal to the term inside the first square root in \Cref{eq:move-another-cauchy-schwarz}, which we showed was at most~$1$.
Finally, we want to show that
\begin{equation}\label{eq:change-another}
\eqref{eq:change-one}
\approx_{\sqrt{\zeta_{\mathrm{variance}}}}
\E_{\bu, \bv} \sum_{(o, h) \in S_{\bu}} \bra{\psi} (A^{\bu}_{h(\bu)} \cdot M^{\bu}_o \cdot A^{\bu}_{h(\bu)}) \otimes T_h  \ket{\psi}.
\end{equation} 
To show this, we bound the magnitude of the difference.
\begin{multline*}
\Big|
\E_{\bu, \bv} \sum_{(o, h) \in S_{\bu}} \bra{\psi} (A^{\bu}_{h(\bu)} \cdot M^{\bu}_o \cdot (A^{\bv}_{h(\bv)} - A^{\bu}_{h(\bu)})) \otimes T_h  \ket{\psi}\Big|\\
\leq
\sqrt{\E_{\bu, \bv} \sum_{(o, h) \in S_{\bu}} \bra{\psi} (A^{\bu}_{h(\bu)} \cdot M^{\bu}_o \cdot A^{\bu}_{h(\bu)}) \otimes T_h  \ket{\psi}}\\
\cdot \sqrt{\E_{\bu, \bv} \sum_{(o, h) \in S_{\bu}} \bra{\psi} ((A^{\bv}_{h(\bv)} - A^{\bu}_{h(\bu)}) \cdot M^{\bu}_o \cdot (A^{\bv}_{h(\bv)} - A^{\bu}_{h(\bu)})) \otimes T_h  \ket{\psi}}.
\end{multline*}
The term inside the first square root is
\begin{align*}
&\E_{\bu, \bv} \sum_h \bra{\psi} (A^{\bu}_{h(\bu)} \cdot \bigg(\sum_{o:(o, h) \in S_{\bu}} M^{\bu}_o \bigg)\cdot A^{\bu}_{h(\bu)}) \otimes T_h  \ket{\psi}\\
\leq~&\E_{\bu, \bv} \sum_h \bra{\psi} (A^{\bu}_{h(\bu)})^2 \otimes T_h  \ket{\psi}\tag{because~$M$ is a sub-measurement}\\
\leq~&\E_{\bu, \bv} \sum_h \bra{\psi} I \otimes T_h  \ket{\psi}\tag{because~$A^{\bu}_{h(\bu)} \leq I$}\\
=~&1.\tag{because~$T$ is a measurement}
\end{align*}
As for the term inside the second square root, it is equal to the term inside the first square root in \Cref{eq:change-one-cauchy-schwarz}, which we showed was at most $\zeta_{\mathrm{variance}}$.
This concludes the proof with an error of $2\sqrt{2\delta} + 2\sqrt{\zeta_{\mathrm{variance}}}$.
The lemma now follows by observing that $2\delta \leq \zeta_{\mathrm{variance}}$.
\end{proof}

We now show that~$H$ satisfies \Cref{item:self-improvement-completeness,item:self-improvement-A-consistency,item:self-improvement-self,item:self-improvement-boundedness}.

\vspace{\baselineskip}
\noindent
\emph{Proof of \Cref{item:self-improvement-completeness} (Completeness).}
The completeness of~$H$ is
\begin{align}
\sum_h \bra{\psi} H_h \otimes I \ket{\psi}
& = \E_{\bu} \sum_h \bra{\psi} H^{\bu}_h \otimes I \ket{\psi}\nonumber\\
&= \E_{\bu} \sum_h \bra{\psi} (A^{\bu}_{h(\bu)} \cdot T_h \cdot A^{\bu}_{h(\bu)}) \otimes I \ket{\psi}\nonumber\\
&= \E_{\bu} \sum_a \bra{\psi} (A^{\bu}_{a} \cdot T_{[h(\bu) = a]} \cdot A^{\bu}_{a}) \otimes I \ket{\psi}.\label{eq:bracketize-the-expression}
\end{align}
We claim that
\begin{equation}\label{eq:yet-another-move-a}
\eqref{eq:bracketize-the-expression}
\approx_{2\sqrt{\delta}}
\E_{\bu} \sum_a \bra{\psi} ( T_{[h(\bu)=a]} \cdot A^{\bu}_{a}) \otimes A^{\bu}_{a} \ket{\psi}.
\end{equation}
To show this, we bound the magnitude of the difference
\begin{multline*}
\Big|\E_{\bu} \sum_a \bra{\psi} (A^{\bu}_{a} \otimes I - I \otimes A^{\bu}_{a}) \cdot (( T_{[h(\bu) = a]} \cdot A^{\bu}_{a}) \otimes I) \ket{\psi}\Big|\\
\leq
\sqrt{\E_{\bu} \sum_a \bra{\psi} (A^{\bu}_{a} \otimes I - I \otimes A^{\bu}_{a})^2 \ket{\psi}}
\cdot \sqrt{\E_{\bu} \sum_a \bra{\psi} (A^{\bu}_{a} \cdot  T^2_{[h(\bu) = a]} \cdot A^{\bu}_{a}) \otimes I \ket{\psi}}.
\end{multline*}
The expression inside the first square root is at most $2\delta \leq 4 \delta$
by \Cref{prop:simeq-to-approx} and the self-consistency of~$A$,
and the expression inside the second square root is at most~$1$ because $T_{[h(\bu) = a]} \leq I$.
Next, we claim that 
\begin{equation}\label{eq:mysterious-case-of-the-disappearing-a}
\eqref{eq:yet-another-move-a}
\approx_{\sqrt{\delta}}
\E_{\bu} \sum_a \bra{\psi} ( T_{[h(\bu)=a]} \cdot A^{\bu}_{a}) \otimes I \ket{\psi}.
\end{equation}
To show this, we bound the magnitude of the difference.
\begin{multline*}
\Big|\E_{\bu} \sum_a \bra{\psi} ( T_{[h(\bu)=a]} \cdot A^{\bu}_{a}) \otimes (I-A^{\bu}_{a}) \ket{\psi}\Big|\\
\leq
\sqrt{\E_{\bu} \sum_a \bra{\psi} ( T_{[h(\bu)=a]})^2 \otimes I \ket{\psi}}
\cdot \sqrt{\E_{\bu} \sum_a \bra{\psi} ( A^{\bu}_{a})^2 \otimes (I-A^{\bu}_{a})^2 \ket{\psi}}.
\end{multline*}
The expression inside the first square root is at most~$1$ because~$T$ is a sub-measurement.
By the projectivity of~$A$, the expression inside the second square root is equal to
\begin{equation*}
\E_{\bu} \sum_a \bra{\psi} ( A^{\bu}_{a}) \otimes (I-A^{\bu}_{a}) \ket{\psi}
= \E_{\bu} \sum_{a\neq b} \bra{\psi} A^{\bu}_{a} \otimes A^{\bu}_{b} \ket{\psi},
\end{equation*}
which is at most~$\delta$ by the self-consistency of~$A$.
We can rewrite \Cref{eq:mysterious-case-of-the-disappearing-a} as
\begin{align*}
&\E_{\bu} \sum_h \bra{\psi} ( T_{h} \cdot A^{\bu}_{h(\bu)}) \otimes I \ket{\psi}\\
 =~&  \sum_h \bra{\psi} ( T_{h} \cdot\E_{\bu} A^{\bu}_{h(\bu)}) \otimes I \ket{\psi}\\
 =~&  \sum_h \bra{\psi} ( T_{h} \cdot Z) \otimes I \ket{\psi}\tag{by \Cref{eq:swap-Z-for-A}}\\
=~&  \bra{\psi}  Z \otimes I \ket{\psi}.\tag{because~$T$ is a measurement}
\end{align*}
We pause and record what we have shown so far:
\begin{equation}\label{eq:gonna-use-this-later-H-versus-Z}
\sum_h \bra{\psi} H_h \otimes I \ket{\psi} \geq \bra{\psi}  Z \otimes I \ket{\psi} - 3\sqrt{\delta}.
\end{equation}
At this point, we can lower-bound 
\begin{align*}
\bra{\psi}  Z \otimes I \ket{\psi}
\geq~& \sum_g \bra{\psi} Z \otimes G_g \ket{\psi} \tag{because~$G$ is a sub-measurement}\\
\geq~& \sum_g \bra{\psi} (\E_{\bu} A^{\bu}_{g(\bu)}) \otimes G_g \ket{\psi} \tag{by \Cref{eq:Z-greater-than-A}}\\
=~& \E_{\bu} \sum_a \bra{\psi}  A^{\bu}_{a} \otimes G_{[g(\bu)=a]} \ket{\psi}\\
\geq~& 1-\nu.\tag{by \Cref{item:self-improvement-G-consistency} and \Cref{prop:simeq-for-measurements}}
\end{align*}
Thus, the completeness is at least $1  - \nu - 3\sqrt{\delta}$. 
The proof follows from noting that $3\sqrt{\delta} \leq \zeta$.

\vspace{\baselineskip}
\noindent
\emph{Proof of \Cref{item:self-improvement-A-consistency} (Consistency with~$A$).}
The inconsistency of~$H$ with~$A$ is
\begin{equation}\label{eq:consistency-with-A-baby-step}
\E_{\bu \sim \F_q^m} \sum_{a \neq b} \bra{\psi} A^{\bu}_a \otimes H_{[h(\bu) = b]} \ket{\psi}
= \E_{\bu \sim \F_q^m} \sum_{a, h: h(\bu) \neq a} \bra{\psi} A^{\bu}_a \otimes H_h \ket{\psi}.
\end{equation}
We now apply \Cref{lem:add-in-u} to the right-hand side of \Cref{eq:consistency-with-A-baby-step}.
To do so, we set $\calO = \F_q^m$, $M = A$, and $S_u = \{(a, h) : h(u) \neq a\}$.
Then \Cref{lem:add-in-u} implies that
\begin{equation*}
\eqref{eq:consistency-with-A-baby-step}
\approx_{4 \sqrt{\zeta_{\mathrm{variance}}}}
\E_{\bu \sim \F_q^m} \sum_{a, h: h(\bu) \neq a} \bra{\psi} (A^{\bu}_{h(\bu)} \cdot A^{\bu}_a  \cdot A^{\bu}_{h(\bu)})\otimes T_h \ket{\psi},
\end{equation*}
which is equal to~$0$ because~$A$ is projective.
This implies that
\begin{equation}\label{eq:explicit-bound-for-A-consistency}
\E_{\bu \sim \F_q^m} \sum_{a \neq b} \bra{\psi} A^{\bu}_a \otimes H_{[h(\bu) = b]} \ket{\psi} \leq 4 \sqrt{\zeta_{\mathrm{variance}}}.
\end{equation}
The proof follows from noting that
\begin{align*}
4\sqrt{\zeta_{\mathrm{variance}}}
&= 4 \cdot \sqrt{24 m \cdot\Big(\eps + \delta + \frac{md}{q}\Big)}\\
&\leq 20 m \cdot \Big(\eps^{1/2} + \delta^{1/2} + (d/q)^{1/2}\Big)\\
&\leq \zeta.
\end{align*}

\vspace{\baselineskip}
\noindent
\emph{Proof of \Cref{item:self-improvement-self} (Strong self-consistency).}
We begin by recording the following facts which follow from the definition of $H^u_h$ and the projectivity of~$A$:
\begin{align}
H^u_h &= A^u_{h(u)} \cdot T_h \cdot A^u_{h(u)} = A^u_{h(u)} \cdot H^u_h \cdot A^u_{h(u)},\label{eq:h-sandwich}\\
A^u_{h(u)} \cdot H^u_{h'} \cdot A^u_{h(u)} &=  H^u_{h'} \cdot A^u_{h(u)} = (A^u_{h(u)} \cdot T_{h'} \cdot A^u_{h(u)}) \cdot \bone[h(u) = h'(u)].\label{eq:h-blt}
\end{align}
The strong self-consistency of~$H$ is
\begin{equation}\label{eq:self-consistency-baby-step}
\sum_{h\in\polyfunc{m}{q}{d}} \bra{\psi} H_h \otimes H_h \ket{\psi}
= \E_{\bu \sim \F_q^m} \sum_{h \in \polyfunc{m}{q}{d}} \bra{\psi} H_h^{\bu} \otimes H_h \ket{\psi}.
\end{equation}
We now apply \Cref{lem:add-in-u} to the right-hand side of \Cref{eq:self-consistency-baby-step}.
To do so, we set
\begin{equation*}
\text{$\calO =\polyfunc{m}{q}{d}$, $M = H$, and $S_u = \{(h, h) : h \in \polyfunc{m}{q}{d}\}$.}
\end{equation*}
Then \Cref{lem:add-in-u} implies that
\begin{equation}\label{eq:release-the-kraken}
\eqref{eq:self-consistency-baby-step}
\approx_{4 \sqrt{\zeta_{\mathrm{variance}}}}
\E_{\bu \sim \F_q^m} \sum_{h \in \polyfunc{m}{q}{d}} \bra{\psi} (A^{\bu}_{h(\bu)} \cdot H_h^{\bu} \cdot A^{\bu}_{h(\bu)}) \otimes T_h \ket{\psi}.
\end{equation}
Having placed two $A$'s on the left-hand side, we want to show that
\begin{equation}\label{eq:threw-in-h-prime}
\eqref{eq:release-the-kraken}
\approx_{2\sqrt{\zeta_{\mathrm{variance}}} + \frac{md}{q}}
\E_{\bu \sim \F_q^m} \sum_{h, h' \in \polyfunc{m}{q}{d}} \bra{\psi} (A^{\bu}_{h(\bu)} \cdot H^{\bu}_{h'} \cdot A^{\bu}_{h(\bu)}) \otimes T_h \ket{\psi}.
\end{equation}
We note that \eqref{eq:threw-in-h-prime} is at least as big as \eqref{eq:release-the-kraken}. Thus, we want to upper-bound
\begin{align}
\eqref{eq:threw-in-h-prime}-\eqref{eq:release-the-kraken}
& = \E_{\bu \sim \F_q^m} \sum_{h \neq h'} \bra{\psi}  (A^{\bu}_{h(\bu)} \cdot H^{\bu}_{h'} \cdot A^{\bu}_{h(\bu)}) \otimes T_h \ket{\psi}\nonumber\\
& = \E_{\bu \sim \F_q^m} \sum_{h \neq h'} \bra{\psi}  (A^{\bu}_{h(\bu)} \cdot T_{h'} \cdot A^{\bu}_{h(\bu)}) \otimes T_h \ket{\psi} \cdot \bone[h(\bu) = h'(\bu)],\label{eq:added-indicator}
\end{align}
where the second equality is by~\Cref{eq:h-blt}.
To do this, we first show that
\begin{equation}\label{eq:swapped-u-for-v}
\eqref{eq:added-indicator}
\approx_{\sqrt{\zeta_{\mathrm{variance}}}}
\E_{\bu,\bv} \sum_{h \neq h'} \bra{\psi}  (A^{\bv}_{h(\bv)} \cdot T_{h'} \cdot A^{\bu}_{h(\bu)}) \otimes T_h \ket{\psi} \cdot \bone[h(\bu) = h'(\bu)].
\end{equation}
To show this, we bound the magnitude of the difference.
\begin{multline}
\Big|\E_{\bu,\bv} \sum_{h \neq h'} \bra{\psi} ( (A^{\bu}_{h(\bu)} - A^{\bv}_{h(\bv)}) \cdot T_{h'} \cdot A^{\bu}_{h(\bu)}) \otimes T_h \ket{\psi} \cdot \bone[h(\bu) = h'(\bu)]\Big|\\
\leq
\sqrt{\E_{\bu,\bv} \sum_{h \neq h'} \bra{\psi}  ( (A^{\bu}_{h(\bu)} - A^{\bv}_{h(\bv)}) \cdot T_{h'} \cdot (A^{\bu}_{h(\bu)} - A^{\bv}_{h(\bv)})) \otimes T_h\ket{\psi}}\\
\cdot \sqrt{\E_{\bu,\bv} \sum_{h \neq h'} \bra{\psi} (A^{\bu}_{h(\bu)} \cdot T_{h'} \cdot A^{\bu}_{h(\bu)}) \otimes T_h\ket{\psi} \cdot \bone[h(\bu) = h'(\bu)]}. \label{eq:swapped-u-for-cauchy-schwarz}
\end{multline}
The term inside the first square root is
\begin{align*}
&\E_{\bu,\bv} \sum_{h} \bra{\psi} ( (A^{\bu}_{h(\bu)} - A^{\bv}_{h(\bv)}) \cdot \bigg(\sum_{h' \neq h} T_{h'}\bigg) \cdot (A^{\bu}_{h(\bu)} - A^{\bv}_{h(\bv)})) \otimes T_h\ket{\psi}\\
\leq~& \E_{\bu,\bv} \sum_{h} \bra{\psi}  (A^{\bu}_{h(\bu)} - A^{\bv}_{h(\bv)})^2 \otimes T_h\ket{\psi},\tag{because $T$ is a sub-measurement}
\end{align*}
But $T \in \polyfunc{m}{q}{d}$, and so by \Cref{lem:global-variance-of-points} this expression is at most~$\zeta_{\mathrm{variance}}$.
The term inside the second square root is equal to
\begin{equation*}
\E_{\bu,\bv} \sum_{h \neq h'} \bra{\psi}  H^{\bu}_{h'} \otimes T_h \ket{\psi} \cdot \bone[h(\bu) = h'(\bu)],
\end{equation*}
which is at most~$1$ because $T$ and $H^{\bu}$ are sub-measurements.
Next, we show that
\begin{equation}\label{eq:swapped-u-for-v-this-time-it's-personal}
\eqref{eq:swapped-u-for-v}
\approx_{\sqrt{\zeta_{\mathrm{variance}}}}
\E_{\bu,\bv} \sum_{h \neq h'} \bra{\psi}  (A^{\bv}_{h(\bv)} \cdot T_{h'} \cdot A^{\bv}_{h(\bv)}) \otimes T_h \ket{\psi} \cdot \bone[h(\bu) = h'(\bu)].
\end{equation}
To show this, we bound the magnitude of the difference.
\begin{multline*}
\Big|\E_{\bu,\bv} \sum_{h \neq h'} \bra{\psi}  (A^{\bv}_{h(\bv)} \cdot T_{h'} \cdot (A^{\bu}_{h(\bu)} - A^{\bv}_{h(\bv)})) \otimes T_h \ket{\psi} \cdot \bone[h(\bu) = h'(\bu)]\Big|\\
\leq
\sqrt{\E_{\bu,\bv} \sum_{h \neq h'} \bra{\psi}  (A^{\bv}_{h(\bv)} \cdot T_{h'} \cdot A^{\bv}_{h(\bv)}) \otimes T_h\ket{\psi} \cdot \bone[h(\bu) = h'(\bu)]}\\
\cdot \sqrt{\E_{\bu,\bv} \sum_{h \neq h'} \bra{\psi} ( (A^{\bu}_{h(\bu)} - A^{\bv}_{h(\bv)}) \cdot T_{h'} \cdot (A^{\bu}_{h(\bu)} - A^{\bv}_{h(\bv)}))\otimes T_h\ket{\psi}}.
\end{multline*}
The term inside the first square root is at most
\begin{align}
\E_{\bu,\bv} \sum_{h \neq h'} \bra{\psi}  (A^{\bv}_{h(\bv)} \cdot T_{h'} \cdot A^{\bv}_{h(\bv)})\otimes T_h\ket{\psi}
& = \E_{\bu,\bv} \sum_{h} \bra{\psi}  (A^{\bv}_{h(\bv)} \cdot \bigg(\sum_{h' \neq h} T_{h'}\bigg) \cdot A^{\bv}_{h(\bv)}) \otimes T_h\ket{\psi}\\
& \leq \E_{\bu,\bv} \sum_{h} \bra{\psi}   (A^{\bv}_{h(\bv)})^2\otimes T_h\ket{\psi}\tag{$T$ is a sub-measurement}\\
& \leq \E_{\bu,\bv} \sum_{h} \bra{\psi}   I \otimes T_h\ket{\psi}\tag{because $A^{\bv}_{h(\bv)} \leq I$}\\
& \leq 1, \label{eq:gonna-use-this-later}
\end{align}
where the last step again uses the fact that~$T$ is a sub-measurement.
As for the term inside the second square root, it is equal to the term inside the first square root in \Cref{eq:swapped-u-for-cauchy-schwarz},
which we showed was at most $\zeta_{\mathrm{variance}}$.
Finally, \Cref{eq:swapped-u-for-v-this-time-it's-personal} is equal to
\begin{align*}
&\E_{\bv} \sum_{h \neq h'} \bra{\psi} (A^{\bv}_{h(\bv)} \cdot T_{h'} \cdot A^{\bv}_{h(\bv)}) \otimes T_h\ket{\psi} \cdot \E_{\bu}\bone[h(\bu) = h'(\bu)]\\
\leq~& \E_{\bv} \sum_{h \neq h'} \bra{\psi} (A^{\bv}_{h(\bv)} \cdot T_{h'} \cdot A^{\bv}_{h(\bv)}) \otimes T_h \ket{\psi} \cdot \frac{md}{q} \tag{by Schwartz-Zippel}\\
\leq~& \frac{md}{q}. \tag{by \Cref{eq:gonna-use-this-later}}
\end{align*}

By~\Cref{eq:h-blt}, \Cref{eq:threw-in-h-prime} is equal to
\begin{equation}\label{eq:delete-an-A}
\E_{\bu \sim \F_q^m} \sum_{h, h' \in \polyfunc{m}{q}{d}} \bra{\psi} (H^{\bu}_{h'} \cdot A^{\bu}_{h(\bu)}) \otimes T_h \ket{\psi}.
\end{equation}
Now, we show that
\begin{equation}\label{eq:swap-u-for-v-attack-of-the-clones}
\eqref{eq:delete-an-A}
\approx_{\sqrt{\zeta_{\mathrm{variance}}}}
\E_{\bu, \bv} \sum_{h, h' \in \polyfunc{m}{q}{d}} \bra{\psi} (H^{\bu}_{h'} \cdot A^{\bv}_{h(\bv)}) \otimes T_h \ket{\psi}.
\end{equation}
To show this, we bound the magnitude of the difference.
\begin{align*}
&\Big|\E_{\bu, \bv} \sum_{h, h' \in \polyfunc{m}{q}{d}} \bra{\psi} (H^{\bu}_{h'} \cdot (A^{\bu}_{h(\bu)} - A^{\bv}_{h(\bv)})) \otimes T_h \ket{\psi}\Big|\\
&\leq
\sqrt{\E_{\bu, \bv} \sum_{h, h' \in \polyfunc{m}{q}{d}} \bra{\psi} H^{\bu}_{h'} \otimes T_h\ket{\psi}}\\
&\quad\cdot \sqrt{\E_{\bu, \bv} \sum_{h, h' \in \polyfunc{m}{q}{d}} \bra{\psi} ((A^{\bu}_{h(\bu)} - A^{\bv}_{h(\bv)}) \cdot H^{\bu}_{h'} \cdot (A^{\bu}_{h(\bu)} - A^{\bv}_{h(\bv)})) \otimes T_h\ket{\psi}}.
\end{align*}
The expression inside the first square root is at most~$1$ because~$T$ and~$H^{\bu}$ are sub-measurements.
The term inside the second square root is
\begin{align*}
&\E_{\bu, \bv} \sum_{h} \bra{\psi} ((A^{\bu}_{h(\bu)} - A^{\bv}_{h(\bv)}) \cdot \bigg( \sum_{h'} H^{\bu}_{h'}\bigg) \cdot (A^{\bu}_{h(\bu)} - A^{\bv}_{h(\bv)})) \otimes T_h \ket{\psi}\\
\leq~& \E_{\bu, \bv} \sum_{h} \bra{\psi}  (A^{\bu}_{h(\bu)} - A^{\bv}_{h(\bv)})^2 \otimes T_h\ket{\psi}.
\tag{because $H^{\bu}$ is a sub-measurement}
\end{align*}
But $T \in \polysub{m}{q}{d}$, and so by \Cref{lem:global-variance-of-points} this expression is at most~$\zeta_{\mathrm{variance}}$.
Next, we show that
\begin{equation}\label{eq:move-over-v}
\eqref{eq:swap-u-for-v-attack-of-the-clones}
\approx_{\sqrt{2\delta}}
\E_{\bu, \bv} \sum_{h, h' \in \polyfunc{m}{q}{d}} \bra{\psi} H^{\bu}_{h'}  \otimes (T_h \cdot A^{\bv}_{h(\bv)}) \ket{\psi}.
\end{equation}
To show this, we bound the magnitude of the difference.
\begin{align*}
&\Big|\E_{\bu, \bv} \sum_{h, h' \in \polyfunc{m}{q}{d}} \bra{\psi}  (H^{\bu}_{h'}\otimes T_h) \cdot (A^{\bv}_{h(\bv)} \otimes I - I \otimes A^{\bv}_{h(\bv)}) \ket{\psi}\Big|\\
=~&\Big|\E_{\bu, \bv} \sum_{a \in \F_q}  \bra{\psi}  \bigg( \sum_{h'}H^{\bu}_{h'} \otimes \sum_{h:h(\bv) = a} T_h \bigg) \cdot (A^{\bv}_a \otimes I - I \otimes A^{\bv}_a) \ket{\psi}\Big|\\
\leq~&
\sqrt{\E_{\bu, \bv} \sum_{a \in \F_q}  \bra{\psi}  \bigg( \sum_{h'}H^{\bu}_{h'} \otimes \sum_{h:h(\bv) = a} T_h\bigg)^2\ket{\psi}}
\cdot \sqrt{\E_{\bu, \bv} \sum_{a \in \F_q}  \bra{\psi}(A^{\bv}_a \otimes I - I \otimes A^{\bv}_a)^2 \ket{\psi}}.
\end{align*}
The expression inside the first square root is at most~$1$ because~$T$ and~$H^{\bu}$ are sub-measurements,
and the expression inside the second square root is at most~$2\delta$ by \Cref{prop:simeq-to-approx} and the self-consistency of~$A$.
Now, \Cref{eq:move-over-v} is equal to
\begin{align*}
&\E_{\bu} \sum_{h, h'} \bra{\psi} H^{\bu}_{h'} \otimes (T_h\cdot \E_{\bv} A^{\bv}_{h(\bv)}) \ket{\psi}\\
=~& \E_{\bu} \sum_{h, h'} \bra{\psi} H^{\bu}_{h'} \otimes  (T_h\cdot Z)\ket{\psi} \tag{by \Cref{eq:swap-Z-for-A}}\\
=~& \E_{\bu} \sum_{ h'} \bra{\psi} H^{\bu}_{h'}  \otimes Z\ket{\psi} \tag{because~$T$ is a measurement}\\
\geq~& \E_{\bu} \sum_{ h'} \bra{\psi}  H^{\bu}_{h'} \otimes  \E_{\bv} A^{\bv}_{h'(\bv)} \ket{\psi} \tag{by \Cref{eq:Z-greater-than-A}}\\
=~& \E_{\bu, \bv} \sum_{ h'} \bra{\psi}  H^{\bu}_{h'} \otimes  A^{\bv}_{h'(\bv)} \ket{\psi}\\
=~& \E_{\bv} \sum_{a} \bra{\psi}  H_{[h(\bv)=a]} \otimes  A^{\bv}_{a} \ket{\psi}\\
=~& \E_{\bv} \sum_{a} \bra{\psi}  H_{[h(\bv)=a]} \otimes  I \ket{\psi} - \E_{\bv} \sum_{a\neq b} \bra{\psi}  H_{[h(\bv)=a]} \otimes  A^{\bv}_{b} \ket{\psi}\tag{because~$A$ is a measurement}\\
\geq~& \E_{\bv} \sum_{a} \bra{\psi}  H_{[h(\bv)=a]} \otimes  I \ket{\psi} - 4 \sqrt{\zeta_{\mathrm{variance}}} \tag{by \Cref{eq:explicit-bound-for-A-consistency}}\\
=~&  \sum_{h} \bra{\psi}  H_h \otimes  I \ket{\psi} - 4 \sqrt{\zeta_{\mathrm{variance}}}.
\end{align*}
In total, we have shown that
\begin{equation*}
\sum_{h\in\polyfunc{m}{q}{d}} \bra{\psi} H_h \otimes H_h \ket{\psi} \geq  \sum_{h} \bra{\psi}  H_h \otimes  I \ket{\psi} - 7 \sqrt{\zeta_{\mathrm{variance}}} - \sqrt{2\delta} - \frac{md}{q} - 4\sqrt{\zeta_{\mathrm{variance}}}.
\end{equation*}
Because
\begin{align*}
11 \sqrt{\zeta_{\mathrm{variance}}} + \sqrt{2\delta} + \frac{md}{q}
& = 11 \sqrt{24m\cdot\Big(\eps + \delta + \frac{md}{q}\Big)} + \sqrt{2\delta} + \frac{md}{q}\\
&\leq  55 m \cdot \Big(\eps^{1/2} + \delta^{1/2} + (d/q)^{1/2}\Big) + 2\sqrt{\delta} + m \cdot (d/q)^{1/2}\\
&\leq  57 m \cdot \Big(\eps^{1/2} + \delta^{1/2} + (d/q)^{1/2}\Big)\\
&\leq \zeta,
\end{align*}
this concludes the proof of \Cref{item:self-improvement-self}.

\vspace{\baselineskip}
\noindent
\emph{Proof of \Cref{item:self-improvement-boundedness} (Boundedness).}
The boundedness of~$H$ is
\ignore{
\begin{align*}
&\bra{\psi} Z \otimes (I - H) \ket{\psi}\\
=~& \bra{\psi} Z \otimes I \ket{\psi} - \sum_h \bra{\psi} Z \otimes H_h \ket{\psi}\\
\leq~& \bra{\psi} Z \otimes I \ket{\psi} - \sum_h \bra{\psi} \E_{\bu} A^{\bu}_{h(\bu)} \otimes H_h \ket{\psi} \tag{by \Cref{eq:Z-greater-than-A}}\\
=~&  \bra{\psi} Z \otimes I \ket{\psi} -\E_{\bu} \sum_h \bra{\psi}  A^{\bu}_{h(\bu)} \otimes H_h \ket{\psi}\\
=~&  \bra{\psi} Z \otimes I \ket{\psi} -\E_{\bu} \sum_h \bra{\psi}  I \otimes H_h \ket{\psi} +\E_{\bu} \sum_{a, h:h(\bu)\neq a} \bra{\psi}  A^{\bu}_{a} \otimes H_h \ket{\psi} \tag{because~$A$ is a measurement}\\
\leq~&  \bra{\psi} Z \otimes I \ket{\psi} -\E_{\bu} \sum_h \bra{\psi}  I \otimes H_h \ket{\psi} + 4 \sqrt{\zeta_{\mathrm{variance}}}\tag{by \Cref{eq:explicit-bound-for-A-consistency}}\\
\leq~& 3 \sqrt{\delta} + 4 \sqrt{\zeta_{\mathrm{variance}}}.\tag{by \Cref{eq:gonna-use-this-later-H-versus-Z}}
\end{align*}
}
\begin{align*}
&\bra{\psi} Z \otimes I \ket{\psi} -\E_{\bu} \sum_a \bra{\psi}  A^{\bu}_{a} \otimes H_{[h(\bu)=a]} \ket{\psi} \\
=~&  \bra{\psi} Z \otimes I \ket{\psi} -\E_{\bu} \sum_h \bra{\psi}  A^{\bu}_{h(\bu)} \otimes H_h \ket{\psi}\\
=~&  \bra{\psi} Z \otimes I \ket{\psi} -\E_{\bu} \sum_h \bra{\psi}  I \otimes H_h \ket{\psi} +\E_{\bu} \sum_{a, h:h(\bu)\neq a} \bra{\psi}  A^{\bu}_{a} \otimes H_h \ket{\psi} \tag{because~$A$ is a measurement}\\
\leq~&  \bra{\psi} Z \otimes I \ket{\psi} -\E_{\bu} \sum_h \bra{\psi}  I \otimes H_h \ket{\psi} + 4 \sqrt{\zeta_{\mathrm{variance}}}\tag{by \Cref{eq:explicit-bound-for-A-consistency}}\\
\leq~& 3 \sqrt{\delta} + 4 \sqrt{\zeta_{\mathrm{variance}}}.\tag{by \Cref{eq:gonna-use-this-later-H-versus-Z}}
\end{align*}
We can bound the error by
\begin{align*}
 3 \sqrt{\delta} + 4 \sqrt{\zeta_{\mathrm{variance}}}
 & =  3 \sqrt{\delta} + 4 \sqrt{24m\cdot\Big(\eps + \delta + \frac{md}{q}\Big)}\\
 & \leq 3 \sqrt{\delta} + 20 m\cdot\Big(\eps^{1/2} + \delta^{1/2} + (d/q)^{1/2}\Big)\\
 &\leq 23 m\cdot\Big(\eps^{1/2} + \delta^{1/2} + (d/q)^{1/2}\Big)\\
 & \leq \zeta.
\end{align*}
This completes the proof.
\qed

\subsection{Self-improving to a projective measurement}
\label{sec:self-improvement-projective}

We now prove the full self-improvement theorem, i.e.\ \Cref{thm:self-improvement-in-induction-section}.
To do so, we will apply the orthonormalization lemma \Cref{thm:orthonormalization}
to the output of \Cref{lem:self-improvement-helper} and argue that it maintains the four properties of~$H$.

\begin{theorem}[Self-improvement; \Cref{thm:self-improvement-in-induction-section} restated]\label{thm:self-improvement}
Let $G \in \polymeas{m}{q}{d}$ be a measurement with the following properties:
\begin{enumerate}
\ignore{
\item (Completeness): \label{item:self-improvement-projective-G-completeness}  If $G =  \sum_g G_g$, then
  	\begin{equation*}
	\bra{\psi} G \otimes I \ket{\psi} \geq 1 - \kappa.
	\end{equation*}
}
  \item (Consistency with~$A$): \label{item:self-improvement-projective-G-consistency} On average over $\bu \sim \F_q^{m}$,
	  \begin{equation*}
	  A^{u}_a \otimes I \simeq_{\nu} I \otimes G_{[g(u)=a]}.
	  \end{equation*}
	  \end{enumerate}
Let
\begin{equation*}
\zeta = 3000m\cdot \Big(\eps^{1/32} + \delta^{1/32} + (d/q)^{1/32}\Big).
\end{equation*}
Then there exists a projective sub-measurement $H \in \polysub{m}{q}{d}$ with the following properties:
\begin{enumerate}
    \item(Completeness):  \label{item:self-improvement-projective-completeness} If $H =  \sum_h H_h$, then
    \begin{equation*}
    \bra{\psi} H \otimes I \ket{\psi} \geq (1-\nu)-\zeta.
    \end{equation*}
    \item(Consistency with~$A$):\label{item:self-improvement-projective-A-consistency} On average over $\bu \sim \F_q^m$,
        \begin{equation*}
            A^u_a \otimes I \simeq_{\zeta} I \otimes H_{[h(u) = a]}.
        \end{equation*}
    \item(Strong self-consistency): \label{item:self-improvement-projective-self}
    	\begin{equation*}
		H_h \otimes I \approx_{\zeta} I \otimes H_h.
	\end{equation*}
    \item(Boundedness): \label{item:self-improvement-projective-boundedness} There exists a positive-semidefinite matrix~$Z$ such that
    \begin{equation*}
        \bra{\psi} Z \otimes (I - H) \ket{\psi} \leq \zeta
    \end{equation*}
    and for each $h \in \polyfunc{m}{q}{d}$,
	\begin{equation*}
	Z \geq \left(\E_{\bu} A^{\bu}_{h(\bu)}\right).
	\end{equation*}
\end{enumerate}
\end{theorem}
\begin{proof}
We note that the bound we are proving is trivial when at least one of $\eps$, $\delta$, or $d/q$ is $\geq 1$.
In this case, $\zeta \geq 3000$.
Hence, we may assume that $\gamma, \zeta, d/q \leq 1$.
This will aid us when carrying out the error calculations near the end of the proof,
as it allows us to bound terms like $\eps^{1/2}$ by terms like $\eps^{1/4}$.

To begin,  apply \Cref{lem:self-improvement-helper} to~$G$.
Let $\widehat{H} \in \polysub{m}{q}{d}$ be the sub-measurement it outputs
and let $\widehat{\zeta}$ be the error
\begin{equation*}
\widehat{\zeta} = 100m \cdot \Big(\eps^{1/2} + \delta^{1/2} + (d/q)^{1/2}\Big).
\end{equation*}
By \Cref{item:self-improvement-self} of \Cref{lem:self-improvement-helper},
\begin{equation*}
\sum_{h} \bra{\psi} \widehat{H}_h \ot \widehat{H}_h \ket{\psi}
	\geq \bra{\psi} \widehat{H} \ot I \ket{\psi} - \widehat{\zeta}.
\end{equation*}
Thus, \Cref{thm:orthonormalization} implies the existence of a projective sub-measurement
$H \in \polysub{m}{q}{d}$ such that
\begin{equation}\label{eq:approx-between-H-with-and-without-hat}
\widehat{H}_h \otimes I \approx_{\widehat{\zeta}_{\mathrm{ortho}}} H_{h} \otimes I,
\end{equation}
where
\begin{equation*}
\widehat{\zeta}_{\mathrm{ortho}} = 100 \widehat{\zeta}^{1/4}.
\end{equation*}
In addition, \Cref{prop:self-consistency-implies-data-processing} implies that
\begin{equation}\label{eq:approx-data-processed}
\widehat{H}_{[h(u)=a]} \otimes I \approx_{\widehat{\zeta}_{\mathrm{dataprocess}}} H_{[h(u)=a]} \otimes I,
\end{equation}
where
\begin{equation*}
\widehat{\zeta}_{\mathrm{dataprocess}} = 8 \widehat{\zeta} + 8 \sqrt{\widehat{\zeta}_{\mathrm{ortho}}}.
\end{equation*}

Now we prove the four properties of this theorem.
We will show that each quantity is bounded by some error,
and then at the end of this proof we will show that all four errors are bounded by~$\zeta$.
\begin{enumerate}
    \item(Completeness):
\Cref{prop:completeness-transfer-self-consistent-A} implies that
\begin{align*}
 \bra{\psi} H \otimes I \ket{\psi}
 & \geq  \bra{\psi} \widehat{H} \otimes I \ket{\psi} - \widehat{\zeta} - 2\sqrt{\widehat{\zeta}_{\mathrm{ortho}}}\\
 & \geq  (1-\nu) -  2\widehat{\zeta} - 2\sqrt{\widehat{\zeta}_{\mathrm{ortho}}}. \tag{by \Cref{item:self-improvement-projective-completeness} of \Cref{lem:self-improvement-helper}}
\end{align*}
	\item (Consistency with~$A$):
\Cref{prop:triangle-sub} applied to
 \Cref{item:self-improvement-projective-A-consistency} of \Cref{lem:self-improvement-helper}
implies that
\begin{equation*}
A^u_a \otimes I \simeq_{\widehat{\zeta} + \sqrt{\widehat{\zeta}_{\mathrm{dataprocess}}}}
		I \otimes H_{[h(u) = a]}.
\end{equation*}
	\item (Strong self-consistency): \Cref{prop:two-notions-of-self-consistency} implies that
		\begin{equation*}
		\widehat{H}_h \otimes I \approx_{2\widehat{\zeta}} I \otimes \widehat{H}_h.
		\end{equation*}
		Thus,
		\begin{equation*}
			H_h \otimes I
			\approx_{\widehat{\zeta}_{\mathrm{ortho}}} \widehat{H}_{h} \otimes I
			\approx_{2\widehat{\zeta}} I \otimes \widehat{H}_{h}
			\approx_{\widehat{\zeta}_{\mathrm{ortho}}} I \otimes H_h.
		\end{equation*}
		Hence, by \Cref{prop:triangle-inequality-for-approx_delta},
		\begin{equation*}
		H_h \otimes I \approx_{6 \widehat{\zeta} + 6\widehat{\zeta}_{\mathrm{ortho}}} I \otimes H_h.
		\end{equation*}
	\item (Boundedness):
		The boundedness of~$H$ is
		\begin{align}
		&\bra{\psi} Z \otimes (I - H) \ket{\psi}\nonumber\\
		=~& \bra{\psi} Z \otimes I \ket{\psi} - \sum_h \bra{\psi} Z \otimes H_h \ket{\psi}\nonumber\\
		\leq~& \bra{\psi} Z \otimes I \ket{\psi} - \sum_h \bra{\psi} \E_{\bu} A^{\bu}_{h(\bu)} \otimes H_h \ket{\psi} \tag{by \Cref{item:self-improvement-projective-boundedness} of \Cref{lem:self-improvement-helper}}\\
		=~&  \bra{\psi} Z \otimes I \ket{\psi} -\E_{\bu} \sum_h \bra{\psi}  A^{\bu}_{h(\bu)} \otimes H_h \ket{\psi}\nonumber\\
		=~&  \bra{\psi} Z \otimes I \ket{\psi} -\E_{\bu} \sum_a \bra{\psi}  A^{\bu}_{a} \otimes H_{[h(\bu)=a]} \ket{\psi}. \label{eq:almost-there-self-improvement-edition}
		\end{align}
		By \Cref{prop:easy-approx-from-approx-delta},
		\begin{equation*}
		\E_{\bu} \sum_a \bra{\psi}  A^{\bu}_{a} \otimes H_{[h(\bu)=a]} \ket{\psi}
		\approx_{\sqrt{\widehat{\zeta}_{\mathrm{dataprocess}}}} \E_{\bu} \sum_a \bra{\psi}  A^{\bu}_{a} \otimes \widehat{H}_{[h(\bu)=a]} \ket{\psi}.
		\end{equation*}
		Hence,
		\begin{align*}
		\eqref{eq:almost-there-self-improvement-edition}
		& \leq \bra{\psi} Z \otimes I \ket{\psi} -\E_{\bu} \sum_a \bra{\psi}  A^{\bu}_{a} \otimes \widehat{H}_{[h(\bu)=a]} \ket{\psi} + \sqrt{\widehat{\zeta}_{\mathrm{dataprocess}}}\\
		& \leq \widehat{\zeta} + \sqrt{\widehat{\zeta}_{\mathrm{dataprocess}}}. \tag{by \Cref{item:self-improvement-boundedness} of \Cref{lem:self-improvement-helper}}
		\end{align*}
\end{enumerate}

This shows the four properties hold with errors
\begin{equation*}
2\widehat{\zeta} + 2\sqrt{\widehat{\zeta}_{\mathrm{ortho}}},\quad
\widehat{\zeta} + \sqrt{\widehat{\zeta}_{\mathrm{dataprocess}}},\quad
6 \widehat{\zeta} + 6\widehat{\zeta}_{\mathrm{ortho}},\quad
\text{and} \quad \widehat{\zeta} + \sqrt{\widehat{\zeta}_{\mathrm{dataprocess}}},
\end{equation*}
respectively. We now show that these four are bounded by~$\zeta$.
First, using the fact that $100^{1/4} \leq 4$, we note that
\begin{align*}
\widehat{\zeta}_{\mathrm{ortho}}
= 100 \widehat{\zeta}^{1/4}
&= 100 \Big(100m \cdot \Big(\eps^{1/2} + \delta^{1/2} + (d/q)^{1/2}\Big)\Big)^{1/4}\\
&\leq 400 m \cdot \Big(\eps^{1/8} + \delta^{1/8} + (d/q)^{1/8}\Big).
\end{align*}
Hence,
\begin{align*}
6 \widehat{\zeta} + 6\widehat{\zeta}_{\mathrm{ortho}}
& \leq 6 \Big(100m \cdot \Big(\eps^{1/2} + \delta^{1/2} + (d/q)^{1/2}\Big)\Big)
	+ 6 \Big(400 m \cdot \Big(\eps^{1/8} + \delta^{1/8} + (d/q)^{1/8}\Big)\Big)\\
&\leq 600m \cdot \Big(\eps^{1/8} + \delta^{1/8} + (d/q)^{1/8}\Big)
	+2400 m \cdot \Big(\eps^{1/8} + \delta^{1/8} + (d/q)^{1/8}\Big)\\
&= 3000m\cdot \Big(\eps^{1/8} + \delta^{1/8} + (d/q)^{1/8}\Big),
\end{align*}
which is less than~$\zeta$.
In addition, using the fact that~$\sqrt{400} = 20$, we also note that
\begin{align*}
\widehat{\zeta}_{\mathrm{dataprocess}}
&= 8 \widehat{\zeta} + 8 \sqrt{\widehat{\zeta}_{\mathrm{ortho}}}\\
&\leq 8 \Big(100m \cdot \Big(\eps^{1/2} + \delta^{1/2} + (d/q)^{1/2}\Big)\Big)
	+ 8 \sqrt{400 m \cdot \Big(\eps^{1/8} + \delta^{1/8} + (d/q)^{1/8}\Big)}\\
&\leq 800m \cdot \Big(\eps^{1/16} + \delta^{1/16} + (d/q)^{1/16}\Big)
	+ 160m \cdot \Big(\eps^{1/16} + \delta^{1/16} + (d/q)^{1/16}\Big)\\
&= 960m \cdot \Big(\eps^{1/16} + \delta^{1/16} + (d/q)^{1/16}\Big),
\end{align*}
which is clearly less than $\zeta$.
Thus, $2\widehat{\zeta} + 2\sqrt{\widehat{\zeta}_{\mathrm{ortho}}} \leq \widehat{\zeta}_{\mathrm{dataprocess}} \leq \zeta$.
Finally, using $\sqrt{960} \leq 31$,
\begin{align*}
\widehat{\zeta} + \sqrt{\widehat{\zeta}_{\mathrm{dataprocess}}}
&\leq  100m \cdot \Big(\eps^{1/2} + \delta^{1/2} + (d/q)^{1/2}\Big)
		+ \sqrt{960m \cdot \Big(\eps^{1/16} + \delta^{1/16} + (d/q)^{1/16}\Big)}\\
&\leq  100m \cdot \Big(\eps^{1/32} + \delta^{1/32} + (d/q)^{1/32}\Big)
		+ 31 m \cdot \Big(\eps^{1/32} + \delta^{1/32} + (d/q)^{1/32}\Big)\\
&=  131m \cdot \Big(\eps^{1/32} + \delta^{1/32} + (d/q)^{1/32}\Big),
\end{align*}
which is also less than~$\zeta$.
This completes the proof.
\end{proof}

\section{Commutativity of the points measurements}
\label{sec:commutativity-points}

\begin{theorem}\label{thm:commutativity-points}
Let $(\psi, A, B, L)$ be an $(\eps, \delta, \gamma)$-good symmetric strategy for the
$(m,q,d)$ low individual degree test.
On average over independent and uniformly random $\bu, \bv \sim \F_q^m$,
\begin{equation*}
(A^{u}_a \cdot A^{v}_b) \otimes I \approx_{32\gamma m} (A^{v}_b\cdot A^{u}_a) \otimes I.
\end{equation*}
\end{theorem}
\begin{proof}
The strategy passes the diagonal lines test with probability $1-\gamma$.
Therefore, it passes the $m$-restricted diagonal lines test with probability $1- \gamma\cdot m$.
This means that
\begin{equation*}
A^{u}_a \otimes I \simeq_{\gamma \cdot m} I \otimes L^{\ell}_{[f(u)=a]},
\end{equation*}
on average over a uniformly random $\bu \sim \F_q^{m}$ and a uniformly random line $\bell$ in $\F_q^m$
containing $\bu$.
By~\Cref{prop:simeq-to-approx}, this implies that
\begin{equation}\label{eq:point-diagonal-line-approx}
A^{u}_a \otimes I \approx_{2\cdot \gamma \cdot m} I \otimes L^{\ell}_{[f(u)=a]}. 
\end{equation}

Let $\bu$ and $\bv$ be independent and uniformly random points in $\F_q^m$.
Let $\bell$ be a uniformly random line in $\F_q^m$ containing both points.
(If $\bu$ and $\bv$ are distinct, then $\bell$ is just the unique line that passes through both of them.
Otherwise, $\bu = \bv$, and $\bell$ is a uniformly random line passing through $\bu$.)
Then the marginal distribution on $\bu$ and $\bell$ is as in \Cref{eq:point-diagonal-line-approx},
as is the marginal distribution on $\bv$ and $\bell$.
As a result,
\begin{align*}
(A^{u}_a \cdot A^{v}_b) \otimes I
&\approx_{2\cdot \gamma \cdot m} A^{u}_a \otimes L^{\ell}_{[f(v)=b]} \tag{by \Cref{eq:point-diagonal-line-approx}}\\
&\approx_{2\cdot \gamma \cdot m} I \otimes (L^{\ell}_{[f(v)=b]} \cdot L^{\ell}_{[f'(u)=a]}) \tag{by \Cref{eq:point-diagonal-line-approx}}\\
&= I \otimes (L^{\ell}_{[f'(u)=a]} \cdot L^{\ell}_{[f(v)=b]}) \tag{because~$L$ is projective}\\
&\approx_{2\cdot \gamma \cdot m} A^{v}_b \otimes L^{\ell}_{[f'(u)=a]} \tag{by \Cref{eq:point-diagonal-line-approx}}\\
&\approx_{2\cdot \gamma \cdot m} (A^{v}_b \cdot A^{u}_a) \otimes I. \tag{by \Cref{eq:point-diagonal-line-approx}}
\end{align*}
The theorem now follows from
\Cref{prop:triangle-inequality-for-approx_delta}.
\end{proof}


\section{Commutativity}
\label{sec:g-comm}

Let $(\bu, \bx)$ and $(\bv, \by)$ be sampled independently and uniformly at random from $\F_q^{m+1}$.
In \Cref{sec:G-commutes-after-evaluation},
we will show that the $G$ measurements approximately commute ``after evaluation";
namely, that $G^{\bx}_{[g(\bu)=a]}$ commutes with $G^{\by}_{[g(\bv)=b]}$.
Then, in \Cref{sec:G-commutes},
we will use this to show that $G^{\bx}_g$ approximately commutes with $G^{\by}_h$.
This is necessary if we wish to ``paste" together several~$G^x$ measurements at different points~$x \in \F_q$
to produce a single global measurement~$H \in \polysub{m+1}{q}{d}$, as we do in \Cref{sec:ld-pasting} below.

\subsection{Commutativity of~$G$ after evaluation}\label{sec:G-commutes-after-evaluation}

\begin{lemma}[Commutativity of~$G$ after evaluation]
  \label{lem:comm-data-processed-g}
  Let $(\psi, A, B, L)$ be an $(\eps, \delta, \gamma)$-good symmetric strategy for the $(m+1,q,d)$ low individual degree test.
  Let $\{G^x\} \in \polysub{m}{q}{d}$ be a collection of projective sub-measurements indexed by $x \in \F_q$ with the following properties:
  \begin{enumerate}
  \item (Consistency with~$A$): \label{item:data-processed-consistency} On average over $(\bu, \bx) \sim \F_q^{m+1}$,
	  \begin{equation*}
	  A^{u, x}_a \otimes I \simeq_{\zeta} I \otimes G^x_{[g(u)=a]}.
	  \end{equation*}
  \item (Strong self-consistency): \label{item:data-processed-self-consistency} On average over~$\bx \sim \F_q$,
  	\begin{equation*}
		G^x_g \otimes I \approx_{\zeta} I \otimes G^x_g.
	\end{equation*}
  \item (Boundedness): \label{item:data-processed-boundedness} There exists a positive-semidefinite matrix $Z^x$ for each $x \in \F_q$ such that
  	\begin{equation*}
		\E_{\bx} \bra{\psi} (I-G^{\bx})\otimes Z^{\bx} \ket{\psi} \leq \zeta
	\end{equation*}
	and for each $x \in \F_q$ and $g \in \polyfunc{m}{q}{d}$,
	\begin{equation*}
	 Z^x \geq \left(\E_{\bu} A^{\bu, x}_{g(\bu)}\right).
	\end{equation*}
  \end{enumerate}
  Let
  \begin{equation*}
  \nu = 48 m \cdot (\gamma^{1/2} + \zeta^{1/2}).
  \end{equation*}
  Then on average over independent and uniformly random $(\bu,\bx), (\bv,\by) \sim \F_q^{m+1}$,
  \begin{equation*}
    G^x_{[g(u) = a]} G^y_{[h(v) = b]} \ot I \approx_{\nu} G^y_{[h(v) = b]}
    G^x_{[g(u) = a]} \ot I.
  \end{equation*}
\end{lemma}

\begin{proof}
	For notational convenience we use the abbreviation $G^{u,x}_a = G^{x}_{[g(u) = a]}$ for all $(u,x) \in \F_q^{m+1}$. 
We expand the square:
\begin{align}
	& \E_{\bu,\bv,\bx,\by} \sum_{a,b} \bra{\psi} ( G^{\bu,\bx}_a G^{\bv,\by}_b - G^{\bv,\by}_b G^{\bu,\bx}_a )^\dagger \cdot( G^{\bu,\bx}_a G^{\bv,\by}_b - G^{\bv,\by}_b G^{\bu,\bx}_a ) \ot I \ket{\psi} \nonumber\\
	=~& \E_{\bu,\bv,\bx,\by} \sum_{a,b} \bra{\psi} ( G^{\bv,\by}_b G^{\bu,\bx}_a  - G^{\bu,\bx}_a G^{\bv,\by}_b  ) \cdot( G^{\bu,\bx}_a G^{\bv,\by}_b - G^{\bv,\by}_b G^{\bu,\bx}_a ) \ot I \ket{\psi} \nonumber\\
	=~& 2\cdot \E_{\bu,\bv,\bx,\by} \sum_{a,b} \Big ( \bra{\psi} G^{\bv,\by}_b G^{\bu,\bx}_a G^{\bv,\by}_b \ot I \ket{\psi} -  \bra{\psi} G^{\bu,\bx}_a  G^{\bv,\by}_b G^{\bu,\bx}_a G^{\bv,\by}_b  \ot I \ket{\psi} \Big ),  \label{eq:gcom8}
\end{align}
where the last step uses the projectivity of~$G$.

We will show that the second term of \Cref{eq:gcom8} is close to the first term.
To begin, we note that for each $(u, x) \in \F_q^{m+1}$,
\begin{equation}\label{eq:sum-of-gux}
G^{u, x} = \sum_a G^{u, x}_a =  \sum_a G^{x}_{[g(u) = a]} = G^x.
\end{equation}
As a result, 
 \Cref{item:data-processed-consistency} and \Cref{prop:cons-sub-meas} imply that
\begin{align}
G^{u,x}_a \ot I 
&\approx_{4\zeta} G^{u, x} \otimes A^{u, x}_a \nonumber\\
&= G^{x} \otimes A^{u, x}_a, \label{eq:add-an-a}
\end{align}
where the second step is by \Cref{eq:sum-of-gux}.
We can therefore approximate the second term of \Cref{eq:gcom8}
as
\begin{align}
&\E_{\bu,\bv,\bx,\by} \sum_{a,b} \bra{\psi} G^{\bu,\bx}_a  G^{\bv,\by}_b G^{\bu,\bx}_a G^{\bv,\by}_b  \ot I \ket{\psi} \nonumber\\
\approx_{2\sqrt{\zeta}}& \E_{\bu,\bv,\bx,\by} \sum_{a,b} \bra{\psi} G^{\bu,\bx}_a  G^{\bv,\by}_b G^{\bu,\bx}_a G^{\by}  \ot A^{\bv,\by}_b \ket{\psi}.\label{eq:apply-add-an-a-once}
\end{align}
using \Cref{prop:closeness-of-ip} and \Cref{eq:add-an-a}.
Next, we claim that
\begin{equation}
\eqref{eq:apply-add-an-a-once}
\approx_{\sqrt{\zeta}}\E_{\bu,\bv,\bx,\by} \sum_{a,b} \bra{\psi} G^{\bu,\bx}_a  G^{\bv,\by}_b G^{\bu,\bx}_a  \ot A^{\bv,\by}_b \ket{\psi}. \label{eq:gcom9} 
\end{equation}
This is proved in \Cref{clm:g-comm-stability} below. Continuing, we have
\begin{align}
\eqref{eq:gcom9} &\approx_{2\sqrt{\zeta}} \E_{\bu,\bv,\bx,\by} \sum_{a,b} \bra{\psi} G^{\bu,\bx}_a  G^{\bv,\by}_b G^{\bx}  \ot A^{\bv,\by}_b A^{\bu,\bx}_a \ket{\psi} \nonumber\\
&\approx_{6\sqrt{\gamma (m+1)}} \E_{\bu,\bv,\bx,\by} \sum_{a,b} \bra{\psi} G^{\bu,\bx}_a  G^{\bv,\by}_b  G^{\bx}  \ot A^{\bu,\bx}_a A^{\bv,\by}_b \ket{\psi}. \label{eq:dunno-what-i-should-call-this}
\end{align}
The first approximation again uses  \Cref{prop:closeness-of-ip} and \Cref{eq:add-an-a}. The second approximation follows from  \Cref{prop:closeness-of-ip} and \Cref{thm:commutativity-points}.
Next, we claim that
\begin{equation}
\eqref{eq:dunno-what-i-should-call-this}
\approx_{\sqrt{\zeta}+6\sqrt{\gamma(m+1)}} \E_{\bu,\bv,\bx,\by} \sum_{a,b} \bra{\psi} G^{\bu, \bx}_a  G^{\bv,\by}_b \ot A^{\bu,\bx}_a  A^{\bv,\by}_b \ket{\psi}.
\label{eq:gcom10}
\end{equation}
This is proved in \Cref{clm:g-comm-stability2} below.
We now apply \Cref{eq:add-an-a} twice with the help of \Cref{prop:closeness-of-ip}.
\begin{align}
\eqref{eq:gcom10}
&= \E_{\bu,\bv,\bx,\by} \sum_{a,b} \bra{\psi} G^{\bx}G^{\bu, \bx}_a  G^{\bv,\by}_b \ot A^{\bu,\bx}_a  A^{\bv,\by}_b \ket{\psi} \tag{because~$G$ is projective}\nonumber\\
&\approx_{2\sqrt{\zeta}} \E_{\bu,\bv,\bx,\by} \sum_{a,b} \bra{\psi}G^{\bu,\bx}_a G^{\bu,\bx}_a  G^{\bv,\by}_b \ot A^{\bv,\by}_b \ket{\psi} \nonumber\\
&= \E_{\bu,\bv,\bx,\by} \sum_{a,b} \bra{\psi} G^{\bu,\bx}_a  G^{\bv,\by}_b \ot A^{\bv,\by}_b \ket{\psi} \tag{because~$G$ is projective}\nonumber\\
&= \E_{\bu,\bv,\bx,\by} \sum_{a,b} \bra{\psi} G^{\bu,\bx}_a  G^{\bv,\by}_b G^{\by} \ot A^{\bv,\by}_b \ket{\psi} \tag{because~$G$ is projective}\nonumber\\
&\approx_{2\sqrt{\zeta}} \E_{\bu,\bv,\bx,\by} \sum_{a,b} \bra{\psi} G^{\bu,\bx}_a  G^{\bv,\by}_b G^{\bv, \by}_b \ot I \ket{\psi}. \label{eq:gonna-cite-this-in-just-a-bit}
\end{align} 
Now, \Cref{item:data-processed-self-consistency} and the fact that~$G$ is projective allows us to apply \Cref{prop:two-notions-of-self-consistency}, which states that~$G$ is $\zeta/2$-strongly self-consistent. Hence, \Cref{prop:two-notions-of-self-consistency-after-evaluation} says that we can ``post-process'' its measurement outcomes:
\begin{equation*}
G^x_{[g(u)=a]} \ot I \approx_{\zeta} I \ot G^x_{[g(u)=a]}.
\end{equation*}
In other words, using our abbreviation,
\begin{equation}\label{eq:new-fact-that-i-derived}
G^{u,x}_a \ot I \approx_{\zeta} I \otimes G^{u,x}_a.
\end{equation}
Applying \Cref{eq:new-fact-that-i-derived} twice with the help of \Cref{prop:closeness-of-ip},
we conclude that
\begin{align*}
\eqref{eq:gonna-cite-this-in-just-a-bit}
&\approx_{\sqrt{\zeta}} \E_{\bu,\bv,\bx,\by} \sum_{a,b} \bra{\psi} G^{\bu,\bx}_a  G^{\bv,\by}_b  \ot G^{\bv, \by}_b \ket{\psi}\\
&\approx_{\sqrt{\zeta}} \E_{\bu,\bv,\bx,\by} \sum_{a,b} \bra{\psi} G^{\bv, \by}_b G^{\bu,\bx}_a  G^{\bv,\by}_b  \ot I \ket{\psi}.
\end{align*}
Hence we have shown that the second term of \Cref{eq:gcom8} is close to the first term.
\ignore{
\begin{align*}
\labelcref{eq:gcom10} &\approx_{\sqrt{2\zeta}} \E_{\bu,\bv,\bx,\by} \sum_{a,b} \bra{\psi}G^{\bu,\bx}_a G^{\bu,\bx}_a  G^{\bv,\by}_b \ot A^{\bv,\by}_b \ket{\psi} \tag{by \Cref{item:data-processed-consistency} and \Cref{prop:simeq-to-approx}} \\
&= \E_{\bu,\bv,\bx,\by} \sum_{a,b} \bra{\psi} G^{\bu,\bx}_a  G^{\bv,\by}_b \ot A^{\bv,\by}_b \ket{\psi} \tag{because~$G$ is projective}\\
&\approx_{\sqrt{\zeta}} \E_{\bu,\bv,\bx,\by} \sum_{a,b} \bra{\psi} G^{\bv,\by}_b G^{\bu,\bx}_a  G^{\bv,\by}_b \ot I \ket{\psi}. \tag{by \Cref{item:data-processed-consistency} and \Cref{prop:simeq-to-approx}}
\end{align*} }

Putting everything together, this shows that \Cref{eq:gcom8} is bounded by
\begin{align*}
&2\cdot\Big(2\sqrt{\zeta} + \sqrt{\zeta} + 2\sqrt{\zeta} +6\sqrt{\gamma(m+1)} + \sqrt{\zeta} + 6\sqrt{\gamma(m+1)} + 2\sqrt{\zeta} + 2\sqrt{\zeta} + \sqrt{\zeta} + \sqrt{\zeta}\Big)\\
=~& 24 \cdot(\sqrt{\gamma(m+1)}+\sqrt{\zeta})\\
\leq~& 24\sqrt{m+1} \cdot(\sqrt{\gamma}+\sqrt{\zeta})\\
\leq~& 48 m \cdot (\sqrt{\gamma} + \sqrt{\zeta}),
\end{align*}
and this completes the proof of the lemma, modulo the proofs of  \Cref{clm:g-comm-stability} and \Cref{clm:g-comm-stability2}.
We now prove these claims.
\begin{claim}
\label{clm:g-comm-stability}
\begin{align*}
	&\E_{\bu,\bv,\bx,\by} \sum_{a,b} \bra{\psi} G^{\bu,\bx}_a  G^{\bv,\by}_b G^{\bu,\bx}_a G^{\by}  \ot A^{\bv,\by}_b \ket{\psi} \label{eq:g-comm-stab1} \\
	\approx_{\sqrt{\zeta}}& \E_{\bu,\bv,\bx,\by} \sum_{a,b} \bra{\psi} G^{\bu,\bx}_a  G^{\bv,\by}_b G^{\bu,\bx}_a  \ot A^{\bv,\by}_b \ket{\psi}.
\end{align*}
\end{claim}
\begin{proof}
	Recall that $G_b^{v,y} = G^y_{[g(v)=b]}= \sum_{g : g(v) = b} G_g^y$. For all $y \in \F_q$ and $g \in \polyfunc{m}{q}{d}$, define the matrix
	\[
		R^y_g = \E_{\bu,\bx} \sum_a G^{\bu,\bx}_a G^y_g G^{\bu,\bx}_a.
	\]
	Then because~$G$ is a sub-measurement,
	\begin{equation*}
	\sum_{g} R^y_g
		= \E_{\bu,\bx} \sum_a G^{\bu,\bx}_a\cdot \Big(\sum_g G^y_g\Big) \cdot G^{\bu,\bx}_a
		\leq\E_{\bu,\bx} \sum_a G^{\bu,\bx}_a
		\leq I.
	\end{equation*}
	As a result, $R^y$ is a sub-measurement in $\polysub{m}{q}{d}$.
	
	Our goal is to bound the magnitude of
	\begin{align}
	&\E_{\bu,\bv,\bx,\by} \sum_{a,b} \bra{\psi} G^{\bu,\bx}_a  G^{\bv,\by}_b G^{\bu,\bx}_a (I-G^{\by})  \ot A^{\bv,\by}_b \ket{\psi}\nonumber\\
	=&\E_{\bu,\bv,\bx,\by} \sum_{a,b} \sum_{g : g(\bv) = b} \bra{\psi} G^{\bu,\bx}_a  G^{\by}_g G^{\bu,\bx}_a (I-G^{\by})  \ot A^{\bv,\by}_b \ket{\psi}\nonumber\\
	=&\E_{\bu,\bv,\bx,\by} \sum_{a,g}  \bra{\psi} G^{\bu,\bx}_a  G^{\by}_g G^{\bu,\bx}_a (I-G^{\by})  \ot A^{\bv,\by}_{g(\bv)} \ket{\psi}\nonumber\\
	=& \E_{\bv,\by} \sum_{g} \bra{\psi}R^{\by}_g (I-G^{\by}) \otimes A^{\bv, \by}_{g(\bv)}  \ket{\psi}.\label{eq:bound-this-right-now!}
	\end{align}
	We can now bound the magnitude as follows.
	\begin{multline*}
	|\eqref{eq:bound-this-right-now!}|
	= \Big| \E_{\bv,\by} \sum_{g} \bra{\psi}\Big(\sqrt{R^{\by}_g} \otimes I\Big)
		\cdot  \Big(\sqrt{R^{\by}_g} (I-G^{\by}) \otimes A^{\bv, \by}_{g(\bv)}\Big)\ket{\psi}\Big|\\
	\leq \sqrt{\E_{\bv,\by} \sum_{g} \bra{\psi}R^{\by}_g \otimes I\ket{\psi}}
		\cdot \sqrt{\E_{\bv,\by} \sum_{g}\bra{\psi}(I-G^{\by}) R^{\by}_g (I-G^{\by}) \otimes A^{\bv, \by}_{g(\bv)}\ket{\psi}}.
	\end{multline*}
	The term inside the first square root is at most~$1$ because~$R^y$ is a sub-measurement.
	The term inside the second square root is
	\begin{align*}
	&\E_{\bv,\by} \sum_{g}\bra{\psi}(I-G^{\by}) R^{\by}_g (I-G^{\by}) \otimes A^{\bv, \by}_{g(\bv)}\ket{\psi}\\
	=~&\E_{\by} \sum_{g}\bra{\psi}(I-G^{\by}) R^{\by}_g (I-G^{\by}) \otimes \left(\E_{\bv} A^{\bv, \by}_{g(\bv)}\right)\ket{\psi}\\
	\leq~&\E_{\by} \sum_{g}\bra{\psi}(I-G^{\by}) R^{\by}_g (I-G^{\by}) \otimes Z^{\by}\ket{\psi}\tag{by~\Cref{item:data-processed-boundedness}}\\
	\leq~&\E_{\by} \bra{\psi}(I-G^{\by}) \otimes Z^{\by}\ket{\psi}\\
	\leq~&\zeta.\tag{by~\Cref{item:data-processed-boundedness}}
	\end{align*}
	This concludes the proof.
\end{proof}

\begin{claim}
\label{clm:g-comm-stability2}
\begin{align*}
	&\E_{\bu,\bv,\bx,\by} \sum_{a,b} \bra{\psi} G^{\bu,\bx}_a  G^{\bv,\by}_b G^{\bx}  \ot A^{\bu,\bx}_a A^{\bv,\by}_b  \ket{\psi}  \\
	\approx_{\sqrt{\zeta}+6\sqrt{\gamma(m+1)}}& \E_{\bu,\bv,\bx,\by} \sum_{a,b} \bra{\psi} G^{\bu,\bx}_a  G^{\bv,\by}_b \ot A^{\bu,\bx}_a A^{\bv,\by}_b  \ket{\psi}.
\end{align*}
\end{claim}
\begin{proof}
	Our goal is to bound the magnitude of
	\begin{align}
	&\E_{\bu,\bv,\bx,\by} \sum_{a,b} \bra{\psi} G^{\bu,\bx}_a  G^{\bv,\by}_b (I-G^{\bx})  \ot A^{\bu,\bx}_a A^{\bv,\by}_b  \ket{\psi}\nonumber\\
	\approx_{6\sqrt{\gamma(m+1)}} &\E_{\bu,\bv,\bx,\by} \sum_{a,b} \bra{\psi} G^{\bu,\bx}_a  G^{\bv,\by}_b (I-G^{\bx})  \ot  A^{\bv,\by}_b A^{\bu,\bx}_a \ket{\psi}.\label{eq:just-got-commuted}
	\end{align}
	where the second line follows from \Cref{prop:closeness-of-ip} and \Cref{thm:commutativity-points}.
	We now proceed nearly identically to \Cref{clm:g-comm-stability}.
	\begin{align}
	\eqref{eq:just-got-commuted}
	&=\E_{\bu,\bv,\bx,\by} \sum_{a,b}\sum_{g:g(\bu)=a} \bra{\psi} G^{\bx}_g  G^{\bv,\by}_b (I-G^{\bx})  \ot A^{\bv,\by}_b A^{\bu,\bx}_a \ket{\psi} \nonumber\\
	&=\E_{\bu,\bv,\bx,\by} \sum_{g,b} \bra{\psi} G^{\bx}_g  G^{\bv,\by}_b (I-G^{\bx})  \ot A^{\bv,\by}_b A^{\bu,\bx}_{g(\bu)}   \ket{\psi}.\label{eq:g-comm-stab7}
	\end{align}
	Then we can bound the magnitude as follows.
	\begin{multline*}
	|\eqref{eq:g-comm-stab7}|
	= \Big|\E_{\bu,\bv,\bx,\by} \sum_{g,b} \bra{\psi} \Big(\sqrt{G^{\bx}_g}\otimes A^{\bv,\by}_b\Big) \cdot \Big(\sqrt{G^{\bx}_g}  G^{\bv,\by}_b (I-G^{\bx})  \ot  A^{\bu,\bx}_{g(\bu)}\Big)\ket{\psi} \Big|\\
	\leq \sqrt{\E_{\bu,\bv,\bx,\by} \sum_{g,b} \bra{\psi} G^{\bx}_g\otimes A^{\bv,\by}_b\ket{\psi}}
		\cdot \sqrt{\E_{\bu,\bv,\bx,\by} \sum_{g,b} \bra{\psi}(I-G^{\bx}) G^{\bv,\by}_b G^{\bx}_g  G^{\bv,\by}_b (I-G^{\bx})  \ot  A^{\bu,\bx}_{g(\bu)}\ket{\psi}}.
	\end{multline*}
	The term inside the first square root is at most~$1$ because~$G$ and~$A$ are sub-measurements.
	The term inside the second square root is
	\begin{align*}
	&\E_{\bu,\bv,\bx,\by} \sum_{g,b} \bra{\psi}(I-G^{\bx}) G^{\bv,\by}_b G^{\bx}_g  G^{\bv,\by}_b (I-G^{\bx})  \ot  A^{\bu,\bx}_{g(\bu)}\ket{\psi}\\
	=~&\E_{\bv,\bx,\by} \sum_{g,b} \bra{\psi}(I-G^{\bx}) G^{\bv,\by}_b G^{\bx}_g  G^{\bv,\by}_b (I-G^{\bx})  \ot  \left(\E_{\bu} A^{\bu,\bx}_{g(\bu)}\right)\ket{\psi}\\
	\leq~&\E_{\bv,\bx,\by} \sum_{g,b} \bra{\psi}(I-G^{\bx}) G^{\bv,\by}_b G^{\bx}_g  G^{\bv,\by}_b (I-G^{\bx})  \ot  Z^{\bx}\ket{\psi}\tag{by~\Cref{item:data-processed-boundedness}}\\
	\leq~&\E_{\bv,\bx,\by} \sum_{b} \bra{\psi}(I-G^{\bx}) G^{\bv,\by}_b  (I-G^{\bx})  \ot  Z^{\bx}\ket{\psi}\\
	\leq~&\E_{\bv,\bx,\by} \bra{\psi}(I-G^{\bx})  \ot  Z^{\bx}\ket{\psi}\\
	\leq~&\zeta.\tag{by~\Cref{item:data-processed-boundedness}}
	\end{align*}
	This concludes the proof.
	\end{proof}
Having proved \Cref{clm:g-comm-stability} and \Cref{clm:g-comm-stability2}, we conclude the proof of \Cref{lem:comm-data-processed-g}.
\end{proof}

\subsection{Commutativity of~$G$}\label{sec:G-commutes}

\begin{theorem}[Commutativity of~$G$]\label{thm:com-main}
  Let $(\psi, A, B, L)$ be an $(\eps, \delta, \gamma)$-good symmetric strategy for the $(m,q,d)$ low individual degree test.
  Let $\{G^x\}_{x \in \F_q}$ denote a set of projective sub-measurements in $\polysub{m}{q}{d}$ with the following properties:
  \begin{enumerate}
  \item (Consistency with~$A$): \label{item:commuting-consistency} On average over $(\bu, \bx) \sim \F_q^{m+1}$,
	  \begin{equation*}
	  A^{u, x}_a \otimes I \simeq_{\zeta} I \otimes G^x_{[g(u)=a]}.
	  \end{equation*}
  \item (Strong self-consistency): \label{item:commuting-self-consistency} On average over~$\bx \sim \F_q$,
  	\begin{equation*}
		G^x_g \otimes I \approx_{\zeta} I \otimes G^x_g.
	\end{equation*}
  \item (Boundedness): \label{item:commuting-boundedness} There exists a positive-semidefinite matrix $Z^x$ for each $x \in \F_q$ such that
  	\begin{equation*}
		\E_{\bx} \bra{\psi} (I-G^{\bx})\otimes Z^{\bx} \ket{\psi} \leq \zeta
	\end{equation*}
	and for each $x \in \F_q$ and $g \in \polyfunc{m}{q}{d}$,
	\begin{equation*}
	Z^x \geq \left(\E_{\bu} A^{\bu, x}_{g(\bu)}\right).
	\end{equation*}
  \end{enumerate}
  Let
  \begin{equation*}
  \nu = 30m \cdot \left(\gamma^{1/4}+ \zeta^{1/4} + (d/q)^{1/4}  \right).
  \end{equation*}
  Then on average over independent and uniformly random $(\bu,\bx), (\bv,\by) \sim \F_q^{m+1}$,
  \begin{equation*}
    G^{x}_{g} G^{y}_h \ot I \approx_{\nu} G^y_h G^x_g \ot I
  \end{equation*}
\end{theorem}

\begin{proof}
We note that the bound we are proving is trivial when at least one of $\gamma$, $\zeta$, or $d/q$ is $\geq 1$.
In this case, $\nu \geq 30$.
On the other hand,
\begin{align*}
&\E_{\bx,\by} \sum_{g, h} \Vert(G^{\bx}_g G^{\by}_h \ot I - G^{\by}_h G^{\bx}_g \ot I) \ket{\psi}\Vert^2\\
=~&\E_{\bx,\by} \sum_{g, h} \Vert(G^{\bx}_g G^{\by}_h \ot I) \ket{\psi} + (- G^{\by}_h G^{\bx}_g \ot I) \ket{\psi}\Vert^2\\
\leq~&\E_{\bx,\by} \sum_{g, h}
	2\cdot \Big(\Vert(G^{\bx}_g G^{\by}_h \ot I) \ket{\psi} \Vert^2 + \Vert(G^{\by}_h G^{\bx}_g \ot I) \ket{\psi}\Vert^2\Big) \tag{by \Cref{prop:triangle-inequality-for-vectors-squared}}\\
=~&\E_{\bx,\by} \sum_{g, h}
	4\cdot \Vert(G^{\bx}_g G^{\by}_h \ot I) \ket{\psi} \Vert^2 \tag{by symmetry of the terms}\\
=~&4 \cdot \E_{\bx, \by}\sum_{g, y} \bra{\psi} (G^{\by}_h G^{\bx}_g G^{\by}_h \ot I) \ket{\psi}\\
\leq~&4, \tag{because~$G$ is a sub-measurement}
\end{align*}
which is therefore less than~$\nu$.
Hence, we may assume that $\gamma, \zeta, d/q \leq 1$.
This will aid us when carrying out the error calculations near the end of the proof,
as it allows us to bound terms like $\zeta^{1/2}$ by terms like $\zeta^{1/4}$.

Our goal is to bound
\begin{align}
&\E_{\bx, \by} \sum_{g, h} \bra{\psi} (G^{\bx}_{g} G^{\by}_h -
    G^{\by}_h G^{\bx}_g)^\dagger \cdot(G^{\bx}_{g} G^{\by}_h - G^{\by}_h
                   G^{\bx}_g) \ot I \ket{\psi}\nonumber\\
=~&\E_{\bx, \by} \sum_{g, h} \bra{\psi} ( G^{\by}_h G^{\bx}_{g}-
    G^{\bx}_g G^{\by}_h )\cdot (G^{\bx}_{g} G^{\by}_h - G^{\by}_h
                   G^{\bx}_g) \ot I \ket{\psi}\nonumber\\
=~&2\cdot\E_{\bx, \by} \sum_{g,h} \bra{\psi} G^\bx_g G^{\by}_h
      G^{\bx}_g \ot I \ket{\psi}
- 2\cdot \E_{\bx, \by} \sum_{g,h}
      \bra{\psi} G^{\bx}_g G^{\by}_h G^{\bx}_g G^{\by}_h \ot I
      \ket{\psi}.\label{eq:gcomterms}
\end{align}
  We will show that both terms in \Cref{eq:gcomterms} are close to $\bra{\psi} G \otimes G \ket{\psi}$, where we recall that $G = \E_\bx G^\bx$.
  With the errors we derive, this will imply \Cref{thm:com-main}.
  
  For the first term in \Cref{eq:gcomterms}, we have
  \begin{align*}
   \E_{\bx, \by} \sum_{g,h} \bra{\psi} G^\bx_g G^{\by}_h
      G^{\bx}_g \ot I \ket{\psi}&= \E_{\bx} \sum_{g}
                                \bra{\psi}( G^{\bx}_g\cdot G\cdot G^{\bx}_g) \ot I \ket{\psi}
    \\
    &\approx_{2\sqrt{\zeta}} \E_{\bx}\sum_g \bra{\psi}
      G \ot G^{\bx}_g \ket{\psi} \\
     &= \bra{\psi} G \ot G \ket{\psi}, 
  \end{align*}
  where the approximation is by \Cref{prop:switch-sandwich} and \Cref{item:commuting-self-consistency}.

  For the second term in \Cref{eq:gcomterms}, we begin with the following lemma.
  \begin{lemma}\label{lem:normalization-condition}
  Let $P = \{P_a\}$ be a sub-measurement and $Q = \{Q_b\}$ be a projective sub-measurement.
  Define $C_{a, b} = Q_b \cdot P_a \cdot Q_b$. 
  Then
  $
  \sum_a (\sum_b C_{a, b})^\dagger (\sum_b C_{a, b}) = \sum_a (\sum_b C_{a, b}) (\sum_b C_{a, b})^\dagger \leq I.
  $
  \end{lemma}
  \begin{proof}
  The first equality follows from the fact that $C_{a, b}$ is Hermitian.
  As for the second equality,
  \begin{align*}
  \sum_a \Big(\sum_b C_{a, b}\Big) \Big(\sum_b C_{a, b}\Big)^\dagger
  &= \sum_a \sum_{b, b'} C_{a, b} \cdot C_{a, b'}\tag{$C_{a, b'}$ is Hermitian}\\
  &= \sum_a \sum_{b, b'} (Q_b \cdot P_a \cdot Q_b) \cdot (Q_{b'} \cdot P_a \cdot Q_{b'})\\
  &= \sum_a \sum_{b} Q_b \cdot P_a \cdot Q_b \cdot P_a \cdot Q_{b}\tag{because $Q$ is projective}\\
  &\leq \sum_a \sum_{b} Q_b \cdot P_a  \cdot Q_{b}\tag{because $Q_b \leq I$}\\
  &\leq I,
  \end{align*}
  because~$P$ and~$Q$ are sub-measurements.
  \end{proof}

  First, we show that
  \begin{equation}
  	\E_{\bx,\by} \sum_{g,h} \bra{\psi} G^{\bx}_g G^{\by}_h G^{\bx}_g G^{\by}_h \ot I \ket{\psi} \approx_{\sqrt{\zeta}} \E_{\bx,\by} \sum_{g,h} \bra{\psi} G^{\by}_h G^{\bx}_g G^{\by}_h  \ot G^{\bx}_g \ket{\psi}
	\label{eq:gcom4}
\end{equation}
This follows from \Cref{prop:closeness-of-ip} where we let ``$A^x_a$'' and ``$B^x_a$'' in the Proposition be $G^x_g \ot I$ and $I \ot G^x_g$, respectively, and let ``$C^x_{a,b}$'' denote $G^y_h G^x_g G^y_h$. The closeness between ``$A^x_a$'' and ``$B^x_a$'' follows from \Cref{item:commuting-self-consistency}, and the normalization condition on ``$C^x_{a,b}$'' follows from the projectivity of the~$\{G^y_h\}$ measurements and \Cref{lem:normalization-condition}. 
Next, we show that
\begin{align}\label{eq:evaluate-gcom-at-points}
\eqref{eq:gcom4}
= \E_{\bx,\by} \sum_{g,h} \bra{\psi} G^{\by}_h G^{\bx}_g G^{\by}_h  \ot G^{\bx}_g \ket{\psi}
\approx_{\frac{dm}{q}} \E_{\bu,\bx,\by} \sum_{a,h} \bra{\psi} G^{\by}_h G^{\bx}_{[g(\bu) = a]} G^{\by}_h  \ot G^{\bx}_{[g(\bu) = a]} \ket{\psi}.
\end{align}
To do so, we first compute the difference as
\begin{equation}\label{eq:gcom4-diff}
\E_{\bu, \bx, \by} \sum_{g \neq g', h} \bone[g(\bu) = g'(\bu)] \cdot \bra{\psi} G^{\by}_h G^{\bx}_g G^{\by}_h  \ot G^{\bx}_{g'} \ket{\psi}.
\end{equation}
This is nonnegative and real, so it suffices to upper bound it, which we do as follows.
\begin{align*}
\eqref{eq:gcom4-diff}
& = \E_{\bx, \by} \sum_{ g \neq g',h}   \bra{\psi} G^{\by}_h G^{\bx}_g G^{\by}_h  \ot G^{\bx}_{g'} \ket{\psi} \cdot \E_{\bu} \bone[g(\bu) = g'(\bu)]\\
& \leq \E_{\bx, \by} \sum_{g \neq g', h}   \bra{\psi} G^{\by}_h G^{\bx}_g G^{\by}_h  \ot G^{\bx}_{g'} \ket{\psi} \cdot \frac{dm}{q} \tag{by Schwartz-Zippel}\\
& \leq \frac{dm}{q}. \tag{because~$G$ is a sub-measurement}
\end{align*}
Next, we show that
\begin{align}
\eqref{eq:evaluate-gcom-at-points}
&= \E_{\bu,\bx,\by} \sum_{a,h} \bra{\psi} G^{\by}_h G^{\bx}_{[g(\bu) = a]} G^{\by}_h  \ot G^{\bx}_{[g(\bu) = a]} \ket{\psi}\tag{\Cref{eq:evaluate-gcom-at-points} rewriten}\\
&\approx_{\sqrt{\zeta}}\E_{\bu,\bx,\by} \sum_{a,h} \bra{\psi} G^{\bx}_{[g(\bu) = a]} G^{\by}_h G^{\bx}_{[g(\bu) = a]} G^{\by}_h  \ot I \ket{\psi}\nonumber\\
&\approx_{\sqrt{\zeta}} \E_{\bu,\bx,\by} \sum_{a,h} \bra{\psi} G^{\bx}_{[g(\bu) = a]} G^{\by}_h G^{\bx}_{[g(\bu) = a]} \ot G^{\by}_h \ket{\psi}.\label{eq:don't-understand-the-numbering-system}
\end{align}
These two approximations are derived as follows.
\begin{enumerate}
\item
In the first approximation we used \Cref{prop:closeness-of-ip} where we let ``$A^x_a$'' and ``$B^x_a$'' in the Proposition be $G^x_{[g(u)=a]} \ot I$ and $I \ot G^x_{[g(u) = a]}$, respectively, and let ``$C^x_{a,b}$'' denote $G^y_h G^x_{[g(u)=a]} G^y_h$. The closeness between ``$A^x_a$'' and ``$B^x_a$'' follows from \Cref{item:commuting-self-consistency}, and the normalization condition on ``$C^x_{a,b}$'' follows from the projectivity of the~$\{G^y_h\}$ measurements and \Cref{lem:normalization-condition}.
\item
In the second approximation we used \Cref{prop:closeness-of-ip} where we let ``$A^x_a$'' and ``$B^x_a$'' in the Proposition be $G^y_h \ot I$ and $I \ot G^y_h$, respectively, and let ``$C^x_{a,b}$'' denote $G^x_{[g(u)=a]} G^y_h G^x_{[g(u)=a]}$. The closeness between ``$A^x_a$'' and ``$B^x_a$'' follows from \Cref{item:commuting-self-consistency}, and the normalization condition on ``$C^x_{a,b}$'' follows from the projectivity of the~$\{G^x_{[g(u)=a]}\}$ measurements and \Cref{lem:normalization-condition}.
\end{enumerate}
Analogously to the derivation of \Cref{eq:evaluate-gcom-at-points}, we now show that
\begin{align}
\eqref{eq:don't-understand-the-numbering-system}
& = \E_{\bu,\bx,\by} \sum_{a,h} \bra{\psi} G^{\bx}_{[g(\bu) = a]} G^{\by}_h G^{\bx}_{[g(\bu) = a]} \ot G^{\by}_h \ket{\psi} \tag{\Cref{eq:don't-understand-the-numbering-system} restated}\\
& \approx_{\frac{dm}{q}} \E_{\bu,\bv,\bx,\by} \sum_{a,b} \bra{\psi} G^{\bx}_{[g(\bu) = a]} G^{\by}_{[h(\bv)=b]} G^{\bx}_{[g(\bu) = a]} \ot G^{\by}_{[h(\bv)=b]} \ket{\psi}.\label{eq:evaluate-gcom-at-points-part-dos}
\end{align}
To do so, we first compute the difference as
\begin{equation}\label{eq:eq:don't-understand-the-numbering-system-diff}
\E_{\bu,\bv, \bx,\by} \sum_{a,h\neq h'} \bone[h(\bv) = h'(\bv)] \cdot \bra{\psi} G^{\bx}_{[g(\bu) = a]} G^{\by}_h G^{\bx}_{[g(\bu) = a]} \ot G^{\by}_{h'} \ket{\psi}.
\end{equation}
This is nonnegative and real, so it suffices to upper bound it, which we do as follows.
\begin{align*}
\eqref{eq:eq:don't-understand-the-numbering-system-diff}
& = \E_{\bu, \bx,\by} \sum_{a,h\neq h'} \bra{\psi} G^{\bx}_{[g(\bu) = a]} G^{\by}_h G^{\bx}_{[g(\bu) = a]} \ot G^{\by}_{h'} \ket{\psi} \cdot \E_{\bv} \bone[h(\bv) = h'(\bv)]\\
& \leq \E_{\bu, \bx,\by} \sum_{a,h\neq h'} \bra{\psi} G^{\bx}_{[g(\bu) = a]} G^{\by}_h G^{\bx}_{[g(\bu) = a]} \ot G^{\by}_{h'} \ket{\psi} \cdot \frac{dm}{q} \tag{by Schwartz-Zippel}\\
& \leq \frac{dm}{q}. \tag{because~$G$ is a sub-measurement}
\end{align*}
Now, we set
\begin{equation*}
\nu_{\mathrm{evaluation}} = 48 m \cdot (\gamma^{1/2} + \zeta^{1/2})
\end{equation*}
to be the approximation error from \Cref{lem:comm-data-processed-g}.
We conclude by showing that
\begin{align*}
	\eqref{eq:evaluate-gcom-at-points-part-dos} &= \E_{\bu,\bv,\bx,\by} \sum_{a,b} \bra{\psi} G^{\bx}_{[g(\bu) = a]} G^{\by}_{[h(\bv) = b]} G^{\bx}_{[g(\bu) = a]} \ot G^{\by}_{[h(\bv) = b]} \ket{\psi} \tag{\Cref{eq:evaluate-gcom-at-points-part-dos} restated}\\
	&\approx_{\sqrt{\nu_{\mathrm{evaluation}}}} \E_{\bu,\bv,\bx,\by} \sum_{a,b} \bra{\psi} G^{\bx}_{[g(\bu) = a]} G^{\bx}_{[g(\bu) = a]} G^{\by}_{[h(\bv) = b]}  \ot G^{\by}_{[h(\bv) = b]} \ket{\psi} \\
	&= \E_{\bu,\bv,\bx,\by} \sum_{a,b} \bra{\psi} G^{\bx}_{[g(\bu) = a]} G^{\by}_{[h(\bv) = b]}  \ot G^{\by}_{[h(\bv) = b]} \ket{\psi} \\
	&\approx_{\sqrt{\zeta}} \E_{\bu,\bv,\bx,\by} \sum_{a,b} \bra{\psi} G^{\bx}_{[g(\bu) = a]} \ot G^{\by}_{[h(\bv) = b]} G^{\by}_{[h(\bv) = b]} \ket{\psi} \\
	&= \E_{\bu,\bv,\bx,\by} \sum_{a,b} \bra{\psi} G^{\bx}_{[g(\bu) = a]} \ot G^{\by}_{[h(\bv) = b]} \ket{\psi} \\
	&= \bra{\psi} G \ot G \ket{\psi}.
\end{align*}
The third and fifth lines follow from the projectivity of the $G$ measurements.
The two approximations are derived as follows.
\begin{enumerate}
\item
In the first approximation we used \Cref{prop:closeness-of-ip} where we let ``$A^x_a$'' and ``$B^x_a$'' in the Proposition be $G^x_{[g(u)=a]} G^y_{[h(v)=b]} \ot I$ and $G^y_{[h(v)=b]}  G^x_{[g(u) = a]} \ot I$, respectively, and let ``$C^x_{a,b}$'' denote $G^x_{[g(u)=a]} \ot G^y_{[h(v)=b]}$. The closeness between ``$A^x_a$'' and ``$B^x_a$'' follows from \Cref{lem:comm-data-processed-g}, and the normalization condition on ``$C^x_{a,b}$'' follows from it being a sub-measurement.
\item
In the second approximation we used \Cref{prop:closeness-of-ip} where we let ``$A^x_a$'' and ``$B^x_a$'' in the Proposition be $G^y_{[h(v)=b]} \ot I$ and $I \ot G^y_{[h(v)=b]}$, respectively, and let ``$C^x_{a,b}$'' denote $G^x_{[g(u)=a]} \ot G^y_{[h(v)=b]}$. The closeness between ``$A^x_a$'' and ``$B^x_a$'' follows from \Cref{item:commuting-self-consistency}, and the normalization condition on ``$C^x_{a,b}$'' follows from it being a sub-measurement.
\end{enumerate}
This shows that the second term in \Cref{eq:gcomterms} is approximately $\bra{\psi} G \ot G \ket{\psi}$, as desired.
In total, we have incurred an error of
\begin{align*}
&2\cdot\left(2\sqrt{\zeta} +\sqrt{\zeta} + \frac{dm}{q} + \sqrt{\zeta} + \sqrt{\zeta} + \frac{dm}{q} + \sqrt{\nu_{\mathrm{evaluation}}} + \sqrt{\zeta} \right)\\
=~&12 \zeta^{1/2} + 2 m \cdot (d/q) + 2 \cdot \sqrt{48 m \cdot (\gamma^{1/2} + \zeta^{1/2})}\\
\leq~&12 \zeta^{1/2} + 2 m \cdot (d/q) + 14m \cdot \left(\gamma^{1/4} + \zeta^{1/4}\right)\\
\leq~& 12 \zeta^{1/4} + 2m \cdot (d/q)^{1/4} + 14m \cdot \left(\gamma^{1/4} + \zeta^{1/4}\right) \\
\leq~& 30m \cdot \left(\gamma^{1/4}+ \zeta^{1/4} + (d/q)^{1/4}  \right).
\end{align*}
  This concludes the proof of \Cref{thm:com-main}.
\end{proof}


\newcommand{\rmcom}{\mathrm{commute}}
\newcommand{\rmcons}{\mathrm{consistency}}
\newcommand{\rmcomp}{\mathrm{completeness}}
\newcommand{\distinct}[1]{\mathsf{Distinct}_{#1}}
\newcommand{\outc}{\mathsf{Outcomes}}
\newcommand{\sfO}{\mathsf{O}}
\newcommand{\wH}{\widehat{H}}
\newcommand{\wG}{\widehat{G}}

\section{Pasting}\label{sec:ld-pasting}

\begin{theorem}[Pasting]\label{thm:ld-pasting}
  Let $(\psi, A, B, L)$ be an $(\eps, \delta, \gamma)$-good symmetric strategy for the $(m+1,q,d)$ low individual degree test.
  Let $\{G^x\}_{x \in \F_q}$ denote a set of projective sub-measurements in $\polysub{m}{q}{d}$ with the following properties:
  \begin{enumerate}
  \item (Completeness): \label{item:ld-pasting-completeness}  If $G = \E_{\bx}\sum_g G_g^{\bx}$, then
  	\begin{equation*}
	\bra{\psi} G \otimes I \ket{\psi} \geq 1 - \kappa.
	\end{equation*}
  \item (Consistency with~$A$): \label{item:ld-pasting-consistency} On average over $(\bu, \bx) \sim \F_q^{m+1}$,
	  \begin{equation*}
	  A^{u, x}_a \otimes I \simeq_{\zeta} I \otimes G^x_{[g(u)=a]}.
	  \end{equation*}
  \item (Strong self-consistency): \label{item:ld-pasting-self-consistency} On average over~$\bx \sim \F_q$,
  	\begin{equation*}
		G^x_g \otimes I \approx_{\zeta} I \otimes G^x_g.
	\end{equation*}
  \item (Boundedness): \label{item:ld-pasting-boundedness} There exists a positive-semidefinite matrix $Z^x$ for each $x \in \F_q$ such that
  	\begin{equation*}
		\E_{\bx} \bra{\psi} (I-G^{\bx})\otimes Z^{\bx} \ket{\psi} \leq \zeta
	\end{equation*}
	and for each $x \in \F_q$ and $g \in \polyfunc{m}{q}{d}$,
	\begin{equation*}
	Z^x \geq \left(\E_{\bu} A^{\bu, x}_{g(\bu)}\right).
	\end{equation*}
  \end{enumerate}
  Let $k \geq 400md$ be an integer. Let
 \begin{align*}
 \nu &= 100 k^2m \cdot \left(\eps^{1/32} + \delta^{1/32} + \gamma^{1/32} + \zeta^{1/32} + (d/q)^{1/32}\right),\\
 \sigma & = \kappa \cdot \left(1 + \frac{1}{100m}\right)
    + 2\nu + e^{- k/(80000m^2)}.
 \end{align*}
  Then there exists a ``pasted" measurement $H \in \polymeas{m+1}{q}{d}$ which satisfies the following property.
  \begin{enumerate}
  \item (Consistency with~$A$): \label{item:ld-pasting-N-consistency} On average over $\bu \sim \F_q^{m+1}$,
	  \begin{equation*}
	  A^{u}_a \otimes I \simeq_{\sigma} I \otimes H_{[h(u)=a]}.
	  \end{equation*}
  \end{enumerate}
\end{theorem}

We note that the bound we are proving is trivial when at least one of $\eps$, $\delta$, $\gamma$, $\zeta$, or $d/q$ is $\geq 1$,
as $\nu$ is at least~$1$ in that case.
Hence, we may assume that $\eps,\delta,\gamma, \zeta, d/q \leq 1$.
This will aid us when carrying out the error calculations,
as it allows us to bound terms like $\zeta^{1/2}$ by terms like $\zeta^{1/4}$.

It will be convenient to state
certain consistency relations that follow easily from the hypotheses
of the theorem.
In this section, we will  only use the axis-parallel lines test in the $(m+1)$-st direction,
		i.e.\ in the case when the line is of the form $\ell = \{(u, x) \mid x \in \F_q\}$.
		Such a line is specified by a point $u \in \F_q^m$.
		As a result, we can use the shorthand introduced in \Cref{not:conditioned-on-last-direction},
		where instead of writing $B^\ell_f$ for such a line, we write $B^u_f$.
		We will also make use of the other shorthand from \Cref{not:conditioned-on-last-direction}
		in which for a function $f:\ell \rightarrow \F_q$,
		we will sometimes write $f(x)$ instead of $f(u,x)$.
		
		We know that $(\psi, A, B, L)$ passes the axis-parallel lines test with probability $1-\eps$,
		and so it passes this test with probability $1-(m+1) \cdot \eps$ conditioned on the random 		direction $\bi \in \{1, \ldots, m+1\}$ from the test
		being equal to $m+1$.
		This means the following: if $(\bu, \bx) \sim \F_q^{m+1}$ is drawn uniformly at random and $\bell = \{(\bu, y) \mid y \in \F_q\}$
		is the line through it in the $(m+1)$-st direction, then
		\begin{equation*}
		A^{u, x}_a \ot I \simeq_{(m+1) \eps} I \ot B^{\ell}_{[f(x)=a]} = I \ot B^u_{[f(x)=a]}.
		\end{equation*}
		Hence, \Cref{prop:simeq-to-approx} and the fact that $(m+1) \leq 2m$ imply that
		\begin{equation*}
		A^{u, x}_a \ot I \approx_{4m \eps} I \ot B^u_{[f(x)=a]}.
		\end{equation*}

Next, from the assumption that the given strategy is $(\eps,
\delta, \gamma)$-good, \Cref{prop:simeq-to-approx} implies that
\begin{equation*}
I \otimes A^{u, x}_a \approx_{2\delta} A^{u, x}_a \otimes I \approx_{4m\eps} I \otimes B^{u}_{[f(x)=a]}.
\end{equation*}
As a result, \Cref{prop:triangle-inequality-for-approx_delta} implies that
\begin{equation}\label{eq:ld-abcon}
I \otimes A^{u, x}_a \approx_{8m\eps + 4\delta} I\otimes B^{u}_{[f(x)=a]}.
\end{equation}
\Cref{prop:triangle-sub} applied to \Cref{item:ld-pasting-consistency} and \Cref{eq:ld-abcon} implies
\begin{equation}\label{eq:ld-gbcon}
G^x_{[g(u) = a]} \ot I \simeq_{\nu_1} I \ot B^{u}_{[f(x)=a]},
\end{equation}
where
\begin{equation}\label{eq:ld-nu1-def}
\nu_1 = \zeta + \sqrt{8 m\eps + 4\delta}.
\end{equation}
Finally, \Cref{thm:com-main} implies that
\begin{equation}\label{eq:quote-com-main}
G^x_g G^y_h \otimes I \approx_{\nu_{\rmcom}} G^y_h G^x_g \otimes I,
\end{equation}
where
\begin{equation*}
\nu_{\rmcom} = 30 m \cdot \left( \gamma^{1/4} +\zeta^{1/4} + (d/q)^{1/4}\right).
\end{equation*}

\subsection{From measurements to sub-measurements}

Rather than designing the pasted measurement guaranteed by \Cref{thm:ld-pasting},
it is convenient to first design a pasted \emph{sub}-measurement,
and then convert it to a measurement.
The next lemma shows the bounds we achieve for this sub-measurement.

\begin{lemma}\label{lem:ld-pasting-sub-measurement}
There exists a ``pasted" sub-measurement $H \in \polysub{m+1}{q}{d}$ which satisfies the following properties.
  \begin{enumerate}
  \item (Consistency with~$A$): \label{item:ld-pasting-N-consistency-sub-measurement} On average over $\bu \sim \F_q^{m+1}$,
	  \begin{equation*}
	  A^{u}_a \otimes I \simeq_{\nu} I \otimes H_{[h(u)=a]}.
	  \end{equation*}
  \item (Completeness): \label{item:ld-pasting-N-completeness-sub-measurement} If $H =  \sum_h H_h$, then
  	\begin{equation*}
	\bra{\psi} H \otimes I \ket{\psi} \geq 1 - \kappa \cdot \left(1 + \frac{1}{100m}\right)
    - \nu - e^{- k/(80000m^2)}.
	\end{equation*}
  \end{enumerate}
\end{lemma}
\begin{proof}[Proof of \Cref{thm:ld-pasting} assuming \Cref{lem:ld-pasting-sub-measurement}]
Let~$H$ be the sub-measurement in $\polysub{m+1}{q}{d}$ guaranteed by \Cref{lem:ld-pasting-sub-measurement}.
Let $h^*$ be an arbitrary polynomial in $\polyfunc{m+1}{q}{d}$.
We define the measurement~$H_{\mathrm{meas}} \in \polymeas{m+1}{q}{d}$ as follows.
\begin{equation*}
(H_{\mathrm{meas}})_h
= \left\{\begin{array}{cl}
		H_{h^*} + (I- H) & \text{if $h = h^*$,}\\
		H_h & \text{otherwise.}
		\end{array}
		\right.
\end{equation*}
This is clearly a measurement, as the sum of its POVM elements is $H + (I -H) = H$.
In addition,
\begin{align*}
&\E_{\bu} \sum_{a \neq b} \bra{\psi} A^{\bu}_a \ot (H_{\mathrm{meas}})_{[h(\bu)=b]} \ket{\psi}\\
 ={}& \E_{\bu} \sum_{a} \sum_{h:h(\bu) \neq a} \bra{\psi} A^{\bu}_a \ot (H_{\mathrm{meas}})_{h} \ket{\psi}\\
 ={}& \E_{\bu} \sum_{a} \sum_{h:h(\bu) \neq a} \bra{\psi} A^{\bu}_a \ot H_{h} \ket{\psi}
 	+\E_{\bu} \sum_{a:h^*(\bu) \neq a}\bra{\psi} A^{\bu}_a \ot (I-H) \ket{\psi}\\
 \leq{}& \E_{\bu} \sum_{a} \sum_{h:h(\bu) \neq a} \bra{\psi} A^{\bu}_a \ot H_{h} \ket{\psi}
 	+\E_{\bu} \bra{\psi} I \ot (I-H) \ket{\psi}\\
\leq{}& (\nu) +  \left(\kappa \cdot \left(1 + \frac{1}{100m}\right)
    				+ \nu + e^{- k/(80000m^2)}\right)
					\tag{by \Cref{item:ld-pasting-N-consistency-sub-measurement,item:ld-pasting-N-completeness-sub-measurement}}\\
={}& \sigma.
\end{align*}
Thus, $A^u_a \ot I \simeq_{\sigma} I \ot (H_{\mathrm{meas}})_{[h(u)=a]}$. This completes the proof.
\end{proof}

We will now spend the rest of this section proving \Cref{lem:ld-pasting-sub-measurement},
i.e.\ designing the pasted sub-measurement~$H$.

\subsection{The pasted sub-measurement}

\begin{definition}
For $k \geq 1$, we let $\distinct{k}$ be the set of tuples $(x_1, \ldots, x_k) \in \F_q^k$ such that $x_i \neq x_j$ for all $i \neq j$.
We write $(\bx_1, \ldots, \bx_k) \sim \distinct{k}$ for a uniformly random element of this set.
\end{definition}

In manipulations involving our pasted measurement, it will often be useful to switch
the expectation over $ \distinct{k}$ with an expectation over uniformly
random $k$-tuples of elements of $\F_q$. The following proposition
bounds the distance between these two distributions.

\begin{proposition}\label{prop:ld-dnoteq}
Let $\bx = (\bx_1, \ldots, \bx_k) \sim \F_q^k$ be sampled uniformly at random
and let $\by = (\by_1, \ldots, \by_k) \sim \distinct{k}$.
Then
\begin{equation*}
d_{\mathrm{TV}}(\bx, \by) \leq \frac{k^2}{q}.
\end{equation*}
\end{proposition}
\begin{proof}
For any $z = (z_1, \ldots, z_k) \in \F_q^k$, 
\begin{align*}
\Pr[\bx = z] 
& = \frac{1}{q^k},\\
\Pr[\by = z] 
& = \left\{\begin{array}{cl}
		\frac{1}{\binom{q}{k} k!} & \text{if $z \in \distinct{k}$,}\\
		0 & \text{otherwise.}
		\end{array}\right.
\end{align*}
Because $\frac{1}{\binom{q}{k} k!} \geq \frac{1}{q^k}$, $\Pr[\bx = z] \geq \Pr[\by = z]$ if and only if $z \notin \distinct{k}$.
Hence,
\begin{align*}
d_{\mathrm{TV}}(\bx, \by)
&= \max_{S \subseteq \F_q^k}\{\Pr[\bx \in S] - \Pr[\by \in S]\}\\
&= \Pr[\bx \in \overline{\distinct{k}}] - \Pr[\by \in \overline{\distinct{k}}]
= \Pr[\bx \in \overline{\distinct{k}}].
\end{align*}
We can upper-bound this probability as follows.
\begin{align*}
\Pr[\bx \in \overline{\distinct{k}}]
&= \Pr[\exists i : \bx_i \in \{\bx_1, \ldots, \bx_{i-1}\}]\\
&\leq \sum_{i=2}^k \Pr[\bx_i \in \{\bx_1, \ldots, \bx_{i-1}\}]
\leq \sum_{i=2}^k \left(\frac{i-1}{q}\right)
= \frac{k(k-1)}{2q}.
\end{align*}
This concludes the proof.
\end{proof}

To begin, we will describe our construction of the global pasted measurement.
In fact, we will consider \emph{two} separate constructions of the global pasted measurement.
The first is given in \Cref{sec:construction-numba-one} below.
It is the more natural of the two,
but unfortunately we do not know how to prove that it works correctly.
This motivates our second construction,
given in \Cref{sec:construction-numba-two} below,
which is designed to circumvent the problems in the first construction.

\subsubsection{The first construction}\label{sec:construction-numba-one}

The first construction of $H = \{H_h\}$ is conceptually simple:
we perform the $G$ sub-measurement $d+1$ times to produce $d+1$ polynomials $g_1, \ldots, g_{d+1} \in \polyfunc{m}{q}{d}$.
We then perform polynomial interpolation to produce a single global polynomial $h \in \polyfunc{m+1}{q}{d}$.
In more detail, this construction involves three steps.

\begin{enumerate}
\item (Pasting):
Let $x_1, \ldots, x_{d+1} \in \F_q$.
We will define an initial ``sandwiched" measurement as follows:
\begin{equation*}
\widehat{H}^{x_1, \ldots, x_{d+1}}_{g_1, \ldots, g_{d+1}} = G^{x_1}_{g_1} \cdot G^{x_2}_{g_2} \cdots G^{x_{d+1}}_{g_{d+1}} \cdots G^{x_2}_{g_2} \cdot G^{x_1}_{g_1}.
\end{equation*}
$\widehat{H}^{x_1, \ldots, x_{d+1}}$ has a natural interpretation as the sub-measurement in which one performs the sub-measurements $G^{x_1}, G^{x_2}, \ldots, G^{x_{d+1}}$ one after another and outputs their results.
\item (Interpolation)
Next, let $(x_1, \ldots, x_{d+1}) \in \distinct{d+1}$.
We define the interpolated measurement
\begin{equation*}
H^{x_1, \ldots, x_{d+1}}_{h}
= \widehat{H}^{x_1, \ldots, x_{d+1}}_{h|_{x_1}, \ldots, h|_{x_{d+1}}}.
\end{equation*}
This performs the $\widehat{H}^{x_1, \ldots, x_{d+1}}$ measurement and outputs the~$h$
which is consistent with the outcomes~$g_1, \ldots, g_{d+1}$ if one exists.
(We note that this is not necessarily a sub-measurement if $(x_1, \ldots, x_{d+1}) \notin \distinct{d+1}$. This is because there may exist $h \neq h'$ for which $h|_{x_i} = (h')|_{x_i}$ for all $1 \leq i \leq d+1$.)
\item (Averaging):
Finally, we randomize over the choice of $(x_1, \ldots, x_{d+1})$. In other words, we define
\begin{equation*}
H_{h} = \E_{(\bx_1, \ldots, \bx_{d+1}) \sim \distinct{d+1}} H^{\bx_1, \ldots, \bx_{d+1}}_h.
\end{equation*}
\end{enumerate}

To show that this construction works,
we need to show two things: (i) that $H_h$ has good agreement with $A$,
and (ii) that $H$'s completeness is close to~$G$'s.
Showing (i) is simple and follows from
the approximate commutativity of $G^{x}$ and $G^{y}$ established in \Cref{thm:com-main}.
What we do not know how to do is to show (ii),
at least for general~$d$.
In fact, at first glance it even looks like it should be false!
To see why,
we note that it is possible to show that the completeness of~$H$ can be approximated as
\begin{equation*}
\bra{\psi} H \otimes I \ket{\psi} \approx \bra{\psi} G^{d+1} \otimes I \ket{\psi},
\end{equation*}
because~$H$ is performing the~$G$ measurement $d+1$ times.
We would therefore like to show that
\begin{equation}\label{eq:subtract-a-G}
\bra{\psi} G^{d+1}\otimes I \ket{\psi}
\approx \bra{\psi} G\otimes I \ket{\psi}.
\end{equation}
But if $\bra{\psi} G \otimes I \ket{\psi} = 1-\kappa$,
a ``naive analysis" would lead to the conclusion that
\begin{equation}\label{eq:dumbo-bound-for-idiots}
\bra{\psi} G^{d+1}\otimes I \ket{\psi} \approx 1 - (d+1) \cdot \kappa,
\end{equation}
which would be too great of a loss in completeness for our proof strategy to work.

However, we can show \Cref{eq:subtract-a-G} is actually correct
and therefore this ``naive analysis" is incorrect,
at least in the case of constant~$d$.
To see how this is possible,
note that \Cref{eq:dumbo-bound-for-idiots} is in fact correct when all of~$G$'s eigenvalues are equal to $1-\kappa$,
in which case $G^{d+1} = (1-\kappa)^d \cdot G$.
So we would like to show that, on the contrary,
the fact that $G$'s incompleteness is equal to $1-\kappa$
is because a ``$(1-\kappa)$ fraction" of its eigenvalues are equal to~$1$,
and the remaining ``$\kappa$ fraction" of its eigenvalues are equal to~$0$.
In this case, $G^{d+1} = G$, and so \Cref{eq:subtract-a-G} holds.

We now sketch the argument that shows this holds, at least for constant~$d$.
To do this, it is first possible to show that
\begin{equation}\label{eq:step-one-of-grand-plan}
\bra{\psi} G^{d+2} \otimes I \ket{\psi}
\approx_{\Delta} \bra{\psi} G^{d+1} \otimes I \ket{\psi},
\end{equation}
where $\Delta = \poly(m) \cdot \poly(\eps, \delta, \gamma, \zeta, d/q)$ is small.
Intuitively, this is because after performing $(d+1)$ $G$ measurements,
the outcome of another $G$ measurement is essentially determined due to interpolation.
(The proof is of this fact is slightly subtle and involves the $Z$-boundedness condition.)

From here, it is possible to derive \Cref{eq:subtract-a-G} as follows.
Write the eigendecomposition $G = \sum_i \lambda_i \ket{v_i}\bra{v_i}$ of~$G$.
This defines a probability distribution $\mu$ over eigenvectors, where eigenvector~$i$ occurs with probability $\mu(i) = \bra{\psi} (\ket{v_i}\bra{v_i} \ot I) \ket{\psi}$.
In this language, for any power~$k$ we can write
\begin{equation*}
\bra{\psi} G^k \ot I \ket{\psi}
= \bra{\psi} \Big(\sum_i \lambda_i^k \ket{v_i}\bra{v_i}\Big) \ot I \ket{\psi}
= \sum_i \lambda_i^k \cdot \mu(i)
= \E_{\bi \sim \mu}[\lambda_{\bi}^k].
\end{equation*}
Thus, \Cref{eq:step-one-of-grand-plan} is equivalent to the statement that
\begin{equation}\label{eq:equivalent-way-of-writing-grand-plan}
\Delta
\geq \bra{\psi} G^{d+1} \otimes I \ket{\psi} - \bra{\psi} G^{d+2} \otimes I \ket{\psi}
= \E_{\bi \sim \mu}[\lambda_{\bi}^{d+1}] - \E_{\bi \sim \mu}[\lambda_{\bi}^{d+2}]
= \E_{\bi \sim\mu}[\lambda_{\bi}^{d+1}(1-\lambda_{\bi})].
\end{equation}
By \Cref{eq:subtract-a-G}, our goal is to bound
\begin{equation*}
\bra{\psi} G \otimes I \ket{\psi} - \bra{\psi} G^{d+1} \otimes I \ket{\psi}
= \E_{\bi \sim \mu}[\lambda_{\bi}] - \E_{\bi \sim \mu}[\lambda_{\bi}^{d+1}]
= \E_{\bi \sim\mu}[\lambda_{\bi}(1-\lambda_{\bi}^d)].
\end{equation*}
To do so, we will use the following lemma.

\begin{lemma}\label{lem:looks-easy-but-took-me-a-while}
For any real number $0 \leq \lambda \leq 1$,
\begin{equation*}
\lambda (1-\lambda^d) \leq 2 \cdot \Big(\lambda^{d+1}(1-\lambda)\Big)^{1/(d+1)}.
\end{equation*}
\end{lemma}
\begin{proof}
The lemma is trivial for $\lambda = 1$, and so we will assume that $\lambda \neq 1$.
We note that
\begin{align*}
\lambda^{d+1} (1-\lambda^d)^{d+1}
\leq \lambda^{d+1} (1-\lambda^d)
&= \lambda^{d+1}(1-\lambda) \cdot \left(\frac{1-\lambda^d}{1-\lambda}\right)\\
&= \lambda^{d+1}(1-\lambda) \cdot (1 + \lambda + \cdots + \lambda^{d-1})
\leq d \cdot \lambda^{d+1}(1-\lambda).
\end{align*}
The lemma now follows by taking the $(d+1)$-st root of both sides and noting that $d^{1/(d+1)} \leq 2$ for all integers $d \geq 1$.
\end{proof}

Then \Cref{lem:looks-easy-but-took-me-a-while} implies that
\begin{align*}
\E_{\bi \sim\mu}[\lambda_{\bi}(1-\lambda_{\bi}^d)]
&\leq 2\cdot \E_{\bi \sim\mu}\Big[\Big(\lambda_{\bi}^{d+1}(1-\lambda_{\bi})\Big)^{1/(d+1)}\Big]\\
&\leq2\cdot \Big(\E_{\bi \sim\mu}[\lambda_{\bi}^{d+1}(1-\lambda_{\bi})]\Big)^{1/(d+1)}
							\tag{because $a \mapsto a^{1/(d+1)}$ is concave}\\
&\leq 2 \cdot \Delta^{1/(d+1)}. \tag{by \Cref{eq:equivalent-way-of-writing-grand-plan}}
\end{align*}
Thus, we have established the bound
\begin{equation*}
\bra{\psi}G \otimes I \ket{\psi}
- \bra{\psi}G^{d+1} \otimes I \ket{\psi}
\leq 2\cdot \Delta^{1/{d+1}}
= 2\cdot ( \poly(m) \cdot \poly(\eps, \delta, \gamma, \zeta,d/q))^{1/{d+1}}.
\end{equation*}
For constant~$d$, this bound is suitable for our proof, as the right-hand side is still a polynomial in the relevant parameters.
However, when~$d$ is larger, for this bound to be meaningful, we need $\eps$, $\delta$, etc. to be exponentially small in~$d$, which is a more stringent condition than we generally allow (unless, again, $d$ is a constant). 
That said, we believe this large error may be an artifact of the proof strategy
rather than something intrinsic to the construction.
We leave this to future work.

\subsubsection{The second construction}\label{sec:construction-numba-two}

The second construction of $H = \{H_h\}$ 
is designed to circumvent the problem of the first construction,
which is that its completeness was difficult to analyze.
Instead, we will design a pasted measurement
in which the ``naive analysis" actually gets us the bound we want,
which is that~$H$'s completeness is close to~$G$'s completeness.
Before describing the construction, we need the following definitions.

\begin{definition}[$G$'s incomplete part]
For each $x \in \F_q$, we write $G^x = \sum_g G^x_g$ and $G^x_{\bot} = I - G^x$
for the ``complete" and ``incomplete" parts of $G^x$, respectively.

It will be convenient to sometimes regard $G^x$ as a complete measurement
by throwing in the additional measurement outcome ``$\bot$".
To distinguish this from $G^x$ as a sub-measurement,
we will use the notation ``$\widehat{G}^x$".
In other words, we let $\widehat{G} = \{\widehat{G}^x_g\}$
be the projective measurement defined as
\begin{equation*}
\widehat{G}^x_g
= \left\{\begin{array}{rl}
	G^x_g & \text{if } g \in \polyfunc{m}{q}{d},\\
	G^x_{\bot} & \text{if } g = \bot.
	\end{array}\right.
\end{equation*}
This measurement has outcomes ranging over the set 
$\calP^+(m,q,d) := \polyfunc{m}{q}{d} \cup \{\bot\}$.
\end{definition}

\begin{definition}[Types]
A type $\tau$ is an element of $\{0, 1\}^k$ for some integer~$k$.
We write $|\tau| = \tau_1 + \cdots + \tau_k$ for the Hamming weight of~$\tau$.
We will also associate $\tau$ with the set $\{i \mid \tau_i = 1\}$ and write $i \in \tau$ if $\tau_i = 1$.
\end{definition}

Suppose we perform the $\wG$ measurement $k$ times in succession,
generating the random outcomes $\bg_1, \ldots, \bg_k$.
Let us write $\btau \in \{0, 1\}^k$ for the ``type" of these outcomes, where
\begin{equation*}
\btau_i =
\left\{\begin{array}{rl}
	1 & \text{if $\bg_i \in \polyfunc{m}{q}{d}$},\\
	0 & \text{if $\bg_i = \bot$.}
	\end{array}
	\right.
\end{equation*}
Assuming the $\bg_i$'s are not inconsistent,
then we can interpolate them to produce a global polynomial~$\bh$
whenever $|\btau| \geq d+1$.
Hence, we would like to understand the probability that $|\btau| \geq d+1$
and ensure that it is as large as possible.
The probability that the measurement $\wG$ returns a polynomial $g \in \polyfunc{m}{q}{d}$
is equal to the completeness of~$G$, which is $1-\kappa$.
This tells us that the probability that $\btau_1 = 1$ is $1-\kappa$.
We might naively expect that the same holds for the other $\btau_i$'s as well.
We might also naively expect that the $\btau_i$'s are independent.
These two assumptions should not be expected to hold in general,
as they ignore correlations between the measurements
and the fact that each measurement perturbs the state $\ket{\psi}$ for subsequent measurements to use.
However, if we make these assumptions, then we at least have a simple toy model
for the measurement outcomes: $\btau \sim \mathrm{Binomial}(k, 1-\kappa)$.

In this toy model,
we expect $|\btau| \approx k \cdot (1-\kappa)$ on average.
This was the problem with the ``naive analysis" from the first construction:
if $k = d+1$ and $\kappa$ is reasonably large (say, on the order of $1/d$),
then we don't expect $|\btau|$ to be $\geq d+1$ with high probability,
and so we can't interpolate to produce a global polynomial.
This suggests an alternative strategy:
simply choose $k$ large enough so that $k \cdot (1 - \kappa) \gg d+1$.
In fact, as we are aiming for $H$ to have completeness close to $1-\kappa$,
we should choose~$k$ so large that $|\btau| \geq d+1$ with probability 
roughly $1-\kappa$.
This is easily done with a Chernoff bound,
which is responsible for the exponential error term in \Cref{item:ld-pasting-N-completeness-sub-measurement}
of \Cref{lem:ld-pasting-sub-measurement}.
On the other hand, if we set $k$ \emph{too} large,
then we increase the risk that our $k$ outcomes $g_1, \ldots, g_k$ are inconsistent with each other,
which is an additional source of error.
This is responsible for the tradeoff between ``large" and ``small" $k$ discussed in \Cref{sec:self-improvement-and-pasting} above.

Although, this ``naive analysis" only holds in this toy model,
it still motivates our second construction of~$H$,
which we state below.
We will show that the naive analysis,
in which we treat $\btau$ as a binomial random variable
and bound $|\btau|$ using a Chernoff bound,
can actually be made formal.

\begin{definition}[The pasted measurement]
Let $k \geq d+1$ be an integer.
\begin{enumerate}
\item (Pasting):
Let $x_1, \ldots, x_k \in \F_q$.
We will define an initial ``sandwiched" measurement as follows:
\begin{equation*}
\widehat{H}^{x_1, \ldots, x_{k}}_{g_1, \ldots, g_{k}} = \widehat{G}^{x_1}_{g_1} \cdot \widehat{G}^{x_2}_{g_2} \cdots \widehat{G}^{x_{k}}_{g_{k}} \cdots \widehat{G}^{x_2}_{g_2} \cdot \widehat{G}^{x_1}_{g_1}.
\end{equation*}
\item (Interpolation):
Next, let $(x_1, \ldots, x_{k}) \in \distinct{k}$.
For any string $w \in \{0, 1\}^k$ and polynomial $h \in \polyfunc{m+1}{q}{d}$,
we define $h_w$  to be the tuple $(g_1, \ldots, g_k)  \in \calP^+(m, q, d)^k$
where $g_i = \bot$ if $w_i = 0$ and $g_i = h|_{x_i}$ otherwise.
We define the interpolated measurement
\begin{equation*}
H^{x_1, \ldots, x_{k}}_{h}
= \sum_{w : |w| \geq d+1} \widehat{H}^{x_1, \ldots, x_{k}}_{h_w}.
\end{equation*}
\item (Averaging):
Finally, we randomize over the choice of $(x_1, \ldots, x_{k})$. In other words, we define
\begin{equation*}
H_{h} = \E_{(\bx_1, \ldots, \bx_{k}) \sim \distinct{k}} H^{\bx_1, \ldots, \bx_{k}}_h.
\end{equation*}
\end{enumerate}
\end{definition}

To analyze the second construction,
we first need to show that the $\wG$ measurement
satisfies some basic properties, like commutation with itself.
These are shown in \Cref{sec:hat-consistency},
where they follow from the fact that similar properties hold for $G$.
Using this, we prove that $\widehat{H}$ is consistent with~$B$ in \Cref{sec:ld-sandwiching},
which we use to prove that $H$ is consistent with~$A$ in \Cref{sec:consistency-of-h-with-a}.
Finally, we analyze the completeness of~$H$ in \Cref{sec:completeness-of-h-low-degree}.

\subsection{Strong self-consistency and commutation of $\widehat{G}$}
\label{sec:hat-consistency}

In this section, we show that $\wG$ is strongly self-consistent and commutes with itself.
As we already know this holds for the sub-measurement~$G$,
our task essentially reduces to showing that these properties also hold for $G$'s incomplete part,
i.e.\ $G_{\perp}$.
As it is more convenient to work with $G = I - G_{\perp}$ rather than $G_{\perp}$,
we will first show that these properties hold for~$G$;
the fact that they also hold for $G_{\perp}$ will then follow as an immediate corollary.

\subsubsection{Strong self-consistency of $G_{\perp}$}

\begin{lemma}[Strong self-consistency of $G$'s complete part]\label{lem:g-complete-self-consistency}
\begin{equation*}
G^{x} \ot I \approx_{\zeta}  I \ot G^{x}.
\end{equation*}
\end{lemma}
\begin{proof}
Because~$G$ is a projective measurement, \Cref{prop:two-notions-of-self-consistency}
implies that the strong self-consistency of~$G$ from \Cref{item:ld-pasting-self-consistency} is equivalent to 
\begin{equation}\label{eq:ld-g-self-consistency}
\E_{\bx} \sum_{g} \bra{\psi} G^{\bx}_{g} \ot G^{\bx}_{g}
    \ket{\psi} \geq \E_{\bx} \sum_{g} \bra{\psi} G^{\bx}_{g} \ot
                 I\ket{\psi} - \frac{1}{2}\cdot \zeta.
\end{equation}
Our goal is to bound
\begin{align*}
&\E_{\bx} \Vert (G^{\bx}\ot I - I \ot G^{\bx}) \ket{\psi} \Vert^2\\
=~&2\cdot \E_{\bx} \bra{\psi} G^{\bx} \otimes I \ket{\psi} - 2\cdot \E_{\bx} \bra{\psi} G^{\bx} \otimes G^{\bx} \ket{\psi}\\
=~&2\cdot \E_{\bx} \sum_{g_1} \bra{\psi} G^{\bx}_{g} \otimes I \ket{\psi} - 2\cdot \E_{\bx} \sum_{g, h} \bra{\psi} G^{\bx}_{g} \otimes G^{\bx}_{h} \ket{\psi}\\
\leq~&2\cdot \E_{\bx} \sum_{g} \bra{\psi} G^{\bx}_{g} \otimes I \ket{\psi} - 2\cdot \E_{\bx} \sum_{g} \bra{\psi} G^{\bx}_{g} \otimes G^{\bx}_{g} \ket{\psi}.
\end{align*}
But this is at most~$\zeta$ by \Cref{eq:ld-g-self-consistency}.
\end{proof}

\begin{corollary}[Strong self-consistency of~$G$'s incomplete part]\label{cor:g-bot-self-consistency}
\begin{equation*}
G^{x}_{\bot} \ot I \approx_{\zeta}  I \ot G^{x}_{\bot}.
\end{equation*}
\end{corollary}
\begin{proof}
 For any~$x$,
 \begin{equation*}
 G^x_{\bot} \otimes I - I \otimes G^x_{\bot}
 = (I - G^x) \otimes I - I \otimes (I - G^x)
 = I \otimes G^x - G^x \otimes I.
 \end{equation*}
Thus,
\begin{equation*}
\E_{\bx}  \| (G^{\bx}_{\bot} \ot I -I \ot
      G^{\bx}_{\bot}) \ket{\psi} \|^2
= \E_{\bx}  \| (I \otimes G^{\bx} - G^{\bx} \otimes I) \ket{\psi} \|^2,
\end{equation*}
which is at most $\zeta$ by \Cref{lem:g-complete-self-consistency}.
\end{proof}

\subsubsection{Commutativity of $G_{\perp}$}

\begin{lemma}[Commutativity with~$G_g^x$ implies commutativity with~$G^x$]\label{lem:commutativity-switcheroo}
Let $M = \{M^x_o\}$ be a projective sub-measurement with outcomes in some set~$\calO$.
Suppose that
\begin{equation}\label{eq:M-self-consistent}
M^x_o \otimes I \approx_{\omega} I \otimes M^x_o,
\end{equation}
and
\begin{equation}\label{eq:M-commutes-with-G}
G^x_g M^y_o \otimes I \approx_{\chi} M^y_o G^x_g \otimes I
\end{equation}
over independent and uniformly random $\bx, \by \sim \F_q$. Then
\begin{equation*}
G^x M^y_o \otimes I \approx_{6\sqrt{\zeta} + 6\sqrt{\omega} + 4\sqrt{\chi}} M^y_o G^x \otimes I.
\end{equation*}
\end{lemma}
\begin{proof}
The error we wish to bound is
\begin{align}
&\E_{\bx} \sum_o \Vert (G^{\bx} M^{\by}_o - M^{\by}_o G^{\bx}) \otimes I \ket{\psi}\Vert^2\nonumber\\
=~&\E_{\bx, \by} \sum_o \bra{\psi} M^{\by}_o (G^{\bx})^2 M^{\by}_o \otimes I \ket{\psi} + \E_{\bx, \by} \sum_o \bra{\psi} G^{\bx} (M^{\by}_o)^2 G^{\bx} \otimes I \ket{\psi}\nonumber\\
& \qquad - \E_{\bx, \by} \sum_o  \bra{\psi} G^{\bx} M^{\by}_o
        G^{\bx} M^{\by}_o \ot I \ket{\psi}
        	- \E_{\bx, \by} \sum_g \bra{\psi}
        M^{\by}_o G^{\bx} M^{\by}_o G^{\bx} \ot I \ket{\psi}.\label{eq:g-commute-with-gg-error}
\end{align}
We will show that all four terms in \Cref{eq:g-commute-with-gg-error} are close to $\bra{\psi} G \otimes M \ket{\psi}$,
where $M = \E_{\by} \sum_o M^{\by}_o$.

For the first term in \Cref{eq:g-commute-with-gg-error}, we have
\begin{align*}
\E_{\bx, \by} \sum_o \bra{\psi} M^{\by}_o (G^{\bx})^2 M^{\by}_o\ot I \ket{\psi}
&= \E_{\bx, \by} \sum_o \bra{\psi} M^{\by}_o G^{\bx} M^{\by}_o\ot I \ket{\psi} \tag{because~$G$ is projective}\\
&= \E_{\by}\sum_o \bra{\psi} M^{\by}_o G M^{\by}_o \ot I \ket{\psi}\\
&\approx_{2\sqrt{\omega}} \E_{\by}\sum_o \bra{\psi} G \ot M^{\by}_o  \ket{\psi} \tag{by \Cref{prop:switch-sandwich} and \Cref{eq:M-self-consistent}}\\
&= \bra{\psi} G \ot M \ket{\psi}.
\end{align*}
Similarly, for the second term in \Cref{eq:g-commute-with-gg-error}, we have
\begin{align*}
\E_{\bx, \by} \sum_o \bra{\psi} G^{\bx} (M^{\by}_o)^2 G^{\bx} \otimes I \ket{\psi}
&=\E_{\bx, \by} \sum_o \bra{\psi} G^{\bx} M^{\by}_o G^{\bx} \otimes I \ket{\psi} \tag{because~$M$ is projective}\\
&=\E_{\bx}  \bra{\psi} G^{\bx} M G^{\bx} \otimes I \ket{\psi}\\
&\approx_{2\sqrt{\zeta}} \E_{\bx} \bra{\psi} M \ot G^{\bx}  \ket{\psi} \tag{by \Cref{prop:switch-sandwich} and \Cref{lem:g-complete-self-consistency}}\\
&= \bra{\psi} M \ot G \ket{\psi}.
\end{align*}

For the third term in \Cref{eq:g-commute-with-gg-error}, we begin by claiming that
\begin{align}
\E_{\bx, \by} \sum_o \bra{\psi} G^{\bx} M^{\by}_o
        G^{\bx} M^{\by}_o \ot I \ket{\psi}
& = 
\E_{\bx, \by} \sum_{o, g}  \bra{\psi} G^{\bx} M^{\by}_o
        G^{\bx}_{g} M^{\by}_o \ot I \ket{\psi}\nonumber\\
& = 
\E_{\bx, \by} \sum_{o, g}  \bra{\psi} G^{\bx} M^{\by}_o G^{\bx}_g
        G^{\bx}_{g} M^{\by}_o \ot I \ket{\psi} \tag{because~$G$ is projective}\\
& \approx_{\sqrt{\chi}} 
\E_{\bx, \by} \sum_{o, g}  \bra{\psi}  G^{\bx} M^{\by}_o G^{\bx}_g
        M^{\by}_o G^{\bx}_{g}  \ot I \ket{\psi}.\label{eq:split-G-and-commute}
\end{align}
To show this, we bound the magnitude of the difference.
\begin{multline*}
\Big|\E_{\bx, \by} \sum_{o, g}  \bra{\psi}  (G^{\bx}
        M^{\by}_{o} G^{\bx}_g \ot I) \cdot ((G^{\bx}_g M^{\by}_o - M^{\by}_o G^{\bx}_{g}) \otimes I)\ket{\psi}\Big|\\
\leq \sqrt{\E_{\bx, \by} \sum_{o, g}  \bra{\psi} (G^{\bx} M^{\by}_o G^{\bx}_g
        M^{\by}_{o} G^{\bx}) \ot I\ket{\psi}}\\
\cdot \sqrt{\E_{\bx, \by} \sum_{o, g}  \bra{\psi}  ((M^{\by}_o G^{\bx}_g  -  G^{\bx}_{g} M^{\by}_o) \cdot (G^{\bx}_g M^{\by}_o - M^{\by}_o G^{\bx}_{g}) \otimes I) \ket{\psi}}.
\end{multline*}
The expression inside the first square root is at most~$1$ because~$G$ and~$M$ are sub-measurements.
The expression inside the second square root is at most $\chi$ by \Cref{eq:M-commutes-with-G}.
Next, we claim that
\begin{equation}\label{eq:move-G-for-great-justice}
\eqref{eq:split-G-and-commute}
\approx_{\sqrt{\zeta}} \E_{\bx, \by} \sum_{o, g}  \bra{\psi}   G^{\bx} M^{\by}_o
        G^{\bx}_{g} M^{\by}_o \ot G^{\bx}_g\ket{\psi}.
\end{equation}
To show this, we bound the magnitude of the difference.
\begin{align*}
&\Big|  \E_{\bx, \by} \sum_{o, g}  \bra{\psi} (G^{\bx} M^{\by}_{o} G^{\bx}_g \ot I)
	\cdot ((M^{\by}_{o} \otimes I) \cdot(G^{\bx}_g \otimes I - I \otimes G^{\bx}_g))\ket{\psi}\Big|\\
&\leq  \sqrt{\E_{\bx, \by} \sum_{o, g} \bra{\psi} (G^{\bx} M^{\by}_o G^{\bx}_g M^{\by}_{o} G^{\bx}) \ot I \ket{\psi}}\\
&\quad \cdot \sqrt{ \E_{\bx, \by} \sum_{o, g}  \bra{\psi}   ((G^{\bx}_g \otimes I - I \otimes G^{\bx}_g) \cdot (M^{\by}_{o} \otimes I) \cdot (G^{\bx}_g \otimes I - I \otimes G^{\bx}_g)) \ket{\psi}}.
\end{align*}
The expression inside the first square root is at most~$1$ because~$G$ and~$M$ are sub-measurements.
The expression inside the second square root is
\begin{align*}
& \E_{\bx}\sum_g \bra{\psi}(G^{\bx}_g \ot I - I \ot G^{\bx}_g) \cdot\Big(\E_{\by} \sum_o M^{\by}_{o} \ot I\Big) \cdot (G^{\bx}_g \ot I - I \ot G^{\bx}_g) \ket{\psi}\\
\leq~&\E_{\bx}\sum_g \bra{\psi}(G^{\bx}_g \ot I - I \ot G^{\bx}_g)^2 \ket{\psi}. \tag{because~$M$ is a sub-measurement}
\end{align*} 
This is at most~$\zeta$ by \Cref{item:ld-pasting-self-consistency}.
Next, we claim that
\begin{equation}\label{eq:move-G-for-even-greater-justice}
\eqref{eq:move-G-for-great-justice}
\approx_{\sqrt{\zeta}} \E_{\bx, \by} \sum_{o, g}  \bra{\psi}  G^{\bx}_g  G^{\bx} M^{\by}_o
        G^{\bx}_{g} M^{\by}_o \ot I\ket{\psi}.
\end{equation}
To show this, we bound the magnitude of the difference.
\begin{align*}
&\Big|  \E_{\bx, \by} \sum_{o, g}  \bra{\psi} ((G^{\bx}_g \ot I - I \ot G^{\bx}_g) \cdot (G^{\bx} M^{\by}_o \otimes I)) 
	\cdot  (G^{\bx}_{g} M^{\by}_o \otimes I) 
        \ket{\psi}\Big|\\
&\leq \sqrt{ \E_{\bx, \by} \sum_{o, g}  \bra{\psi} ((G^{\bx}_g \ot I - I \ot G^{\bx}_g) \cdot (G^{\bx} M^{\by}_o G^{\bx} \ot I) \cdot (G^{\bx}_g \ot I - I \ot G^{\bx}_g))\ket{\psi}}\\
&\quad \cdot \sqrt{ \E_{\bx, \by} \sum_{o, g}  \bra{\psi}  (M^{\by}_{o} G^{\bx}_g M^{\by}_o) \otimes I \ket{\psi}}.
\end{align*}
The expression inside the first square root is
\begin{align*}
& \E_{\bx} \sum_{g}  \bra{\psi} ((G^{\bx}_g \ot I - I \ot G^{\bx}_g) \cdot \Big(\E_{\by} \sum_o G^{\bx} M^{\by}_o G^{\bx} \ot I\Big) \cdot (G^{\bx}_g \ot I - I \ot G^{\bx}_g))\ket{\psi}\\
\leq~&\E_{\bx} \sum_{g}  \bra{\psi} (G^{\bx}_g \ot I - I \ot G^{\bx}_g)^2\ket{\psi}. \tag{because~$G$ and~$M$ are sub-measurements}
\end{align*} 
This is at most~$\zeta$ by \Cref{item:ld-pasting-self-consistency}.
The expression inside the second square root is at most~$1$ because~$G$ and~$M$ are sub-measurements.
Next, we claim that
\begin{align}
\eqref{eq:move-G-for-even-greater-justice}
& = \E_{\bx, \by} \sum_{o, g}  \bra{\psi}   G^{\bx}_{g} M^{\by}_o
        G^{\bx}_{g} M^{\by}_o \ot I\ket{\psi} \tag{because~$G$ is projective}\\
& \approx_{\sqrt{\chi}} 
\E_{\bx, \by} \sum_{o,g}  \bra{\psi}   M^{\by}_{o} G^{\bx}_g G^{\bx}_g
        M^{\by}_{o}  \ot I\ket{\psi}. \label{eq:commute-the-G-yet-again}
\end{align}
To show this, we bound the magnitude of the difference.
\begin{multline*}
\Big|\E_{\bx, \by} \sum_{o, g}  \bra{\psi} ((G^{\bx}_g
        M^{\by}_{o} - M^{\by}_o G^{\bx}_g)  \ot I)  \cdot(G^{\bx}_{g} M^{\by}_o  \ot I ) \ket{\psi}\Big|\\
\leq\sqrt{\E_{\bx, \by} \sum_{o, g} \bra{\psi} ((G^{\bx}_g M^{\by}_o - M^{\by}_o G^{\bx}_g) \cdot (M^{\by}_o
        G^{\bx}_{g} - G^{\bx}_g M^{\by}_o)  \ot I)\ket{\psi}}\\
        \cdot \sqrt{\E_{\bx, \by} \sum_{o, g}  \bra{\psi} (M^{\by}_o G^{\bx}_g M^{\by}_o) \otimes I \ket{\psi}}.
\end{multline*}
The expression inside the first square root is at most $\chi$ by \Cref{eq:M-commutes-with-G}.
The expression inside the second square root is at most~$1$ because~$G$ and~$M$ are sub-measurements.
Finally,
\begin{align*}
\eqref{eq:commute-the-G-yet-again}
&=\E_{\bx, \by} \sum_{o, g}  \bra{\psi}   M^{\by}_{o} G^{\bx}_g
        M^{\by}_{o}  \ot I\ket{\psi} \tag{because~$G$ is projective}\\
&=\E_{\by} \sum_o \bra{\psi} M^{\by}_{o} G M^{\by}_{o} \otimes I \ket{\psi}\\
&\approx_{2\sqrt{\omega}} \E_{\by} \sum_o \bra{\psi} G \ot M^{\by}_o  \ket{\psi} \tag{by \Cref{prop:switch-sandwich} and \Cref{eq:M-self-consistent}}\\
&= \bra{\psi} G \ot M \ket{\psi}.
\end{align*}
In total, this shows that
\begin{equation}\label{eq:term-three-for-use-right-now}
\E_{\bx, \by} \sum_o  \bra{\psi} G^{\bx} M^{\by}_o
        G^{\bx} M^{\by}_o \ot I \ket{\psi}
        \approx_{2 \sqrt{\zeta} + 2 \sqrt{\omega} + 2 \sqrt{\chi}} \bra{\psi} G \ot M \ket{\psi}.
\end{equation}
The fourth term in \Cref{eq:g-commute-with-gg-error} is the Hermitian conjugate of the third term. As a result, \Cref{eq:term-three-for-use-right-now} implies that
\begin{equation*}
\E_{\bx, \by} \sum_o  \bra{\psi} M^{\by}_o G^{\bx} M^{\by}_o G^{\bx} \ot I \ket{\psi}
        \approx_{2 \sqrt{\zeta} + 2 \sqrt{\omega} + 2 \sqrt{\chi}} \bra{\psi} G \ot M \ket{\psi}
\end{equation*}
as well.

In total, this gives an error of
\begin{equation*}
2\sqrt{\omega}
+ 2\sqrt{\zeta}
+ 2 \cdot\left(2 \sqrt{\zeta} + 2 \sqrt{\omega} + 2 \sqrt{\chi}\right)
= 6 \sqrt{\zeta} + 6 \sqrt{\omega} + 4 \sqrt{\chi}.
\end{equation*}
This proves the claimed bound.
\end{proof}

\begin{corollary}[Commutativity of~$G$'s complete part]\label{cor:commuting-with-G-complete}
The following commutation relations hold.
\begin{align*}
G^x_g  G^y \ot I &\approx_{\nu_2}  G^y G^x_g \otimes I\\
G^x G^y \ot I &\approx_{\nu_2}  G^y G^x \otimes I,
\end{align*}
where
\begin{equation*}
\nu_2 = 36 m \cdot \left(\gamma^{1/16} +  \zeta^{1/16} + (d/q)^{1/16}\right).
\end{equation*}
\end{corollary}
\begin{proof}
By \Cref{eq:quote-com-main},
\begin{equation}\label{eq:com-main-copy}
G^x_g G^y_h \otimes I \approx_{\nu_{\rmcom}} G^y_h G^x_g \otimes I.
\end{equation}
Now we apply \Cref{lem:commutativity-switcheroo} to \Cref{eq:com-main-copy}.
To do so, we set the ``$\{M^y_o\}$" sub-measurement to be $\{G^x_g\}$,
and therefore $\calO = \polyfunc{m}{q}{d}$.
This implies that
\begin{equation}\label{eq:applied-the-lemma-to-com-main-copy}
G^x_g G^y \otimes I \approx_{\theta_1} G^y G^x_g \otimes I
\end{equation}
for $\theta_1 = 12 \sqrt{\zeta} + 4 \sqrt{\nu_{\rmcom}}$.

Next, we apply \Cref{lem:commutativity-switcheroo} to \Cref{eq:applied-the-lemma-to-com-main-copy}.
This time, we let $\calO$ be a set containing a single outcome,
and for this outcome~$o$, we set $M^y_o = G^y$.
This implies that
\begin{equation*}
G^x G^y \otimes I \approx_{\theta_2} G^y G^x \otimes I
\end{equation*}
for $\theta_2 = 12 \sqrt{\zeta} + 4 \sqrt{\theta_1}$.
This uses \Cref{lem:g-complete-self-consistency} for the strong self-consistency of $\{M^y_o\}$.

We now show that $\nu_2$ bounds $\theta_1$ and $\theta_2$.
First, using $\sqrt{30} \leq 6$, we have
\begin{align*}
\theta_1 = 12 \sqrt{\zeta} + 4 \sqrt{\nu_{\rmcom}}
&= 12 \sqrt{\zeta} + 4 \sqrt{30 m \cdot \left(\gamma^{1/4} + \zeta^{1/4} + (d/q)^{1/4}\right)}\\
&\leq 12 \zeta^{1/8} + 24m \cdot \left(\gamma^{1/8} +\zeta^{1/8} +  (d/q)^{1/8}\right)\\
&\leq 36 m \cdot \left(\gamma^{1/8} + \zeta^{1/8} +  (d/q)^{1/8}\right).
\end{align*}
This is clearly less than $\nu_2$.
Next,  we have
\begin{align*}
\theta_2 = 12 \sqrt{\zeta} + 4 \sqrt{\theta_1}
&\leq 12 \sqrt{\zeta} + 4 \sqrt{36 m \cdot \left(\gamma^{1/8} + \zeta^{1/8} +  (d/q)^{1/8}\right)}\\
&\leq 12 \zeta^{1/16} + 24m \cdot \left(\gamma^{1/16} + \zeta^{1/16} + (d/q)^{1/16}\right)\\
&\leq 36 m \cdot \left( \gamma^{1/16} + \zeta^{1/16} + (d/q)^{1/16}\right).
\end{align*}
This is equal to $\nu_2$, which completes the proof.
\end{proof}

\begin{corollary}[Commutativity of~$G$'s incomplete part]\label{cor:commuting-with-G-incomplete}
The following commutation relations hold.
\begin{align*}
G^x_g  G^y_{\bot} \ot I &\approx_{\nu_2}  G^y_{\bot} G^x_g \otimes I\\
G^x_{\bot} G^y_{\bot} \ot I &\approx_{\nu_2}  G^y_{\bot} G^x_{\bot} \otimes I,
\end{align*}
where $\nu_2$ is as in \Cref{cor:commuting-with-G-complete}.
\end{corollary}
\begin{proof}
First, we note that
\begin{equation*}
G^x_g  G^y_{\bot} - G^y_{\bot} G^x_g
= G^x_g  \cdot (I - G^y) - (I- G^y)\cdot G^x_g
= G^y G^x_g - G^x_g G^y.
\end{equation*}
Hence, by \Cref{cor:commuting-with-G-complete},
\begin{equation*}
\E_{\bx, \by}\sum_g \Vert (G^{\bx}_g  G^{\by}_{\bot} - G^{\by}_{\bot} G^{\bx}_g) \otimes I \ket{\psi} \Vert^2
= \E_{\bx, \by}\sum_g \Vert (G^{\by} G^{\bx}_g - G^{\bx}_g G^{\by}) \otimes I \ket{\psi} \Vert^2
\leq \nu_2.
\end{equation*}
Next, we note that
\begin{align*}
G^x_{\bot} G^y_{\bot}
- G^y_{\bot} G^x_{\bot}
&= (I - G^x) \cdot (I - G^y) - (I - G^y) \cdot (I - G^x)\\
& = (I - G^x - G^y + G^x G^y) - (I - G^y - G^x + G^y G^x)\\
&= G^x G^y - G^y G^x.
\end{align*}
As a result, by \Cref{cor:commuting-with-G-complete},
\begin{equation*}
\E_{\bx, \by} \Vert (G^{\bx}_{\bot}  G^{\by}_{\bot} - G^{\by}_{\bot} G^{\bx}_{\bot}) \otimes I \ket{\psi} \Vert^2
= \E_{\bx, \by} \Vert (G^{\by} G^{\bx} - G^{\bx} G^{\by}) \otimes I \ket{\psi} \Vert^2
\leq \nu_2.\qedhere
\end{equation*}
\end{proof}

\subsubsection{Putting everything together}

Now we combine the results of the previous two sections
to show our strong self-consistency and commutation results for~$\wG$.

\begin{corollary}[Strong self-consistency and commutation of~$\wG$]\label{cor:G-hat-facts}
  $\widehat{G}$ obeys the
  following strong self-consistency and commutation properties.
  \begin{align}
    \widehat{G}^{x}_{g} \ot I &\approx_{2\zeta} I \ot \widehat{G}^{x}_{g}, \label{eq:gselfconall}\\
    \widehat{G}^{x}_{g}\widehat{G}^{y}_{h} \ot I &\approx_{\nu_3} \widehat{G}^{y}_{h}
                               \widehat{G}^{x}_{g} \ot I, \label{eq:gcomall}
  \end{align}
  where
  \[ \nu_{3} = 138 m \cdot \left(\zeta^{1/16} + \gamma^{1/16} + (d/q)^{1/16}\right). \]
\end{corollary}
\begin{proof}
We begin with \Cref{eq:gselfconall}. To prove this, we wish to bound
\begin{align*}
&\E_{\bx} \sum_{g} \Vert(\widehat{G}^{\bx}_g \otimes I - I \otimes \widehat{G}^{\bx}_g) \ket{\psi} \Vert^2\\
=~&\E_{\bx} \sum_{g \in \polyfunc{m}{q}{d}} \Vert(G^{\bx}_g \otimes I - I \otimes G^{\bx}_g) \ket{\psi} \Vert^2
	+ \E_{\bx} \Vert(G^{\bx}_{\bot} \otimes I - I \otimes G^{\bx}_{\bot}) \ket{\psi} \Vert^2\\
\leq~& \zeta + \zeta = 2\zeta,
\end{align*}
by \Cref{item:ld-pasting-self-consistency} and \Cref{cor:g-bot-self-consistency}.

Next, we show \Cref{eq:gcomall}.
To prove this, we wish to bound
\begin{align*}
&\E_{\bx, \by} \sum_{g, h} \Vert(\widehat{G}^{\bx}_g \widehat{G}^{\by}_h - \widehat{G}^{\by}_h \widehat{G}^{\bx}_g) \otimes I \ket{\psi}\Vert^2\\
=~&\E_{\bx, \by} \sum_{g, h \in \polyfunc{m}{q}{d}} \Vert(G^{\bx}_g G^{\by}_h - G^{\by}_h G^{\bx}_g) \otimes I \ket{\psi}\Vert^2
	+ \E_{\bx, \by} \Vert(G^{\bx}_{\bot} G^{\by}_{\bot} - G^{\by}_{\bot} G^{\bx}_{\bot}) \otimes I \ket{\psi}\Vert^2\\
& \quad + \E_{\bx, \by} \sum_{g\in \polyfunc{m}{q}{d}} \Vert(G^{\bx}_g G^{\by}_{\bot} - G^{\by}_{\bot} G^{\bx}_g) \otimes I \ket{\psi}\Vert^2
	+ \E_{\bx, \by} \sum_{h\in \polyfunc{m}{q}{d}} \Vert(G^{\bx}_{\bot} G^{\by}_h - G^{\by}_h G^{\bx}_{\bot}) \otimes I \ket{\psi}\Vert^2\\
\leq~& \nu_{\rmcom} + 3\nu_2,
\end{align*}
by \Cref{eq:quote-com-main} and \Cref{cor:commuting-with-G-incomplete}.
We can therefore bound this by
\begin{align*}
\nu_{\rmcom} + 3\nu_2
&= 30 m \cdot \left(\gamma^{1/4} +\zeta^{1/4} +  (d/q)^{1/4}\right)
	+ 3 \cdot 36 m \cdot \left(\gamma^{1/16} + \zeta^{1/16} +  (d/q)^{1/16}\right)\\
&\leq 30 m \cdot \left(\gamma^{1/16} + \zeta^{1/16} + (d/q)^{1/16}\right)
	+ 3 \cdot 36 m \cdot \left(\gamma^{1/16} + \zeta^{1/16} +  (d/q)^{1/16}\right)\\
&= 138 m \cdot \left( \gamma^{1/16} + \zeta^{1/16} +(d/q)^{1/16}\right).	
\end{align*}
This completes the proof.
\end{proof}

\subsection{Consistency of the sandwich $\widehat{H}$ with~$B$}
\label{sec:ld-sandwiching}

In the next lemmas, we will show that $\widehat{H}$ is consistent with the lines
measurement $B$. To start, we show some self-consistency and
commutativity properties of $\widehat{H}$.

\begin{lemma}[Commuting past multiple $\widehat{G}$'s]
\label{lem:commute-g-half-sandwich}
For all $k \geq 2$,
  \[ \widehat{G}^{x_1}_{g_1} \widehat{G}^{x_2}_{g_2} \cdots \widehat{G}^{x_k}_{g_k} \ot I \approx_{\nu_4}
    \widehat{G}^{x_2}_{g_2} \cdots \widehat{G}^{x_k}_{g_k} \widehat{G}^{x_1}_{g_1} \ot I, \]
where
\begin{equation*}
\nu_4 = 426 k^2 m \cdot \left(\gamma^{1/16} +\zeta^{1/16} +  (d/q)^{1/16}\right).
\end{equation*}
\end{lemma}
\begin{proof}
This proof will consist of multiple applications of \Cref{eq:gselfconall,eq:gcomall};
for each line, we will specify which equation to apply.
Each line will also involve an application of \Cref{prop:cab-approx-delta},
which we will specify only implicitly.
\begin{align*}
&  \widehat{G}^{x_1}_{g_1} \widehat{G}^{x_2}_{g_2} \cdots \widehat{G}^{x_k}_{g_k} \ot I\\
\approx_{2\zeta}~& \widehat{G}^{x_1}_{g_1} \widehat{G}^{x_2}_{g_2} \cdots \widehat{G}^{x_{k-1}}_{g_{k-1}} \ot \widehat{G}^{x_k}_{g_k} \tag{by \Cref{eq:gselfconall}}\\
\cdots~&\\
\approx_{2\zeta}~& \widehat{G}^{x_1}_{g_1} \widehat{G}^{x_2}_{g_2} \ot \widehat{G}^{x_k}_{g_k} \cdots \widehat{G}^{x_{3}}_{g_{3}}\tag{by \Cref{eq:gselfconall}}\\
\approx_{\nu_{3}}~&  \widehat{G}^{x_2}_{g_2}G^{x_1}_{g_1} \ot \widehat{G}^{x_k}_{g_k} \cdots \widehat{G}^{x_{3}}_{g_{3}}\tag{by \Cref{eq:gcomall}}\\
\approx_{2\zeta}~&  \widehat{G}^{x_2}_{g_2}\widehat{G}^{x_1}_{g_1} \widehat{G}^{x_3}_{g_3} \ot \widehat{G}^{x_k}_{g_k} \cdots \widehat{G}^{x_{4}}_{g_{4}}\tag{by \Cref{eq:gselfconall}}\\
\approx_{\nu_{3}}~&  \widehat{G}^{x_2}_{g_2}\widehat{G}^{x_3}_{g_3} \widehat{G}^{x_1}_{g_1} \ot \widehat{G}^{x_k}_{g_k} \cdots \widehat{G}^{x_{4}}_{g_{4}}\tag{by \Cref{eq:gcomall}}\\
\cdots~&\\
\approx_{2\zeta}~&  \widehat{G}^{x_2}_{g_2}\cdots \widehat{G}^{x_{k-1}}_{g_{k-1}} \widehat{G}^{x_1}_{g_1} \widehat{G}^{x_k}_{g_k} \ot I\tag{by \Cref{eq:gselfconall}}\\
\approx_{\nu_{3}}~&  \widehat{G}^{x_2}_{g_2}\cdots \widehat{G}^{x_{k-1}}_{g_{k-1}} \widehat{G}^{x_k}_{g_k} \widehat{G}^{x_1}_{g_1} \ot I.\tag{by \Cref{eq:gcomall}}
\end{align*}
In total, we have $(k-2) + (k-2) \leq 2k$ applications of \Cref{eq:gselfconall} with error $2\zeta$ each
and $(k-1) \leq k$ applications of \Cref{eq:gcomall} with error $\nu_3$ each.
By \Cref{prop:triangle-inequality-for-approx_delta}, this implies that 
\begin{equation*}
G^{x_1}_{g_1} G^{x_2}_{g_2} \cdots G^{x_k}_{g_k} \ot I
\approx_{3k \cdot (4k \zeta + k\nu_3)} G^{x_2}_{g_2}\cdots G^{x_{k-1}}_{g_{k-1}} G^{x_k}_{g_k} G^{x_1}_{g_1} \ot I,
\end{equation*}
as claimed.
We can bound this as follows.
\begin{align*}
3k \cdot (4k \zeta + k\nu_3)
& = 12k^2 \zeta + 3k^2 \nu_3\\
&= 12k^2 \zeta + 3 k^2 \cdot 138 m \cdot \left(\gamma^{1/16} + \zeta^{1/16} + (d/q)^{1/16}\right)\\
&\leq 12k^2 \zeta^{1/16} + 414 k^2 m \cdot \left(\gamma^{1/16} + \zeta^{1/16} + (d/q)^{1/16}\right)\\
&\leq 426 k^2 m \cdot \left(\gamma^{1/16} + \zeta^{1/16} + (d/q)^{1/16}\right).
\end{align*}
This completes the proof.
\end{proof}

\begin{lemma}[Consistency of $\widehat{H}$ with $B$]
\label{lem:ld-sandwich-line-one-point}
For any $1 \leq i \leq k$,
\begin{equation*}
\E_{\bu}  \E_{\bx_1, \dots, \bx_k} \sum_{g_1, \dots, g_k: g_i \neq \bot}  \sum_{a \neq
    g_i(\bu)} \bra{\psi}
  \widehat{H}^{\bx_1, \dots, \bx_k}_{g_1, \dots, g_k} \ot B^{\bu}_{[f(\bx_i) =
    a]} \ket{\psi}  \leq \nu_5,
 \end{equation*}
where
\begin{equation*}
    \nu_5 = 43 km \cdot \left(\eps^{1/32} + \delta^{1/32} + \gamma^{1/32} + \zeta^{1/32} + (d/q)^{1/32}\right).
\end{equation*}
\end{lemma}
\begin{proof}
To begin, we note that
\begin{align*}
\sum_{g_{i+1}, \ldots, g_k} \widehat{H}^{x_1, \ldots, x_k}_{g_1, \ldots, g_k}
& = \sum_{g_{i+1}, \ldots, g_k} \widehat{G}^{x_1}_{g_1} \cdots \widehat{G}^{x_k}_{g_k} \cdots \widehat{G}^{x_1}_{g_1}\\
& = \sum_{g_{i+1}, \ldots, g_{k-1}} \widehat{G}^{x_1}_{g_1} \cdots \Big(\sum_{g_k} \widehat{G}^{x_k}_{g_k} \Big)\cdots \widehat{G}^{x_1}_{g_1}\\
& = \sum_{g_{i+1}, \ldots, g_{k-1}} \widehat{G}^{x_1}_{g_1} \cdots \widehat{G}^{x_{k-1}}_{g_{k-1}} \cdot I \cdot \widehat{G}^{x_{k-1}}_{g_{k-1}}\cdots \widehat{G}^{x_1}_{g_1} \tag{because $\widehat{G}^{x_k}$ is a measurement}\\
& = \sum_{g_{i+1}, \ldots, g_{k-1}} \widehat{G}^{x_1}_{g_1} \cdots \widehat{G}^{x_{k-1}}_{g_{k-1}} \cdots \widehat{G}^{x_1}_{g_1} \tag{because $\widehat{G}^{x_{k-1}}$ is projective}\\
& \cdots \\
& = \widehat{G}^{x_1}_{g_1} \cdots \widehat{G}^{x_{i}}_{g_{i}} \cdots \widehat{G}^{x_1}_{g_1}\\
& = \widehat{H}^{x_1, \ldots, x_i}_{g_1, \ldots, g_i}.
\end{align*}
As a result,
\begin{align}
&\E_{\bu}  \E_{\bx_1, \dots, \bx_k} \sum_{g_1, \dots, g_k: g_i \neq \bot}  \sum_{a \neq
    g_i(\bu)} \bra{\psi}
  \widehat{H}^{\bx_1, \dots, \bx_k}_{g_1, \dots, g_k} \ot B^{\bu}_{[f(\bx_i) =
    a]} \ket{\psi}\nonumber\\
=~&\E_{\bu}  \E_{\bx_1, \dots, \bx_i} \sum_{g_1, \dots, g_i: g_i \neq \bot}  \sum_{a \neq
    g_i(\bu)} \bra{\psi}
  \widehat{H}^{\bx_1, \dots, \bx_i}_{g_1, \dots, g_i} \ot B^{\bu}_{[f(\bx_i) =
    a]} \ket{\psi}.\label{eq:delete-extraneous-coordinates}
\end{align}
For shorthand, we write
\begin{equation*}
\widehat{G}^{x_{<i}}_{g_{<i}} = \widehat{G}^{x_1}_{g_1} \cdots \widehat{G}^{x_{i-1}}_{g_{i-1}}.
\end{equation*}
Then we claim that
\begin{align}
\eqref{eq:delete-extraneous-coordinates}
& = \E_{\bu}  \E_{\bx_1, \dots, \bx_i} \sum_{g_1, \dots, g_i: g_i \neq \bot}   \bra{\psi}
  \widehat{G}^{\bx_{<i}}_{g_{<i}}  \cdot \widehat{G}^{\bx_i}_{g_i} \cdot \widehat{G}^{\bx_i}_{g_i} \cdot (\widehat{G}^{\bx_{<i}}_{g_{<i}} )^\dagger \ot (I - B^{\bu}_{[f(\bx_i) =
    g_i(\bu)]}) \ket{\psi}\nonumber\\
& \approx_{\sqrt{\nu_4}} \E_{\bu}  \E_{\bx_1, \dots, \bx_i} \sum_{g_1, \dots, g_i: g_i \neq \bot} \bra{\psi}\widehat{G}^{\bx_i}_{g_i} \cdot
  \widehat{G}^{\bx_{<i}}_{g_{<i}}  \cdot  \widehat{G}^{\bx_i}_{g_i}\cdot (\widehat{G}^{\bx_{<i}}_{g_{<i}} )^\dagger \ot (I - B^{\bu}_{[f(\bx_i) =
    g_i(\bu)]}) \ket{\psi}.\label{eq:gonna-need-a-bigger-cauchy-schwarz}
\end{align}
To show this, we bound the magnitude of the difference.
\begin{multline*}
\E_{\bu}  \E_{\bx_1, \dots, \bx_i} \sum_{g_1, \dots, g_i: g_i \neq \bot}  \bra{\psi}
  ((\widehat{G}^{\bx_{<i}}_{g_{<i}}  \cdot \widehat{G}^{\bx_i}_{g_i} - \widehat{G}^{\bx_i}_{g_i}  \cdot \widehat{G}^{\bx_{<i}}_{g_{<i}}) \otimes I) \cdot (\widehat{G}^{\bx_i}_{g_i} \cdot (\widehat{G}^{\bx_{<i}}_{g_{<i}} )^\dagger \ot (I - B^{\bu}_{[f(\bx_i) =
    g_i(\bu)]})) \ket{\psi}\\
    \leq 
    \sqrt{\E_{\bu}  \E_{\bx_1, \dots, \bx_i} \sum_{g_1, \dots, g_i: g_i \neq \bot}   \bra{\psi}
  ((\widehat{G}^{\bx_{<i}}_{g_{<i}}  \cdot \widehat{G}^{\bx_i}_{g_i} - \widehat{G}^{\bx_i}_{g_i}  \cdot \widehat{G}^{\bx_{<i}}_{g_{<i}})  \cdot  (\widehat{G}^{\bx_i}_{g_i}\cdot (\widehat{G}^{\bx_{<i}}_{g_{<i}})^\dagger - (\widehat{G}^{\bx_{<i}}_{g_{<i}})^\dagger \cdot \widehat{G}^{\bx_i}_{g_i})) \otimes I \ket{\psi}}\\
  \cdot \sqrt{\E_{\bu}  \E_{\bx_1, \dots, \bx_i} \sum_{g_1, \dots, g_i: g_i \neq \bot}  \bra{\psi}
  (\widehat{G}^{\bx_{<i}}_{g_{<i}}  \cdot \widehat{G}^{\bx_i}_{g_i} \cdot (\widehat{G}^{\bx_{<i}}_{g_{<i}} )^\dagger) \ot (I - B^{\bu}_{[f(\bx_i) =
    g_i(\bu)]})^2 \ket{\psi}}.
\end{multline*}
The term inside the first square root is at most 
\begin{equation}\label{eq:add-in-the-bot}
\E_{\bu}  \E_{\bx_1, \dots, \bx_i} \sum_{g_1, \dots, g_i}   \bra{\psi}
  ((\widehat{G}^{\bx_{<i}}_{g_{<i}}  \cdot \widehat{G}^{\bx_i}_{g_i} - \widehat{G}^{\bx_i}_{g_i}  \cdot \widehat{G}^{\bx_{<i}}_{g_{<i}})  \cdot  (\widehat{G}^{\bx_i}_{g_i}\cdot (\widehat{G}^{\bx_{<i}}_{g_{<i}})^\dagger - (\widehat{G}^{\bx_{<i}}_{g_{<i}})^\dagger \cdot \widehat{G}^{\bx_i}_{g_i})) \otimes I \ket{\psi},
\end{equation}
which is at most~$\nu_4$ by \Cref{lem:commute-g-half-sandwich}.
The term inside the second square root is at most~$1$ because~$B$ and~$\widehat{G}$ are measurements.
Next, we claim that
\begin{equation}\label{eq:even-bigger-CS}
\eqref{eq:gonna-need-a-bigger-cauchy-schwarz}
\approx_{\sqrt{\nu_4}}
\E_{\bu}  \E_{\bx_1, \dots, \bx_i} \sum_{g_1, \dots, g_i: g_i \neq \bot} \bra{\psi}\widehat{G}^{\bx_i}_{g_i} \cdot
  \widehat{G}^{\bx_{<i}}_{g_{<i}}  \cdot (\widehat{G}^{\bx_{<i}}_{g_{<i}} )^\dagger \cdot  \widehat{G}^{\bx_i}_{g_i}\ot (I - B^{\bu}_{[f(\bx_i) =
    g_i(\bu)]}) \ket{\psi}.
\end{equation}
To show this, we bound the magnitude of the difference.
\begin{multline*}
\E_{\bu}  \E_{\bx_1, \dots, \bx_i} \sum_{g_1, \dots, g_i: g_i \neq \bot}  \bra{\psi}
   (\widehat{G}^{\bx_i}_{g_i} \cdot \widehat{G}^{\bx_{<i}}_{g_{<i}} \ot (I - B^{\bu}_{[f(\bx_i) =
    g_i(\bu)]})) \cdot (((\widehat{G}^{\bx_{<i}}_{g_{<i}})^\dagger  \cdot \widehat{G}^{\bx_i}_{g_i} - \widehat{G}^{\bx_i}_{g_i}  \cdot (\widehat{G}^{\bx_{<i}}_{g_{<i}})^\dagger) \otimes I) \ket{\psi}\\
    \leq 
    \sqrt{\E_{\bu}  \E_{\bx_1, \dots, \bx_i} \sum_{g_1, \dots, g_i: g_i \neq \bot}  \bra{\psi}
  (\widehat{G}^{\bx_i}_{g_i} \cdot\widehat{G}^{\bx_{<i}}_{g_{<i}}  \cdot  (\widehat{G}^{\bx_{<i}}_{g_{<i}} )^\dagger \cdot \widehat{G}^{\bx_i}_{g_i}) \ot (I - B^{\bu}_{[f(\bx_i) =
    g_i(\bu)]})^2 \ket{\psi}}\\
    \cdot
    \sqrt{\E_{\bu}  \E_{\bx_1, \dots, \bx_i} \sum_{g_1, \dots, g_i: g_i \neq \bot}   \bra{\psi}
  ((\widehat{G}^{\bx_{<i}}_{g_{<i}}  \cdot \widehat{G}^{\bx_i}_{g_i} - \widehat{G}^{\bx_i}_{g_i}  \cdot \widehat{G}^{\bx_{<i}}_{g_{<i}})  \cdot  ((\widehat{G}^{\bx_{<i}}_{g_{<i}})^\dagger  \cdot \widehat{G}^{\bx_i}_{g_i} - \widehat{G}^{\bx_i}_{g_i}  \cdot (\widehat{G}^{\bx_{<i}}_{g_{<i}}))^\dagger) \otimes I \ket{\psi}}.
\end{multline*}
The term inside the first square root is at most~$1$ because~$B$ and $\widehat{G}$ are measurements.
The term inside the second square root is at most \Cref{eq:add-in-the-bot},
which is at most~$\nu_4$ by \Cref{lem:commute-g-half-sandwich}.
But
\begin{align*}
\sum_{g_1, \ldots, g_{i-1}} \widehat{G}^{x_{<i}}_{g_{<i}} \cdot (\widehat{G}^{x_{<i}}_{g_{<i}})^\dagger
&= \sum_{g_1, \ldots, g_{i-1}} \widehat{G}^{x_1}_{g_1} \cdots \widehat{G}^{x_{i-1}}_{g_{i-1}} \cdot \widehat{G}^{x_{i-1}}_{g_{i-1}}  \cdots \widehat{G}^{x_{1}}_{g_{1}}\\
&= \sum_{g_1, \ldots, g_{i-2}} \widehat{G}^{x_1}_{g_1} \cdots \Big(\sum_{g_{i-1}} \widehat{G}^{x_{i-1}}_{g_{i-1}} \Big)  \cdots \widehat{G}^{x_{1}}_{g_{1}}\\
&= \sum_{g_1, \ldots, g_{i-2}} \widehat{G}^{x_1}_{g_1} \cdots I  \cdots \widehat{G}^{x_{1}}_{g_{1}}\\
& \cdots\\
& = I.
\end{align*}
Thus,
\begin{equation*}
\eqref{eq:even-bigger-CS}
= 
\E_{\bu}  \E_{\bx_1, \dots, \bx_i} \sum_{g_i: g_i \neq \bot} \bra{\psi}\widehat{G}^{\bx_i}_{g_i} \ot (I - B^{\bu}_{[f(\bx_i) =
    g_i(\bu)]}) \ket{\psi} \leq \nu_1,
\end{equation*}
by \Cref{eq:ld-gbcon}.
In total, using $\sqrt{426} \leq 21$, this gives an error of
\begin{align*}
\nu_1 + 2 \sqrt{\nu_4}
& = \zeta + \sqrt{8m\eps + 4\delta} + 2\cdot \sqrt{426 k^2 m \cdot \left(\gamma^{1/16} + \zeta^{1/16} +  (d/q)^{1/16}\right)}\\
& \leq \zeta^{1/32} + 3m \eps^{1/32} + 2 \delta^{1/32} + 42 km \cdot \left(\gamma^{1/32} + \zeta^{1/32} + (d/q)^{1/32}\right)\\
& \leq 43 km \cdot \left(\eps^{1/32} + \delta^{1/32} +\gamma^{1/32} + \zeta^{1/32} +  (d/q)^{1/32}\right).
\end{align*}
This completes the proof.
\end{proof}

\subsection{Consistency of~$H$ with~$A$}
\label{sec:consistency-of-h-with-a}

\begin{lemma}[Consistency of~$H$ with~$B$]\label{lem:h-b-consistency}
\begin{equation*}
H_{[h|_u =f]} \otimes I \simeq_{\nu_6} I \otimes B^u_f,
\end{equation*}
where
\begin{equation*}
\nu_6 = 44 k^2m \cdot \left(\eps^{1/32} + \delta^{1/32} + \gamma^{1/32} + \zeta^{1/32} + (d/q)^{1/32}\right).
\end{equation*}
\end{lemma}
\begin{proof}
Let $\bx_1, \ldots, \bx_k \sim \F_q$ be independent and uniformly random.
Our goal is to bound
\begin{align}
&\E_{\bu} \sum_{f \neq f'} \bra{\psi} H_{[h|_{\bu} = f']} \otimes B^{\bu}_f \ket{\psi}\nonumber\\
=~&\E_{\bu} \sum_{h}\sum_{f \neq h|_{\bu}} \bra{\psi} H_{h} \otimes B^{\bu}_f \ket{\psi}\nonumber\\
=~&\E_{\bu}\E_{\bx_1, \ldots, \bx_k} \sum_{h}\sum_{f \neq h|_{\bu}} \bra{\psi} H^{\bx_1, \ldots, \bx_k}_{h} \otimes B^{\bu}_f \ket{\psi}\nonumber\\
=~&\E_{\bu}\E_{\bx_1, \ldots, \bx_k} \sum_{h} \sum_{w: |w| \geq d+1} \sum_{f \neq h|_{\bu}} \bra{\psi} \widehat{H}^{\bx_1, \ldots, \bx_k}_{h_w} \otimes B^{\bu}_f \ket{\psi}\nonumber\\
=~&\E_{\bu}\E_{\bx_1, \ldots, \bx_k} \sum_{h} \sum_{w: |w| \geq d+1} \sum_{(g_1, \ldots, g_k) = h_w} \sum_{f \neq h|_{\bu}} \bra{\psi} \widehat{H}^{\bx_1, \ldots, \bx_k}_{g_1, \ldots, g_k} \otimes B^{\bu}_f \ket{\psi}.\label{eq:keep-on-expandin}
\end{align}
We note that the sum over $(g_1, \ldots, g_k) = h_w$ in the final step is trivial
because there is only ever one tuple $(g_1, \ldots, g_k)$ which is equal to $h_w$.
Because $|w| \geq d+1$,
there exist at least $(d+1)$ coordinates~$i$ such that $g_i \neq \bot$
and hence $g_i = h|_{\bx_i}$.
Since~$f$ is degree-$d$, if it is not equal to $h|_{\bu}$,
then there must exist an $i$ such that $g_i \neq \bot$ and $g_i(\bu) \neq f(\bx_i)$.
Thus,
\begin{align}
\eqref{eq:keep-on-expandin}
~=~& \E_{\bu}\E_{\bx_1, \ldots, \bx_k} \sum_{h} \sum_{w: |w| \geq d+1} \sum_{(g_1, \ldots, g_k) = h_w} \sum_{\substack{f: \exists i: g_i \neq \bot, \\g_i(\bu) \neq f(\bx_i)} }\bra{\psi} \widehat{H}^{\bx_1, \ldots, \bx_k}_{g_1, \ldots, g_k} \otimes B^{\bu}_f \ket{\psi}\nonumber\\
\leq~&\E_{\bu}\E_{\bx_1, \ldots, \bx_k} \sum_{g_1, \ldots, g_k} \sum_{\substack{f: \exists i: g_i \neq \bot, \\g_i(\bu) \neq f(\bx_i)}} \bra{\psi} \widehat{H}^{\bx_1, \ldots, \bx_k}_{g_1, \ldots, g_k} \otimes B^{\bu}_f \ket{\psi}.\label{eq:keep-on-contractin}
\end{align}
Let $(\by_1, \ldots, \by_k) \sim \distinct{k}$. Then by \Cref{prop:ld-dnoteq}, \Cref{eq:keep-on-contractin}
is $(k^2/q)$-close to
\begin{align*}
&\E_{\bu}\E_{\by_1, \ldots, \by_k} \sum_{g_1, \ldots, g_k} \sum_{\substack{f: \exists i: g_i \neq \bot, \\g_i(\bu) \neq f(\by_i)}} \bra{\psi} \widehat{H}^{\by_1, \ldots, \by_k}_{g_1, \ldots, g_k} \otimes B^{\bu}_f \ket{\psi}\\
\leq~&
\sum_i \E_{\bu}\E_{\by_1, \ldots, \by_k} \sum_{g_1, \ldots, g_k} \sum_{\substack{f: g_i \neq \bot, \\g_i(\bu) \neq f(\by_i)}} \bra{\psi} \widehat{H}^{\by_1, \ldots, \by_k}_{g_1, \ldots, g_k} \otimes B^{\bu}_f \ket{\psi} \tag{by the union bound}\\
=~&
\sum_i \E_{\bu}\E_{\by_1, \ldots, \by_k} \sum_{g_1, \ldots, g_k:g_i \neq \bot} \sum_{a \neq g_i(\bu)} \bra{\psi} \widehat{H}^{\by_1, \ldots, \by_k}_{g_1, \ldots, g_k} \otimes B^{\bu}_{[f(\by_i)=a]} \ket{\psi}\\
\leq~& \sum_i \nu_5 \tag{by \Cref{lem:ld-sandwich-line-one-point}}\\
=~& k \cdot \nu_5.
\end{align*}
In total, using $1/q \leq (d/q)^{1/32}$, this gives an error of
\begin{align*}
\frac{k^2}{q} + k \cdot \nu_5
& = \frac{k^2}{q} + k \cdot 43 km \cdot \left(\eps^{1/32} + \delta^{1/32} + \gamma^{1/32} + \zeta^{1/32} + (d/q)^{1/32}\right)\\
& = 44 k^2m \cdot \left(\eps^{1/32} + \delta^{1/32} + \gamma^{1/32} + \zeta^{1/32} + (d/q)^{1/32}\right).
\end{align*}
This completes the proof.
\end{proof}

\begin{corollary}[Consistency of~$H$ with~$A$; Proof of \Cref{item:ld-pasting-N-consistency-sub-measurement} in \Cref{lem:ld-pasting-sub-measurement}]
\begin{equation*}
H_{[h(u, x)=a]} \otimes I \simeq_{\nu} I \otimes A^{u,x}_a.
\end{equation*}
\end{corollary}
\begin{proof}
\Cref{lem:h-b-consistency} implies that
\begin{equation}\label{eq:h-b-consistency-at-a-point}
H_{[h(u, x)=a]} \otimes I \simeq_{\nu_6} I \otimes B^{u}_{[f(x)=a]}.
\end{equation}
\Cref{prop:triangle-sub} applied to \Cref{eq:h-b-consistency-at-a-point} and \Cref{eq:ld-abcon} implies that
\begin{equation*}
H_{[h(u, x)=a]} \otimes I \simeq_{\nu_6 + \sqrt{8m \eps + 4\delta}} I \otimes A^{u,x}_a.
\end{equation*}
We can bound this error by
\begin{align*}
\sqrt{8m \eps + 4\delta} + \nu_6
& = \sqrt{8m \eps + 4\delta}
	+ 44 k^2m \cdot \left(\eps^{1/32} + \delta^{1/32} +  \gamma^{1/32} +\zeta^{1/32} + (d/q)^{1/32}\right)\\
& \leq 3m \eps^{1/32} + 2\delta^{1/32}
	+ 44 k^2m \cdot \left(\eps^{1/32} + \delta^{1/32} +  \gamma^{1/32} + \zeta^{1/32} +(d/q)^{1/32}\right)\\
& \leq  47 k^2m \cdot \left(\eps^{1/32} + \delta^{1/32} +  \gamma^{1/32} + \zeta^{1/32} +(d/q)^{1/32}\right).
\end{align*}
This is clearly at most~$\nu$, which completes the proof.
\end{proof}

\subsection{Completeness of~$H$}
\label{sec:completeness-of-h-low-degree}

\begin{definition}
Let $\tau \in \{0, 1\}^k$ be a type.
We define the following two subsets of $\calP^+(m,q,d)^k$.
\begin{itemize}
\item
We define $\mathsf{Outcomes}_{\tau}$ to be the set of tuples $(g_1, \ldots, g_k)$ such that $g_i  \in \polyfunc{m}{q}{d}$ for each $i \in \tau$ and $g_i = \bot$ for each $i \notin \tau$.
This is the set of possible outcomes of the $\widehat{H}$ measurement of type~$\tau$.
\item
Let $x_1, \ldots, x_k \in \F_q$.
We define $\mathsf{Global}_{\tau}(x)$ to be the subset of $\mathsf{Outcomes}_{\tau}$ containing only those tuples which are consistent with a global polynomial. In other words, $(g_1, \ldots, g_k) \in \mathsf{Global}_{\tau}(x)$ if there exists an $h \in \polyfunc{m+1}{q}{d}$ such that $g_i = h|_{x_i}$ for each $i \in \tau$.
Next, we define
\begin{equation*}
\overline{\mathsf{Global}_{\tau}(x)} = \mathsf{Outcomes}_{\tau} \setminus \mathsf{Global}_{\tau}(x).
\end{equation*}
This contains those tuples of type~$\tau$ with no consistent global polynomial.
\end{itemize}
\end{definition}

\begin{lemma}
\label{lem:over-all-outcomes}
Let $\bx_1, \ldots, \bx_k \sim \F_q$ be sampled uniformly at random.
Then
\begin{equation*}
\bra{\psi} H \otimes I \ket{\psi}
\approx_{\nu_7}  \E_{\bx_1, \ldots, \bx_k} \sum_{\tau:|\tau| \geq d+1} \sum_{(g_1, \ldots, g_k) \in \mathsf{Outcomes}_{\tau}} \bra{\psi} \widehat{H}^{\bx_1, \ldots, \bx_k}_{g_1, \ldots, g_k} \otimes I \ket{\psi},
\end{equation*}
where
\begin{equation*}
\nu_7 = 46 k^2m \cdot \left(\eps^{1/32} + \delta^{1/32} +\gamma^{1/32} +  \zeta^{1/32} + (d/q)^{1/32}\right).
\end{equation*}
\end{lemma}
\begin{proof}
Let $(\by_1, \ldots, \by_k) \sim \mathsf{Distinct}_k$.
By definition,
\begin{align}
\bra{\psi} H \otimes I \ket{\psi}
&= \sum_h \bra{\psi} H_h \otimes I \ket{\psi}\nonumber\\
&= \E_{\by_1, \ldots, \by_k} \sum_h \bra{\psi} H^{\by_1, \ldots, \by_k}_h \otimes I \ket{\psi}\nonumber\\
&= \E_{\by_1, \ldots, \by_k} \sum_h \sum_{\tau:|\tau| \geq d+1} \bra{\psi} \widehat{H}^{\by_1, \ldots, \by_k}_{h_{\tau}} \otimes I \ket{\psi}\nonumber\\
&= \E_{\by_1, \ldots, \by_k} \sum_h \sum_{\tau:|\tau| \geq d+1} \sum_{(g_1, \ldots, g_k) = h_\tau} \bra{\psi} \widehat{H}^{\by_1, \ldots, \by_k}_{g_1, \ldots, g_k} \otimes I \ket{\psi}\nonumber\\
&= \E_{\by_1, \ldots, \by_k}  \sum_{\tau:|\tau| \geq d+1} \sum_{(g_1, \ldots, g_k) \in \mathsf{Global}_{\tau}(\by)} \bra{\psi} \widehat{H}^{\by_1, \ldots, \by_k}_{g_1, \ldots, g_k} \otimes I \ket{\psi}. \label{eq:sum-restricted-to-global-polynomial}
\end{align}
The sum in \Cref{eq:sum-restricted-to-global-polynomial}
is over $g_1, \ldots, g_k$ which are consistent with a global polynomial.
We will now show that the value of this sum
remains largely unchanged if we drop this condition.
In particular, we claim that
\begin{equation}\label{eq:remove-the-restriction}
\eqref{eq:sum-restricted-to-global-polynomial}
\approx_{\frac{k^2}{q} + k \cdot\nu_5 + \frac{md}{q}}
\E_{\by_1, \ldots, \by_k}  \sum_{\tau:|\tau| \geq d+1} \sum_{(g_1, \ldots, g_k) \in \mathsf{Outcomes}_{\tau}} \bra{\psi} \widehat{H}^{\by_1, \ldots, \by_k}_{g_1, \ldots, g_k} \otimes I \ket{\psi}.
\end{equation}
To show this, we note that $\eqref{eq:remove-the-restriction} \geq \eqref{eq:sum-restricted-to-global-polynomial}$. Hence, it suffices to upper bound their difference.
\begin{align}
\eqref{eq:remove-the-restriction} - \eqref{eq:sum-restricted-to-global-polynomial}
&= \E_{\by_1, \ldots, \by_k} \sum_{\tau:|\tau| \geq d+1} \sum_{(g_1, \ldots, g_k) \in \overline{\mathsf{Global}_{\tau}(\by)}} \bra{\psi} \widehat{H}^{\by_1, \ldots, \by_k}_{g_1, \ldots, g_k} \otimes I \ket{\psi}\nonumber\\
&=  \E_{\bu} \E_{\by_1, \ldots, \by_k}  \sum_{\tau:|\tau| \geq d+1} \sum_f \sum_{(g_1, \ldots, g_k) \in \overline{\mathsf{Global}_{\tau}(\by)}} \bra{\psi} \widehat{H}^{\by_1, \ldots, \by_k}_{g_1, \ldots, g_k} \otimes B^{\bu}_f \ket{\psi},\label{eq:B-appears-out-of-thin-air}
\end{align}
because~$B$ is a measurement. Next, we claim that
\begin{equation}\label{eq:add-in-indicator-for-f-and-g}
\eqref{eq:B-appears-out-of-thin-air}
\approx_{\frac{k^2}{q} + k \cdot\nu_5} 
\E_{\bu} \E_{\by_1, \ldots, \by_k}   \sum_{\tau:|\tau| \geq d+1} \sum_f \sum_{(g_1, \ldots, g_k) \in \overline{\mathsf{Global}_{\tau}(\by)}} \bra{\psi} \widehat{H}^{\by_1, \ldots, \by_k}_{g_1, \ldots, g_k} \otimes B^{\bu}_f \ket{\psi} \cdot \bone[\forall i\in \tau, f(\by_i) = g_i(\bu)].
\end{equation}
To show this, we note that $\eqref{eq:B-appears-out-of-thin-air} \geq \eqref{eq:add-in-indicator-for-f-and-g}$.
Thus, it suffices to upper bound their difference $\eqref{eq:B-appears-out-of-thin-air} - \eqref{eq:add-in-indicator-for-f-and-g}$. This is given by
\begin{equation}\label{eq:about-to-swap-x-for-y}
\E_{\bu} \E_{\by_1, \ldots, \by_k}  \sum_{\tau: |\tau|\geq d+1} \sum_f  \sum_{(g_1, \ldots, g_k) \in \overline{\mathsf{Global}_{\tau}(\by)}} \bra{\psi} \widehat{H}^{\by_1, \ldots, \by_k}_{g_1, \ldots, g_k} \otimes B^{\bu}_f \ket{\psi} \cdot \bone[\exists i \in \tau, f(\by_i) \neq g_i(\bu)].
\end{equation}
Recall that $\bx_1, \ldots, \bx_k \sim \F_q$ are sampled independently and uniformly at random.
Then \Cref{prop:ld-dnoteq} implies that \Cref{eq:about-to-swap-x-for-y} is $(k^2/q)$-close to
\begin{align*}
 &\E_{\bu} \E_{\bx_1, \ldots, \bx_k}  \sum_{\tau: |\tau|\geq d+1} \sum_f  \sum_{(g_1, \ldots, g_k) \in \overline{\mathsf{Global}_{\tau}(\bx)}} \bra{\psi} \widehat{H}^{\bx_1, \ldots, \bx_k}_{g_1, \ldots, g_k} \otimes B^{\bu}_f \ket{\psi} \cdot \bone[\exists i \in \tau, f(\bx_i) \neq g_i(\bu)]\\
 \leq~& \E_{\bu} \E_{\bx_1, \ldots, \bx_k}   \sum_{\tau:|\tau| \geq d+1} \sum_f \sum_{(g_1, \ldots, g_k) \in \mathsf{Outcomes}_{\tau}} \bra{\psi} \widehat{H}^{\bx_1, \ldots, \bx_k}_{g_1, \ldots, g_k} \otimes B^{\bu}_f \ket{\psi} \cdot \bone[\exists i \in \tau, f(\bx_i) \neq g_i(\bu)]\\
 \leq~& \E_{\bu} \E_{\bx_1, \ldots, \bx_k}   \sum_{\tau} \sum_f \sum_{(g_1, \ldots, g_k) \in \mathsf{Outcomes}_{\tau}} \bra{\psi} \widehat{H}^{\bx_1, \ldots, \bx_k}_{g_1, \ldots, g_k} \otimes B^{\bu}_f \ket{\psi} \cdot \bone[\exists i \in \tau, f(\bx_i) \neq g_i(\bu)]\\
 \leq~&  \E_{\bu} \E_{\bx_1, \ldots, \bx_k}  \sum_{\tau} \sum_f \sum_{(g_1, \ldots, g_k) \in \mathsf{Outcomes}_{\tau}} \bra{\psi} \widehat{H}^{\bx_1, \ldots, \bx_k}_{g_1, \ldots, g_k} \otimes B^{\bu}_f \ket{\psi} \cdot \Big(\sum_{i\in \tau} \bone[f(\bx_i) \neq g_i(\bu)]\Big)\\
 =~& \sum_i  \E_{\bu} \E_{\bx_1, \ldots, \bx_k}  \sum_{g_1, \ldots, g_k : g_i \neq \bot} \sum_{f: f(\bx_i) \neq g_i(\bu)} \bra{\psi} \widehat{H}^{\bx_1, \ldots, \bx_k}_{g_1, \ldots, g_k} \otimes B^{\bu}_f \ket{\psi}\\
 =~& \sum_i  \E_{\bu} \E_{\bx_1, \ldots, \bx_k}  \sum_{g_1, \ldots, g_k : g_i \neq \bot} \sum_{a \neq g_i(\bu)} \bra{\psi} \widehat{H}^{\bx_1, \ldots, \bx_k}_{g_1, \ldots, g_k} \otimes B^{\bu}_{[f(\bx_i)=a]} \ket{\psi}\\
 \leq~& \sum_i (\nu_5) \tag{by \Cref{lem:ld-sandwich-line-one-point}}\\
 =~& k \cdot \nu_5.
\end{align*}

Returning to \Cref{eq:add-in-indicator-for-f-and-g},
we introduce the  notation $\mathsf{Consistent}_{\tau}(g, \by, \bu)$
to indicate whether there is a consistent degree-$d$ polynomial
interpolating the $g_i$'s along the line in direction~$u$.
In other words,
\begin{equation*}
\mathsf{Consistent}_{\tau}(g, \by, \bu)
= \left\{\begin{array}{rl}
	1 & \text{if $\exists f$ such that $f(\by_i) = g_i(\bu)$ for all $i \in \tau$,}\\
	0 & \text{otherwise.}
	\end{array}\right.
\end{equation*}
Clearly, for any $f$,
\begin{equation*}
\bone[\forall i\in \tau, f(\by_i) = g_i(\bu)] \leq \bone[\mathsf{Consistent}_{\tau}(g, \by, \bu)]
\end{equation*}
because the left-hand side indicates whether~$f$ is the consistent degree-$d$ polynomial
interpolating along direction~$u$.
As a result,
\begin{align}
\eqref{eq:add-in-indicator-for-f-and-g}
\leq~&
\E_{\bu} \E_{\by_1, \ldots, \by_k} \sum_{\tau:|\tau|\geq d+1} \sum_f  \sum_{(g_1, \ldots, g_k) \in \overline{\mathsf{Global}_{\tau}(\by)}} \bra{\psi} \widehat{H}^{\by_1, \ldots, \by_k}_{g_1, \ldots, g_k} \otimes B^{\bu}_f \ket{\psi} \cdot \bone[\mathsf{Consistent}_{\tau}(g,\by,\bu)]\nonumber\\
=~&
\E_{\bu} \E_{\by_1, \ldots, \by_k} \sum_{\tau:|\tau|\geq d+1} \sum_{(g_1, \ldots, g_k) \in \overline{\mathsf{Global}_{\tau}(\by)}} \bra{\psi} \widehat{H}^{\by_1, \ldots, \by_k}_{g_1, \ldots, g_k} \otimes \Big(\sum_f B^{\bu}_f\Big) \ket{\psi} \cdot \bone[\mathsf{Consistent}_{\tau}(g,\by,\bu)]\nonumber\\
=~&
\E_{\bu} \E_{\by_1, \ldots, \by_k}   \sum_{\tau:|\tau|\geq d+1} \sum_{(g_1, \ldots, g_k)  \in \overline{\mathsf{Global}_{\tau}(\by)}} \bra{\psi} \widehat{H}^{\by_1, \ldots, \by_k}_{g_1, \ldots, g_k} \otimes I \ket{\psi} \cdot \bone[\mathsf{Consistent}_{\tau}(g,\by,\bu)]\nonumber\\
=~&
 \E_{\by_1, \ldots, \by_k}   \sum_{\tau:|\tau|\geq d+1} \sum_{(g_1, \ldots, g_k) \in \overline{\mathsf{Global}_{\tau}(\by)}} \bra{\psi} \widehat{H}^{\by_1, \ldots, \by_k}_{g_1, \ldots, g_k} \otimes I \ket{\psi} \cdot \Pr_{\bu}[\mathsf{Consistent}_{\tau}(g,\by,\bu)].\label{eq:consistent-indicator}
\end{align}

Let us now fix $(y_1, \ldots, y_k) \in \mathsf{Distinct}_k$,
a type $\tau$ such that $|\tau| \geq d+1$,
and $(g_1, \ldots, g_k) \in \overline{\mathsf{Global}_{\tau}(y)}$.
Suppose without loss of generality that $\tau_1 = \cdots = \tau_{d+1} = 1$,
so that $g_1, \ldots, g_{d+1} \in \polyfunc{m}{q}{d}$.
Then there is a unique polynomial $h^* \in \polyfunc{m+1}{q}{d}$ which interpolates $g_1, \ldots, g_{d+1}$.
In other words, for all $1 \leq i \leq d+1$,
\begin{equation*}
h^*(u, y_i) = g_i(u).
\end{equation*}
In addition, for any $u$, 
$(h^*)|_u$ is the unique degree-$d$ polynomial
which interpolates $g_1, \ldots, g_{d+1}$ along the line in direction~$u$.
Thus, if there is a consistent degree-$d$ polynomial interpolating all the $g_i$'s along the line in direction~$u$,
then it is $(h^*)|_u$.
In math,
\begin{equation*}
\mathsf{Consistent}_{\tau}(g, y,u) = 1
\qquad \text{if and only if} \qquad
\forall i \in \tau, g_i(u) = h^*(u, y_i).
\end{equation*}
On the other hand, because $(g_1, \ldots, g_k) \in \overline{\mathsf{Global}_{\tau}(y)}$, 
there exists an $i^* \in \tau$ such that $g_{i^*} \neq (h^*)|_{y_{i^*}}$.
Hence,
\begin{equation*}
\Pr_{\bu}[\mathsf{Consistent}_{\tau}(g, y,\bu)]
\leq \Pr_{\bu}[g_{i^*}(\bu) = h^*(\bu, y_{i^*})]
= \Pr_{\bu}[g_{i^*}(\bu) = (h^*)|_{y_{i^*}}(\bu)]
\leq \frac{md}{q},
\end{equation*}
by Schwartz-Zippel.
As a result,
\begin{align*}
\eqref{eq:consistent-indicator}
& \leq 
 \E_{\by_1, \ldots, \by_k}   \sum_{\tau:|\tau|\geq d+1} \sum_{(g_1, \ldots, g_k) \in \overline{\mathsf{Global}_{\tau}(\by)}} \bra{\psi} \widehat{H}^{\by_1, \ldots, \by_k}_{g_1, \ldots, g_k} \otimes I \ket{\psi} \cdot \frac{md}{q}\\
 &\leq
  \E_{\by_1, \ldots, \by_k}   \sum_{g_1, \ldots, g_k} \bra{\psi} \widehat{H}^{\by_1, \ldots, \by_k}_{g_1, \ldots, g_k} \otimes I \ket{\psi} \cdot \frac{md}{q}\\
  & = \frac{md}{q}.
\end{align*}
This establishes \Cref{eq:remove-the-restriction}.
Finally, \Cref{prop:ld-dnoteq} implies that \Cref{eq:remove-the-restriction} is $(k^2/q)$-close to
\begin{equation*}
\E_{\bx_1, \ldots, \bx_k}  \sum_{\tau:|\tau| \geq d+1} \sum_{(g_1, \ldots, g_k) \in \mathsf{Outcomes}_{\tau}} \bra{\psi} \widehat{H}^{\bx_1, \ldots, \bx_k}_{g_1, \ldots, g_k} \otimes I \ket{\psi}.
\end{equation*}
In total, this gives an error of
\begin{align*}
2 \frac{k^2}{q}  + \frac{md}{q}+ k \cdot \nu_5
& = 2 \frac{k^2}{q}  + \frac{md}{q} + k \cdot 43 km \cdot \left(\eps^{1/32} + \delta^{1/32} + \gamma^{1/32} + \zeta^{1/32} +  (d/q)^{1/32}\right)\\
& \leq 3 \frac{k^2 md}{q}  +  43 k^2m \cdot \left(\eps^{1/32} + \delta^{1/32} +\gamma^{1/32} +\zeta^{1/32} +  (d/q)^{1/32}\right)\\
& \leq 3 k^2 m (d/q)^{1/32}  +  43 k^2m \cdot \left(\eps^{1/32} + \delta^{1/32} +\gamma^{1/32} +  \zeta^{1/32} + (d/q)^{1/32}\right)\\
& \leq  46 k^2m \cdot \left(\eps^{1/32} + \delta^{1/32} + \gamma^{1/32} +  \zeta^{1/32} +(d/q)^{1/32}\right).
\end{align*}
This completes the proof of the lemma.
\end{proof}

\begin{lemma}
\label{lem:from-H-to-G}
Let $\bx_1, \ldots, \bx_k \sim \F_q$ be sampled uniformly at random.
Then
\begin{equation*}
\E_{\bx_1, \dots, \bx_k} \sum_{\tau:|\tau| \geq d+1} 
	\sum_{(g_1, \ldots, g_{k}) \in \outc_{\tau}} \bra{\psi} \widehat{H}^{\bx_1, \dots, \bx_k}_{g_1, \dots, g_k} \ot I \ket{\psi}
\approx_{\nu_8}
\sum_{i = d+1}^k \binom{k}{i} \bra{\psi} G^i (I-G)^{k-i} \otimes I \ket{\psi},
\end{equation*}
where
\begin{equation*}
\nu_8 = 46 k m \cdot \left( \gamma^{1/32} + \zeta^{1/32} +(d/q)^{1/32}\right).
\end{equation*}
\end{lemma}
\begin{proof}
We begin by introducing some notation that we will use throughout the proof.
Let $\tau \in \{0, 1\}^k$ be a type.
Then we define
\begin{equation*}
\tau_{< \ell}  = (\tau_1, \ldots, \tau_{\ell-1}) \in \{0, 1\}^{\ell-1},
\quad
\tau_{> \ell}  = (\tau_{\ell+1}, \ldots, \tau_{k}) \in \{0,1\}^{k-\ell},
\end{equation*}
and we define $\tau_{\leq \ell}$ and $\tau_{\geq \ell}$ similarly.
In addition, given $(g_1, \ldots, g_k) \in \calP^+(m, q, d)^{k}$, we define
\begin{equation*}
g_{< \ell}  = (g_1, \ldots, g_{\ell-1}) \in \calP^+(m, q, d)^{\ell-1},
\quad
g_{> \ell}  = (g_{\ell+1}, \ldots, g_{k}) \in \calP^+(m, q, d)^{k-\ell}
\end{equation*}
and we define $g_{\leq \ell}$ and $g_{\geq \ell}$ similarly.
Using this notation, we can write
\begin{equation*}
 \wH^{x_{\geq \ell}}_{g_{\geq \ell}} = \wH^{x_{\ell}, \ldots, x_k}_{g_{\ell}, \ldots, g_k}.
\end{equation*}
Next, we introduce the notation
\begin{equation*}
\wG^{x_{\geq \ell}}_{g_{\geq \ell}} = \wG^{x_{\ell}}_{g_{\ell}} \cdots \wG^{x_k}_{g_k}.
\end{equation*}
This satisfies the recurrence relation
\begin{equation}\label{eq:G-recurrence}
\wG^{x_{\geq \ell}}_{g_{\geq \ell}} = \wG^{x_{\ell}}_{g_{\ell}} \cdot \wG^{x_{> \ell}}_{g_{> \ell}}.
\end{equation}
Furthermore, we can write
\begin{equation}\label{eq:split-H-into-two-Gs}
 \wH^{x_{\geq \ell}}_{g_{\geq \ell}} = (\wG^{x_{\geq \ell}}_{g_{\geq \ell}}) \cdot (\wG^{x_{\geq \ell}}_{g_{\geq \ell}})^\dagger.
\end{equation}
Finally, we will write $\sfO_{\tau}$ as shorthand for $\outc_{\tau}$.

To prove the lemma, we will show that for each $1 \leq \ell \leq k$,
\begin{align}
&\E_{\bx_{\geq \ell}} \sum_{\tau:|\tau| \geq d+1}  \sum_{g_{\geq \ell} \in \sfO_{\tau_{\geq \ell}}}
	\bra{\psi} \widehat{H}^{\bx_{\geq \ell}}_{g_{\geq \ell}} \ot (G^{|\tau_{<\ell}|}\cdot (I-G)^{(\ell-1)-|\tau_{<\ell}|}) \ket{\psi}\nonumber\\
\approx_{2\sqrt{2\zeta} + 2 \sqrt{\nu_4}}~&\E_{\bx_{>\ell}} \sum_{\tau:|\tau| \geq d+1} \sum_{g_{> \ell} \in \sfO_{\tau_{> \ell}}}
	\bra{\psi} \widehat{H}^{\bx_{> \ell}}_{g_{> \ell}} \ot (G^{|\tau_{\leq \ell}|} \cdot (I-G)^{\ell-|\tau_{\leq\ell}|}) \ket{\psi}.
	\label{eq:i-think-this-is-what-i'm-supposed-to-prove}
\end{align}
If we then repeatedly apply \Cref{eq:i-think-this-is-what-i'm-supposed-to-prove} for $\ell = 1, \ldots, k$,
we derive
\begin{align*}
&\E_{\bx_1, \dots, \bx_k} \sum_{\tau:|\tau| \geq d+1} 
	\sum_{(g_1, \ldots, g_{k}) \in \outc_{\tau}} \bra{\psi} \widehat{H}^{\bx_1, \dots, \bx_k}_{g_1, \dots, g_k} \ot I \ket{\psi}\\
=~&\E_{\bx_{\geq 1}} \sum_{\tau:|\tau| \geq d+1} 
	\sum_{g_{\geq 1} \in \sfO_{\tau}} \bra{\psi} \widehat{H}^{\bx_{\geq 1}}_{g_{\geq 1}} \ot I \ket{\psi}\\
\approx_{2\sqrt{2\zeta} + 2 \sqrt{\nu_4}}~&\E_{\bx_{\geq 2}} \sum_{\tau:|\tau| \geq d+1} 
	\sum_{g_{\geq 2} \in \sfO_{\tau_{\geq 2}}} \bra{\psi} \widehat{H}^{\bx_{\geq 2}}_{g_{\geq 2}} \ot (G^{|\tau_{\leq 1}|} \cdot (I-G)^{1-|\tau_{\leq1}|}) \ket{\psi}\\
\approx_{2\sqrt{2\zeta} + 2 \sqrt{\nu_4}}~&\E_{\bx_{\geq 3}} \sum_{\tau:|\tau| \geq d+1} 
	\sum_{g_{\geq 3} \in \sfO_{\tau_{\geq 3}}} \bra{\psi} \widehat{H}^{\bx_{\geq 3}}_{g_{\geq 3}} \ot (G^{|\tau_{\leq 2}|} \cdot (I-G)^{2-|\tau_{\leq2}|}) \ket{\psi}\\
	\cdots&\\
\approx_{2\sqrt{2\zeta} + 2 \sqrt{\nu_4}}~& \sum_{\tau:|\tau| \geq d+1} 
	 \bra{\psi}I  \ot (G^{|\tau|} \cdot (I-G)^{k-|\tau|}) \ket{\psi}\\
=~&\sum_{i = d+1}^k \binom{k}{i} \bra{\psi} I \otimes ( G^i (I-G)^{k-i} ) \ket{\psi}.
\end{align*}
In total, using $\sqrt{426}\leq 21$, this gives an error of
\begin{align*}
k \cdot(2\sqrt{2\zeta} + 2 \sqrt{\nu_4})
& = k \cdot 2 \sqrt{2 \zeta} + 2 \cdot \sqrt{426 k^2 m \cdot \left( \gamma^{1/16} +\zeta^{1/16} + (d/q)^{1/16}\right)}\\
& \leq 4k \cdot \zeta^{1/32} + 42 k m \cdot \left(\gamma^{1/32} + \zeta^{1/32} + (d/q)^{1/32}\right)\\
& \leq 46 k m \cdot \left(\gamma^{1/32} +\zeta^{1/32} +  (d/q)^{1/32}\right).
\end{align*}
This proves the lemma.

We now prove \Cref{eq:i-think-this-is-what-i'm-supposed-to-prove}.
To begin, for each $1 \leq \ell \leq k+1$ and $\tau_{\geq \ell} \in \{0, 1\}^{k - \ell + 1}$,
we define the matrix
\begin{equation}\label{eq:S-def}
S_{\tau_{\geq \ell}} = \sum_{\tau_{< \ell} : |\tau| \geq d+1} G^{|\tau_{<\ell}|}\cdot (I-G)^{(\ell-1)-|\tau_{<\ell}|}.
\end{equation}
Then the statement in \Cref{eq:i-think-this-is-what-i'm-supposed-to-prove} can be rewritten as
\begin{align}
&\E_{\bx_{\geq \ell}} \sum_{\tau_{\geq \ell}}  \sum_{g_{\geq \ell} \in \sfO_{\tau_{\geq \ell}}}
	\bra{\psi} \widehat{H}^{\bx_{\geq \ell}}_{g_{\geq \ell}} \ot S_{\tau_{\geq \ell}}\ket{\psi}\nonumber\\
\approx_{2\sqrt{2\zeta} + 2 \sqrt{\nu_4}}~&\E_{\bx_{> \ell}} \sum_{\tau_{>\ell}} \sum_{g_{>\ell} \in \sfO_{\tau_{> \ell}}}
	\bra{\psi} \widehat{H}^{\bx_{>\ell}}_{g_{>\ell}} \ot S_{\tau_{>\ell}} \ket{\psi}.
	\label{eq:i-think-this-is-what-i'm-supposed-to-prove-2}
\end{align}
To prove this, we will use several facts about $S_{\tau_{\geq \ell}}$.
First, $S$ is Hermitian and positive semidefinite.
This is because each term in \Cref{eq:S-def}
is a product of~$G$ and $(I-G)$.
These matrices commute with each other,
and both are Hermitian and positive semidefinite.
Next, $S$ is bounded:
\begin{align}
S_{\tau_{\geq \ell}}
&= \sum_{\tau_{< \ell} : |\tau| \geq d+1} G^{|\tau_{<\ell}|}\cdot (I-G)^{(\ell-1)-|\tau_{<\ell}|}\nonumber\\
&\leq \sum_{\tau_{< \ell} } G^{|\tau_{<\ell}|}\cdot (I-G)^{(\ell-1)-|\tau_{<\ell}|}\nonumber\\
& = (G + (I-G))^{\ell-1}\nonumber\\
& = I. \label{eq:S-bound}
\end{align}
In addition, for any $\tau_{\ell} \in \{0, 1\}$,
\begin{align}
\Big(\E_{\bx_{\ell}} \sum_{g_{\ell} \in \sfO_{\tau_{\ell}}} \wG^{\bx_{\ell}}_{g_{\ell}}\Big)
& = \left\{ \begin{array}{cl}
		G & \text{if } \tau_{\ell} = 1,\\
		(I-G) & \text{if } \tau_{\ell} = 0,
	\end{array}\right.\nonumber\\
&= G^{\tau_{\ell}} \cdot (I - G)^{1 - \tau_{\ell}}.\label{eq:explicit-formula-for-G-expectation}
\end{align}
Thus, for any $\tau_{> \ell}$,
\begin{align}
\sum_{\tau_{\ell}} 
	S_{\tau_{\geq \ell}} \cdot \Big(\E_{\bx_{\ell}} \sum_{g_{\ell} \in \sfO_{\tau_{\ell}}} \wG^{\bx_{\ell}}_{g_{\ell}}\Big)
&= \sum_{\tau_{\ell}}  S_{\tau_{\geq \ell}} \cdot(G^{\tau_{\ell}} \cdot (I - G)^{1 - \tau_{\ell}})\nonumber\\
&= \sum_{\tau_{\ell}}  \sum_{\tau_{< \ell} : |\tau| \geq d+1} G^{|\tau_{<\ell}|}\cdot (I-G)^{(\ell-1)-|\tau_{<\ell}|}
	\cdot (G^{\tau_{\ell}} \cdot (I - G)^{1 - \tau_{\ell}})\nonumber\\
&= \sum_{\tau_{\ell}}  \sum_{\tau_{< \ell} : |\tau| \geq d+1} G^{|\tau_{\leq \ell}|}\cdot (I-G)^{\ell-|\tau_{\leq\ell}|}\nonumber\\
&= \sum_{\tau_{\leq \ell} : |\tau| \geq d+1} G^{|\tau_{\leq \ell}|}\cdot (I-G)^{\ell-|\tau_{\leq\ell}|}\nonumber\\
&= S_{\tau_{> \ell}}.\label{eq:S-recurrence}
\end{align}
Finally, for any $\tau_{\geq \ell}$,
\begin{align}
&S_{\tau_{\geq \ell}}
	\cdot \Big(\E_{\bx_{\ell}} \sum_{g_{\ell} \in \sfO_{\tau_{\ell}}} \wG^{\bx_{\ell}}_{g_{\ell}}\Big)
	\cdot S_{\tau_{\geq \ell}}\nonumber\\
 =~& S_{\tau_{\geq \ell}}
	\cdot (G^{\tau_{\ell}} \cdot (I - G)^{1 - \tau_{\ell}})
	\cdot S_{\tau_{\geq \ell}} \tag{by \Cref{eq:explicit-formula-for-G-expectation}}\\
 =~& \sqrt{G^{\tau_{\ell}} \cdot (I - G)^{1 - \tau_{\ell}}} \cdot (S_{\tau_{\geq \ell}})^2
	\cdot \sqrt{G^{\tau_{\ell}} \cdot (I - G)^{1 - \tau_{\ell}}}\tag{because $S_{\tau_{\geq \ell}}$ commutes with $G$ and $(I-G)$}\\
\leq~& \sqrt{G^{\tau_{\ell}} \cdot (I - G)^{1 - \tau_{\ell}}} \cdot I
	\cdot \sqrt{G^{\tau_{\ell}} \cdot (I - G)^{1 - \tau_{\ell}}} \tag{by \Cref{eq:S-bound}}\\
=~& G^{\tau_{\ell}} \cdot (I - G)^{1 - \tau_{\ell}}\nonumber\\
=~& \Big(\E_{\bx_{\ell}} \sum_{g_{\ell} \in \sfO_{\tau_{\ell}}} \wG^{\bx_{\ell}}_{g_{\ell}}\Big), \label{eq:S-sandwich}
\end{align}
where the last step uses \Cref{eq:explicit-formula-for-G-expectation} again.
This concludes the set of facts we will need about $S_{\tau_{\geq \ell}}$.

Now we prove \Cref{eq:i-think-this-is-what-i'm-supposed-to-prove-2}.
To start, we write $\wH$ as a sandwich of $\wG$ operators, and move the
rightmost $\wG^{\bx_{\ell}}_{g_{\ell}}$ to the second tensor factor.
\begin{align}
&\E_{\bx_{\geq \ell}} \sum_{\tau_{\geq \ell}}  \sum_{(g_{\ell}, \ldots,g_k) \in \sfO_{\tau_{\geq \ell}}}
	\bra{\psi} \widehat{H}^{\bx_{\ell}, \dots, \bx_{k}}_{g_{\ell}, \dots, g_{k}} \ot S_{\tau_{\geq \ell}}\ket{\psi}\nonumber\\
=~&\E_{\bx_{\geq \ell}} \sum_{\tau_{\geq \ell}}  \sum_{(g_{\ell}, \ldots,g_k) \in \sfO_{\tau_{\geq \ell}}}
	\bra{\psi} (\wG^{\bx_{\geq \ell}}_{g_{\geq \ell}} \cdot (\wG^{\bx_{\geq \ell}}_{g_{\geq \ell}})^\dagger) \ot S_{\tau_{\geq \ell}}\ket{\psi}\tag{by \Cref{eq:split-H-into-two-Gs}}\\
=~&\E_{\bx_{\geq \ell}} \sum_{\tau_{\geq \ell}}  \sum_{(g_{\ell}, \ldots,g_k) \in \sfO_{\tau_{\geq \ell}}}
	\bra{\psi} (\wG^{\bx_{\geq \ell}}_{g_{\geq \ell}} \cdot (\wG^{\bx_{> \ell}}_{g_{> \ell}})^\dagger \cdot \wG^{\bx_{\ell}}_{g_{\ell}}) \ot S_{\tau_{\geq \ell}}\ket{\psi}\tag{by \Cref{eq:G-recurrence}}\\
\approx_{\sqrt{2\zeta}}~&\E_{\bx_{\geq \ell}} \sum_{\tau_{\geq \ell}}  \sum_{(g_{\ell}, \ldots,g_k) \in \sfO_{\tau_{\geq \ell}}}
	\bra{\psi} (\wG^{\bx_{\geq \ell}}_{g_{\geq \ell}} \cdot (\wG^{\bx_{> \ell}}_{g_{> \ell}})^\dagger) \ot (S_{\tau_{\geq \ell}} \cdot \wG^{\bx_{\ell}}_{g_{\ell}})\ket{\psi}.\label{eq:move-g-over-there}
\end{align}
To justify the approximation, we bound the error.
\begin{multline}\label{eq:call-again-later-part-tres}
\Big|\E_{\bx_{\geq \ell}} \sum_{\tau_{\geq \ell}}  \sum_{(g_{\ell}, g_{>\ell}) \in \mathsf{O}_{\tau_{\geq \ell}}}
	\bra{\psi} (\wG^{\bx_{\geq \ell}}_{g_{\geq \ell}}
		\ot S_{\tau_{\geq \ell}}) \cdot (((\wG^{\bx_{> \ell}}_{g_{> \ell}})^\dagger \ot I) \cdot (\wG^{\bx_{\ell}}_{g_{\ell}} \ot I - I \ot \wG^{\bx_{\ell}}_{g_{\ell}} ))\ket{\psi}\Big|\\
\leq \sqrt{\E_{\bx_{\geq \ell}} \sum_{\tau_{\geq \ell}}  \sum_{(g_{\ell}, g_{>\ell}) \in \mathsf{O}_{\tau_{\geq \ell}}}
	\bra{\psi} (\wG^{\bx_{\geq \ell}}_{g_{\geq \ell}} \cdot (\wG^{\bx_{\geq \ell}}_{g_{\geq \ell}})^\dagger)
		\ot (S_{\tau_{\geq \ell}})^2 \ket{\psi}}\\
\cdot \sqrt{\E_{\bx_{\geq \ell}} \sum_{\tau_{\geq \ell}}  \sum_{
	(g_{\ell}, g_{>\ell}) \in \mathsf{O}_{\tau_{\geq \ell}}}
		\bra{\psi} ((\wG^{\bx_{\ell}}_{g_{\ell}} \ot I - I \ot \wG^{\bx_{\ell}}_{g_{\ell}} )
	\cdot (\wG^{\bx_{> \ell}}_{g_{> \ell}} \cdot(\wG^{\bx_{> \ell}}_{g_{> \ell}})^\dagger\ot I)
	\cdot (\wG^{\bx_{\ell}}_{g_{\ell}} \ot I - I \ot \wG^{\bx_{\ell}}_{g_{\ell}} ))\ket{\psi}}.
\end{multline}
The expression inside the first square root is equal to 
\begin{equation*}
\E_{\bx_{\geq \ell}} \sum_{\tau_{\geq \ell}}  \sum_{(g_{\ell}, g_{>\ell}) \in \mathsf{O}_{\tau_{\geq \ell}}}
	\bra{\psi} \wH^{\bx_{\geq \ell}}_{g_{\geq \ell}}
		\ot (S_{\tau_{\geq \ell}})^2 \ket{\psi},
\end{equation*}
which is at most~$1$ because $S_{\tau_{\geq \ell}} \leq I$ and $\widehat{H}$ is a sub-measurement.
The expression inside the second square root is equal to
\begin{align*}
&\E_{\bx_{\geq \ell}} \sum_{\tau_{\geq \ell}}  \sum_{
	(g_{\ell}, g_{>\ell}) \in \mathsf{O}_{\tau_{\geq \ell}}}
		\bra{\psi} ((\wG^{\bx_{\ell}}_{g_{\ell}} \ot I - I \ot \wG^{\bx_{\ell}}_{g_{\ell}} )
	\cdot (\wH^{\bx_{> \ell}}_{g_{> \ell}} \ot I)
	\cdot (\wG^{\bx_{\ell}}_{g_{\ell}} \ot I - I \ot \wG^{\bx_{\ell}}_{g_{\ell}} ))\ket{\psi}\\
\leq~&\E_{\bx_{\geq \ell}} \sum_{g_{\ell}, \ldots, g_{k}}
		\bra{\psi} ((\wG^{\bx_{\ell}}_{g_{\ell}} \ot I - I \ot \wG^{\bx_{\ell}}_{g_{\ell}} )
	\cdot (\wH^{\bx_{> \ell}}_{g_{> \ell}} \ot I)
	\cdot (\wG^{\bx_{\ell}}_{g_{\ell}} \ot I - I \ot \wG^{\bx_{\ell}}_{g_{\ell}} ))\ket{\psi}\\
\leq~& \E_{\bx_{\geq \ell}} \sum_{g_{\ell}}
		\bra{\psi} (\wG^{\bx_{\ell}}_{g_{\ell}} \ot I - I \ot \wG^{\bx_{\ell}}_{g_{\ell}} )^2\ket{\psi}
			\tag{because $\widehat{H}$ is a sub-measurement}\\
\leq~& 2\zeta. \tag{by \Cref{cor:G-hat-facts}}
\end{align*}
We will now commute the leftmost $\wG^{x_\ell}_{g_\ell}$ in
\Cref{eq:move-g-over-there} to the right in two stages. In the first stage, we have:
\begin{align}
\eqref{eq:move-g-over-there}
&= \E_{\bx_{\geq \ell}} \sum_{\tau_{\geq \ell}}  \sum_{(g_{\ell}, g_{>\ell}) \in \sfO_{\tau_{\geq \ell}}}
	\bra{\psi} (\wG^{\bx_{\ell}}_{g_{\ell}} \cdot \wG^{\bx_{> \ell}}_{g_{> \ell}} \cdot (\wG^{\bx_{> \ell}}_{g_{> \ell}})^\dagger) \ot (S_{\tau_{\geq \ell}} \cdot \wG^{\bx_{\ell}}_{g_{\ell}})\ket{\psi}\nonumber\\
& \approx_{\sqrt{\nu_4}} \E_{\bx_{\geq \ell}} \sum_{\tau_{\geq \ell}}  \sum_{(g_{\ell}, g_{>\ell}) \in \sfO_{\tau_{\geq \ell}}}
	\bra{\psi} (\wG^{\bx_{> \ell}}_{g_{> \ell}} \cdot \wG^{\bx_{\ell}}_{g_{\ell}} \cdot (\wG^{\bx_{> \ell}}_{g_{> \ell}})^\dagger) \ot (S_{\tau_{\geq \ell}} \cdot \wG^{\bx_{\ell}}_{g_{\ell}})\ket{\psi}.
\label{eq:commute-g-part-one}
\end{align}
To justify the approximation, we bound the magnitude of the difference.
\begin{multline}\label{eq:call-this-later}
\Big|  \E_{\bx_{\geq \ell}} \sum_{\tau_{\geq \ell}}
  \sum_{(g_{\ell}, g_{>\ell}) \in \sfO_{\tau_{\geq \ell}}}
  \bra{\psi}     ([ \wG^{\bx_{>\ell}}_{g_{>\ell}},
    \wG^{\bx_\ell}_{g_\ell} ] \otimes I) \cdot     ((\wG^{\bx_{>\ell}}_{g_{>\ell}})^\dagger \ot
    (S_{\tau_{\geq \ell}} \cdot
  \wG^{\bx_\ell}_{g_\ell})) \ket{\psi} \Big|  \\
  \leq \sqrt{ \E_{\bx_{\geq \ell}} \sum_{\tau_{\geq \ell}}
  \sum_{(g_{\ell}, g_{>\ell}) \in \sfO_{\tau_{\geq \ell}}} \bra{\psi} ([ \wG^{\bx_{>\ell}}_{g_{>\ell}},
    \wG^{\bx_\ell}_{g_\ell} ]) \cdot ([ \wG^{\bx_{>\ell}}_{g_{>\ell}},
    \wG^{\bx_\ell}_{g_\ell} ])^\dagger \ot I \ket{\psi}}  \\
  \quad \cdot \sqrt{\E_{\bx_{\geq \ell}} \sum_{\tau_{\geq \ell}}
   \sum_{(g_{\ell}, g_{>\ell}) \in \sfO_{\tau_{\geq \ell}}} \bra{\psi}
    (\wG^{\bx_{>\ell}}_{g_{>\ell}}\cdot (
    \wG^{\bx_{>\ell}}_{g_{>\ell}})^\dagger) \ot
    (\wG^{\bx_{\ell}}_{g_\ell}  \cdot (S_{\tau_{\geq \ell}})^2
    \cdot \wG^{\bx_{\ell}}_{g_\ell}) \ket{\psi}}.
\end{multline}
The quantity inside the first square root is at most
\begin{equation*}
\E_{\bx_{\geq \ell}} \sum_{g_{\ell}, \ldots, g_k} \bra{\psi} ([ \wG^{\bx_{>\ell}}_{g_{>\ell}},
    \wG^{\bx_\ell}_{g_\ell} ]) \cdot ([ \wG^{\bx_{>\ell}}_{g_{>\ell}},
    \wG^{\bx_\ell}_{g_\ell} ])^\dagger \ot I \ket{\psi},
\end{equation*}
which is at most~$\nu_4$ by~\Cref{lem:commute-g-half-sandwich}.
The quantity inside the second square root is equal to
\begin{align*}
&\E_{\bx_{\geq \ell}} \sum_{\tau_{\geq \ell}}
   \sum_{(g_{\ell}, g_{>\ell}) \in \sfO_{\tau_{\geq \ell}}} \bra{\psi}
    \wH^{\bx_{>\ell}}_{g_{>\ell}} \ot
    (\wG^{\bx_{\ell}}_{g_\ell}  \cdot (S_{\tau_{\geq \ell}})^2
    \cdot \wG^{\bx_{\ell}}_{g_\ell}) \ket{\psi}\\
\leq~
&\E_{\bx_{\geq \ell}} \sum_{\tau_{\geq \ell}}
   \sum_{(g_{\ell}, g_{>\ell}) \in \sfO_{\tau_{\geq \ell}}} \bra{\psi}
    \wH^{\bx_{>\ell}}_{g_{>\ell}} \ot
    (\wG^{\bx_{\ell}}_{g_\ell}  \cdot I
    \cdot \wG^{\bx_{\ell}}_{g_\ell}) \ket{\psi} \tag{because $S_{\tau_{\geq \ell}} \leq I$}\\
=~
&\E_{\bx_{\geq \ell}} \sum_{\tau_{\geq \ell}}
   \sum_{(g_{\ell}, g_{>\ell}) \in \sfO_{\tau_{\geq \ell}}} \bra{\psi}
    \wH^{\bx_{>\ell}}_{g_{>\ell}} \ot
    \wG^{\bx_{\ell}}_{g_\ell} \ket{\psi} \\
\leq~&  1. \tag{because $\wG$ and $\wH$ are sub-measurements}
\end{align*}
We continue commuting the leftmost $\wG^{x_{\ell}}_{g_\ell}$ to the
right. 
\begin{align}
  \eqref{eq:commute-g-part-one} =&
  \E_{\bx_{\geq \ell}} \sum_{\tau_{\geq \ell}}  \sum_{(g_{\ell}, g_{>\ell}) \in \sfO_{\tau_{\geq \ell}}}
	\bra{\psi} (\wG^{\bx_{> \ell}}_{g_{> \ell}} \cdot \wG^{\bx_{\ell}}_{g_{\ell}} \cdot (\wG^{\bx_{> \ell}}_{g_{> \ell}})^\dagger) \ot (S_{\tau_{\geq \ell}} \cdot \wG^{\bx_{\ell}}_{g_{\ell}})\ket{\psi} \nonumber\\
  \approx_{\sqrt{\nu_4}}~& 
   \E_{\bx_{\geq \ell}} \sum_{\tau_{\geq \ell}}  \sum_{(g_{\ell}, g_{>\ell}) \in \sfO_{\tau_{\geq \ell}}}
	\bra{\psi} (\wG^{\bx_{> \ell}}_{g_{> \ell}} \cdot (\wG^{\bx_{> \ell}}_{g_{> \ell}})^\dagger \cdot \wG^{\bx_{\ell}}_{g_{\ell}}) \ot (S_{\tau_{\geq \ell}} \cdot \wG^{\bx_{\ell}}_{g_{\ell}})\ket{\psi}. \label{eq:commute-g-part-two}
\end{align}
To justify the approximation, we will need to be slightly more
clever this time. First, as always, we bound the magnitude of the
difference:
\begin{multline}\label{eq:call-again-later-part-dos}
  \Big|  \E_{\bx_{\geq \ell}} \sum_{\tau_{\geq \ell}}  \sum_{(g_{\ell}, g_{>\ell}) \in \sfO_{\tau_{\geq \ell}}}
  \bra{\psi} (\wG^{\bx_{>\ell}}_{g_{>\ell}}         \ot
    (S_{\tau_{\geq \ell}} \cdot
  \wG^{\bx_\ell}_{g_\ell})) \cdot ([ (\wG^{\bx_{>\ell}}_{g_{>\ell}})^\dagger,
    \wG^{\bx_\ell}_{g_\ell} ] \ot I) \ket{\psi} \Big|\\
  \leq   \sqrt{\E_{\bx_{\geq \ell}} \sum_{\tau_{\geq \ell}}  \sum_{(g_{\ell}, g_{>\ell}) \in \sfO_{\tau_{\geq \ell}}}
    \bra{\psi}
    (\wG^{\bx_{>\ell}}_{g_{>\ell}})\cdot (
    \wG^{\bx_{>\ell}}_{g_{>\ell}})^\dagger \ot (S_{\tau_{\geq \ell}} \cdot
    \wG^{\bx_{\ell}}_{g_\ell}  \cdot S_{\tau_{\geq \ell}})
    \ket{\psi}} \\
  \cdot \sqrt{ \E_{\bx_{\geq \ell}}
  \sum_{\tau_{\geq \ell}}  \sum_{(g_{\ell}, g_{>\ell}) \in \sfO_{\tau_{\geq \ell}}}
   \bra{\psi} ([ (\wG^{\bx_{>\ell}}_{g_{>\ell}})^\dagger,
    \wG^{\bx_\ell}_{g_\ell} ])^\dagger \cdot ([ (\wG^{\bx_{>\ell}}_{g_{>\ell}})^\dagger,
    \wG^{\bx_\ell}_{g_\ell} ]) \ot I \ket{\psi}}.
\end{multline}
The term inside the first square root is equal to
\begin{align*}
&\E_{\bx_{\geq \ell}} \sum_{\tau_{\geq \ell}}  \sum_{(g_{\ell}, g_{>\ell}) \in \sfO_{\tau_{\geq \ell}}}
    \bra{\psi}
    \wH^{\bx_{>\ell}}_{g_{>\ell}}
     \ot (S_{\tau_{\geq \ell}} \cdot
    \wG^{\bx_{\ell}}_{g_\ell}  \cdot S_{\tau_{\geq \ell}})
    \ket{\psi}\\
=~&\E_{\bx_{> \ell}} \sum_{\tau_{\geq \ell}}  \sum_{g_{>\ell} \in \sfO_{\tau_{> \ell}}}
    \bra{\psi}
    \wH^{\bx_{>\ell}}_{g_{>\ell}}
     \ot (S_{\tau_{\geq \ell}} \cdot
    \Big(\E_{\bx_{\ell}} \sum_{g_{\ell} \in \sfO_{\tau_{\ell}}}\wG^{\bx_{\ell}}_{g_\ell}  \Big)\cdot S_{\tau_{\geq \ell}})
    \ket{\psi}\\
\leq~&\E_{\bx_{> \ell}} \sum_{\tau_{\geq \ell}}  \sum_{g_{>\ell} \in \sfO_{\tau_{> \ell}}}
    \bra{\psi}
    \wH^{\bx_{>\ell}}_{g_{>\ell}}
     \ot 
    \Big(\E_{\bx_{\ell}} \sum_{g_{\ell} \in \sfO_{\tau_{\ell}}}\wG^{\bx_{\ell}}_{g_\ell}  \Big)
    \ket{\psi} \tag{by~\Cref{eq:S-sandwich}}\\
=~&\E_{\bx_{\geq \ell}} \sum_{\tau_{\geq \ell}}  \sum_{(g_{\ell}, g_{>\ell}) \in \sfO_{\tau_{\geq \ell}}}
    \bra{\psi}
    \wH^{\bx_{>\ell}}_{g_{>\ell}}
     \ot
    \wG^{\bx_{\ell}}_{g_\ell}  
    \ket{\psi}\\
\leq~& 1. \tag{because $\wG$ and $\wH$ are sub-measurements}
\end{align*}
The term inside the second square root is equal to the term inside the first square root of~\Cref{eq:call-this-later}, which we bounded by~$\nu_4$.
We are now ready for the final step, which is to bring the leftmost
$\wG^{x_\ell}_{g_\ell}$ over to the second tensor factor.
\begin{align}
  \eqref{eq:commute-g-part-two} &= 
  \E_{\bx_{\geq \ell}} \sum_{\tau_{\geq \ell}}  \sum_{(g_{\ell}, g_{>\ell}) \in \sfO_{\tau_{\geq \ell}}}
	\bra{\psi} (\wG^{\bx_{> \ell}}_{g_{> \ell}} \cdot (\wG^{\bx_{> \ell}}_{g_{> \ell}})^\dagger \cdot \wG^{\bx_{\ell}}_{g_{\ell}}) \ot (S_{\tau_{\geq \ell}} \cdot \wG^{\bx_{\ell}}_{g_{\ell}})\ket{\psi}\nonumber\\
  &\approx_{\sqrt{2\zeta}} \E_{\bx_{\geq \ell}} \sum_{\tau_{\geq \ell}}  \sum_{(g_{\ell}, g_{>\ell}) \in \sfO_{\tau_{\geq \ell}}}
	\bra{\psi} (\wG^{\bx_{> \ell}}_{g_{> \ell}} \cdot (\wG^{\bx_{> \ell}}_{g_{> \ell}})^\dagger) \ot (S_{\tau_{\geq \ell}} \cdot \wG^{\bx_{\ell}}_{g_{\ell}}\cdot \wG^{\bx_{\ell}}_{g_{\ell}})\ket{\psi}. \label{eq:h-ot-mgg}
\end{align}
To justify the approximation, we bound the magnitude of the
difference.
\begin{multline*}
  \Big|\E_{\bx_{\geq \ell}} \sum_{\tau_{\geq \ell}}
   \sum_{(g_{\ell}, g_{>\ell}) \in \sfO_{\tau_{\geq \ell}}}
   	\bra{\psi} (\wG^{\bx_{> \ell}}_{g_{> \ell}}
		\ot (S_{\tau_{\geq \ell}}
		\cdot \wG^{\bx_{\ell}}_{g_{\ell}}))
		\cdot
		(((\wG^{\bx_{> \ell}}_{g_{> \ell}})^\dagger \ot I)
		\cdot (\wG^{\bx_{\ell}}_{g_{\ell}} \ot I - I \ot \wG^{\bx_{\ell}}_{g_{\ell}}))\ket{\psi}
  \Big|  \\
  \leq \sqrt{\E_{\bx_{\geq \ell}}\sum_{\tau_{\geq \ell}}
   \sum_{(g_{\ell}, g_{>\ell}) \in \sfO_{\tau_{\geq \ell}}}
   	\bra{\psi} (\wG^{\bx_{> \ell}}_{g_{> \ell}} \cdot (\wG^{\bx_{> \ell}}_{g_{> \ell}})^\dagger)
		\ot (S_{\tau_{\geq \ell}}
		\cdot \wG^{\bx_{\ell}}_{g_{\ell}} \cdot S_{\tau_{\geq \ell}}) \ket{\psi}} \nonumber \\
  \cdot \sqrt{\E_{\bx_{\geq \ell}} \sum_{\tau_{\geq \ell}}
   \sum_{(g_{\ell}, g_{>\ell}) \in \sfO_{\tau_{\geq \ell}}} \bra{\psi}
   ( (\wG^{\bx_\ell}_{g_\ell} \ot I - I \ot \wG^{\bx_\ell}_{g_\ell})\cdot
   (\wG^{\bx_{> \ell}}_{g_{> \ell}} \cdot (\wG^{\bx_{> \ell}}_{g_{> \ell}})^\dagger \ot I)
    \cdot (\wG^{\bx_\ell}_{g_\ell} \ot I - I \ot \wG^{\bx_\ell}_{g_\ell})).
    \ket{\psi}}
\end{multline*}
The expression inside the first is equal to the expression inside the first square root in~\Cref{eq:call-again-later-part-dos},
which we bounded by~$1$.
The expression inside the second square root is equal to the expression inside the second square root in~\Cref{eq:call-again-later-part-tres},
which we bounded by~$2\zeta$.
We end by noting that
\begin{align*}
\eqref{eq:h-ot-mgg}
& =
\E_{\bx_{\geq \ell}} \sum_{\tau_{\geq \ell}}  \sum_{(g_{\ell}, g_{>\ell}) \in \sfO_{\tau_{\geq \ell}}}
	\bra{\psi} (\wG^{\bx_{> \ell}}_{g_{> \ell}} \cdot (\wG^{\bx_{> \ell}}_{g_{> \ell}})^\dagger) \ot (S_{\tau_{\geq \ell}} \cdot \wG^{\bx_{\ell}}_{g_{\ell}})\ket{\psi} \tag{because $\wG$ is projective}\\
& =
\E_{\bx_{\geq \ell}} \sum_{\tau_{\geq \ell}}  \sum_{(g_{\ell}, g_{>\ell}) \in \sfO_{\tau_{\geq \ell}}}
	\bra{\psi} \wH^{\bx_{> \ell}}_{g_{> \ell}} \ot (S_{\tau_{\geq \ell}} \cdot \wG^{\bx_{\ell}}_{g_{\ell}})\ket{\psi}\\
& =\E_{\bx_{> \ell}} \sum_{\tau_{> \ell}}  \sum_{g_{>\ell} \in \sfO_{\tau_{> \ell}}}
	\bra{\psi} \wH^{\bx_{> \ell}}_{g_{> \ell}} \ot \Big(\sum_{\tau_{\ell}}S_{\tau_{\geq \ell}} \cdot
		\Big(\E_{\bx_{\ell}} \sum_{g_\ell \in \sfO_{\tau_{\ell}}}\wG^{\bx_{\ell}}_{g_{\ell}}\Big)\Big)\ket{\psi}\\
& =\E_{\bx_{> \ell}} \sum_{\tau_{> \ell}}  \sum_{g_{>\ell} \in \sfO_{\tau_{> \ell}}}
	\bra{\psi} \wH^{\bx_{> \ell}}_{g_{> \ell}} \ot S_{\tau_{>\ell}}\ket{\psi}. \tag{by \Cref{eq:S-recurrence}}
\end{align*}
This concludes the proof of
\Cref{eq:i-think-this-is-what-i'm-supposed-to-prove-2}
and therefore proves the lemma.
\end{proof}

\begin{lemma}
  \label{lem:chernoff-bernoulli-matrix}
  Let $0 < \theta < 1$ and let $k, d > 0$ be integers such that $k \geq 2d/\theta$.
  Define the matrix-valued function $F$ by
  \[ F(X) = \sum_{r = d+1}^{k} \binom{k}{r} X^r (I - X)^{r - k}. \]
  Then for any Hermitian matrix $X$ such that $0 \leq X \leq I$ and $\bra{\psi} X \ot I
  \ket{\psi} \geq 1 - \kappa$, it holds that
  \[ \bra{\psi} F(X) \ot I \ket{\psi} \geq 1 - \frac{\kappa}{1 - \theta}
    - e^{-\theta^2 k/2}. \]
  \end{lemma}
  \begin{proof}
    Let $\rho$ be the reduced state of $\ket{\psi}$ on one prover's
    subsystem (since the state is assumed to be symmetric, it does not
    matter which prover we take).
    Write the eigendecomposition $X = \sum_{i} \lambda_i
    \ket{v_i}\bra{v_i}$ of $X$, where we allow some of the
    eigenvalues to be $0$ so that the set of eigenvectors
    $\{\ket{v_i}\}$ forms an orthonormal basis of the space. This defines a probability
    distribution $\mu$ over eigenvectors, where eigenvector $i$
    occurs with probability $\mu(i) = \bra{v_i} \rho \ket{v_i}$. The
    given condition $\bra{\psi} X \ot I \ket{\psi} = \Tr(X \rho) \geq
    1 -\kappa$ implies that
    \begin{align*}
      \E_{\bi \sim \mu} \lambda_{\bi}
      = \sum_i \lambda_i \cdot \mu(i)
       &= \sum_i \lambda_i \cdot \bra{v_i} \rho \ket{v_i}\\
       &=\sum_i \lambda_i \cdot \tr(\ket{v_i} \bra{v_i}  \cdot \rho)\\
       &=\tr \Big(\sum_i \lambda_i  \ket{v_i} \bra{v_i} \cdot\rho\Big)
       = \tr(X \rho)
       \geq 1 - \kappa.
     \end{align*}
     Or, equivalently,
     \begin{equation*}
      \E_{\bi \sim \mu} ( 1 -\lambda_{\bi}) \leq \kappa.
    \end{equation*}
      By Markov's inequality, for any $0 < \theta < 1$, we have
     \begin{equation*}
     \Pr_{\bi \sim \mu}[(1-\lambda_{\bi}) \geq (1 - \theta)]
     \leq \frac{\E_{\bi \sim \mu} (1-\lambda_{\bi})}{1-\theta}
     \leq \frac{\kappa}{1-\theta}.
     \end{equation*} 
     In other words,
     \begin{equation}\label{eq:in-other-words}
     \Pr_{\bi \sim \mu}[\lambda_{\bi} > \theta]
     = \Pr_{\bi \sim \mu}[(1 - \lambda_{\bi}) < (1 - \theta)]
     = 1 - \Pr_{\bi \sim \mu}[(1 - \lambda_{\bi}) \geq (1 - \theta)]
     \geq 1 - \frac{\kappa}{1-\theta}.  
     \end{equation}
      Thus, it  holds that with probability at least $1 - \frac{\kappa}{1 - \theta}$ over $\bi \sim
    \mu$, $\lambda_{\bi} > \theta$.

    We now evaluate $\bra{\psi} F(X) \ot I \ket{\psi}$. We will
    essentially do this eigenvalue by eigenvalue.
    To begin, we consider a hypothetical eigenvalue $d/k \leq p \leq 1$.
    Observe that $F(p)$ is precisely the
    probability 
    probability of observing at least $d+1$ successes out of $k$
    i.i.d.\ Bernoulli trials, each of which succeeds with probability
    $p$.
    In other words, it is the probability that $\bY := \bY_1 + \cdots + \bY_k \geq d+1$,
    where $\bY_1, \ldots, \bY_k \sim \mathrm{Bernoulli}(p)$.
    We can bound this probability by the additive Chernoff bound (see the second additive bound in~\cite{Blu11}):
    \begin{align*}
    \Pr[\bY \leq d]
    &= \Pr[\bY \leq p k - (pk - d)] \\
    &=\Pr\Big[\bY \leq pk - \Big(p - \frac{d}{k}\Big)\cdot k\Big]\\
    	&\leq  \exp\Big(-2 \Big(p - \frac{d}{k}\Big)^2 \cdot k\Big).
    \end{align*}
    Thus,
    \begin{equation}\label{eq:by-chernoff}
    F(p) = \Pr[\bY \geq d+1] = 1 - \Pr[\bY \leq d] \geq 1 - \exp\Big(-2 \Big(p - \frac{d}{k}\Big)^2 \cdot k\Big).
    \end{equation}
        Putting the pieces together, we compute $\bra{\psi} f(X) \ot I
    \ket{\psi}$:
    \begin{align}
      \bra{\psi} F(X) \ot I \ket{\psi} &= \Tr(F(X) \rho) \nonumber\\
                                       &= \sum_{i} F(\lambda_i)
                                         \bra{v_i} \rho \ket{v_i} \nonumber\\
                                       &=\E_{\bi \sim \mu}
                                         F(\lambda_{\bi}) \nonumber\\
                                       &\geq \Pr_{\bi \sim
                                         \mu}[\lambda_{\bi} \geq \theta]
                                         \cdot F(\theta) \nonumber\\
                                       &\geq \Big( 1 - \frac{\kappa}{1 - \theta} \Big) \cdot F(\theta)  \tag{by \Cref{eq:in-other-words}}\\
      &\geq \Big( 1 - \frac{\kappa}{1 - \theta} \Big) \cdot \Big(1 - \exp \Big(-2\Big(\theta -
        \frac{d}{k}\Big)^2 \cdot k \Big) \Big) \label{eq:almost-done-with-this-giant-proof},
        \end{align}
        where the last step uses \Cref{eq:by-chernoff}.
    Next, we claim that if $a, b, c \geq 0$ satisfy $a \geq (1-b) \cdot (1-c)$,
    then $a \geq 1 - b - c$. This is because if either $b$ or $c$ is at least~$1$,
    then the conclusion is trivially true, and if both are less than~$1$,
    then
    \begin{equation*}
    (1-b) \cdot (1-c) = 1\cdot (1-c) -  b \cdot (1-c) \geq 1 \cdot (1-c) - b \cdot 1 = 1-c-b.
    \end{equation*}
    Since $\bra{\psi} F(X) \ot I \ket{\psi}$ is manifestly positive,
    we can apply this to \Cref{eq:almost-done-with-this-giant-proof}, yielding
    \begin{equation*}
    \bra{\psi} F(X) \ot I \ket{\psi}
    \geq  1 - \frac{\kappa}{1 - \theta} - \exp \Big(-2\Big(\theta -
        \frac{d}{k}\Big)^2 \cdot k \Big).
    \end{equation*}
    Finally, we note that because $k \geq 2d/\theta$,
    we have $\theta/2 \geq d/k$.
    This implies that
    $\theta-d/k \geq \theta - \theta/2 = \theta/2$,
    and so    
    $(\theta - d/k)^2 \geq (\theta/2)^2 = \theta^2/4$.
    As a result,
    \begin{equation*}
    \exp \Big(2\Big(\theta -
        \frac{d}{k}\Big)^2 \cdot k \Big)
        \geq 
        \exp \Big(\theta^2 k/2 \Big).
    \end{equation*}
    Equivalently,
    \begin{equation*}
    \exp \Big(-2\Big(\theta -
        \frac{d}{k}\Big)^2 \cdot k \Big)
        \leq 
        \exp \Big(-\theta^2 k/2 \Big).
    \end{equation*}
    Thus, we conclude
    \begin{equation*}
     \bra{\psi} F(X) \ot I \ket{\psi}
    \geq  1 - \frac{\kappa}{1 - \theta} - \exp \Big(-\theta^2 k/2 \Big).\qedhere
    \end{equation*}
  \end{proof}
\begin{corollary}[Completeness of~$H$; Proof of \Cref{item:ld-pasting-N-completeness-sub-measurement} of \Cref{lem:ld-pasting-sub-measurement}]
  \label{cor:ld-pasting-N-completeness}
  Let $k \geq 400md$. Then
  \[ \bra{\psi} H \otimes I \ket{\psi}
\geq 1 - \kappa \cdot \left(1 + \frac{1}{100m}\right)
    - \nu - e^{- k/(80000m^2)}. \]
\end{corollary}
\begin{proof}
We begin by approximating the completeness as follows.
\begin{align*}
\bra{\psi} H \otimes I \ket{\psi}
& \approx_{\nu_7} \E_{\bx_1, \ldots, \bx_k} \sum_{\tau:|\tau| \geq d+1} \sum_{(g_1, \ldots, g_k) \in \mathsf{Outcomes}_{\tau}} \bra{\psi} \widehat{H}^{\bx_1, \ldots, \bx_k}_{g_1, \ldots, g_k} \otimes I \ket{\psi} \tag{by \Cref{lem:over-all-outcomes}}\\
& \approx_{\nu_8} \sum_{i = d+1}^k \binom{k}{i} \bra{\psi} G^i (I-G)^{k-i} \otimes I \ket{\psi} \tag{by \Cref{lem:from-H-to-G}}\\
& = \bra{\psi} F(G) \ot I \ket{\psi}.
\end{align*}
We can bound the error incurred here by
\begin{align*}
\nu_7 + \nu_8
&= 46 k^2m \cdot \left(\eps^{1/32} + \delta^{1/32} +  \gamma^{1/32} + \zeta^{1/32} +(d/q)^{1/32}\right)
	+ 46 k m \cdot \left(\gamma^{1/32} + \zeta^{1/32} + (d/q)^{1/32}\right)\\
&\leq 100 k^2m \cdot \left(\eps^{1/32} + \delta^{1/32} +  \gamma^{1/32} + \zeta^{1/32} +(d/q)^{1/32}\right)\\
&= \nu.
\end{align*}
Let $\theta = 1/(200 m)$. Note that
\begin{equation*}
\frac{1}{1-\theta}
= \frac{1}{1- 1/(200m)}
= \frac{200m}{200m-1}
= 1 + \frac{1}{200m-1}
\leq 1 + \frac{1}{100m}.
\end{equation*}
Then $k \geq 2d/\theta$ and, therefore, $k \geq d+1$.
As a result, $k$ and~$\theta$ satisfy the hypothesis of \Cref{lem:chernoff-bernoulli-matrix},
namely that $k \geq \max\{d+1, 2d/\theta\}$.
Thus, \Cref{lem:chernoff-bernoulli-matrix} implies that
\begin{align*}
\bra{\psi} F(G) \ot I \ket{\psi}
&\geq 1 - \frac{\kappa}{1 - \theta}
    - e^{-\theta^2 k/2}\\
&= 1 - \frac{\kappa}{1 - \theta}
    - e^{- k/(80000m^2)}\\
&\geq 1 - \kappa \cdot \left(1 + \frac{1}{100m}\right)
    - e^{- k/(80000m^2)}.
\end{align*}
In total, we have
\begin{equation*}
\bra{\psi} H \otimes I \ket{\psi}
\geq 1 - \kappa \cdot \left(1 + \frac{1}{100m}\right)
    - \nu - e^{- k/(80000m^2)}.
\end{equation*}
This completes the proof.
\end{proof}

\bibliographystyle{alpha}

\begin{thebibliography}{ALM{\etalchar{+}}98}

\bibitem[AHO97]{AHO97}
Farid Alizadeh, Jean-Pierre Haeberly, and Michael Overton.
\newblock Complementarity and nondegeneracy in semidefinite programming.
\newblock {\em Mathematical programming}, 77(1):111--128, 1997.

\bibitem[ALM{\etalchar{+}}98]{ALM+98}
Sanjeev Arora, Carsten Lund, Rajeev Motwani, Madhu Sudan, and Mario Szegedy.
\newblock Proof verification and the hardness of approximation problems.
\newblock {\em Journal of the ACM}, 45(3):501--555, 1998.

\bibitem[Ara02]{aravind2002simple}
PK~Aravind.
\newblock A simple demonstration of {B}ell's theorem involving two observers
  and no probabilities or inequalities.
\newblock {\em arXiv preprint quant-ph/0206070}, 2002.

\bibitem[AS98]{AS98}
Sanjeev Arora and Shmuel Safra.
\newblock Probabilistic checking of proofs: a new characterization of {NP}.
\newblock {\em Journal of the ACM}, 45(1):70--122, 1998.

\bibitem[BFL91]{BFL91}
L{\'a}szl{\'o} Babai, Lance Fortnow, and Carsten Lund.
\newblock Non-deterministic exponential time has two-prover interactive
  protocols.
\newblock {\em Computational complexity}, 1(1):3--40, 1991.

\bibitem[BLR93]{BLR93}
Manuel Blum, Michael Luby, and Ronitt Rubinfeld.
\newblock Self-testing/correcting with applications to numerical problems.
\newblock {\em Journal of computer and system sciences}, 47(3):549--595, 1993.

\bibitem[Blu11]{Blu11}
Avrim Blum.
\newblock Lecture {5} from {15-859(M)}:\ {R}andomized {A}lgorithms.
\newblock Found at
  \url{http://www.cs.cmu.edu/~avrim/Randalgs11/lectures/lect0124.pdf}, 2011.

\bibitem[BV04]{cvxbook}
Stephen Boyd and Lieven Vandenberghe.
\newblock {\em Convex Optimization}.
\newblock 2004.

\bibitem[CMS17]{CMS17}
Alessandro Chiesa, Peter Manohar, and Igor Shinkar.
\newblock On axis-parallel tests for tensor product codes.
\newblock In {\em Proceedings of the 21st Annual International Workshop on
  Randomization and Computation}, 2017.

\bibitem[FHS94]{FHS94}
Katalin Friedl, Zsolt Hatsagi, and Alexander Shen.
\newblock Low-degree tests.
\newblock In {\em Proceedings of the 5rd Annual ACM-SIAM Symposium on Discrete
  Algorithms}, pages 57--64, 1994.

\bibitem[IV12]{IV12}
Tsuyoshi Ito and Thomas Vidick.
\newblock A multi-prover interactive proof for {NEXP} sound against entangled
  provers.
\newblock In {\em Proceedings of the 53rd Annual IEEE Symposium on Foundations
  of Computer Science}, pages 243--252, 2012.

\bibitem[JNV{\etalchar{+}}20]{JNV+20}
Zhengfeng Ji, Anand Natarajan, Thomas Vidick, John Wright, and Henry Yuen.
\newblock $\mathsf{MIP}^* = \mathsf{RE}$.
\newblock Technical report, arXiv:2001.04383, 2020.

\bibitem[KS08]{KS08}
Tali Kaufman and Madhu Sudan.
\newblock Algebraic property testing: the role of invariance.
\newblock In {\em Proceedings of the 40th Annual ACM Symposium on Theory of
  Computing}, pages 403--412, 2008.

\bibitem[KV11]{KV11}
Julia Kempe and Thomas Vidick.
\newblock Parallel repetition of entangled games.
\newblock In {\em Proceedings of the 43rd Annual ACM Symposium on Theory of
  Computing}, pages 353--362, 2011.

\bibitem[Mer90]{mermin1990simple}
David Mermin.
\newblock Simple unified form for the major no-hidden-variables theorems.
\newblock {\em Physical Review Letters}, 65(27):3373, 1990.

\bibitem[NV18a]{NV18a}
Anand Natarajan and Thomas Vidick.
\newblock Low-degree testing for quantum states, and a quantum entangled games
  {PCP}.
\newblock In {\em Proceedings of the 59th Annual IEEE Symposium on Foundations
  of Computer Science}, 2018.

\bibitem[NV18b]{NV18b}
Anand Natarajan and Thomas Vidick.
\newblock Two-player entangled games are {NP}-hard.
\newblock In {\em Proceedings of the 33rd Annual IEEE Conference on
  Computational Complexity}, 2018.

\bibitem[NW19]{NW19}
Anand Natarajan and John Wright.
\newblock $\mathsf{NEEXP} \subseteq \mathsf{MIP}^*$.
\newblock In {\em Proceedings of the 60th Annual IEEE Symposium on Foundations
  of Computer Science}, pages 510--518, 2019.

\bibitem[Per90]{peres1990incompatible}
Asher Peres.
\newblock Incompatible results of quantum measurements.
\newblock {\em Physics Letters A}, 151(3-4):107--108, 1990.

\bibitem[PS94]{PS94}
Alexander Polishchuk and Daniel Spielman.
\newblock Nearly-linear size holographic proofs.
\newblock In {\em Proceedings of the 26th Annual ACM Symposium on Theory of
  Computing}, pages 194--203, 1994.

\bibitem[RS97]{RS97}
Ran Raz and Shmuel Safra.
\newblock A sub-constant error-probability low-degree test, and a sub-constant
  error-probability {PCP} characterization of {NP}.
\newblock In {\em Proceedings of the 29th Annual ACM Symposium on Theory of
  Computing}, pages 475--484, 1997.

\bibitem[Sch80]{Sch80}
Jacob Schwartz.
\newblock Fast probabilistic algorithms for verification of polynomial
  identities.
\newblock {\em Journal of the ACM}, 27(4):701--717, 1980.

\bibitem[Sud11]{Sud11}
Madhu Sudan.
\newblock Guest column: testing linear properties: some general theme.
\newblock {\em ACM SIGACT News}, 42(1):59--80, 2011.

\bibitem[Vid11]{Vid11}
Thomas Vidick.
\newblock {\em The complexity of entangled games}.
\newblock PhD thesis, University of California, Berkeley, 2011.

\bibitem[Vid16]{Vid16}
Thomas Vidick.
\newblock Three-player entangled {XOR} games are {NP}-hard to approximate.
\newblock {\em SIAM Journal on Computing}, 45(3):1007--1063, 2016.

\bibitem[Zip79]{Zip79}
Richard Zippel.
\newblock Probabilistic algorithms for sparse polynomials.
\newblock In {\em Proceedings of the 2nd International Symposium on Symbolic
  and Algebraic Manipulation}, pages 216--226, 1979.

\end{thebibliography}
\newcommand{\etalchar}[1]{$^{#1}$}

\end{document}